\documentclass[12pt]{report}

\usepackage[algochapter,ruled,linesnumbered,vlined]{algorithm2e}
\usepackage{sectsty}

\chapternumberfont{\LARGE} 
\chaptertitlefont{\Huge}
\usepackage{amsmath,amssymb,amsthm,amsfonts,enumerate,eucal,hyperref,mdframed,tikz,parskip,subfig,url,xparse, xstring,xr,xr-hyper}
\pgfdeclarelayer{background}
\usetikzlibrary{backgrounds}
\usetikzlibrary{shapes.geometric}
\usepackage{graphicx}
\usepackage[nottoc]{tocbibind}
\usepackage{setspace}
\def\AA{\mathcal{A}}
\def\CC{\mathcal{C}}
\def\DD{\mathcal{D}}
\def\FF{\mathcal{F}}

\def\PP{\mathcal{P}}
\def\QQ{\mathcal{Q}}
\def\WW{\mathcal{W}}

\def\UIG{\mathrm{UIG}}
\def\APUD{\mathrm{APUD}}
\def\EX{\mathit{ex}}
\def\ES#1{{\EuScript#1}}

\def\lseg#1#2{\overline{#1#2}}

\def\upseg{\mathop{\mbox{\sl ups}}}
\def\uptan{\mathop{\mbox{\sl upt}}}
\def\lowseg{\mathop{\mbox{\sl los}}}
\def\lowtan{\mathop{\mbox{\sl lot}}}
\let\lotan\lowtan

\def\camera#1#2#3{
	\begin{scope}[shift={#1}, rotate=#2, scale=0.06]
		\coordinate (0) at (-4,-2);
		\coordinate (1) at (-4,2);
		\coordinate (2) at (4,2);
		\coordinate (3) at (4,1);
		\coordinate (4) at (6,2);
		\coordinate (5) at (6,-2);
		\coordinate (6) at (4,-1);
		\coordinate (7) at (4,-2);
		\draw[ultra thick] (0) foreach \i in {1,...,7} {--(\i)}  -- cycle;
		\fill[gray!80!blue] (0) foreach \i in {1,...,7} {--(\i)}  -- cycle;
		\node[text=white] at (0,0) {\small  #3};
	\end{scope}
}
\def\robot#1#2{
	\begin{scope}[shift={#1},scale=0.09]
		\coordinate (0) at (0,5);
		\coordinate (1) at (1,5);
		\coordinate (2) at (1,3);
		\coordinate (3) at (2,3);
		\coordinate (4) at (2,2);
		\coordinate (5) at (3,2);
		\coordinate (6) at (3,4);
		\coordinate (7) at (2.55,4.45);
		\coordinate (8) at (2.5,4.95);
		\coordinate (9) at (2.85,4.45);
		\coordinate (10) at (3.1,4.05);
		\coordinate (11) at (3.4,4.05);
		\coordinate (12) at (3.75,4.45);
		\coordinate (13) at (3.55,4.95);
		\coordinate (14) at (4.05,4.45);
		\coordinate (15) at (3.5,4);
		\coordinate (16) at (3.5,1.5);
		\coordinate (17) at (2,1.5);
		\coordinate (18) at (2,0);
		\coordinate (19) at (3,-1);
		\coordinate (20) at (3,-2);
		\coordinate (21) at (0,-2);
		\draw[ultra thick] (0) foreach \i in {1,...,21} {--(\i)};
		\draw[ultra thick, gray!40!white] (21)--(0);
		\fill[gray!40!white] (0) foreach \i in {1,...,21} {--(\i)} -- cycle;
		\node[shape=circle, inner sep=1pt, fill=red] at (0.6,4.5) {};
		\begin{scope}[yscale=1,xscale=-1]
			\coordinate (0) at (0,5);
			\coordinate (1) at (1,5);
			\coordinate (2) at (1,3);
			\coordinate (3) at (2,3);
			\coordinate (4) at (2,2);
			\coordinate (5) at (3,2);
			\coordinate (6) at (3,4);
			\coordinate (7) at (2.55,4.45);
			\coordinate (8) at (2.5,4.95);
			\coordinate (9) at (2.85,4.45);
			\coordinate (10) at (3.1,4.05);
			\coordinate (11) at (3.4,4.05);
			\coordinate (12) at (3.75,4.45);
			\coordinate (13) at (3.55,4.95);
			\coordinate (14) at (4.05,4.45);
			\coordinate (15) at (3.5,4);
			\coordinate (16) at (3.5,1.5);
			\coordinate (17) at (2,1.5);
			\coordinate (18) at (2,0);
			\coordinate (19) at (3,-1);
			\coordinate (20) at (3,-2);
			\coordinate (21) at (0,-2);
			\draw[ultra thick] (0) foreach \i in {1,...,21} {--(\i)};
			\draw[gray!40!white] (21)--(0);
			\fill[gray!40!white] (0) foreach \i in {1,...,21} {--(\i)} -- cycle;
			\node[shape=circle, inner sep=1pt, fill=red] at (0.6,4.5) {};
		\end{scope}
		\node at (0,0) {#2};
		
	\end{scope}
}

\def\clause#1#2#3#4{
	\begin{scope}[square/.style = {regular polygon, regular polygon sides=4, scale=1.2}]
		
		\node[fill=black, opacity=1, draw=black, square] at (#1,#2) {};
		\node[fill=black, opacity=1, draw=black, square] at (#1 + 0.6, #3) {};
		\node[fill=black, opacity=1, draw=black, square] at (#1 + 1.2,#4) {};
		
		\draw (#1,#2) -- (#1,0);
		\draw (#1 + 0.6, #3) -- (#1+ 0.6,0);
		\draw (#1 + 1.2,#4) -- (#1+ 1.2,0);
	\end{scope}
}

\def\crossing#1#2{
	\begin{scope}[shift={#1}, scale=#2]
		\node (A) at (0,0.7) {}; % A
		\node (B) at (0.6,1.5) {}; % B
		
		\node (C) at (0.4,-0.1) {}; % C
		\node (D) at (0.6,0.3) {}; % D
		
		\node (E) at (1.4,0.1) {}; % E
		\node (F) at (1.4,0.7) {}; % F
		
		\node (G) at (1.2,1.1) {}; % G
		\node (H) at (1.4,2) {}; % H
		\node (I) at (1.85,0.75) {}; % I
		\node (J) at (2.1,1.6) {}; % J
		
		\node (K) at (2.2,-0.2) {}; % K
		\node (L) at (2.2,0.8) {}; % L
		
		\node (M) at (1.3,-1) {}; % M
		\node (N) at (1.3,-0.5) {}; % N
		
		\node (O) at (2,-1.4) {}; % M
		\node (P) at (2,-0.8) {}; % N
		
		\node (Q) at (2.7,0) {}; % M
		\node (R) at (3.1,0) {}; % N 
		
		\node (L1) at (-2,1) {};
		\node (L2) at (-1.2,0.6) {};
		
		\node (L3) at (-0.6,0.9) {};
		\node (L4) at (-0.6,0.3) {};
		
		\node (R1) at (3.5,0.3) {};
		\node (R2) at (3.5,-0.3) {};
		
		\node (R3) at (4,0) {};
		\node (R4) at (4,-0.6) {};
		
		\node (R5) at (4.5,0.7) {};
		\node (R6) at (4.5,-0.3) {};
		
		\node (U1) at (1,2.5) {};
		\node (U2) at (1.6,2.5) {};
		
		\node (U3) at (1.4,3) {};
		\node (U4) at (1.4,4) {};
		
		\node (D1) at (1.8,-2.1) {};
		\node (D2) at (2.4,-1.9) {};

		\draw[thick] (A)--(B);
		\draw[thick] (C)--(D);
		\draw[thick] (E)--(F);
		\draw[thick] (G)--(H);
		\draw[thick] (I)--(J);
		\draw[thick] (K)--(L);
		\draw[thick] (M)--(N);
		\draw[thick] (O)--(P);
		\draw[thick] (Q)--(R);
		\draw[thick] (L1)--(L2);
		\draw[thick] (L3)--(L4);
		\draw[thick] (R1)--(R2);
		\draw[thick] (R3)--(R4);
		\draw[thick] (R5)--(R6);
		\draw[thick] (D1)--(D2);
		\draw[thick] (U1)--(U2);
		\draw[thick] (U3)--(U4);
		
	\end{scope}
}

\def\hedgeT#1#2{
	\begin{scope}[shift={#1}]
		\foreach [evaluate={\j = int(mod(\i,2)); \x = int(mod(\i,3));}] \i in {1,...,#2}
		{
			\ifthenelse{\x = 1}
			{
				\node[BLUE] (\i) at (\i - 0.2, 0) {};
				\node[RED] (\i') at (\i + 0.2 , 0.3) {};
			}
			{}
			
			\ifthenelse{\x = 2}
			{
				\node[GREEN] (\i) at (\i - 0.2, 0) {};
				\node[BLUE] (\i') at (\i + 0.2 , 0.3) {};
				
			}
			{}
			
			\ifthenelse{\x = 0}
			{
				\node[RED] (\i) at (\i - 0.2, 0) {};
				\node[GREEN] (\i') at (\i + 0.2 , 0.3) {};
			}
			{}

			\ifthenelse{\j = 0}
			{
				\pgfmathtruncatemacro\k{\i-1};
				\draw[black, thick] (\i)--(\k);
				\draw[black, thick] (\i')--(\k');		
			}
			{}			
		}
		
	\end{scope}
}

\def\hedgeF#1#2{
	\begin{scope}[shift={#1}]
		\foreach [evaluate={\j = int(mod(\i,2)); \x = int(mod(\i,3));}] \i in {1,...,#2}
		{
			\ifthenelse{\x = 1}
			{
				\node[RED] (\i) at (\i - 0.2, 0) {};
				\node[BLUE] (\i') at (\i + 0.2 , 0.3) {};
			}
			{}
			
			\ifthenelse{\x = 2}
			{
				\node[GREEN] (\i) at (\i - 0.2, 0) {};
				\node[RED] (\i') at (\i + 0.2 , 0.3) {};
				
			}
			{}
			
			\ifthenelse{\x = 0}
			{
				\node[BLUE] (\i) at (\i - 0.2, 0) {};
				\node[GREEN] (\i') at (\i + 0.2 , 0.3) {};
			}
			{}

			\ifthenelse{\j = 0}
			{
				\pgfmathtruncatemacro\k{\i-1};
				\draw[black, thick] (\i)--(\k);
				\draw[black, thick] (\i')--(\k');		
			}
			{}			
		}
		
	\end{scope}
}

\def\hedgeL#1#2{
	\begin{scope}[shift={#1}]
		\foreach [evaluate={\j = int(mod(\i,2)); \x = int(mod(\i,3));}] \i in {1,...,#2}
		{
			\ifthenelse{\x = 1}
			{
				\node[BLUE] (\i) at (\i + 0.2, 0) {};
				\node[RED] (\i') at (\i - 0.2 , 0.3) {};
			}
			{}
			
			\ifthenelse{\x = 2}
			{
				\node[RED] (\i) at (\i + 0.2, 0) {};
				\node[GREEN] (\i') at (\i - 0.2 , 0.3) {};
				
			}
			{}
			
			\ifthenelse{\x = 0}
			{
				\node[GREEN] (\i) at (\i + 0.2, 0) {};
				\node[BLUE] (\i') at (\i - 0.2 , 0.3) {};
			}
			{}

			\ifthenelse{\j = 0}
			{
				\pgfmathtruncatemacro\k{\i-1};
				\draw[black, thick] (\i)--(\k);
				\draw[black, thick] (\i')--(\k');		
			}
			{}			
		}
		
	\end{scope}
}

\def\nae#1{
	\begin{scope}[shift={#1}]
		\node (A) at (-1.4,0.4) {}; % A
		\node (B) at (-1.4,-0.4) {}; % A
		
		\node (C) at (-1,0) {}; 
		\node (D) at (-0.4,0) {}; 
		
		\node (E) at (1.4,0.4) {}; 
		\node (F) at (1.4,-0.4) {}; 
		
		\node (G) at (1,0) {}; 
		\node (H) at (0.4,0) {}; 
		
		\node (I) at (-0.4,1.4) {}; 
		\node (J) at (0.4,1.4) {}; 
		
		\node (K) at (0,1) {}; 
		\node (L) at (0,0.4) {}; 
		
		\node (L1) at (-2.7,0.4) {};
		\node (L2) at (-2.7,-0.4) {};
		
		\node (L3) at (-2,0.8) {};
		\node (L4) at (-2,0) {};
		
		\node (R1) at (2.6,0.4) {};
		\node (R2) at (2.6,-0.4) {};
		
		\node (R3) at (2,0.8) {};
		\node (R4) at (2,0) {};
		
		\node (U1) at (-0.4,2.5) {};
		\node (U2) at (0.4,2.5) {};
		
		\node (U3) at (0,1.9) {};
		\node (U4) at (0.8,1.9) {};
		
		\draw[thick] (A)--(B);
		\draw[thick] (C)--(D);
		\draw[thick] (E)--(F);
		\draw[thick] (G)--(H);
		\draw[thick] (I)--(J);
		\draw[thick] (K)--(L);
		\draw[thick] (L1)--(L2);
		\draw[thick] (L3)--(L4);
		\draw[thick] (R1)--(R2);
		\draw[thick] (R3)--(R4);
		\draw[thick] (U1)--(U2);
		\draw[thick] (U3)--(U4);

	\end{scope}
}

\def\vedgeT#1#2{
	\begin{scope}[shift={#1}]
		\foreach [evaluate={\j = int(mod(\i,2)); \x = int(mod(\i,3));}] \i in {1,...,#2}
		{
			\ifthenelse{\x = 1}
			{
				\node[GREEN] (\i) at (0, -\i - 0.3) {};
				\node[BLUE] (\i') at (0.25, -\i + 0.3) {};
			}
			{}
			
			\ifthenelse{\x = 2}
			{
				\node[BLUE] (\i) at (0, -\i - 0.3) {};
				\node[RED] (\i') at (0.25, -\i + 0.3) {};
			}
			{}
			
			\ifthenelse{\x = 0}
			{
				\node[RED] (\i) at (0, -\i - 0.3) {};
				\node[GREEN] (\i') at (0.25, -\i + 0.3) {};
			}
			{}

			\ifthenelse{\j = 0}
			{
				\pgfmathtruncatemacro\k{\i-1};
				\draw[black, thick] (\i)--(\k);
				\draw[black, thick] (\i')--(\k');		
			}
			{}			
		}
		
	\end{scope}
}

\def\vedgeCT#1#2{
	\begin{scope}[shift={#1}]
		\node[GREEN] (A) at (-0.3, 0) {};
		\node[RED] (B) at (0, -0.4) {};
		\draw[black, thick] (A)--(B);
		\foreach [evaluate={\j = int(mod(\i,2)); \x = int(mod(\i,3));}] \i in {1,...,#2}
		{
			\ifthenelse{\x = 1}
			{
				\node[BLUE] (\i) at (0, -\i - 0.3) {};
				\node[GREEN] (\i') at (0.25, -\i + 0.3) {};
			}
			{}
			
			\ifthenelse{\x = 2}
			{
				\node[GREEN] (\i) at (0, -\i - 0.3) {};
				\node[RED] (\i') at (0.25, -\i + 0.3) {};
			}
			{}
			
			\ifthenelse{\x = 0}
			{
				\node[RED] (\i) at (0, -\i - 0.3) {};
				\node[BLUE] (\i') at (0.25, -\i + 0.3) {};
			}
			{}
			
			\ifthenelse{\j = 0}
			{
				\pgfmathtruncatemacro\k{\i-1};
				\draw[black, thick] (\i)--(\k);
				\draw[black, thick] (\i')--(\k');		
			}
			{}			
		}
	\end{scope}
}

\def\vedgeCF#1#2{
	\begin{scope}[shift={#1}]
		\node[GREEN] (A) at (-0.3, 0) {};
		\node[BLUE] (B) at (0, -0.4) {};
		\draw[black, thick] (A)--(B);
		\foreach [evaluate={\j = int(mod(\i,2)); \x = int(mod(\i,3));}] \i in {1,...,#2}
		{
			\ifthenelse{\x = 1}
			{
				\node[RED] (\i) at (0, -\i - 0.3) {};
				\node[GREEN] (\i') at (0.25, -\i + 0.3) {};
			}
			{}
			
			\ifthenelse{\x = 2}
			{
				\node[GREEN] (\i) at (0, -\i - 0.3) {};
				\node[BLUE] (\i') at (0.25, -\i + 0.3) {};
			}
			{}
			
			\ifthenelse{\x = 0}
			{
				\node[BLUE] (\i) at (0, -\i - 0.3) {};
				\node[RED] (\i') at (0.25, -\i + 0.3) {};
			}
			{}
			
			\ifthenelse{\j = 0}
			{
				\pgfmathtruncatemacro\k{\i-1};
				\draw[black, thick] (\i)--(\k);
				\draw[black, thick] (\i')--(\k');		
			}
			{}			
		}
	\end{scope}
}

\def\vedgeL#1#2{
	\begin{scope}[shift={#1}]
		\foreach [evaluate={\j = int(mod(\i,2)); \x = int(mod(\i,3));}] \i in {1,...,#2}
		{
			
			\ifthenelse{\x = 1}
			{
				\node[RED] (\i) at (0, -\i + 0.3) {};
				\node[GREEN] (\i') at (0.25, -\i - 0.3) {};
			}
			{}
			
			\ifthenelse{\x = 2}
			{
				\node[BLUE] (\i) at (0, -\i + 0.3) {};
				\node[RED] (\i') at (0.25, -\i - 0.3) {};
			}
			{}
			\ifthenelse{\x = 0}
			{
				
				\node[GREEN] (\i) at (0, -\i + 0.3) {};
				\node[BLUE] (\i') at (0.25, -\i - 0.3) {};
			}
			{}
			
			\ifthenelse{\j = 0}
			{
				\pgfmathtruncatemacro\k{\i-1};
				\draw[black, thick] (\i)--(\k);
				\draw[black, thick] (\i')--(\k');		
			}
			{}			
		}
	\end{scope}
}

\def\polyV#1#2{%
	\begin{scope}[shift={#1}, rotate=#2]
		
		\coordinate (1) at (0,0);
		\coordinate (2) at (0.6, -0.8);
		\coordinate (3) at (0.8, -0.6);
		\coordinate (4) at (1.4, 0);
		\coordinate (5) at (0.8, 0.6);
		\coordinate (6) at (0.6, 0.8);
		\coordinate (7) at (0, 1.4);
		\coordinate (8) at (-0.6, 0.8);
		\coordinate (9) at (-0.8, 0.6);
		\coordinate (10) at (-1.4, 0);
		\coordinate (11) at (-0.8, -0.6);
		\coordinate (12) at (-0.6, -0.8);
		
		\draw[thick] (1)--(2);
		\draw[thick] (3)--(4)--(5);
		\draw[thick] (6)--(7)--(8);
		\draw[thick] (9)--(10)--(11);
		\draw[thick] (12)--(1);

		\fill[gray,opacity=0.3] (1) \foreach \i in {2,...,12}
		{
			--(\i)
		};		
		
	\end{scope}
}

\makeatletter
\renewcommand\footnotesize{%
	\@setfontsize\footnotesize\@ixpt{11}%
	\abovedisplayskip 8\p@ \@plus2\p@ \@minus4\p@
	\abovedisplayshortskip \z@ \@plus\p@
	\belowdisplayshortskip 4\p@ \@plus2\p@ \@minus2\p@
	\def\@listi{\leftmargin\leftmargini
		\topsep 4\p@ \@plus2\p@ \@minus2\p@
		\parsep 2\p@ \@plus\p@ \@minus\p@
		\itemsep \parsep}%
	\belowdisplayskip \abovedisplayskip
}
\makeatother

\makeatletter
\def\grd@save@target#1{%
	\def\grd@target{#1}}
\def\grd@save@start#1{%
	\def\grd@start{#1}}
\tikzset{
	grid with coordinates/.style={
		to path={%
			\pgfextra{%
				\edef\grd@@target{(\tikztotarget)}%
				\tikz@scan@one@point\grd@save@target\grd@@target\relax
				\edef\grd@@start{(\tikztostart)}%
				\tikz@scan@one@point\grd@save@start\grd@@start\relax
				\draw[minor help lines] (\tikztostart) grid (\tikztotarget);
				\draw[major help lines] (\tikztostart) grid (\tikztotarget);
				\grd@start
				\pgfmathsetmacro{\grd@xa}{\the\pgf@x/1cm}
				\pgfmathsetmacro{\grd@ya}{\the\pgf@y/1cm}
				\grd@target
				\pgfmathsetmacro{\grd@xb}{\the\pgf@x/1cm}
				\pgfmathsetmacro{\grd@yb}{\the\pgf@y/1cm}
				\pgfmathsetmacro{\grd@xc}{\grd@xa + \pgfkeysvalueof{/tikz/grid with coordinates/major step}}
				\pgfmathsetmacro{\grd@yc}{\grd@ya + \pgfkeysvalueof{/tikz/grid with coordinates/major step}}
				\foreach \x in {\grd@xa,...,\grd@xb}
				\node[anchor=north] at (\x,\grd@ya) {\tiny \pgfmathprintnumber{\x}};
				\foreach \y in {\grd@ya,...,\grd@yb}
				\node[anchor=east] at (\grd@xa,\y) {\tiny \pgfmathprintnumber{\y}};
			}
		}
	},
	minor help lines/.style={
		help lines,
		step=\pgfkeysvalueof{/tikz/grid with coordinates/minor step}
	},
	major help lines/.style={
		help lines,
		line width=\pgfkeysvalueof{/tikz/grid with coordinates/major line width},
		step=\pgfkeysvalueof{/tikz/grid with coordinates/major step}
	},
	grid with coordinates/.cd,
	minor step/.initial=.1,
	major step/.initial=.5,
	major line width/.initial=0.7pt,
}

\newtheoremstyle{thmsty}
{1em} % Space above
{\topsep} % Space below
{\itshape} % Body font
{} % Indent amount
{\bfseries} % Theorem head font
{.} % Punctuation after theorem head
{1em} % Space after theorem head
{} % Theorem head spec (can be left empty, meaning `normal')

\theoremstyle{thmsty}
\newtheorem{thm}{Theorem}[chapter]
\newtheorem{lem}[thm]{Lemma}
\newtheorem{cor}[thm]{Corollary}
\newtheorem{prp}[thm]{Proposition}
\newtheorem{cnj}[thm]{Conjecture}
\newtheorem{rem}[thm]{Remark}
\newtheorem{clm}[thm]{Claim}
\newtheorem{dfn}[thm]{Definition}
\newtheorem{conj}[thm]{Conjecture}
\newtheorem{open}[thm]{Open problem}

\usepackage{subfiles}

\makeatletter
\providecommand*{\input@path}{}
\g@addto@macro\input@path{{./preliminaries//}}
\g@addto@macro\input@path{{./results//}}
\makeatother

\begin{document}
	\thispagestyle{empty}
\begin{center}
	{ \sc Masaryk University\\
		Faculty of Informatics}
	
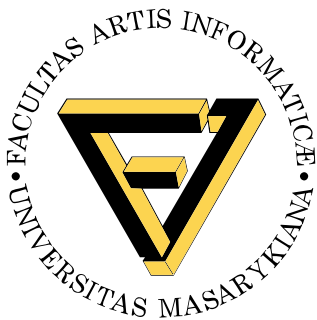
\begin{figure}[htbp]
		\centering
		\includegraphics[width=0.3\linewidth]{fithesis-fi-color}
		\label{fig:logoficnv}
	\end{figure}

	\vspace{0.7cm}
	
	{\Huge \bf Computational Aspects of Problems on Visibility and Disk Graph Representations}
	
	\vspace{1.5cm}
	
	{\large \sc Doctorate Dissertation}

	{\Large \bfseries Onur \c{C}a\u{g}\i r\i c\i}
	\vfill
	
	\textbf{Advisor:} Prof. RNDr. Petr Hlin\v{e}n\'{y}, Ph.D.
	
	\vfill
	
	\begin{center}
		Brno, August 2021
	\end{center}

\end{center}

\pagenumbering{gobble}
\newpage
\thispagestyle{empty}
\mbox{}
\newpage

\section*{Acknowledgments}
During my time in Masaryk University, I have had the chance to meet with very exciting and great people.
I wouldn't have spent such an enjoyable and special time if it wasn't for them.

First of all, I would like to thank to my supervisor, Petr Hlin\v{e}n\'{y} as he generously allowed me to work under his supervision.
Thanks to his support and guidance, I have been able to pursue my research ideas with freedom and confidence. 
Whenever I felt like I was stuck, Petr pushed me to the right direction and helped me to gain intuition about how to tackle challenging research problems.
Without him, I would not be able to complete this piece of work.

I also would like to thank my fellow colleague and officemate Bodhayan Roy for his friendship and support while he was a postdoctoral researcher at the Masaryk University. 
Our discussions with Bodhayan have taught me how to be productive and precise while doing research.

While I was pursuing my doctoral degree, I also had great time with my friends. If it wasn't for the enjoyable times I spent with them during beautiful Brno evenings, I couldn't have found the motivation and courage to go back to work and finish this thesis.
So, thank you, my friends from Brno for all the booze you bought for me.

Although being away from their son were not very comfortable for them, my parents have always gave countenance to me while I was pursuing my PhD.
Thank you, mom, and thank you dad, for all the sacrifices you have made for me to achieve my academic goals.

Last, but absolutely not the least, I would like to thank Deniz, my significant other.
With her endless love and support, I was able to live in a very happy and peaceful environment.
She stood by me no matter what decision I have made, and gave her precious advice whenever I had the need.
I cannot ever overestimate her contributions to the process of preparation of this thesis, and my academic life.
Deniz, I cannot thank you enough for being such an essential part of my life.

\newpage
\thispagestyle{empty}
\mbox{}
\newpage
\begin{abstract}
This thesis focuses on two concepts which are widely studied in the field of computational geometry.
Namely, visibility and unit disk graphs.
In the field of visibility, we have studied the conflict-free chromatic guarding of polygons, for which we have described a polynomial-time algorithm that uses $O(n \log^2 n)$ colors to guard a polygon in a conflict-free setting, and proper coloring of polygon visibility graphs, for which we have described an algorithm that returns a proper 4-coloring for a simple polygon.
Besides, we have shown that the 5-colorability problem is NP-complete on visibility graphs of simple polygons, and 4-colorability is NP-complete on visibility graphs of polygons with holes.

Then, we move further with the notion of visibility, and define a graph class which considers the real-world limitations for the applications of visibility graphs.
That is, no physical object has infinite range, and two objects might not be mutually visible from a certain distance although there are no obstacles in-between.
To model this property, we introduce unit disk visibility graphs, and show that the 3-colorability problem is NP-complete for unit disk visibility graphs of a set of line segments, and a polygon with holes.

After bridging the gap between the visibility and the unit disk graphs, we then present our results on the recognition of unit disk graphs in a restricted setting -- axes-parallel unit disk graphs.
We show that the recognition of unit disk graphs is NP-complete when the disks are centered on pre-given parallel lines.
If, on the other hand, the lines are not parallel to one another, the recognition problem is NP-hard even though the pre-given lines are axes-parallel (i.e. any pair is either parallel or perpendicular).
\end{abstract}

\newpage
\pagenumbering{arabic}
\setcounter{page}{1}
\tableofcontents

\chapter{Introduction} \label{chap:introduction}
The field of computational geometry has gained attention drastically with the latest technological advancements.
Besides two very essential application areas, namely, computer graphics, and 3D design softwares, the computational geometry algorithms are also utilized in the fields of robotics \cite{Latombe_robotmotion}, efficient 3D printing \cite{Gupta_3D}, geographic information systems \cite{pathPlanning}, computer-integrated manufacturing \cite{Scheer_cim}, and many more \cite{bcko-cgaa-08}.
In this thesis, we attempt to tackle the problems which concern the visibility relations and wireless networks.

A very famous visibility problem, called \emph{art gallery problem}, has been studied for a long time along with its variations \cite{o-agta-87,ll-ccagp-86,Bartschi-2014,hoffmann:2015}.
The problem is simply placing guards into an art gallery, such that the guards can observe the whole gallery together.
While tackling this problem, the gallery is modeled as a simple polygon, and each guard is represented by a point in the polygon.

Until recently, the main focus for the art gallery problem and its variations was to minimize the number of guards used.
However, as the wireless technologies advanced, the very same problem is revisited with the motivation of robot motion planning via wireless communications \cite{Latombe_robotmotion}.
As the wireless sensors became cheaper and cheaper, the objective of the problem has evolved into minimizing the number of unique frequencies used by the sensors, instead of minimizing the number of the guards.

This very problem constitutes the main motivation of our thesis, and is also closely related to frequency assignment problem in wireless networks \cite{freqAssignment,suri-conflict}.
One can easily solve this problem by placing a sensor at each corner of the room, and assigning a different frequency to each sensor.
However, this method becomes very expensive as the number of sensors grow \cite{freqAssignment,cf-app}.
Therefore, the main goal in this problem is minimize the number of different frequencies assigned to sensors.
Since the cost of a sensor is comparatively very low, we do not aim to minimize the number of sensors used.

Along with this application-based motivation, we also tackle a theoretical problem, colorability of polygon visibility graphs, which was left open by Ghosh \cite{Ghosh_unsolvedproblems} in 1995, and solved by  \c{C}a\u{g}\i r\i c\i, Hlin\v{e}n\'{y}, and Roy in 2017.
In this problem, every vertex of the given polygon corresponds to a vertex in the visibility graph, and two vertices of the graph are adjacent if, and only if the corresponding vertices see each other.

While dealing with the problems that are mentioned above, we also consider some restricted types of polygon, namely, funnels and weak-visibility polygons.
A funnel consists of two concave polygonal chains which meet at a common point at one end, and connected by an edge at the other end.
A weak-visibility polygon is a polygon which has a specific edge, of whose at least one point is visible from every other vertex.
Although these types of polygons seem an artificial to study on, there are indeed many applications which concern specifically those types of polygons \cite{Ghosh_weakvis,Choi_funnel}.

Since the main motivation of this thesis is based on the wireless sensor networks, we also remark that the conventional visibility relations fall short while modeling real-world scenarios, and suggest a new model to overcome this potential inaccuracy.
In the literature, two objects are said to be mutually visible when there are no obstacles in-between.
However, the ``communication'' between a pair of objects also depends on the distance between them, regardless of the type of the communication (visual, audio, data transfer, etc.).
This phenomenon is usually described by unit disk graphs \cite{udgParameterized,udgConstrained}.

In an effort to bridge the gap between two important fields of study in computational geometry, namely, unit disk graphs and visibility the graphs, we introduce a new graph class: unit disk visibility graphs.
This class considers both the combinatorial problems concerning visibility relations and the limitations of the wireless sensors and is also a superclass of unit disk graphs.

We moreover study unit disk graph recognition problem \cite{Breu_UDrecog} based on the motivation of theoretical aspects of wireless sensor networks \cite{Apnes_theory}.

The results presented in this thesis are published or will be published in the following conferences. 1, 2, and 4 are also submitted to high-quality journals, but the review process has not yet finished in the time of the submission of this thesis. Moreover, 4 has received the ``Best Student Paper'' award from the conference chair, and 5 has just been accepted to Journal of Combinatorial Theory, Series B (JCTB), which is a top journal in its field.

\begin{enumerate}
	\item O. \c{C}a\u{g}{\i}r{\i}c{\i}, P. Hlin\v{e}n\'{y}, B. Roy: On Colourability of Polygon Visibility Graphs in \emph{Foundations of Software Technology and Theoretical Computer Science (FSTTCS)}, December 2017 (Section~\ref{sec:polyproper}).
	
	\item O. {\c{C}}a\u{g}\i{r}\i{c}\i, S. K. Ghosh, P. Hlin\v{e}n\'{y}, B. Roy: On Conflict-Free Chromatic Guarding of Simple Polygons in \emph{Combinatorial Optimization and Applications (COCOA)}, December 2019 (Sections \ref{sec:funnels}, \ref{sec:cfweakvis}, and \ref{sec:polygonconflict})

	\item D. A\u{g}ao\u{g}lu and O. \c{C}a\u{g}{\i}r{\i}c{\i}: Unit Disk Visibility Graphs in \emph{European Conference on Combinatorics, Graph Theory and Applications (EUROCOMB)}, September 2021 (Chapter~\ref{chap:udvg})

	\item O. \c{C}a\u{g}{\i}r{\i}c{\i}: On embeddability of unit disk graphs onto straight lines in \emph{Computer Science in Russia (CSR)}, June 2020 (Chapter~\ref{chap:apud})
	
	\item O. \c{C}a\u{g}{\i}r{\i}c{\i}, P. Hlin\v{e}n\'{y}, F. Pokr\'{y}vka, A. Sankaran: Clique-Width of Point Configurations in \emph{Workshop on Graph-Theoretic Concepts in Computer Science (WG)}, June 2020 (outside of the main focus of this thesis, mentioned briefly in Chapter~\ref{chap:conclusion})
	
	\item O. \c{C}a\u{g}{\i}r{\i}c{\i}, L. Casuso, C. Medina, T. Patino, M. Raggi, E. Roldan-Pensado, G. Salazar, J. Urrutia: On upward straight-line embeddings of oriented paths in \emph{XVII Spanish Meeting on Computational
	Geometry (ECG)}, July 2017 (outside of the main focus of this thesis, mentioned briefly in Chapter~\ref{chap:conclusion})
\end{enumerate}

The contributions of the author of this thesis to the papers that are mentioned above are summarized in Table~\ref{table:authorsContributions}.

\chapter{Preliminaries} \label{chap:preliminaries}
\section{Organization of the chapter}

This chapter introduces necessary definitions and notations that we use throughout this thesis.
We generally use standard notations but some are slightly changed to fit into the context better.
In Section~\ref{sec:graphs}, we briefly introduce graphs and describe the definitions and notations we use.
In Section~\ref{sec:geometry}, we describe the visibility graphs, in Section~\ref{sec:intersection} we focus on the geometric aspects of the intersection graphs, and in Section~\ref{sec:problems}, we define the problems we tackle.

\section{Graphs} \label{sec:graphs}

	This section describes the fundamental graph theory terminology and notations.
	Unless stated explicitly, a graph is always
	\begin{itemize}
		\item simple,
		\item undirected,
		\item unweighted,
		\item connected, and
		\item finite.
	\end{itemize}
	in this thesis.
	Thus, in this section, all the definitions follow accordingly.
	
	\subsubsection{Vertex, edge and adjacency}
	A \emph{graph} is an abstract data structure which defines the relationships among a set of objects.
	Each object is represented by (usually) a circle which is referred to as a \emph{vertex} or a \emph{node}.
	If two objects are related, then there exists an \emph{edge} between the vertices which they are represented with, and is denoted by a line segment.
	A graph $G$ is denoted by $G = (V, E)$ where $V$ is the set of vertices, and $E$ is the set of edges.
	The vertices are named using numbers $1,2,\dots$ or letters $a,b,\dots$.
	See Figure~\ref{fig:exGraph} for an example graph whose vertices are denoted by letters.

	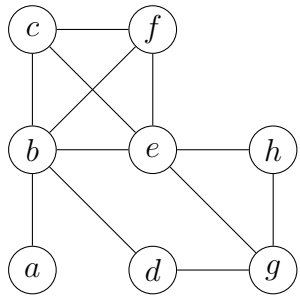
\begin{figure}[htbp]
		\centering
		\begin{tikzpicture}[scale=1.6]
			\tikzstyle{every node}=[draw, fill=white, shape=circle, minimum size=18pt,inner sep=0pt];
			\node (a) at (0,0) {$a$};
			\node (b) at (0,1) {$b$};
			\node (c) at (0,2) {$c$};
			\node (d) at (1,0) {$d$};
			\node (e) at (1,1) {$e$};
			\node (f) at (1,2) {$f$};
			\node (g) at (2,0) {$g$};
			\node (h) at (2,1) {$h$};
			
			\draw (a)--(b)--(c)--(e)--(h)--(g);
			\draw (g)--(d)--(b)--(f)--(e)--(g);
			\draw (c)--(f);
			\draw (b)--(e);
			
		\end{tikzpicture}
		\caption{A graph with 8 vertices, labeled $a$ through $g$, and 12 edges.}
		\label{fig:exGraph}
	\end{figure}

	An edge is denoted by the pair of vertices which it joins together.
	In Figure~\ref{fig:exGraph}, the vertices of the graph are $a$, $b$, $c$, $d$, $e$, $f$, $g$ and the edges of the graph are $ab$, $bc$, $bd$, $be$, $bf$, $ce$, $cf$, $dg$, $ef$, $eg$, $eh$, $gh$.
	There exists an edge $ab$ between a pair $a$ and $b$ of vertices. 
	Thus, $a$ and $b$ are called \emph{adjacent vertices} or \emph{neighbors}, and the edge $ab$ is called \emph{incident to} both $a$ and $b$.
	
	The set of all the vertices that are adjacent to a vertex $v$ is called the \emph{neighborhood} of $v$ and denoted by $N(v)$.
	For instance, in Figure~\ref{fig:exGraph}, $N(a) = \{b\}$ and $N(b) = \{a,c,d,e\}$.
	Note that the neighborhood of a vertex does not include the vertex itself. 
	Although this set is informally referred to as \emph{the neighborhood}, more specifically, it is the \emph{open neighborhood}.
	If we include $v$ in the neighborhood of $v$, then we have the \emph{closed neighborhood} of $v$, which is denoted by $N[v]$. 
	That is, $N[v] = N(v) \cup \{v\}$.
	
	\subsubsection{Subgraphs}
	A \emph{subgraph} $H = (V', E')$ of the graph $G$ is a graph whose vertex set $V' \subseteq V$ is a subset of $V$, and edge set $E' \subseteq E$ is a subset of $E$.
	For instance, let $G$ be the graph given in Figure~\ref{fig:exGraph}.
	Then, $H =\left(\{b,c,e,f\}, \{bc, cf, fe, eb\}\right)$ is a subgraph of $G$.
	An \emph{induced subgraph}, on the other hand, must include every edge which connects a pair of vertices in the vertex subset.
	This means that $H$ is not an induced subgraph without the edges $bf$ and $ce$.
	
	\subsubsection{Paths}
	Let us mention some special structures which might appear in graphs.
	A \emph{path} $P_n$ of length $n$ is a sequence of edges $(e_1, e_2, \dots, e_n)$ which joins a sequence of vertices without repetition.
	Every edge in a path joins two vertices, and there exists no other edges between any pair of the vertices except $e_1,\dots,e_n$.
	In Figure~\ref{fig:exGraph}, $(ab,bd,dg,ge)$ is a 4-path, but it is not an induced path, as the edge $be$ exists in the graph.
	An example of induced 4-path is $(ab,bd,dg,gh)$.
	A path can be denoted by the sequence of the edges as well as the sequence of the vertices which it includes, i.e., $(ab,bd,dg,gh)$ and $(a,b,d,g,h)$ define the same substructure.
	This substructure is also referred to as ``a path between $a$ and $h$'' and denoted by $\pi(a,h)$
	The smallest possible path is a single vertex, and is denoted by $P_0$.
	
	A \emph{shortest path} $\Pi(u,v)$ between a pair $u$ and $v$ of vertices is a path whose length is not greater than any given $\pi(u,v)$ in the same graph.

	\subsubsection{Cycles}
	A \emph{cycle} $C_n$ is simply a path of length $n$, plus one extra edge which joins the first vertex in the sequence with the last vertex in the sequence.
	In Figure~\ref{fig:exGraph}, the vertices $b,d,g,e$ together with the incident edges form an induced 4-cycle.
	A graph without any cycles is called a \emph{tree}.
	In Figure~\ref{fig:exGraph}, the vertices $a,b,d,e,h$ form a tree.
	A \emph{shortest path tree} $\textit{SPT}(u)$ of a vertex $u$ in a graph $G$ is a spanning tree of $G$ rooted at $u$, such that all paths between $u$ and any vertex $v \in \textit{SPT}(u)$ is a shortest path $\Pi^*_{uv}$ in $G$.
	
	\subsubsection{Complete graphs and cliques}
	A \emph{complete graph} $K_n$ is a graph with $n$ vertices, in which every pair of vertices are adjacent.
	When a graph $G$ has a subgraph $Q$ which is a complete graph, then $Q$ is referred to as a \emph{clique}.
	$Q$ is said to be a \emph{maximal clique} if there exists no other vertex in the graph which is adjacent to every vertex in $Q$.
	In other words, a maximal clique cannot be extended by adding other vertices.
	A clique $Q$ is called a \emph{maximum clique} if there exist no other cliques in the graph whose cardinality is strictly greater than $Q$.
	In Figure~\ref{fig:exGraph}, $b,c,e,f$ form a clique of size 4, which is also the maximum clique.
	The smallest non-empty clique, $K_1$ is a single vertex.
	Analogously, $K_2$ is an edge with two vertices, and $K_3$ is a 3-cycle ($C_3$).

\section{Geometry} \label{sec:geometry}

\subsection{Polygons} \label{sec:polygon}
A polygon $\mathcal{P}$ is defined as a closed region in the Euclidean plane, bounded by a finite set $\{e_1, \dots, e_n\}$ of line segments that are referred to as \emph{edges} of the polygon.
By definition, no pair of edges of a polygon intersect in the Euclidan plane.
Every consecutive pair $(e_i, e_{i+1})$ and specifically $(e_n, e_1)$ of line segments shares an endpoint.
These endpoints are called the \emph{vertices} of the polygon.
Such a cyclic sequence of consecutive vertices is referred to as a \emph{polygonal chain}.
The vertices are usually labeled by numbers $1,2,\dots$ in the clockwise order.
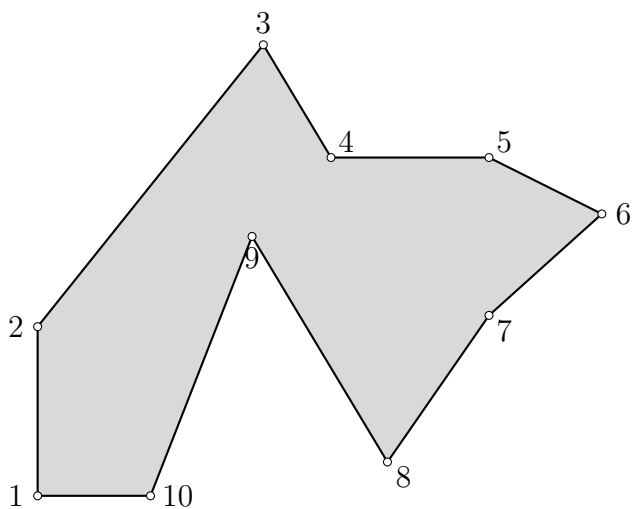
\begin{figure}
	\centering
	\begin{tikzpicture}[scale=1.5]
		\tikzstyle{every node}=[draw, fill=white, shape=circle, minimum size=3pt,inner sep=1pt];
		\coordinate (1) at (0,0);
		\coordinate (2) at (0,1.5);
		\coordinate (3) at (2,4);
		\coordinate (4) at (2.6,3);		
		\coordinate (5) at (4,3);
		\coordinate (6) at (5,2.5);
		\coordinate (7) at (4,1.6);
		\coordinate (8) at (3.1,0.3);
		\coordinate (9) at (1.9,2.3);
		\coordinate (10) at (1,0);
		
		\fill[gray, opacity=0.3] (1) foreach \i in {2,...,10} {--(\i)} --cycle;   
		\draw[thick] (1) foreach \i in {2,...,10} {--(\i)} --cycle;
		
		\node[label=left:$1$] at (1) {};
		\node[label=left:$2$] at (2) {};
		\node[label=above:$3$] at (3) {};
		\node[label=45:$4$] at (4) {};
		\node[label=45:$5$] at (5) {};
		\node[label=right:$6$] at (6) {};
		\node[label=-45:$7$] at (7) {};
		\node[label=-30:$8$] at (8) {};
		\node[label=below:$9$] at (9) {};
		\node[label=right:$10$] at (10) {};

		%\tikzstyle{every path}=[draw, brown];
		%\draw (v1)--(v3);
		%\draw (v1)--(v9);
		%\draw (v2)--(v9);
		%\draw (v2)--(v10);
		%\draw (v3)--(v5);
		%\draw (v3)--(v6);
		%\draw (v3)--(v7);
		%\draw (v3)--(v8);
		%\draw (v3)--(v9);
		%\draw (v3)--(v10);
		%\draw (v4)--(v7);
		%\draw (v4)--(v8);
		%\draw (v4)--(v9);
		%\draw (v5)--(v7);
		%\draw (v5)--(v8);
		%\draw (v6)--(v8);
		%\draw (v8)--(v10);	
		
	\end{tikzpicture}
	\caption{A simple polygon with 10 vertices.}
	\label{fig:polygon}
\end{figure}
We give an example polygon in Figure~\ref{fig:polygon}.
The boundary edges are denoted by bold lines, and the vertices are denoted by small circles.
The interior of the polygon is shaded in our pictures.

Note that the above definition (as well as the polygon given in Figure~\ref{fig:polygon}) is a \emph{simple polygon}, meaning that the boundary edges do not cross, and there are no ``holes'' in the polygon.
A \textit{hole} inside a polygon is simply another polygon which acts as an obstacle that blocks the visibility.
The holes can be polygonal \cite{Wein_voronoi}, disk \cite{Kim_shortestPF} and even mobile \cite{Khaili_movingObs}. 

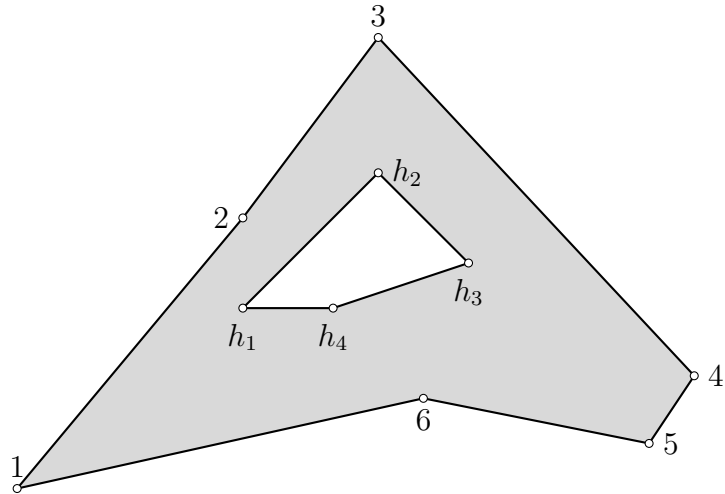
\begin{figure}[htbp]
	\centering
	\begin{tikzpicture}[scale=1.2]
		\tikzstyle{every node}=[draw, fill=white, shape=circle, minimum size=3pt,inner sep=1pt];
		\coordinate (1) at (0,0);
		\coordinate (2) at (2.5,3);
		\coordinate (3) at (4,5);
		\coordinate (4) at (7.5,1.25);		
		\coordinate (5) at (7,0.5);
		\coordinate (6) at (4.5,1);
		
		\coordinate (h1) at (2.5,2);
		\coordinate (h2) at (4,3.5);
		\coordinate (h3) at (5,2.5);
		\coordinate (h4) at (3.5,2);
		
		\fill[gray, opacity=0.3] (1) foreach \i in {2,...,6} {--(\i)} --cycle;   
		\fill[white] (h1) foreach \i in {2,...,4} {--(h\i)} --cycle;  
		
		\draw[thick] (1) foreach \i in {2,...,6} {--(\i)} --cycle;
		\draw[thick] (h1) foreach \i in {2,...,4} {--(h\i)} --cycle;
		
		\node [label=$1$] at (1) {};
		\node [label=left:$2$] at (2) {};
		\node [label=$3$] at (3) {};
		\node [label=right:$4$] at (4) {};
		\node [label=right:$5$] at (5) {};
		\node [label=below:$6$] at (6) {};
		
		\node [label=below:$h_1$] at (h1) {};
		\node [label=right:$h_2$] at (h2) {};
		\node [label=below:$h_3$] at (h3) {};
		\node [label=below:$h_4$] at (h4) {};
		
	\end{tikzpicture}
	\caption{A polygon with 6 vertices, and a hole with 4 vertices which is inside that polygon.}
	\label{fig:polygonholes}
\end{figure}

%\subsection{Terrrains}
%A terrain $T$ is an $x$-monotone polygonal chain in the Euclidean plane.
%That is, given the order of line segments $(e_1,\dots,e_n)$ from left to right, a point on $e_i$ is to the left of a point on $e_j$ if, and only if $i<j$.
%
%\begin{figure}[htbp]
%	\centering
%	\begin{tikzpicture}
%		\tikzstyle{every node}=[draw, fill=white, shape=circle, minimum size=3pt,inner sep=1pt];
%		\coordinate (1) at (-1,3);
%		\coordinate (2) at (0,2.5);
%		\coordinate (3) at (1,0.8);
%		\coordinate (4) at (2,0.1);
%		\coordinate (5) at (3,0.8);
%		\coordinate (6) at (4,1.2);
%		\coordinate (7) at (5,2);
%		\coordinate (8) at (6,2.1);
%		\draw[thick] (1) foreach \i in {2,...,8} {--(\i)};
%		\foreach \i in {1,...,8}
%		{
%			\node at (\i) {};
%		}
%	\end{tikzpicture}
%
%	\caption{An example terrain.}
%	\label{fig:terrhalf}
%\end{figure}

\subsection{Visibility and guarding} \label{sec:visibility}
We describe the visibility relation among a set of geometric objects using \emph{visibility graphs}.
Visibility, despite having various meanings in daily usage, has a specific meaning when it comes to computational geometry.
A point $p$ in a polygon $\PP$ is said to be visible to another point object $q$, if the line segment $\overline{pq}$ drawn between $p$ and $q$ belongs to $\PP$. See Figure~\ref{fig:polygonvisexp} for an example.

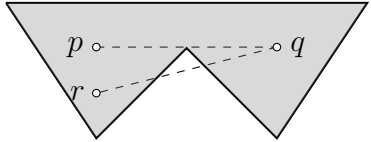
\begin{figure}
	\centering
	\begin{tikzpicture}[scale=1.2]
		\tikzstyle{every node}=[draw, fill=white, shape=circle, minimum size=1pt,inner sep=1pt];
		\coordinate (1) at (0,0);
		\coordinate (2) at (1,1);
		\coordinate (3) at (2,0);
		\coordinate (4) at (0,1.01);		
		\coordinate (5) at (2,1.01);
		\coordinate (6) at (0,0.5);
		\draw[dashed] (4)--(5);
		\draw[dashed] (5)--(6);
		\coordinate (i1) at (-1,1.5);
		\coordinate (i2) at (1,1.5);
		\coordinate (i3) at (3,1.5);
		
		\draw[thick] (i1)--(i2)--(i3)--(3)--(2)--(1)--(i1);
		\fill[gray, opacity=0.3] (i1)--(i2)--(i3)--(3)--(2)--(1)--(i1);

		\node [label=left:$p$] at (4) {};
		\node [label=right:$q$] at (5) {};
		\node [label=left:$r$] at (6) {};
	
	\end{tikzpicture}
	\caption{Two points $p$ and $q$ in a polygon see each other because $\overline{pq}$ is never outside the polygon, even though a vertex of the polygon is on the segment. However, $r$ and $q$ do not see each other because $\overline{rq}$ crosses the exterior of the polygon.}
	\label{fig:polygonvisexp}
\end{figure}

In Figure~\ref{fig:robots}, we show an example room with five robots, $a$, $b$, $c$, $d$, $e$, and two security cameras $X$, $Y$ in it.
The visibility relations of the robots and the cameras are denoted by lines, and the walls of the room are denoted by thick lines.
As seen in the figure, $X$ only sees $a$ since $c$, $d$ and $b$ are on the other side of the wall, and $e$ is behind $a$.
Similarly, $b$ also cannot see $e$ because $d$ blocks the visibility.

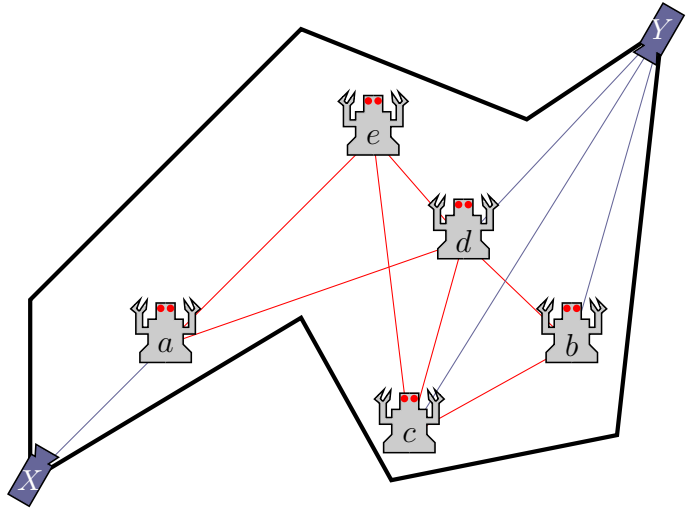
\begin{figure}
	\centering
	\begin{tikzpicture}[scale=1.2]
		\coordinate (1) at (0,0);
		\coordinate (2) at (3,1.8);
		\coordinate (3) at (4,0);
		\coordinate (4) at (6.5,0.5);
		\coordinate (5) at (7,5);
		\coordinate (6) at (5.5,4);
		\coordinate (7) at (3,5);
		\coordinate (8) at (0,2);
		
		\draw[ultra thick] (1) foreach \i in {2,...,8} {--(\i)} -- cycle;
		
		\tikzstyle{every path}=[draw, red];
		\coordinate (a) at (1.5,1.5);
		\coordinate (b) at (6,1.5);
		\coordinate (c) at (4.2,0.5);
		\coordinate (d) at (4.8,2.65);
		\coordinate (e) at (3.8,3.8);
		\draw (b)--(c)--(d)--(e)--(a);
		\draw (a)--(d);
		\draw (b)--(d);
		\draw (c)--(e);
		
		\tikzstyle{every path}=[draw, gray!80!blue];
		\draw (1)--(a);
		\draw (5)--(b);
		\draw (5)--(c);
		\draw (5)--(d);

		\tikzstyle{every path}=[draw, black];
		\robot{(1.5,1.5)}{$a$}
		\robot{(6,1.5)}{$b$}
		\robot{(4.2,0.5)}{$c$}
		\robot{(4.8,2.65)}{$d$}
		\robot{(3.8,3.8)}{$e$}
		
		\camera{(0,0)}{60}{$X$}
		\camera{(7,5)}{-120}{$Y$}
	\end{tikzpicture}
	\caption{Five robots, two cameras, and their visibility relations.}
	\label{fig:robots}
\end{figure}

Note that in Figure~\ref{fig:robots}, the visibility relations are defined in 2D.
In a real-world scenario, the cameras see everywhere in the room as they are usually installed close to the ceiling.
However, let us mention that in the scope of this thesis, we always assume that the visibility relations are defined in the Euclidean plane.
Moreover, unlike the picture shown in Figure~\ref{fig:robots}, the visibility relations are defined between pairs of points.
Thus, an \emph{obstacle} between two points $p$ and $q$ in a polygon $\PP$, is simply another point $o$, which lies on the line segment $\overline{pq}$ drawn between them, and is not a vertex of $\PP$.
The room, on the other hand, is modeled as a \emph{polygon} (see Section~\ref{sec:polygon}).

\emph{Guarding} a polygon $\mathcal{P}$ means placing a set of \emph{guards} $\{g_1, \dots, g_k\}$ into $\mathcal{P}$ such that every point inside $\mathcal{P}$ is seen by at least one of the guards $g_i$ where $1 \leq i \leq k$.
In Figure~\ref{fig:cameras}, two cameras, $X$ and $Y$ are acting as guards and they are guarding the whole room given in Figure~\ref{fig:robots} without the robots.

\begin{figure}
	\centering
	\begin{tikzpicture}[scale=1.2]
		\coordinate (1) at (0,0);
		\coordinate (2) at (3,1.8);
		\coordinate (3) at (4,0);
		\coordinate (4) at (6.5,0.5);
		\coordinate (5) at (7,5);
		\coordinate (6) at (5.5,4);
		\coordinate (7) at (3,5);
		\coordinate (8) at (0,2);
		
		\coordinate (X') at (6.904, 4.142);
		\coordinate (Y') at (0, 0.333);

		\draw[ultra thick] (1) foreach \i in {2,...,8} {--(\i)} -- cycle;
		
		\fill[color=red, opacity=0.2] (X')--(5)--(6)--(7)--(8)--(1)--(2)--cycle;
		\fill[color=blue, opacity=0.2] (Y')--(1)--(2)--(3)--(4)--(5)--cycle;
	
		\camera{(0,0)}{60}{$X$}
		\camera{(7,5)}{-120}{$Y$}
	\end{tikzpicture}
	\caption{Two cameras (guards) seeing the whole room. Orange part is seen (guarded) only by $X$, blue part is seen (guarded) only by $Y$, and the purple part is seen (guarded) by both of the cameras.}
	\label{fig:cameras}
\end{figure}
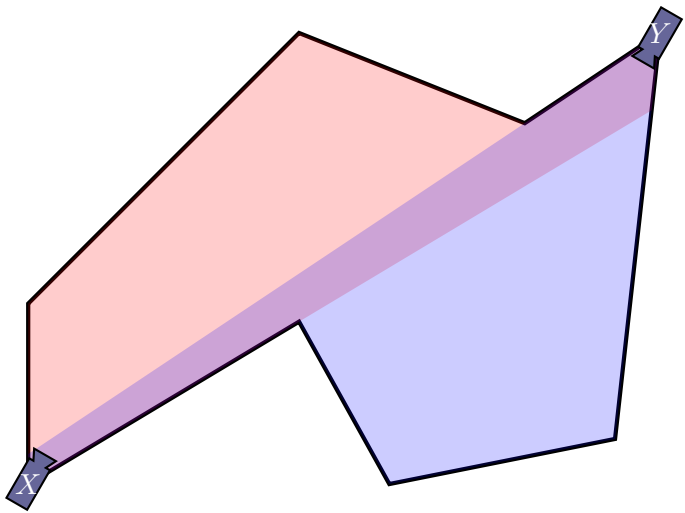

In general, for a polygon with $n$ vertices, $\lfloor n/3 \rfloor$ guards might be needed to guard the whole polygon \cite{c-actpg-75}.
When the guards are restricted to be on the vertices, finding the minimum number is an NP-hard problem \cite{os-snpd-83}.
We give details on this problem in Section~\ref{sec:problems}

Based on the definition of visibility, we give the definitions of two restricted types of polygons, namely, \emph{funnels} and \emph{weak-visibility polygons}.

\subsubsection{Funnels} \label{sec:funneldef}
A funnel is a special type of polygon which consists of two concave polygonal subchains $\mathcal{L} = (l_1, \dots, l_k)$ and $\mathcal{R} = (r_1, \dots, r_m)$ that meet at a vertex $l_k = r_m = \alpha$ called the \emph{apex}, and a \emph{base edge} $l_1r_1$ connecting the other ends of these chains.
In a funnel, a vertex can see only its two immediate neighbors from the same subchain, and some other vertices from the opposite subchain.
In Figure~\ref{fig:funnel}, we see an example funnel with 11 vertices where the length of the left and right polygonal chains are equal to 5.

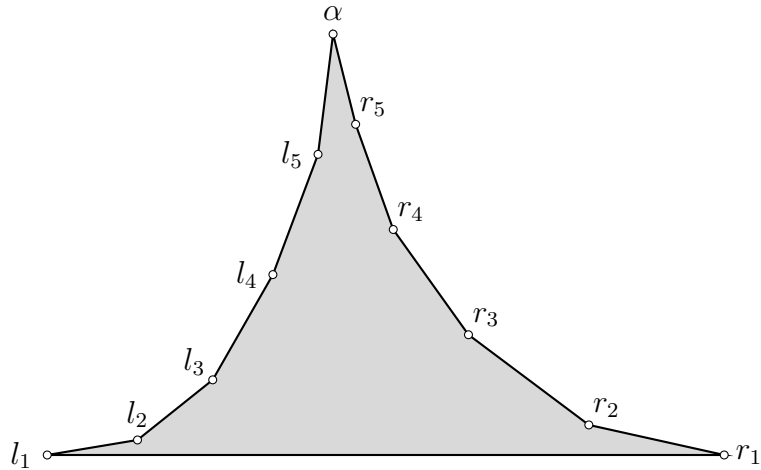
\begin{figure}
	\centering
	\begin{tikzpicture}[scale=2]
		\tikzstyle{every node}=[draw, fill=white, shape=circle, minimum size=3pt,inner sep=1pt];
		\coordinate (1) at (0,0);
		\coordinate (2) at (0.6,0.1);
		\coordinate (3) at (1.1,0.5);
		\coordinate (4) at (1.5,1.2);		
		\coordinate (5) at (1.8,2);
		\coordinate (6) at (1.9,2.8);
		\coordinate (7) at (2.05,2.2);
		\coordinate (8) at (2.3,1.5);
		\coordinate (9) at (2.8,0.8);
		\coordinate (10) at (3.6,0.2);
		\coordinate (11) at (4.5,0);
		
		\fill[gray, opacity=0.3] (1) foreach \i in {2,...,11} {--(\i)} --cycle;   
		\draw[thick] (1) foreach \i in {2,...,11} {--(\i)} --cycle;
		
		\node[label=left:$l_1$] at (1) {};
		\node[label=above:$l_2$] at (2) {};
		\node[label=135:$l_3$] at (3) {};
		\node[label=left:$l_4$] at (4) {};
		\node[label=left:$l_5$] at (5) {};
		\node[label=above:$\alpha$] at (6) {};
		\node[label=45:$r_5$] at (7) {};
		\node[label=45:$r_4$] at (8) {};
		\node[label=45:$r_3$] at (9) {};
		\node[label=45:$r_2$] at (10) {};
		\node[label=right:$r_1$] at (11) {};

	\end{tikzpicture}
	\caption{A funnel with 11 vertices.}
	\label{fig:funnel}
\end{figure}

\subsubsection{Weak-visibility polygons} \label{sec:weakvisdef}
A weak visibility polygon is a polygon where there exists an edge $e_c$, called a \emph{common edge}, or a \emph{base edge} such that every point inside the polygon sees at least one point on $e_c$.
The clockwise ordering on a weak-visibility polygon is usually given such that the common edge is the polygonal edge between the first and the last vertex.
In Figure~\ref{fig:weakvis}, we see a weak-visibility polygon with 25 vertices where the edge between the vertices $1$ and $25$ is the common edge, and every other vertex sees at least one point on that edge.

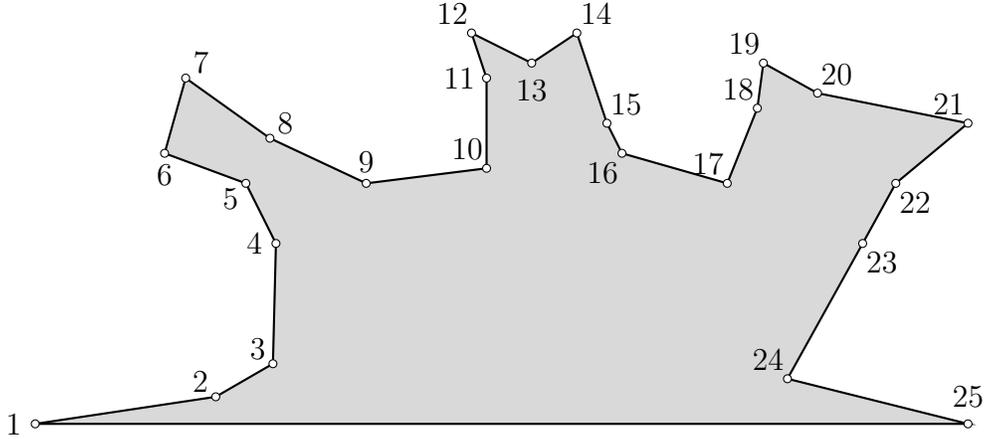
\begin{figure}
	\centering
	\begin{tikzpicture}[scale=0.4]
		\tikzstyle{every node}=[draw, fill=white, shape=circle, minimum size=3pt,inner sep=1pt];
		\coordinate (1) at (0,0) {};
		\coordinate (2) at (6,0.9) {};
		\coordinate (3) at (7.9,2) {};
		\coordinate (4) at (8,6) {};
		\coordinate (5) at (7,8) {};
		\coordinate (6) at (4.3,9) {};
		\coordinate (7) at (5,11.5) {};
		\coordinate (8) at (7.8,9.5) {};
		\coordinate (9) at (11,8) {};
		\coordinate (10) at (15,8.5) {};
		\coordinate (11) at (15,11.5) {};
		\coordinate (12) at (14.5,13) {};
		\coordinate (13) at (16.5,12) {};
		\coordinate (14) at (18,13) {};
		\coordinate (15) at (19,10) {};
		\coordinate (16) at (19.5,9) {};
		\coordinate (17) at (23,8) {};
		\coordinate (18) at (24,10.5) {};
		\coordinate (19) at (24.2,12) {};
		\coordinate (20) at (26,11) {};
		\coordinate (21) at (31,10) {};
		\coordinate (22) at (28.6,8) {};
		\coordinate (23) at (27.5,6) {};
		\coordinate (24) at (25,1.5) {};
		\coordinate (25) at (31,0) {};
		
		\fill[gray, opacity=0.3] (1) foreach \i in {2,...,25} {--(\i)} --cycle;   
		\draw[thick] (1) foreach \i in {2,...,25} {--(\i)} --cycle;
		
		\node[label=left:$1$] at (1) {};
		\node[label=135:$2$] at (2) {};
		\node[label=135:$3$] at (3) {};
		\node[label=left:$4$] at (4) {};
		\node[label=-135:$5$] at (5) {};
		\node[label=below:$6$] at (6) {};
		\node[label=45:$7$] at (7) {};
		\node[label=45:$8$] at (8) {};
		\node[label=above:$9$] at (9) {};
		\node[label=135:$10$] at (10) {};
		\node[label=left:$11$] at (11) {};
		\node[label=135:$12$] at (12) {};
		\node[label=below:$13$] at (13) {};
		\node[label=45:$14$] at (14) {};
		\node[label=45:$15$] at (15) {};
		\node[label=-135:$16$] at (16) {};
		\node[label=135:$17$] at (17) {};
		\node[label=135:$18$] at (18) {};
		\node[label=135:$19$] at (19) {};
		\node[label=45:$20$] at (20) {};
		\node[label=135:$21$] at (21) {};
		\node[label=-45:$22$] at (22) {};
		\node[label=-45:$23$] at (23) {};
		\node[label=135:$24$] at (24) {};
		\node[label=above:$25$] at (25) {};
		
	\end{tikzpicture}
	\caption{A weak-visibility polygon with 25 vertices.}
	\label{fig:weakvis}
\end{figure}
	
\subsubsection{Visibility graphs}
A visibility graph is a simple graph $G = (V,E)$ defined over a set $\mathcal{S}$ of geometric entities where a vertex $u \in V$ represents a geometric entity $s_u \in \mathcal{S}$, and the edge $uv \in E$ exists if and only if $s_u$ and $s_v$ are mutually visible (or see each other).
In the literature, visibility graphs were studied considering various geometric sets such as the vertex set of a simple polygon \cite{o-agta-87}, of a polygon with holes \cite{Wein_voronoi}, a set of points \cite{Cardinal_pointcomplexity}, a set of line segments \cite{Everett_planarsegment}, a set of horizontal line segments \cite{Duchet_planar}, a set of rectangles \cite{Bose_rectVis}, a set of arcs \cite{Hutchinson_arcVis}, along with different visibility models such as line-of-sight visibility \cite{Garey_lineOfSight}, $\alpha$-visibility \cite{Ghodsi_alpha}, $\pi$-visibility \cite{Urrutia_artGalleryAndIllum}, and vertex-edge visibility \cite{ORourke_vertexEdgeVis}.

When we consider the visibilities not only among discrete elements such as vertices of a polygon, but also the \emph{area} which is visible from a point, then we cannot model such a scenario using graphs.
This is simply because there are uncountably infinitely many points in a given finite area, and if we attempt to draw a graph where vertices correspond to those points, then we end up with an infinite graph, which is out of the scope of this thesis.

\section{Intersection graphs} \label{sec:intersection}
In this section, we describe the intersection graphs. 
For the definition of the fundamental graph theory terminology and notations, we refer the reader to Section~\ref{sec:graphs}.

An \emph{intersection graph} is a graph $G = (V,E)$ such that a vertex $u \in V$ is a set, and there exists an edge $uv \in E$ if, and only if the sets $u$ and $v$ intersect.
The elements of these sets might be discrete, e.g. $\mathbb{N}$, $\mathbb{Z}^-$, $\{-23.0001, \sqrt{-13}, \pi,\}$ or continuous, e.g. $\mathbb{R}^+$, $\mathbb{I}$, $S = \{(x,y)\ | \ 0 <x\leq 100,\ -100 < y \leq 0 \}$.
Since a geometric entity (sphere, square, ellipse, line, etc.) is defined by a set of points, a graph which represents the intersection relations among a set of geometric entities is also referred to as an intersection graph.

\subsubsection*{Unit interval graphs}

An \emph{interval} on the number real is defined by a pair $(a,b)$ of real numbers, where $a < b$, and contains all the real numbers lying between $a$ and $b$.
That is, $(a,b) = \{x\in \mathbb{R}\ | \ a < x < b \}$.
If the interval is \emph{closed}, then it is denoted by $[a,b]$, and thus written as  $[a,b] = \{x\in \mathbb{R}\ | \ a \leq x \leq b \}$.

In Figure~\ref{fig:interval}, we see an interval on the real line.
The vertical short lines, labeled $a$ and $b$, denote the beginning and the end of the interval, respectively.
The horizontal long line is the interval itself, and the thin black line is the real line.

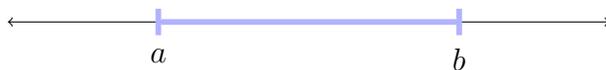
\begin{figure}[htbp]
	\centering
	\begin{tikzpicture}[scale = 2, endpoint/.style = {rectangle, inner sep=0pt, minimum height=10pt, minimum width=0pt, line width=0.75mm, draw}]
		
		\draw[<->] (-2,0) -- (2,0);
		\node[endpoint, draw= blue!30,label=below:$a$] (a) at (-1,0) {};
		\node[endpoint, draw= blue!30,label=below:$b$] (b) at (1,0) {};
		\draw[line width=0.75mm, draw= blue!30] (a)--(b);
	\end{tikzpicture}
	\caption{An interval $(a,b)$ on the real number line.}
	\label{fig:interval}
\end{figure}

The intersection graph of a set of intervals is called an \emph{interval graph}.

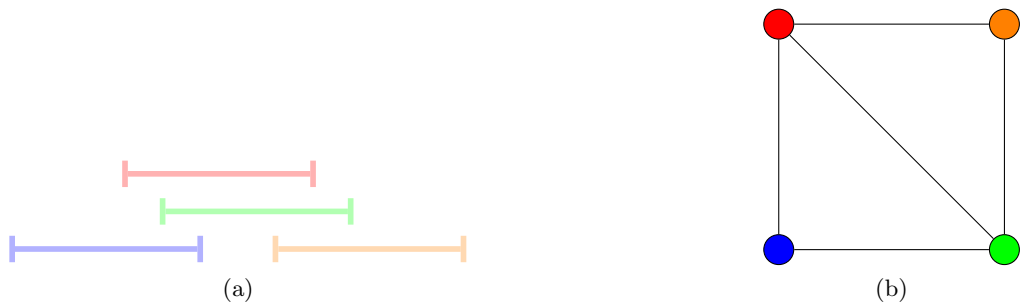
\begin{figure}[htbp]
	\centering
	\subfloat[]{
		
		\begin{tikzpicture}[scale = 2.5, endpoint/.style = {rectangle, inner sep=0pt, minimum height=10pt, minimum width=0pt, line width=0.75mm, draw}]
			
			\node[endpoint, draw= blue!30] (a) at (-0.6,0) {};
			\node[endpoint, draw= blue!30] (b) at (0.4,0) {};
			\node[endpoint, draw= red!30] (c) at (0,0.4) {};
			\node[endpoint, draw= red!30] (d) at (1,0.4) {};
			\node[endpoint, draw= green!30] (e) at (0.2,0.2) {};
			\node[endpoint, draw= green!30] (f) at (1.2,0.2) {};
			\node[endpoint, draw= orange!30] (g) at (0.8,0) {};
			\node[endpoint, draw= orange!30] (h) at (1.8,0) {};
			\draw[line width=0.75mm, draw=blue!30] (a)--(b);
			\draw[line width=0.75mm, draw=red!30] (c)--(d);
			\draw[line width=0.75mm, draw=green!30] (e)--(f);
			\draw[line width=0.75mm, draw=orange!30] (g)--(h);
		\end{tikzpicture}
		
	}
	~\hfill~
	\subfloat[]{
		\begin{tikzpicture}[scale = 0.6]
			\tikzstyle{every node}=[draw, shape=circle, minimum size=6pt,inner sep=4pt];
			\node[fill=red] (blue) at (0,5) {};
			\node[fill=blue] (red) at (0,0) {};
			\node[fill=green] (green) at (5,0) {};
			\node[fill=orange] (orange) at (5,5) {};
			\draw (red)--(blue)--(green)--(red);
			\draw (blue)--(orange)--(green);
		\end{tikzpicture}	
	}
	
	\caption{(a) A set of intervals, and (b) the disk graph that corresponds to the set given in (a).}
	\label{fig:intervalGraph}
\end{figure}

The disks in Figure~\ref{fig:intervalGraph} are of same length.
In this case, we can simply assume that every interval has length 1, and then the graph is called a \emph{unit interval graph}.

\subsubsection*{Unit disk graphs}
A \emph{disk} $\DD$ is a closed region in the Euclidean plane that is defined by a pair of coordinates $(x,y) \in \mathbb{R}^2$ and a radius $r \in \mathbb{R}$.
Formally, it is written as $\DD = \{(x,y) \in \mathbb{R}^2\ |\ (x-a)^2 + (y-b)^2 \leq r^2\}$.
In Figure~\ref{fig:disk}, there is a disk with radius $r$ and centered at $(x,y)$.
The blue area is the \emph{interior} of the disk, the black circle is the \emph{boundary} of the disk, and the line segment of that is drawn from the center to a point on the boundary indicates the radius of $\DD$.

\begin{figure}[htbp]
	\centering
	\begin{tikzpicture}
		\tikzstyle{every node}=[draw, fill=black, shape=circle, minimum size=2pt,inner sep=0pt];
		
		\filldraw[fill=blue!30,draw=black] (0,0) circle (2cm);
		\node[label=120:{$(x,y)$}]  at (0,0) {};
		\draw (0,0)--(1.7320508075,1);
		\node[draw=none,fill=none] at (1,0.3) {$r$};
	\end{tikzpicture}
	\caption{A disk centered at $(x,y)$ with radius $r$.}
	\label{fig:disk}
\end{figure}
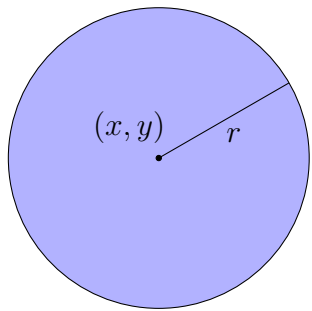

A disk is said to be \emph{open} if the boundary is not an element of the disk.
We always assume that the disks are closed, that is, the disk contains also the boundary.

In Figure~\ref{fig:diskGraph} we see a set of disks, and a graph.
The color of a vertex indicates to which disk it corresponds to.
Two vertices are adjacent if, and only if the corresponding disks intersect.
Such a graph is called a \emph{disk graph}.

\begin{figure}[htbp]
	\centering
	\subfloat[]{
		\begin{tikzpicture}[scale = 0.6]
			
			\filldraw[fill=blue,draw=black,opacity=0.3] (0,0) circle (2cm);
			\filldraw[fill=red,draw=black,opacity=0.3] (0,1.6) circle (2cm);
			\filldraw[fill=green,draw=black,opacity=0.3] (1.5,2.3) circle (2cm);
			\filldraw[fill=orange,draw=black,opacity=0.3] (-1.5,-2.3) circle (2cm);
		\end{tikzpicture}
	}
	~\hfill~
	\subfloat[]{
		\begin{tikzpicture}[scale = 0.6]
			\tikzstyle{every node}=[draw, shape=circle, minimum size=6pt,inner sep=4pt];
			\node[fill=blue] (blue) at (0,5) {};
			\node[fill=red] (red) at (0,0) {};
			\node[fill=green] (green) at (5,0) {};
			\node[fill=orange] (orange) at (5,5) {};
			\draw (red)--(blue)--(green)--(red);
			\draw (blue)--(orange);
		\end{tikzpicture}
	}
	
	\caption{(a) A set of disks in the Euclidean plane, and (b) the disk graph that corresponds to the set given in (a).}
	\label{fig:diskGraph}
\end{figure}
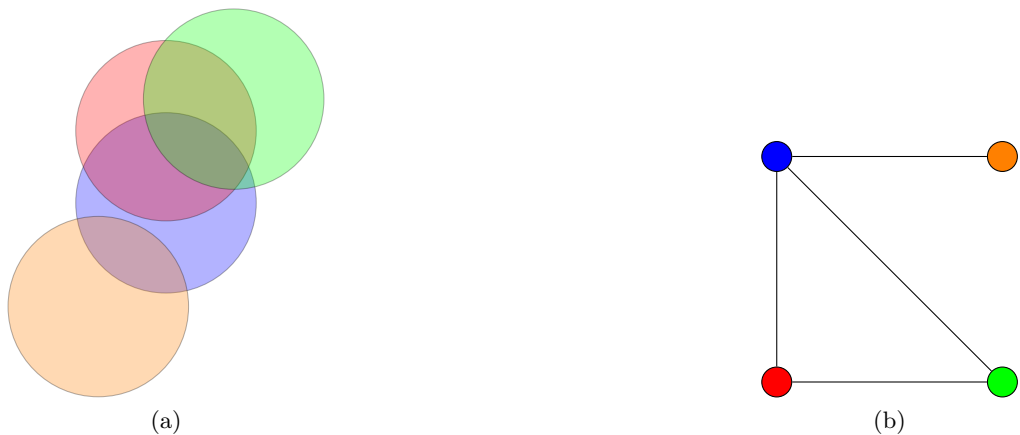

The disks in Figure~\ref{fig:diskGraph} are of same radius.
In this case, regardless of the actual radii, we can simply assume that every disk has radius 1, and then the graph is called a \emph{unit disk graph}.

\subsubsection*{Unit disk visibility graphs}

Unit disk graphs are the intersection graphs of a set of equal radius (or unit) disks in the Euclidean plane. Each vertex in the graph corresponds to a disk, and there exists an edge between a pair of vertices if the disks corresponding to that pair intersect.
Similar to visibility graphs, unit disk graphs have several applications, and they are mainly used to model sensor networks \cite{Clark_UDmaxclique,pathPlanning,rangebased}. 

Transmitters that are embedded into robots, or on sensor nodes form a wireless sensor network that can be modeled using graph data structure.
In the graph of a sensor network, vertices correspond to sensor nodes, and edges correspond to their pairwise communication links.
Note that the sensors have limited sensing range, and a pair of nodes can transmit data to each other if and only if they are inside each others' communication range.

Based on the real world applications, we assume that if a pair of objects (sensors, cameras, guards etc.) are too far from each other, then they do not see each other.
To model this notion, we adapt the unit disk graph model.

Let $\PP$ be a simple polygon, and let $G$ be the visibility graph of $\PP$. 
$G$ is called a \emph{unit disk visibility graph} of $\PP$ if the existence of an edge $uv \in E$ means that $p_u$ and $p_v$ see each other, and the Euclidean distance between them is less than 1 unit.
In other words, an element of the geometric set cannot see another element if they are too far apart.
Note that in our previous definition of the unit disk graph, there exists an edge between two vertices, if the corresponding disks intersect. This means that a distance $\leq 2$ is enough. However, in Chapter~\ref{chap:udvg}, we study this problem with the motivation of wireless sensor networks. 
For two sensor nodes to communicate, they must be inside each other's communication range.
Thus, instead of the radius, we assume that the diameter of a unit disk is one unit, for the sake of accuracy.

\begin{figure}[htbp]
	\centering
	\captionsetup[subfigure]{position=b,justification=centering}
	\subfloat[]{
		\begin{tikzpicture}[scale=3.5]
			\draw (0,0) to[grid with coordinates] (1.5,1);
			
			\tikzstyle{every node}=[draw=black, fill=black, shape=circle, minimum size=3pt,inner sep=0pt];
			
			\node (1) at (0,0) {};
			\node (2) at (0.5,0.6) {};
			\node (3) at (0.8,1) {};
			\node (4) at (1.5,0.5) {};		
			\node (5) at (1.4,0.1) {};
			\node (6) at (0.9,0.2) {};
			
			\tikzstyle{every path}=[red]
			
			\draw (1)--(2)--(3)--(4)--(5);
			\draw (2)--(4);
			
			\foreach \i in {1,...,5} 
			{
				\draw (\i)--(6);
			}
			
			\tikzstyle{every path}=[black, dashed]	
			\draw (1)--(4);
			\draw (1)--(5)--(2);
			\draw (3)--(5);
			\draw (1)--(3);
			
		\end{tikzpicture}
		\label{fig:UDVGa}
	}
	~
	\subfloat[]{
		\begin{tikzpicture}[scale=3.3]
			\draw (0,0) to[grid with coordinates] (1.5,1);
			
			\tikzstyle{every node}=[draw=black, fill=black, shape=circle, minimum size=3pt,inner sep=0pt];
			
			\node (1) at (0,0) {};
			\node (2) at (0.5,0.6) {};
			\node (3) at (0.8,1) {};
			\node (4) at (1.5,0.5) {};		
			\node (5) at (1.4,0.1) {};
			\node (6) at (0.9,0.2) {};
			
			\tikzstyle{every path}=[thick]
			\draw (1)--(2);
			\draw (3)--(6);
			\draw (4)--(5);
			
			\tikzstyle{every path}=[red]
			
			\draw (2)--(3)--(4);

			\foreach \i in {1,2,4,5} 
			{
				\draw (\i)--(6);
			}
			
			\tikzstyle{every path}=[black, dashed]	
			
			\draw (1)--(5);
			\draw (3)--(5);
			\draw (1)--(3);
			
		\end{tikzpicture}
		\label{fig:UDVGb}
	}
	
	\subfloat[]{
		\begin{tikzpicture}[scale=3.5]
			\draw (0,0) to[grid with coordinates] (1.5,1);
			
			\tikzstyle{every node}=[draw=black, fill=black, shape=circle, minimum size=3pt,inner sep=0pt];
			
			\node (1) at (0,0) {};
			\node (2) at (0.5,0.6) {};
			\node (3) at (0.8,1) {};
			\node (4) at (1.5,0.5) {};		
			\node (5) at (1.4,0.1) {};
			\node (6) at (0.9,0.2) {};
			
			\fill[gray, opacity=0.3] (1.center)--(2.center)--(3.center)--(4.center)--(5.center)--(6.center)--(1.center);
			
			\tikzstyle{every path}=[thick]
			
			\draw (1)--(2)--(3)--(4)--(5)--(6)--(1);
			
			\tikzstyle{every path}=[red]
			
			\draw (2)--(4);
			\draw (2)--(6);
			\draw (3)--(6);
			\draw (4)--(6);
			
			\tikzstyle{every path}=[black, dashed]	
			
			\draw (1)--(4);
			\draw (2)--(5);
			\draw (3)--(5);
			
		\end{tikzpicture}
		\label{fig:UDVGc}
	}
	~
	\subfloat[]{
		\begin{tikzpicture}[scale=3.5]
			\draw (0,0) to[grid with coordinates] (1.5,1);
			
			\tikzstyle{every node}=[draw=black, fill=black, shape=circle, minimum size=3pt,inner sep=0pt];
			
			\node (1) at (0,0) {};
			\node (2) at (0.5,0.6) {};
			\node (3) at (0.8,1) {};
			\node (4) at (1.5,0.5) {};		
			\node (5) at (1.4,0.1) {};
			\node (6) at (0.9,0.2) {};
			
			\node (h1) at (0.5,0.4) {};
			\node (h2) at (0.8,0.7) {};
			\node (h3) at (1,0.5) {};
			\node (h4) at (0.7,0.4) {};
			
			\fill[gray, opacity=0.3] (1.center)--(2.center)--(3.center)--(4.center)--(h3.center)--(h2.center)--(h1.center)--(1.center);
			\fill[gray, opacity=0.3] (4.center)--(5.center)--(6.center)--(1.center)--(h1.center)--(h4.center)--(h3.center)--(4.center);		
			
			\tikzstyle{every path}=[thick]
			\draw (h1)--(h2)--(h3)--(h4)--(h1);
			
			\draw (1)--(2)--(3)--(4)--(5)--(6)--(1);
			
			\tikzstyle{every path}=[red]
			
			\draw (1)--(h1);
			\draw (1)--(h4);
			
			\draw (2)--(h1);
			\draw (2)--(h2);
			
			\draw (3)--(h2);
			\draw (3)--(h3);	
			
			\draw (4)--(h2);		
			\draw (4)--(h3);
			\draw (4)--(h4);
			\draw (4)--(6);
			
			\draw (5)--(h3);
			\draw (5)--(h4);
			
			\draw (6)--(h1);
			\draw (6)--(h3);
			\draw (6)--(h4);	
			
			\tikzstyle{every path}=[black, dashed]	
			
			\draw (1)--(4);
			\draw (1)--(h2);
			\draw (3)--(5);

		\end{tikzpicture}
		\label{fig:UDVGd}
	}
	\caption{Unit disk visibility relations of (a) a set of points, (b) a set of line segments, (c) a simple polygon, and (d) a polygon with a hole.}
	
	\label{fig:UDVG}
\end{figure}
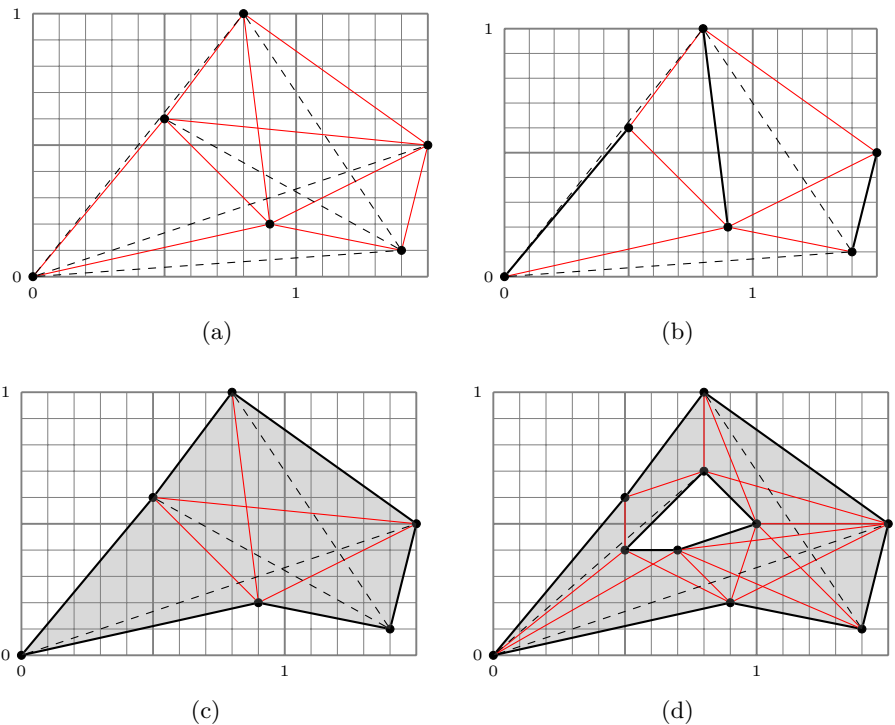

Unit disk point visibility graphs directly follows from this definition while for unit disk segment and polygon visibility graphs, the additional constraints are the followings: $i)$ the edges of unit disk segment visibility graphs cannot to intersect any segment, and $ii)$ the edges of unit disk polygon visibility graphs must totally lie inside the polygon.

\subsubsection*{Axes-parallel unit disk graphs}
\emph{Axes-parallel unit disk graph} is a unit disk graph of a set of disks that are centered on \emph{axes-parallel straight lines}.
The domain of \emph{axes-parallel straight lines} is a set of lines in 2D, where the angle between a pair of lines is either $0$ or $\pi/2$.
This implies that the equation of a straight line is either $y = a$ if it is a horizontal line, or $x = b$ if it is a vertical line, where $a,b \in \mathbb{R}$.

The input for axes-parallel straight lines recognition problem contains two sets, $\mathcal{H},\mathcal{V} \subset \mathbb{R}$, where $\mathcal{H}$ contains the Euclidean distance of each horizontal line from the $x$-axis, and $\mathcal{V}$ contains the Euclidean distance of each vertical line from the $y$-axis.
Thereby in the domain that we use, each vertex is mapped either onto a vertical line, or onto a horizontal line.

We denote the class of axes-parallel unit disk graphs as $\APUD(k,m)$ where $k$ is the number of horizontal lines, and $m$ is the number of vertical lines.

\section{Computational complexity and some combinatorial problems} \label{sec:problems}

\subsubsection{Asymptotic computational complexity}
We first introduce the notation that we use throughout this section.
The set $\Sigma = \{0,1\}$ is called \emph{an alphabet}, and $\Sigma^*$ denotes the set of all strings obtainable from $\Sigma$.
Let $L \subseteq \Sigma^*$ be a \emph{language}, and let $w \in \Sigma^*$ be a \emph{word}.
Throughout this thesis, $\Sigma = \{0,1\}$ always holds, as we do not deal with singleton alphabets.
A \emph{decision problem} $\QQ$ is to determine whether a given word $w \in \Sigma^*$ belongs to a language $L \subseteq \Sigma^*$ over the alphabet $\Sigma$.

Let $\AA$ be an algorithm, $n$ denote the size of the input given to $\AA$, and $f(n)$ be the maximum number of operations executed on any input of size $n$ before $\AA$ terminates.
$\AA$ is called a \emph{polynomial-time algorithm}, if there exists a polynomial $p \in \mathbb{N}_0[X]$ such that, for all $n \in \mathbb{Z}^+$, $f(n) < p(n)$.
A \emph{Turing machine} is an abstract machine which reads and writes on a tape, and is able to simulate the logic of any given algorithm via a finite state automaton.
A Turing machine can be deterministic, i.e. in any given situation, the machine performs at most one pre-determined action, or non-deterministic, i.e. in some situations, the machine might perform one of many actions \cite{computersandintractability}.
This corresponds to the Turing machine having deterministic or non-deterministic controlling automaton.
The notion of \emph{complexity} is defined over the Turing machine.

We now describe four complexity classes which we use throughout this thesis. Namely, P, NP, NP-complete, NP-hard, $\exists\mathbb{R}$, and $\exists\mathbb{R}$-complete.
Let $\QQ$ be a decision problem, and $Q$ be an instance of $\QQ$.

If $Q$ can be correctly decided via a polynomial-time \emph{deterministic Turing machine} algorithm,
then we say that $\QQ$ is a \emph{poly\-no\-mial\-time problem}, and we write ``$\QQ$ is in P'' or ``$\QQ \in$ P.''
In other words, the complexity class P is the set of all decision problems that can be solved in polynomial time.

If, on the other hand, $Q$ can be correctly decided via a polynomial-time \emph{non-deterministic Turing machine} algorithm,
then we say that $\QQ$ is a \emph{non-deterministic polynomial-time problem}, or ``$\QQ$ is an \emph{NP problem},'' in short.
This is equivalent to saying that a given solution to a given instance of any NP problem can be verified in polynomial time.

The complexity class NP-hard is the set $\mathfrak{Q}$ of decision problem $\QQ$ such that there exists a polynomial time reduction from any other NP problem $\QQ'$ to the problem $\QQ$.
Precisely, $\forall \QQ \in \mathfrak{Q}$ and $\forall \QQ' \in \text{NP}$, there exists a polynomial-time reduction from $\QQ'$ to $\QQ$.
This means that if we had a polynomial-time solution to some NP-hard problem, then we would be able to solve any problem in NP in polynomial-time.

A problem $\QQ$ is NP-complete, if $\QQ$ is an NP-hard problem, and is in NP.

Another complexity class, $\exists\mathbb{R}$ (read as ``existential reals'') is defined as follows.
A problem $\QQ$ is referred to as an \emph{$\exists\mathbb{R}$ problem} if any given instance $Q$ of $\QQ$ can be reduced in polynomial time to deciding whether
a given multivariate polynomial has a solution over the real numbers.
Similar to the definition of NP-completeness, a problem $\QQ$ is called an \emph{$\exists\mathbb{R}$-complete problem} if it is an $\exists\mathbb{R}$ problem, and there exists a polynomial-time reduction to an instance of one of  $\exists\mathbb{R}$ problems.
It is still an open problem whether $\exists\mathbb{R}$ problems are in NP.

From this point on, whenever we write ``polynomial-time algorithm,'' we mean ``polynomial-time deterministic Turing machine algorithm.''
Analogously, an \emph{exponential-time algorithm} refers to an algorithm which terminates after executing at most exponentially many operations on a deterministic Turing machine.

We should also note that although the complexity zoo has way more classes,
we only deal with three complexity classes in this thesis: P, NP, and $\exists\mathbb{R}$.

\subsubsection{Proper coloring and chromatic number}
\emph{Assigning a color to a vertex} $v$ means labeling $v$ with a label $\ell \in \mathcal{L}$  from a given set $\mathcal{L}$ of labels.
Conventionally, colors are preferred over labels.
For the sake of simplicity and readability, we also accommodate ourselves to this convention and describe our results using colors instead of labels.
A vertex $v$ is said to have a \emph{unique color} with respect to a set $U$ of vertices when $c(v) \neq c(u)\ \forall u \in U \setminus \{v\}$.

Let $c(v)$ denote the color of a vertex $v$ in a graph $G = (V,E)$.
A proper coloring of a graph is assigning colors to every vertex $v \in V$ in a graph $G = (V,E)$ such that for a pair $u,v \in V$, $c(u) = c(v)$ holds if, and only if $uv \not\in E$.
In other words, every vertex has a unique color among its closed neighborhood.
Chromatic number $\chi(G)$ of a graph $G$ is the minimum number of colors required to obtain a proper coloring of $G$.
The problem of determining the chromatic number, i.e. chromatic number problem, is an NP-complete problem \cite{chromaticNumber}.
Instead of finding the minimum number of colors to properly color a graph, if we ask ``Can a given graph $G$ be properly colored by $k$ colors?'' for some constant $k$, then this problem is called the  $k$-coloring problem.
For $k=1$ and $k=2$, the problem is trivial: a graph $G$ cannot be 1-colored if there exists a pair of adjacent vertices, and can be 2-colored if, and only if $G$ is a bipartite graph.
However, except a few graph classes, $k$-coloring problem remains NP-complete for any value of $k \geq 3$.
Note that if a graph is $k$ colorable, then it is always $k+c$ colorable for $c \in \mathbb{N}$.
Thus, if a graph is 3-colorable, it is also 4-colorable, but the converse does not apply.

\subsubsection{Conflict-free coloring}
A conflict-free coloring of a graph $G$ is assigning colors to some subset $U \subseteq V$ of vertices such that every vertex $v$ has at least one unique color (i.e. a color that does not repeat) in its closed neighborhood $N[v]$.

For a vertex $v$, if there exists no vertex $u \in N[v]$ such that color of $u$ is unique among $N[v] \setminus \{u\}$, then $v$ has a \emph{conflict}.
Note that $u=v$ might hold here due to the definition of a closed neighborhood.
Note that a proper coloring of a graph is also a conflict-free coloring of that graph. 

Similar to proper coloring problem, conflict-free coloring is also NP-complete for a large set of graph classes.
In the most general case, determining whether there exists a conflict-free coloring with only one color in NP-complete.
Following the same notation, we write that conflict-free chromatic number ${\chi_{CF}}(G)$ of  graph $G$ is the minimum number of colors needed to obtain a conflict-free coloring of $G$.

The conflict-free coloring problem was first studied by Biggs with the name \emph{perfect code}, which is essentially conflict-free coloring of a graph using only one color \cite{biggs-1973,DBLP:conf/mfcs/KratochvilK88}.
Later on, this topic arouse interest on polygon visibility graphs when the field of robotics became widespread \cite{chazelle87a,gmr-sjc-97}.

\subsubsection{Polygon guarding}
A point $g$ in a polygon $\mathcal{P}$ is said to be \emph{guarding} a set of points (or an area) $Q$ inside $\mathcal{P}$ if for every $q \in Q$, $g$ and $q$ see each other.
The point $g$, which is chosen as a guard, might be an arbitrary point inside $\mathcal{P}$ as well as a vertex of $\mathcal{P}$.
In this thesis, we assume that our ``guards'' are only placed on the vertices of a polygon.

A polygon $\mathcal{P}$ is \emph{guarded} if, and only if there exists a subset $\Gamma = \{g_1, g_2, \dots, g_k\}$ of vertices of $\mathcal{P}$ are selected as guards, and $\bigcup_{g \in \Gamma} \Sigma(G)$ covers the whole interior of $\mathcal{P}$.

Instead of guarding the whole interior, if our goal is to guard every vertex in a polygon, then we have a different problem.
The main difference is that the former one cannot be expressed as a simple graph, while the latter one can be, since the vertices of a graph are discrete geometric entities.
Thus, if the goal is to guard all the vertices, then a correct solution is referred to as \emph{vertex-to-vertex guarding}.
On the other hand, if the goal is to guard the whole polygon, then a correct solution is referred to as \emph{vertex-to-point}.

One of the best known problems  in computational geometry, the \emph{art gallery problem}, is essentially a vertex-to-point guarding problem \cite{bcko-cgaa-08,o-agta-87}.
The problem is to find the minimum number of guards to guard an art gallery, which is modeled by an $n$-vertex polygon.
This problem  was shown to be NP-hard by Lee and Lin \cite{ll-ccagp-86} and more recently $\exists \mathbb{R}$-complete by Abrahamsen et al. \cite{art-gal-etr}. 
The Art Gallery Theorem, proved by Chv{\'a}tal, shows that $ \lfloor n/3 \rfloor$ guards are sufficient and sometimes necessary to guard a simple polygon \cite{c-actpg-75}.

The guard minimization problem has been studied under many constraints; such as the placement of guards being restricted to the polygonal perimeter or vertices \cite{loglogn-artgal},
the viewers being restricted to vertices, the polygon being terrains \cite{saurabh-artgal,ben-moshe-terrain,terrain-gal}, weakly visible from an edge \cite{weakvis-guard}, with holes or orthogonal \cite{eidenbenz-inapprox,katz-orth-artgal,terrain-np-complete}, with respect to parameterization \cite{para-artgal}, or with respect to approximability \cite{katz-wvp}. 	
For most of these cases the problem remains hard, but interesting approximation algorithms have also been provided \cite{g-apatpp-2010,approx-art-gal-bonnet}.

\subsubsection{Conflict-free guarding of a polygon}

Let $\mathcal{P}$ be a polygon, and $\Gamma = \{g_1, g_2, \dots, g_k\}$ be the guard set which guard $\mathcal{P}$ as a whole.
Suppose that the guards are assigned colors, and when a point $o$ in $\mathcal{P}$ is guarded by some guard $g_i \in \Gamma$, we write ``$o$ sees the color of $g_i$.''

That is, every guard receives a color, and every point inside the polygon must be guarded by at least one unique color.
In addition to above mentioned versions of the art gallery problem (or rather polygon guarding problem), some problems consider not the number of the guards, but the number of colors that are assigned to the guards.
The colors, depending on the scope, determine the types of the guards.
If any observer in the polygon sees at least one guard with a different type, then that polygon has a \emph{conflict-free chromatic guarding} \cite{suri-conflict,Bartschi-2014,DBLP:journals/comgeo/HoffmannKSVW18}.
If every guard that sees any given observer is of different type, then that polygon has a \emph{strong chromatic guarding} \cite{strong-conflict-free}.

\subsubsection{Graph recognition}

Given a graph class $\CC$, and a graph $G$, the \emph{graph recognition problem} is to determine whether $G$ belongs to the class $\CC$.
The class $\CC$ can be any subset of graphs, e.g. polygon visibility graph, bipartite graph, interval grah, unit disk graph, etc.
In our thesis, we deal with a variant of the \emph{unit disk graph recognition} problem.

Analogous to the definition, the \emph{unit disk graph recognition problem} is about deciding whether a given graph $G = (V,E)$ is a unit disk graph.
That is, determining whether there exists a mapping $\Sigma: V \to (\mathbb{R} \times \mathbb{R})$, such that each vertex is the center of a unit disk respecting the intersection property given by the edges of $G$.
The mapping $\Sigma$ is also called the \emph{embedding of $G$ by unit disks}.
This problem is an $\exists\mathbb{R}$-complete problem \cite{sphereAndDotProduct,integerRealization}.

We, therefore, attack this problem by restricting the solution space for the disk centers.
In the conventional setting, the centers can be anywhere in the Euclidean plane.
We introduce the problem called \emph{axes-parallel unit disk graphs recognition}.
In this problem, we use the domain of \emph{axes-parallel straight lines} which is a set of lines in 2D, where the angle between a pair of lines is either $0$ or $\pi/2$.
This implies that the equation of a straight line is either $y = a$ if it is a horizontal line, or $x = b$ if it is a vertical line, where $a,b \in \mathbb{R}$.

The input for axes-parallel straight lines recognition problem contains two sets, $\mathcal{H},\mathcal{V} \subset \mathbb{R}$, where $\mathcal{H}$ contains the Euclidean distance of each horizontal line from the $x$-axis,
and $\mathcal{V}$ contains the Euclidean distance of each vertical line from the $y$-axis.
Thereby in the domain that we use, each vertex is mapped either onto a vertical line, or onto a horizontal line.
We denote the class of axes-parallel unit disk graphs as $\APUD(k,m)$ where $k$ is the number of horizontal lines, and $m$ is the number of vertical lines.
Formally, the input is a graph $G = (V,E)$, where $V = \{1,2,\dots,n\}$, and two sets $\mathcal{H}, \mathcal{V} \subset \mathbb{Q}$ of rational numbers with $|\mathcal{H}| = k$ and $|\mathcal{V}| = m$.
The task is to determine whether there exists a mapping $\Sigma: V \to  (\mathbb{R} \times \mathcal{H}) \cup (\mathcal{V} \times \mathbb{R})$ which is a unit disk realization of $G$.

\chapter{Restricted Types of Polygons} \label{chap:restrictedpolygons}

\section{Summary of the chapter}
In this chapter, study guarding problems mentioned in Section~\ref{sec:problems} on two (restricted) types of polygons, namely the funnels and the weak visibility polygons (defined in Section~\ref{sec:problems}).
We study the relation between those two types of polygons, and we describe our results concerning combinatorial problems on funnels and weak-visibility polygons.

\begin{itemize}
	\item We describe an algorithm which returns the minimum number of guards for a funnel to be fully guarded (Algorithm~\ref{alg:newminfunguard}).
	\item We show that the very same guard set can be used to obtain a conflict-free chromatic guarding of the funnel.
	\item We describe an approximation algorithm for vertex-to-point conflict-free chromatic guarding of a funnel, with only a constant ($+4$) additive error (Algorithm~\ref{alg:apxcfreefunnel}).
	\item We thus prove a direct two-way relation between the optimal number of guards and optimal number of colors needed in a funnel (Section~\ref{sec:cfguardingfunnel}).
	\item By generalizing the previously listed results, we show that every weak visibility polygon on $n$ vertices can be vertex-to-point conflict-free chromatic guarded with only $O (\log^2 n)$ guards (Algorithm~\ref{alg:weakv2pcoloringB}).	
\end{itemize}

Our results concerning the general visibility polygons is in Chapter~\ref{chap:generalpolygons}.

\section{Related work}

Recall that the \emph{guarding} of a polygon is placing ``guards'' into the polygon, in a way that the guards collectively can see the whole polygon.
It is usually assumed that a guard can see any point unless there is an obstacle or a wall between the guard and that point.

In this chapter, we consider the art gallery problem with conflict free guarding property.
In addition to above mentioned versions of art gallery problem (or rather polygon guarding problem), some problems consider not the number of the guards, but the number of colors that are assigned to the guards.
The colors, depending on the scope, determine the types of the guards.
If any observer in the polygon sees at least one guard with a different type, then that polygon has a \emph{conflict-free chromatic guarding} \cite{suri-conflict,Bartschi-2014,DBLP:journals/comgeo/HoffmannKSVW18}.
If every guard that any given observer is of different type, then that polygon has a \emph{strong chromatic guarding} \cite{strong-conflict-free}. 

In our thesis, we focus on the theoretical aspects of this problem.
However, we would like to describe a real-world scenario which, we think, might motivate the reader to both follow the manuscript and improve the current state of art for the sake of potential technological advancements.

Consider a scenario where a mobile robot traverses a room from one point to another, communicating with the wireless sensors placed on the corners of the room.
Even if the robot has full access to the map of the room, it cannot determine its location precisely because of accumulating rounding errors \cite{motionPlanningCourse}.
And thus it needs clear markings in the room to guide itself to the end point in an energy efficient way.
To guide a mobile robot with wireless sensors, two properties must be satisfied.
First one is, no matter where the robot is in the polygon, it should hear from at least one sensor. That is, the placed sensors must together \emph{guard} the whole room and leave no place uncovered.
The second one is, if the robot hears from several sensors, there must be at least one sensor broadcasting with a frequency that is not reused by some other sensor in the range. That is, the sensors must have \emph{conflict-free} frequencies.
If these two properties are satisfied, then the robot can guide itself using the deployed wireless sensors as landmarks.
This problem is also closely related to frequency assignment problem in wireless networks \cite{freqAssignment,suri-conflict}.
One can easily solve this problem by placing a sensor at each corner of the room, and assigning a different frequency to each sensor.
However, as we have mentioned in Chapter~\ref{chap:introduction}, this method becomes very expensive as the number of sensors grow \cite{freqAssignment,cf-app}.
Therefore, the main goal in this problem is minimize the number of different frequencies assigned to sensors.
Since the cost of a sensor is comparatively very low, we do not aim to minimize the number of sensors used.
This notion justifies the conflict-free coloring version of the problem.

The above scenario is geometrically modeled as follows.
The room is a simple polygon with $n$ vertices.
There are $m$ sensors placed in the polygon (usually on some of its vertices),
and two different sensors are given two different colors if, and only if they broadcast in different frequencies.

\section{Funnels} \label{sec:funnels}

In this section, we focus on a special interesting type of polygons -- {\em funnels} (defined in Section~\ref{sec:geometry}).
A polygon $\PP$ is a funnel if, and only if precisely three of the vertices of $\PP$ are convex, and two of the convex vertices share one common edge -- the {\em base} of the funnel~$\PP$.
We denote a funnel by $\FF$ throughout this section.

We use some special notation here. See Figure~\ref{fig:tangents}.
Let the given {funnel} be $\FF$, oriented in the plane as follows. 
On the bottom, there is the horizontal \emph{base} of the funnel -- the line segment $\lseg{l_1}{r_1}$ in the picture.
The topmost vertex of $\FF$ is called the \emph{apex}, and it is denoted by $\alpha$.
There always exists a point $x$ on the base which can see the apex $\alpha$, and then $x$ sees the whole funnel at once.
The vertices on the left side of apex form the \emph{left concave chain}, and analogously, the vertices on the right side of the apex form the 				\emph{right concave chain} of the funnel.
These left and right concave chains are denoted by $\mathcal{L}$ and $\mathcal{R}$ respectively.
We denote the vertices of $\mathcal{L}$ as $l_1, l_2, \ldots, l_k$ from bottom to top.
We denote the vertices of $\mathcal{R}$ as $r_1, r_2, \ldots, r_m$ from bottom to top.
Hence, the apex is $l_k = r_m = \alpha$.

\subsection{Guarding a funnel}

\begin{figure}[tbp]
	\centering
	\begin{tikzpicture}[scale=0.42]
		
		\tikzstyle{every node}=[draw, shape=circle, minimum size=5pt,inner sep=0pt];
		\node[label=$l_1$] (A) at (1.93,4.71) {};
		\node[label=$l_2$,fill=red] (B) at (5,5) {};
		\node[label=$l_3$] (C) at (8.72,6.76) {};
		\node[label=$l_4$] (D) at (10.61,8.21) {};
		\node[label=120:$l_5$] (E) at (13.01,10.71) {};
		\node[label=left:$l_6$] (F) at (14.37,12.5) {};
		\node[label=$\alpha$] (G) at (16.25,16.99) {};
		\node[label=30:$r_7$] (H) at (16.83,14.37) {};
		\node[label=30:$r_6$] (I) at (17.33,13.05) {};
		\node[label=right:$r_5$] (J) at (19.22,10.8) {};
		\node[label=$r_4$] (K) at (22.1,8.49) {};
		\node[label=$r_3$,fill=blue] (L) at (23.5,7.58) {};
		\node[label=$r_2$] (M) at (28,5.1) {};
		\node[label=$r_1$] (N) at (29,4.71) {};
		
		\draw[thick](A)--(B)--(C)--(D)--(E)--(F)--(G)--(H)--(I)--(J)--(K)--(L)--(M)--(N)--(A);
		
		\coordinate (A') at (24.71,6.89);
		\node[minimum size = 3pt, fill=black, label=40:$~p$] (B') at (18.63, 11.46) {};
		\coordinate (C') at (17.25, 13.27);
		\node[minimum size = 3pt, fill=black, label=150:$q~$] (M') at (14.7,13.25) {};
		
		\draw[thick](A)--(B)--(C)--(D)--(E)--(F)--(G)--(H)--(I)--(J)--(K)--(L)--(M)--(N)--(A);
		\draw[very thick, color=red, dashed] (B)--(B');
		\draw[very thick, color=blue, dashed] (L)--(M');
		\node[minimum size = 3pt, fill=black, label=270:$t$] (B'M') at (18,11.15) {};
	\end{tikzpicture}
	
	\caption{A funnel $\FF$ with seven vertices in $\mathcal{L}$ labeled $l_1,
		\ldots, l_7$ from bottom to top, and eight vertices in $\mathcal{R}$ 
		labeled $r_1, \ldots, r_8$, including the apex $\alpha = l_7 = r_8$.
		The picture also shows
		the upper tangent of the vertex $l_2$ of $\mathcal L$ (drawn in
		dashed red), the upper tangent of the vertex $r_3$ of $\mathcal R$
		(drawn in dashed blue), and their intersection~$t$.}
	\label{fig:tangents}
\end{figure}
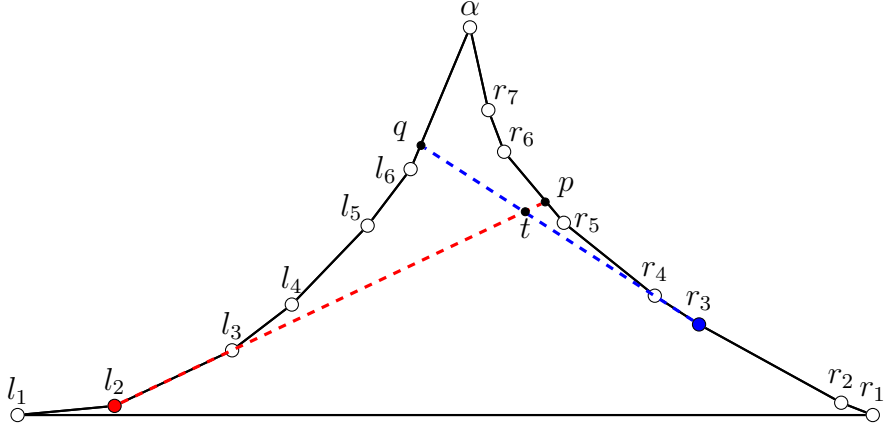

We first consider the problem of minimizing the total number of vertex guards needed to guard all points of a funnel.	
We start by describing a simple procedure (Algorithm~\ref{alg:tightpath}) that provides us with a guard set which may not always be optimal (but very close to the optimum, see Corollary~\ref{cor:alg1-near-optimal}). 
This procedure will be helpful for the subsequent coloring results.
Then, we also refine the simple procedure to compute the optimal number of guards 
in Algorithm~\ref{alg:newminfunguard}.

Let $l_i$ be a vertex on $\mathcal{L}$ which is not the apex. 
We define the \emph{upper tangent} of $l_i$, denoted by $\uptan(l_i)$, as the ray whose origin is $l_i$ and which passes through $l_{i+1}$.
Upper tangents for vertices on $\mathcal{R}$ are defined analogously. 
Let $p$ be the point of intersection of $\mathcal{R}$ and the upper tangent of~$l_i$.
Then we define $\upseg(l_i)$ as the line segment $\lseg{l_{i+1}}p$. 
For the vertices of $\mathcal{R}$, $\upseg$ is defined analogously: 
if $q$ is the point of intersection of $\mathcal{L}$ and the upper tangent of
$r_j \in \mathcal{R}$, then let $\upseg(r_j) := \lseg{r_{j+1}}q$.
See again Figure~\ref{fig:tangents}.

\begin{algorithm}[tbp]
	\KwIn{A funnel $\FF$ with concave chains $\mathcal{L}=(l_1, 
		\ldots, l_k)$ and $\mathcal{R}=(r_1, \ldots, r_m)$.}
	\KwOut{A vertex set guarding all the points of $\FF$.}
	\smallskip
	
	Initialize an auxiliary digraph $G$ with two dummy vertices $x$ and $y$,
	\mbox{and declare $\upseg(x) = \lseg{l_1}{r_1}$}\;
	Initialize $\ES{S} \gets \{x\}$\;
	\While {$\ES{S}$ is not empty}
	{
		Choose an arbitrary $t \in \ES{S}$, and remove $t$ from $\ES{S}$\;
		Let $s=\upseg(t)$	\tcc*{$s$ is a segment inside $\FF$}
		Let $q$ and $p$ be the ends of $s$ on $\mathcal{L}$ and $\mathcal{R}$, respectively\;
		Let $i$ and $j$ be the largest indices such that $l_i$ and
		$r_j$ are not above $q$ and~$p$, resp.\;
		\lIf{$l_{i+1}$ can see whole $s$}{$i' \gets i+1$}
		\lElse(\tcc*[f]{the topmost vertex on the left seeing whole $s$})
		{$i' \gets i$}
		\lIf{$r_{j+1}$ can see whole $s$}{$j' \gets j+1$}
		\lElse(\tcc*[f]{the topmost vertex on the right seeing whole $s$})
		{$j' \gets j$}
		Include the vertices $l_{i'}$ and $r_{j'}$ in~$G$\;
		\ForEach{$z\in\{l_{i'},r_{j'}\}$}
		{\label{lin:zin}%
			Add the directed edge $(t,z)$ to~$G$ \label{lin:wght}\;
			\If{segment $\upseg(z)$ includes the apex $l_k=r_m$}
			{Add the directed edge $(z,y)$ to~$G$
				\tcc*{$y$ is the dummy vertex}}
			\lElse(\tcc*[f]{more guards are needed above~$z$})
			{$\ES{S}\gets \ES{S}\cup\{z\}$}
		}
	}
	Enumerate a shortest path from $x$ to $y$ in~$G$\; 
	Output the shortest path vertices without $x$ and $y$ as the required guard set\;
	\caption{Simple vertex-to-point guarding of funnels (uncolored).}
	\label{alg:tightpath}
\end{algorithm}

The underlying idea of Algorithm~\ref{alg:tightpath} is as follows.
Imagine we proceed bottom-up when building the guard set of a funnel $\FF$.
Then the next guard is placed at the top-most vertex $z$ of $\FF$,
nondeterministically choosing between $z$ on the left and the right chain of
$\FF$, such that no ``unguarded gap'' remains below~$z$.
Note that the unguarded region of $\FF$ after placing a guard at $z$
is bounded from below by~$\upseg(z)$.
The nondeterministic choice of the next guard~$z$ is encoded within a digraph,
in which we then find the desired guard set as a shortest path.
The following is straightforward.

\begin{lem}\label{lem:tightpath-feas}
	Algorithm~\ref{alg:tightpath} runs in polynomial time, and it outputs
	a feasible guard set for all the points of a funnel $\FF$.
\end{lem}
\begin{proof}
	As for the runtime, we observe that the number of considered line segments
	$s$ in the algorithm is, by the definition of $\upseg$, 
	bounded by at most $k+m$ (and it is typically much lower than this bound).
	Each considered segment $\upseg(t)$ of $t \in \ES{S}$ is processed at most once, and it contributes two edges to $G$.
	Overall, a shortest path in $G$ is found in linear time.
	
	We prove feasibility of the output set by induction.
	Let $(x=x_0,x_1,\ldots,x_{a-1},$ $x_a=y)$ be a path in~$G$.
	We claim that, for $0\leq i\leq a$, guards placed at $x_0,x_1,\ldots,x_i$ guard all the points of $\FF$ below $\upseg(x_i)$. 
	This is trivial for $i=0$, and it straightforwardly follows by induction:
	Algorithm~\ref{alg:tightpath} asserts that $x_i$ can see whole $\upseg(x_{i-1})$, and $x_i$ hence also sees the strip between $\upseg(x_{i-1})$ 			and $\upseg(x_{i})$ by basic properties of a funnel.
	Finally, at $x_{a-1}$, we guard whole $\FF$ up to its apex.
\end{proof}

\begin{rem}
	Unfortunately, the guard set produced by Algorithm~\ref{alg:tightpath} may not be optimal under certain circumstances.
	See the example in Figure~\ref{fig:funnelguards}; the algorithm picks the four red vertices, but the funnel can be guarded by the three green 				vertices. 
	Nevertheless, this (possibly non-optimal) guard set will be very useful in the next section in the context of conflict-free coloring.
\end{rem}

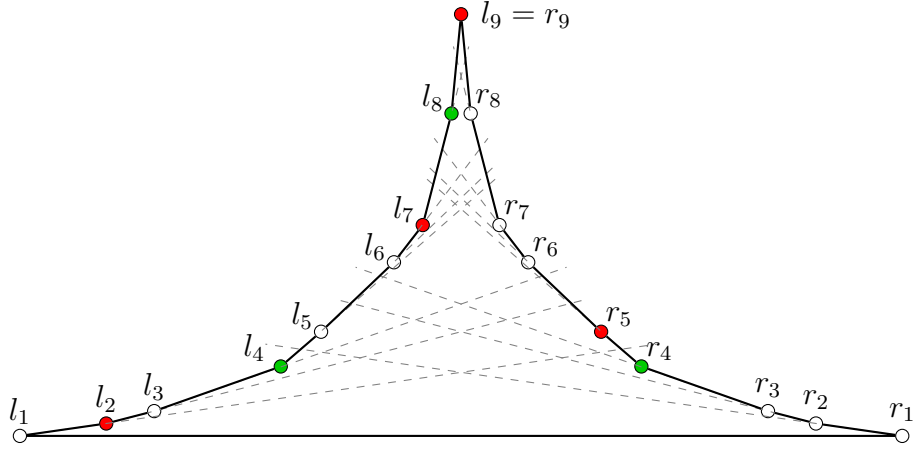
\begin{figure}[tbp]
	\centering
	\begin{tikzpicture}[xscale=0.255, yscale= 0.33]
		
		\tikzstyle{every node}=[draw, shape=circle, minimum size=5pt,inner sep=0pt];
		\node[label=$l_1$] (A) at (-23,0) {};
		\node[fill=red, label=$l_2$] (B) at (-18.5,0.5) {};
		\node[label=$l_3$] (C) at (-16,1) {};
		\node[fill=green!80!black, label=150:$l_4~$] (D) at (-9.4,2.8) {};
		\node[label=150:$l_5$] (E) at (-7.3,4.2) {};
		\node[label=150:$l_6$] (F) at (-3.5,7) {};
		\node[fill=red, label=150:$l_7$] (G) at (-2,8.5) {};
		\node[fill=green!80!black, label=150:$l_8$] (H) at (-0.5,13) {};
		\node[fill=red, label=right:{$~l_9=r_9$}] (I) at (0,17) {};
		
		\node[label=$r_1$] (Q) at (23,0) {};
		\node[label=$r_2$] (P) at (18.5,0.5) {};
		\node[label=$r_3$] (O) at (16,1) {};
		\node[fill=green!80!black, label=30:$r_4$] (N) at (9.4,2.8) {};
		\node[fill=red, label=30:$r_5$] (M) at (7.3,4.2) {};
		\node[label=30:$r_6$] (L) at (3.5,7) {};
		\node[label=30:$r_7$] (K) at (2,8.5) {};
		\node[label=30:$r_8$] (J) at (0.5,13) {};
		
		\coordinate (A') at (10.2,3.7);
		\coordinate (Q') at (-10.2,3.7);
		
		\coordinate (B') at (6.5,5.5);
		\coordinate (P') at (-6.5,5.5);
		\coordinate (C') at (5.5,6.8);
		\coordinate (O') at (-5.5,6.8);
		
		\coordinate (D') at (2,10.5);
		\coordinate (N') at (-2,10.5);
		
		\coordinate (E') at (1.9, 11);
		\coordinate (M') at (-1.9, 11);
		\coordinate (F') at (1.4, 12);
		\coordinate (L') at (-1.4, 12);
		
		\coordinate (G') at (0.4,15.7);
		\coordinate (K') at (-0.4,15.7);
		
		\coordinate (H') at (3.61, 5.6);
		\coordinate (J') at (-3.61, 5.6);
		
		\draw[color=gray, dashed] (A)--(A');
		\draw[color=gray, dashed] (Q)--(Q');
		\draw[color=gray, dashed] (D)--(D');
		\draw[color=gray, dashed] (N)--(N');
		\draw[color=gray, dashed] (E)--(E');
		\draw[color=gray, dashed] (F)--(F');
		\draw[color=gray, dashed] (M)--(M');
		\draw[color=gray, dashed] (L)--(L');
		\draw[color=gray, dashed] (B)--(B');
		\draw[color=gray, dashed] (C)--(C');
		\draw[color=gray, dashed] (P)--(P');
		\draw[color=gray, dashed] (O)--(O');
		\draw[color=gray, dashed] (G)--(G');
		\draw[color=gray, dashed] (K)--(K');
		
		\draw[thick]  (A)--(B)--(C)--(D)--(E)--(F)--(G)--(H)--(I)--(J)--(K)--(L)--(M)--(N)--(O)--(P)--(Q)--(A);
		
	\end{tikzpicture}
	\caption{A symmetric funnel with $17$ vertices. 
		The gray dashed lines show the upper tangents of the vertices.
		It is easy to see that Algorithm~\ref{alg:tightpath} selects $4$ guards,
		up to symmetry, at $l_2,r_5,l_7,l_9$ (the red vertices).
		However, the whole funnel can be guarded by three guards
		at $l_4,r_4,l_8$ (the green vertices), and it will be the task of
		Algorithm~\ref{alg:newminfunguard} to consider such better possibility.}
	\label{fig:funnelguards}
\end{figure}

For the sake of completeness, we now refine the simple approach of Algorithm~\ref{alg:tightpath} to always produce a minimum size guard set.
Recall that Algorithm~\ref{alg:tightpath} always places one next guard based on the position of the previous one guard.
Our refinement is going to consider also pairs of guards (one from the left and one from the right chain) in the procedure. 
We correspondingly extend the definition of $\upseg$ to pairs of vertices as follows.
Let $l_i$ and $r_j$ be vertices of $\FF$ on $\mathcal{L}$ and $\mathcal{R}$, respectively, such that $\upseg(l_i)=\lseg{l_{i+1}}p$ intersects 
$\upseg(r_j)=\lseg{r_{j+1}}q$ in a point~$t$ (see in Figure~\ref{fig:tangents}).
Then we set $\upseg(l_i,r_j)$ as the polygonal line (``$\vee$-shape'') $\lseg pt\cup\lseg qt$.
In case that $\upseg(l_i)\cap\upseg(r_j)=\emptyset$, we simply define
$\upseg(l_i,r_j)$ as the upper one of $\upseg(l_i)$ and	$\upseg(r_j)$.

Algorithm~\ref{alg:newminfunguard}, informally saying,
enriches the two nondeterministic choices of placing the next guard
in Algorithm~\ref{alg:tightpath} with a third choice;
placing a suitable top-most pair of guards $z=(z_1,z_2)$, 
$z_1\in\mathcal{L}$ and $z_2\in\mathcal{R}$,
such that again no ``unguarded gap'' remains below $(z_1,z_2)$.
Figure~\ref{fig:funnelguards} features a funnel in which placing such a pair
of guards $(z_1=l_4,\>z_2=r_4)$ may be strictly better than using any two
consecutive steps of Algorithm~\ref{alg:tightpath}.
On the other hand, we can show that there is no better possibility than one of these three considered steps.
Within the scope of Algorithm~\ref{alg:newminfunguard} (cf.~line~\ref{line:7}),
we extend the definition range of $\upseg(\cdot)$
to include all boundary points of $\mathcal{L}$ and $\mathcal{R}$, as follows.
If $p$ is an internal point of $\lseg{l_i}{l_{i+1}}$, then we set
$\upseg(p):=\upseg(l_i)$.
If $p'$ is an internal point of $\lseg{r_j}{r_{j+1}}$, then we set
$\upseg(p'):=\upseg(r_j)$.

\begin{small}
	\begin{algorithm}[H]
		\KwIn{A funnel $\FF$ with concave chains $\mathcal{L}=(l_1, l_2,
			\ldots, l_k)$ and $\mathcal{R}=(r_1, \ldots, r_m)$.}
		\KwOut{A minimum vertex set guarding all the points of $\FF$.}
		\smallskip
		Initialize an auxiliary digraph $G$ with two dummy vertices $x$ and $y$,
		\mbox{and declare $\upseg(x) = \lseg{l_1}{r_1}$}\;
		Initialize $\ES{S} \gets \{x\}$\;
		%\tcc*{$\ES{S}$ will be a set of segments and $\vee$-shapes}
		\While {$\ES{S}$ is not empty}
		{
			Choose an arbitrary $t \in \ES{S}$, and remove $t$ from $\ES{S}$\;
			Let $s=\upseg(t)$
			\tcc*{$s$ is a segment or a $\vee$-shape}
			Let $i'$ and $j'$ be defined for $s$ as in Algorithm~\ref{alg:tightpath}\;
			
			Let $q$ and $p$ be the ends of $s$ on $\mathcal{L}$ and $\mathcal{R}$, respectively\;
			Let $i''$ and $j''$ be the largest indices such that
			\mbox{$l_{i''}$ lies strictly below $\upseg(p)$ and 
				$r_{j''}$ strictly below $\upseg(q)$}\;
			\label{line:7}
			\tcc*[f]{then $l_{i''}$ and $r_{j''}$ together can see whole $s$}
			
			Include the vertices $l_{i'}$, $r_{j'}$ and $(l_{i''},r_{j''})$ in~$G$\;
			%\ForEach{$z\in\{l_{i'},r_{j'},(l_{i''},r_{j''})\}$}
			% and $t\in\upseg^{-1}(s)$}
		%\tcc*[f]{$\upseg^{-1}(s)$ is actually a singleton element	for every $s\in\ES{S}$}
		
		\ForEach{$z\in\{l_{i'},r_{j'},(l_{i''},r_{j''})\}$}
		{
			Add the directed edge $(t,z)$ to~$G$, and
			\mbox{\quad assign $(t,z)$ weight $2$ if $z=(l_{i''},r_{j''})$,
				and weight $1$ otherwise}\;
			\label{line:9}
			\If{$\upseg(z)$ includes the apex $l_k=r_m$}
			{Add the directed edge $(z,y)$ to~$G$ of weight~$0$\;}
			%\tcc*{$y$ is initial dummy vertex}}
		\lElse(\tcc*[f]{more guards are needed above~$z$})
		{$\ES{S}\gets \ES{S}\cup\{z\}$}
	}
}
Enumerate a shortest weighted path from $x$ to $y$ in~$G$\; 
Output the shortest path vertices without $x$ and $y$,
but considering the possible guard pairs, as the required guard set\;

\caption{Optimum vertex-to-point guarding of funnels.}
\label{alg:newminfunguard}
\end{algorithm}
\end{small}

\begin{thm}\label{thm:sizeguard-funnel}
Algorithm~\ref{alg:newminfunguard} runs in polynomial time, and it outputs
a feasible guard set of minimum size guarding all the points of a funnel $\FF$.
\end{thm}
\begin{proof}
Proving polynomial runtime is analogous to Lemma~\ref{lem:tightpath-feas},
only now obtaining a quadratic worst-scenario bound.
Likewise the proof of feasibility of the obtained solution is analogous to
the previous proof. 
We only need to observe the following new claim:
if $i''$and $j''$ are defined as on line~\ref{line:7} of
Algorithm~\ref{alg:newminfunguard}, then $l_{i''}$ and $r_{j''}$ together
can see whole $s$ and the strip of $\FF$ from $s$ till
$\upseg(l_{i''},r_{j''})$.
The crucial part is to prove optimality.

Having two guard sets $A,B\subseteq V(F)$, we say that $A$ {\em covers} $B$ if there is an injection $c:A\to B$ such that, for each $a\in A$,
the guard $c(a)$ is on the same (left or right) chain of~$\FF$ as $a$ and not higher than~$a$.
Let $G$ be the digraph constructed by Algorithm~\ref{alg:newminfunguard} on~$\FF$. 
Note that the weight of any $x$--$y$ path in $G$ equals the number of guards placed along it.
Hence, together with claimed feasibility of the computed solution, it is enough to prove:
\begin{itemize}
\item For every feasible vertex guard set $D$ of $\FF$,
there exists a feasible guard set $A$ which covers $D$, and
$A$ has a corresponding directed $x$--$y$ path in the graph~$G$
of Algorithm~\ref{alg:newminfunguard}.
\end{itemize}

Let $A$ be any feasible guard set of $\FF$ which covers given $D$,
and such that $A$ is maximal w.r.t.~the cover relation.
Let $\Pi$ be a maximal directed path in $G$, starting from $x$, such that the set of guards $B_P$ listed in the vertices of $\Pi$ (without~$x$)
satisfies~$B_P\subseteq A$.
Obviously, we aim to show that $\Pi$ ends in~$y$.
Suppose not (it may even be that $\Pi$ is a single vertex~$x$ and $B_P=\emptyset$).
Let $t$ be the last vertex of $\Pi$ and denote by $s=\upseg(t)$ and let $q$ and $p$ be the ends of $s$ on $\mathcal{L}$ and $\mathcal{R}$,
respectively, as in the algorithm.

Let $A'=A\setminus B_P\not=\emptyset$.
Then $s$ has to be guarded from $A'$ (while the whole part of $\FF$ below $s$ is already guarded by $B_P$ by feasibility of the algorithm).
Let $i,j$ be such that $l_i\in A'\cap\mathcal{L}$ and $r_j\in A'\cap\mathcal{R}$ are the lowest guards on the left and right chain.
Assume, up to symmetry, that $l_i$ can see whole~$s$.
By our maximal choice of $A$ we have that no vertex on $\mathcal{L}$ above $l_i$ can see whole~$s$,
and so the digraph $G$ contains an edge from $t$ to $l_i$ (line~\ref{line:9} of Algorithm~\ref{alg:newminfunguard}),
which contradicts maximality of the path~$\Pi$.

Otherwise, neither of $l_i$, $r_j$ can see whole $s$,
and so $l_i$ sees the end $p$ and $r_j$ sees the end~$q$.
Consequently, $l_{i}$ is strictly below $\upseg(p)$ and            
$r_{j}$ strictly below $\upseg(q)$, and they are topmost such vertices
again by our maximal choice of $A$.
Hence the digraph $G$ contains an edge from $t$ to $(l_i,r_j)$,
as previously, which is again a contradiction concluding the proof.
\end{proof}

Lastly, we establish that the difference between Algorithms \ref{alg:tightpath} and \ref{alg:newminfunguard} cannot be larger than~$1$ guard,
because we would like to use the simpler Algorithm~\ref{alg:tightpath} instead of the latter one in subsequent applications.
Let $G^1$ with the source $x^1$ be the auxiliary graph produced by Algorithm~\ref{alg:tightpath},
and $G^2$ with the source $x^2$ be the one produced by Algorithm~\ref{alg:newminfunguard}.
We can prove the following detailed statement by induction on~$i\geq0$:
\begin{itemize}
\item Let $P^2=(x^2=x^2_0,x^2_1,\ldots,x^2_i)$ be any directed path in $G^2$ 
of weight $k$, let $Q^2$ denote the set of guards listed in the vertices of $P^2$, and
$L^2=\mathcal{L}\cap Q^2$ and $R^2=\mathcal{R}\cap Q^2$.
Then there exists a directed path $(x^1=x^1_0,x^1_1,\ldots,x^1_k,x^1_{k+1})$ in $G^1$
(of length $k+1$), such that the guard of $x_k$ is at least as high as all the
guards of $L^2$ (if $x_k\in\mathcal{L}$) or of $R^2$ (if
$x_k\in\mathcal{R}$), and the guard of $x_{k+1}$ is strictly higher than all
the guards of $Q^2$.
\end{itemize}

\begin{cor}\label{cor:alg1-near-optimal}
The guard set produced by Algorithm \ref{alg:tightpath} is always by at most one
guard larger than the optimum solution produced by Algorithm~\ref{alg:newminfunguard}.
\end{cor}

\begin{proof}
Let $G^1$ with the source $x^1$ be the auxiliary graph produced by Algorithm~\ref{alg:tightpath},
and $G^2$ with the source $x^2$ be the one produced by Algorithm~\ref{alg:newminfunguard}.
We instead prove the following refined statement by induction on~$i\geq0$.
Recall the detailed inductive statement we are going to prove now:
\begin{itemize}
\item Let $\Pi^2=(x^2=x^2_0,x^2_1,\ldots,x^2_i)$ be any directed path in $G^2$ 
of weight $k$, let $Q^2$ denote the set of guards listed in the vertices of $P^2$, and
$L^2=\mathcal{L}\cap Q^2$ and $R^2=\mathcal{R}\cap Q^2$.
Then there exists a directed path $(x^1=x^1_0,x^1_1,\ldots,x^1_k,x^1_{k+1})$ in $G^1$
(of length $k+1$), such that the guard of $x_k$ is at least as high as all the
guards of $L^2$ (if $x_k\in\mathcal{L}$) or of $R^2$ (if
$x_k\in\mathbb{R}$), and the guard of $x_{k+1}$ is strictly higher than all
the guards of $Q^2$.
\end{itemize}

The claim is trivial for $i=0$, and so we assume that $i\geq1$ and the claim
holds for the shorter path ${\Pi^2}'=(x^2=x^2_0,x^2_1,\ldots,x^2_{i-1})$ of
weight $k'$ in $G^2$, hence providing us with a path
$(x^1=x^1_0,x^1_1,\ldots,x^1_{k'},x^1_{k'+1})$ in $G^1$.
If $k=k'+1$ (i.e., $x^2_i$ represents a single guard),
Algorithm~\ref{alg:tightpath} can ``duplicate'' the move, hence making
$x^2_i$ or a higher vertex $x^2_{i'}$ on the same chain 
an outneighbour of $x^1_{k'+1}$ in $G^1$.
Then we set $x^1_{k'+2}=x^1_{k+1}=x^2_{i'}$ and we are done.

If $k=k'+2$ (i.e., $x^2_i$ represents a pair of guards $z_1,z_2$),
we proceed as follows.
Up to symmetry, assume $x^1_{k'+1}\in\mathcal{L}$ 
and $z_1\in\mathcal{L}$, $z_2\in\mathcal{R}$.
By the induction assumption, we know that $x^1_{k'+1}$ is strictly higher
(on $\mathcal{L}$) than the guards from $L^2\setminus\{z_1\}$.
We choose $x^1_k$ as the outneighbour of $x^1_{k'+1}$ in $G^1$ that lies on
$\mathcal{R}$, and $x^1_{k+1}$ as the outneighbour of $x^1_{k}$ in
$G^1$ that lies back on $\mathcal{L}$. 
From Algorithm~\ref{alg:newminfunguard} (line~\ref{line:7})
it follows that $z_2$ sees $x^1_{k'+1}$, and so $x^1_k$ is at least as high
on $\mathcal{R}$ as~$z_2$.
Consequently, $x^1_{k+1}$ lies on $\mathcal{L}$ strictly higher than~$z_1$
(which sees the highest guard from $R^2\setminus\{z_2\}$), and we are again done.
\end{proof}

\subsection{Conflict-free chromatic guarding of a funnel} \label{sec:cfguardingfunnel}
In this section, we continue to study funnels.
To obtain a conflict-free colored solution, we will simply
consider the guards chosen by Algorithm~\ref{alg:tightpath} in the ascending
order of their vertical coordinates, and color them in the \emph{ruler sequence},
(e.g.,~\cite{guy-1994}) in which the $i^{th}$ term is the exponent of the largest power of $2$ that divides $2i$.\\ 
(The first few terms are $1,2,1,3,1,2,1,4,1,2,1,3,1,2,1,5,1\dots$.)\\
So, if Algorithm~\ref{alg:tightpath} gives $m$ guards, then our approach
will use about $\log m$ colors.

Our aim is to show that this is always very close to the optimum, by giving a lower bound
on the number of necessary colors of order $\log m-O(1)$.
To achieve this, we study the following two sets of guards for a given funnel~$\FF$:
\begin{itemize}
\item The minimal {\em guard set $A$} computed by Algorithm~\ref{alg:tightpath} on~$\FF$
(which is overall nearly optimal by Corollary~\ref{cor:alg1-near-optimal});
if this is not unique, then we fix any such~$A$.
\item A {\em guard set $D$} which achieves the minimum number of colors for
conflict-free guarding; note that $D$ may be much larger than $A$ since it
is the number of colors which matters.
\end{itemize}
On a high level, we are going to show that the coloring of $D$
must (somehow) copy the ruler sequence on~$A$.
For that we will recursively bisect our funnel into smaller ``layers'',
gaining one unique color with each bisection.

Analogous to the upper tangent, we define the {\em lower tangent}
of a vertex $l_i\in\mathcal{L}$, denote by $\lowtan(l_i)$,
as the ray whose origin is $l_i$ and which passes through~$r_j\in\mathcal{R}$
such that $r_j$ is the lowest vertex on $\mathcal{R}$ seeing $l_i$.
Note that $\lowtan(l_i)$ may intersect $\mathcal{R}$ in $r_j$ alone or in a
segment from $r_j$ up.
Let $\lowseg(l_i):=\lseg{l_i}{r_j}$.
The definition of $\lowtan()$ and $\lowseg()$ for vertices of $\mathcal{R}$ is
symmetric.

We now give a definition of ``layers'' of a funnel which is crucial for our proof.

\begin{figure}[tbp]
\centering
\begin{tikzpicture}[yscale=0.75, xscale=1.2]

\tikzstyle{every node}=[draw, shape=circle, fill=black, minimum size=2.5pt,inner sep=0pt];
\node (1) at (-5,0) {};
\node (2) at (-3.9, 0.1) {};
\node (3) at (-3.46, 0.2) {};
\node (4) at (-2.97, 0.32) {};
\node (5) at (-2.62, 0.43) {};
\node (6) at (-2.53, 0.46) {};
\node (7) at (-2.1, 0.7) {};
\node (8) at (-1.65, 0.96) {};
\node (9) at (-1.53, 1.04) {};
\node (10) at (-1.03, 1.54) {};
\node (11) at (-0.66, 2.02) {};
\node (12) at (-0.62, 2.12) {};
\node[fill=red,minimum size=4pt, label=left:$a_3$] (13) at (-0.32, 3.08) {};
\node (14) at (-0.19, 3.7) {};
\node (15) at (-0.14, 4.07) {};
\node[label=left:$q\>$, fill=white,thick,minimum size=4pt] (16) at (-0.03, 5.54) {};
\node (17) at (0.03, 8) {};
\node[fill=red,minimum size=4pt, label=right:$\,a_4$] (18) at (0.26, 6.51) {};
\node (19) at (0.67, 4.09) {};
\node (20) at (0.82, 3.42) {};
\node (21) at (0.99, 2.79) {};
\node (22) at (1.24, 2.23) {};
\node (23) at (1.41, 1.92) {};
\node (24) at (1.6, 1.68) {};
\node (25) at (1.8, 1.43) {};
\node (26) at (2.05, 1.16) {};
\node (27) at (2.23, 0.97) {};
\node[fill=red,minimum size=4pt, label=above:$a_2$] (28) at (2.56, 0.71) {};
\node (29) at (2.83, 0.57) {};
\node (30) at (3.19, 0.42) {};
\node (31) at (3.94, 0.21) {};
\node (32) at (4.33, 0.11) {};
\node (33a) at (4.67, 0.04) {};
\node (33b) at (4.8, 0.03) {};
\node (33) at (5, 0) {};
\foreach \i in {1,...,32} 
{
	\pgfmathtruncatemacro\j{\i+1};
	\draw[thick] (\i)--(\j);
}
\draw[thick] (32)--(33)--(1);

\node[fill=red,minimum size=4pt,thick, label=above:{$\!\!\!a_1=p\!\!\!$}] (1g) at (-4.8,0.02) {};
\coordinate(1') at (2.7,0.62);
\coordinate(16') at (2.03, 0);
\coordinate(1int) at (2.6,0.62); %intersection points]
\coordinate(16int) at (2.02,0);
\draw[dashed, thick] (1g)--(1');
\draw[dashed, thick] (16)--(21);
\draw[dotted, thick] (21)--(16');
\draw[dotted, thick] (2)--(32);
\fill[opacity=0.3,gray] (2) -- (32) -- (1') -- (2);
\fill[opacity=0.5,green] (1'.center)--(2.center)--(3.center)--(4.center)--(5.center)--(6.center)--(7.center)--(8.center)--(9.center)--(10.center)--(11.center)--(12.center)--(13.center)--(14.center)--(15.center)--(16.center)--(21.center)--(22.center)--(23.center)--(24.center)--(25.center)--(26.center)--(27.center)--(28.center)--(29.center)--(1int)--(1.center);
\node[fill=white, inner sep=0.5pt] at (1.82,0.55) {{\small $o$}};
\end{tikzpicture}
\caption{An example of a $2$ -interval $Q$ of a funnel (filled green and
bounded by $s_1=\upseg(p)$ and $s_2=\lowseg(q)$).
The red vertices $a_1=p,a_2,a_3,a_4$ are the guards computed
by Algorithm~\ref{alg:tightpath}, and $a_2,a_3$ belong to
the interval~$Q$.
Note that $p$ and $q$ by definition do not belong to~$Q$. 
The shadow of $Q$ (filled light gray) is
bounded from below by the bottom dotted line, and the
inner point $o$ is the so-called {\em observer} of $Q$
(seeing all vertices of $Q$ and, possibly, some vertices
in the shadow). }
\label{fig:k-interval}
\end{figure}

\begin{dfn}[$t$-interval]\label{df:kinterval}\rm
% \todo{PH: better defined by geometry and not by guards}
Let $\FF$ be a funnel with the %concave 
chains $\mathcal{L}=(l_1, l_2,
\ldots, l_k)$ and $\mathcal{R}=(r_1, r_2, \ldots, r_m)$,
and $A$ be the fixed guard set $A$ computed by Algorithm~\ref{alg:tightpath} on~$\FF$.
Let $s_1$ be the base of $\FF$, or $s_1=\upseg(p)$ for some vertex $p$ 
of $\FF$ (where $p$ is not the apex or its neighbour).
Let $s_2$ be the apex of $\FF$, or $s_2=\lowseg(q)$ for some vertex $q$ 
of $\FF$ (where $q$ is not in the base of~$\FF$).
Assume that $s_2$ is above $s_1$ within~$\FF$.
Then the region $Q$ of $\FF$ bounded from below by $s_1$ and from above by
$s_2$, excluding~$q$ itself, is called an {\em interval of~$\FF$}.
Moreover, $Q$ is called a {\em$t$-interval of~$\FF$} if $Q$ contains at least
$t$ of the guards of~$A$.
See Figure~\ref{fig:k-interval}.

Having an interval $Q$ of the funnel $\FF$, bounded from below by $s_1$,
we define the {\em shadow of $Q$} as follows.
If $s_1=\upseg(l_i)$ ($s_1=\upseg(r_j)$), then the shadow consists of the 
region of $\FF$ between $s_1$ and $\lowseg(l_{i+1})$
(between $s_1$ and $\lowseg(r_{j+1})$, respectively).
If $s_1$ is the base, then the shadow is empty.
\end{dfn}

\begin{lem}\label{lem:11interval}
If $Q$ is a $13$-interval of the funnel $\FF$, then there exists a point
in $Q$ which is not visible from any vertex of $\FF$ outside of~$Q$.
\end{lem}
\begin{proof}
By definition, a $13$-interval has thirteen guards from~$A$ in it. By 
the pigeon-hole principle, at least seven of these guards are on the same chain.
Without the loss of generality, let these seven guards lie on the left chain $\mathcal{L}$ of $\FF$.
Let us denote these guards by $a$, $b$, $c$, $d$, $e$, $f$ and $g$ in the bottom-up order, respectively.
We show that the guard $d$ is not seen by any viewer outside of the $13$-interval.

Suppose that $d$ can be seen by a vertex of $\mathcal{R}$ which lies below $\uptan(a) \cap \mathcal{R}$.
This means that the vertex of $\mathcal{R}$ immediately below $\uptan(a) \cap \mathcal{R}$ 
(say, denoted by~$x$) also sees $d$. 

Since $x$ is below $\uptan(a) \cap \mathcal{R}$, $x$ sees all vertices of $\mathcal{L}$ between $b$ and $d$,
including both $b$ and $d$.
Additionally, if $b$ is not the immediate neighbour of $a$ on $\mathcal{L}$,
then $x$ sees the vertex of $\mathcal{L}$ immediately below $b$ as well.
Thus, $x$ sees all points seen by $b$ or $c$ on $\mathcal{L}$.

Since $x$ sees $d$, all points of $\mathcal{R}$ that are seen by $b$ or $c$ and lie above $x$, are also seen by $d$.
Since $x$ lies below $\uptan(a) \cap \mathcal{R}$, $a$ sees all points of $\mathcal{R}$
between and including $x$ and $\lotan(a) \cap \mathcal{R}$.
Since $a$ lies below $b$ and $c$ on $\mathcal{L}$, none among $b$ and $c$ can see any vertex below $\lotan(a) \cap \mathcal{R}$.
Thus, $d$ and $a$ see all points seen by $b$ or $c$ on $\mathcal{R}$.

The above arguments show that $a$, $d$ and $x$ together see all the points on the two concave chains seen by $b$ and $c$.
Observe that since $x$ is the vertex immediately below $\uptan(a) \cap \mathcal{R}$, Algorithm \ref{alg:tightpath}
includes in $G$ either $x$, or a higher vertex $x'$ of $\mathcal{R}$ which sees the points of $\FF$ exclusively seen by $x$.
This means Algorithm \ref{alg:tightpath} must choose $x$ (or, $x'$) instead of $b$ and $c$ to optimize on the number of 
guards. Hence, we have a contradiction, and no vertex below the $13$-interval can see $d$.

Now suppose that $d$ is seen by a vertex $y$ lying above the $13$-interval. Since the $13$-interval contains
two more guards on $\mathcal{L}$ above $d$, the vertex $y$ can certainly not lie on $\mathcal{L}$.
Therefore, $y$ must lie on $\mathcal{R}$
If the upper segment of the $13$-interval is defined
by $\lowseg(v)$, where $v \in \mathcal{R}$, then $d$, $e$, $f$ and $g$ must lie below
$\lowseg(v) \cap \mathcal{L}$.
This means, $\uptan(d) \cap \mathcal{R}$ must lie below $v$. But to see $d$, $y$ must lie below $v$.
This makes $y$ a vertex contained in the $13$-interval. So, $v$ must lie on~$\mathcal{L}$.

So, we assume that $v$ lies on $\mathcal{L}$. At the worst case, $v$ can be the guard $f$. 
This means, $\uptan(d) \cap \mathcal{R}$ is above $lot(g) \cap \mathcal{R}$, and $y$ lies on the segment
of $\mathcal{R}$ between these two points. But then, by an argument similar to above, $d$, $g$ and $y$ together see 
everything that is seen by $d$, $e$, $f$ and $g$, and so Algorithm \ref{alg:tightpath} would choose only $d$, $g$
and $y$ to get a shortest path in $G$, a contradiction. So, the point $d$ is not visible from any vertex outside of 
the $13$-interval.
\end{proof}

Our second crucial ingredient is the possibility to ``almost privately'' see the vertices of an interval
$Q$ from one point as follows.
If $s_2=\lowseg(q)$, then the intersection point of
$\lowtan(q)$ with $s_1$ is called the {\em observer of~$Q$}.
(Actually, to be precise, we should slightly perturb this position of the
observer $o$ so that the visibility between $o$ and $q$ is blocked.
To keep the presentation simple, we neglect this detail.)
If $s_2$ is the apex, then consider the spine of $\FF$ instead of $\lowtan(q)$.
See again Figure~\ref{fig:k-interval}.

\begin{lem}\label{lem:observer}
The observer $o$ of an interval $Q$ in a funnel $\FF$
can see all the vertices of $Q$, but $o$ cannot see any vertex of $\FF$
which is not in $Q$ and not in the shadow of~$Q$.
\end{lem}
\begin{proof}
Without the loss of generality let $q$ lie on $\mathcal{L}$.
Since the observer $o$ is the intersection point of $\lowtan(q)$ with $s_1$,
$o$ sees all the vertices of $Q$, since all of them lie between 
$s_1$ and  $\lowtan(q)$.
Again, since the observer $o$ is the intersection point of $\lowtan(q)$ with $s_1$, 
$o$ lies below a vertex on $\lowtan(q) \cap \mathcal{R}$. 
Then this vertex must be a blocker between $o$ and any point of $\FF$ above $s_2$.
So, $o$ cannot see any vertex of $\FF$ above $Q$ that is not in $Q$. 
Let $s_1 = \uptan(r)$, where without the loss of generality $r$ lies on $\mathcal{L}$.
Then, since $o$ lies on $s_1$, $o$ cannot see any vertex of $\mathcal{L}$ below $r$.
The shadow of $Q$ on $\mathcal{R}$ extends till $\lowtan(r') \cap \mathcal {R}$, where
$r'$ is the vertex of $\mathcal{L}$ immediately above $r$.
But since $o$ lies on $\upseg(r)$, $o$ is above $r'$, and hence cannot see 
vertices of $\mathcal{R}$ below $\lowtan(r') \cap \mathcal {R}$.
Hence, $o$ cannot see vertices below the shadow of $Q$.
\end{proof}

The last ingredient before the main proof is the notion
of sections of an interval $Q$ of~$\FF$.
Let $s_1$ and $s_2$ form the lower and upper boundary of~$Q$.
Consider a vertex $l_i\in\mathcal{L}$ of~$Q$.
Then the {\em lower section of $Q$ at $l_i$} is the interval of $\FF$ bounded
from below by $s_1$ and from above by $\lowseg(l_i)$.
The {\em upper section of $Q$ at $l_i$} is the interval of $\FF$ bounded
from below by $\upseg(l_i)$ and from above by $s_2$.
Sections of $r_j\in\mathcal{R}$ are defined analogously.

\begin{lem}\label{lem:ksection}
Let $Q$ be a $t$-interval of the funnel $\FF$, and let $Q_1$ and $Q_2$ be its
lower and upper sections at some vertex $p$.
Then $Q_i$, $i=1,2$, is a $t_i$-interval such that $t_1+t_2\geq t-3$.
(In other words, at most $3$ of the $A$-guards in $Q$ are not in $Q_1\cup Q_2$.)
\end{lem}
\begin{proof}
The only vertices of $Q$ which are not included in $Q_1\cup Q_2$ are
$p$ and the vertices of the shadow of~$Q_2$.
Suppose, for a contradiction, that those contain (at least) four guards
from~$A$; in either such case, we easily contradict minimality of the guard
set~$A$ in Algorithm~\ref{alg:tightpath}, by the same argument given in Lemma \ref{lem:11interval}.
The Algorithm~\ref{alg:tightpath} simply chooses $p$ and the topmost and bottommost among these (at least) four guards,
thus choosing only three guards instead of four.
\end{proof}

Now we are ready to prove the advertised lower bound:
\begin{thm}\label{thm:logm-4}
Any conflict-free chromatic guarding of a given funnel 
requires at least $\lfloor\log_2(m+3)\rfloor-3$ colors,
where $m$ is the minimum number of guards needed to guard the whole funnel.
\end{thm}

\begin{proof}
We will prove the following claim by induction on $c\geq0$:
\begin{itemize}
\item If $Q$ is a $t$-interval in the funnel $\FF$
and $t\geq 16\cdot2^c-3$, then any conflict-free coloring of $\FF$
must use at least $c+1$ colors on the vertices of $Q$ or 
of the shadow of~$Q$.
\end{itemize}

In the base $c=0$ of the induction, we have $t\geq16-3=13$.
By Lemma~\ref{lem:11interval}, some point of $Q$ is not seen from outside,
and so there has to be a colored guard in some vertex of $Q$,
thus giving $c+1=1$ color.

Consider now $c>0$.
The observer $o$ of $Q$ (which sees all the vertices of~$Q$) 
must see a guard $g$ of a unique color
where $g$ is, by Lemma~\ref{lem:observer}, 
a vertex of $Q$ or of the shadow of~$Q$.
In the first case, we consider $Q_1$ and $Q_2$,
the lower and upper sections of $Q$ at~$g$.
By Lemma~\ref{lem:ksection}, for some $i\in\{1,2\}$,
$Q_i$ is a $t_i$-interval of $\FF$ such that
$t_i\geq (t-3)/2\geq (16\cdot2^c-6)/2=16\cdot2^{c-1}-3$.
In the second case ($g$ is in the shadow of $Q$),
we choose $g'$ as the lowermost vertex of $Q$ on the same chain as~$g$,
and take only the upper section $Q_1$ of $Q$ at $g'$.
We continue as in the first case with~$i=1$.

By induction assumption for $c-1$, $Q_i$ together with its shadow
carry a set $C$ of at least $c$ colors.
Notice that the shadow of $Q_2$ is included in~$Q$,
and the shadow of $Q_1$ coincides with the shadow of~$Q$,
moreover, the observer of $Q_1$ sees only a subset of the shadow of $Q$ seen
by the observer $o$ of~$Q$.
Since $g$ is not a point of $Q_i$ or its shadow, but our observer 
$o$ sees the color $c_g$ of $g$ and all the colors of~$C$, we have
$c_g\not\in C$ and hence $C\cup\{c_g\}$ has at least $c+1$ colors, as
desired.

Finally, we apply the above claim to $Q=F$.
We have $t\geq m$, and for $t\geq m\geq 16\cdot2^c-3$ we derive that
we need at least $c+1\geq\lfloor\log(m+3)\rfloor-3$ colors for guarding
whole~$\FF$.
\end{proof}

\begin{algorithm}[htbp]
\KwIn{A funnel $\FF$ with concave chains $\mathcal{L}=(l_1, l_2,
\ldots, l_k)$ and $\mathcal{R}=(r_1, \ldots, r_m)$.}
\KwOut{A conflict-free chromatic guard set of $\FF$
using $\leq OPT+4$ colors.}
\smallskip
Run Algorithm~\ref{alg:tightpath} to produce 
a guard sequence $A=(a_1,a_2,\dots,a_t)\>$ (bottom-up)\;
Assign colors to members of $A$ according to the ruler sequence;
the vertex $a_i$ gets color $c_i$
where $c_i$ is the largest integer such that $2^{c_i}$ divides~$2i$\;		
Output colored guards $A$ as the (approximate) solution\;

\caption{Approximate conflict-free chromatic guarding of a funnel.}
\label{alg:apxcfreefunnel}
\end{algorithm}

\begin{cor}
\label{cor:opt+4}
Algorithm~\ref{alg:apxcfreefunnel}, for a given funnel~$\FF$,
outputs in polynomial time a conflict-free chromatic guard set $A$,
such that the number of colors used by $A$ is by at most four larger
than the optimum.
\end{cor}
\begin{proof}
Note the following simple property of the ruler sequence:
if $c_i=c_j$ for some $i\not=j$, then $c_{(i+j)/2}>c_i$.
Hence, for any $i,j$, the largest value occurring among colors
$c_i,c_{i+1},\dots,c_{i+j-1}$ is unique.
Since every point of $\FF$ sees a consecutive subsequence of~$A$,
this is a feasible conflict-free coloring of the funnel~$\FF$.

Let $m$ be the minimum number of guards in~$\FF$.
By Corollary~\ref{cor:alg1-near-optimal}, it is $m+1\geq t=|A|\geq m$.
To prove the approximation guarantee, observe that for $t\leq 2^c-1$,
our sequence $A$ uses $\leq c$ colors.
Conversely, if $t\geq 2^{c-1}$, i.e. $m\geq 2^{c-1}-1$,
then the required number of colors for guarding $\FF$ is
at least~$c-1-3=c-4$, and hence our algorithm uses at most $4$ more colors
than the optimum.
\end{proof}

\section{Weak-visibility polygons} \label{sec:weakvis}
\subsection{Guarding a weak-visibility polygon}
In this section, we give an algorithm to obtain a set of guards to guard a given weak-visibility polygon.
The algorithm is described by Bhattacharya, Ghosh, and Roy in \cite{vertexGuardingWV}.
Note that the mentioned algorithm is not a novel contribution introduced by this thesis, but we describe it nevertheless for the sake of completeness.

Given the graph $G = (V, E)$ of a weak visibility polygon $\WW$, let $\mathit{SPT}(u)$ denote the shortest path tree of $u \in V$.
Moreover, let $p_u(v)$ denote the parent of $v$ in $\textit{SPT}(u)$, and let $\partial(u,v)$ denote the clockwise boundary of $\WW$ from a vertex $u$ to another vertex $v$.
Analogously, let $\partial'(u,v)$ denote the counterclockwise order.

We start by describing the algorithm (Algorithm 2.1 in \cite{vertexGuardingWV}) to obtain a guard set for vertex-to-vertex guarding of $\WW$.
The idea of this algorithm is to find a subset $A$ of vertices, such that the guards set of $A$ sees every vertex in $\WW$.
Thus, instead of processing the whole polygon, the algorithm works only on a subset.
The following algorithm describes how $A$ and the guard set of $A$ is chosen.

\begin{algorithm}[htbp]
	\caption{Computing a guard set $\Gamma_A$ for sove vertices of a given weak-visibility polygon $\WW$}
	\label{alg:wvguarding}
	
	\smallskip
	Compute  $\textit{SPT}(u)$ and  $\textit{SPT}(v)$\;
	Initialize all vertices of $\WW$ as \textit{unmarked}\;
	Initialize $A \gets \emptyset$, $\Gamma_A \gets \emptyset$, and $z \gets u$\;
	
	\While{$z \neq v$}{
		$z \gets$ the vertex next to $z$ in clockwise order on $\partial(u,v)$\;
		\If{$z$ is \textbf{not} marked}
		{
			$A \gets A \cup \{z\}$\;
			$\Gamma_A \gets \Gamma_A \cup \{ p_u(z), p_v(z)\}$\;
			Place guards on $p_u(z)$ and $p_v(z)$\;
			Mark all vertices of $\WW$ that are visible from $p_u(z)$ and $p_v(z)$.
		}
	}
	\Return{$\Gamma_A$}\;
\end{algorithm}

In the light of the algorithm given above, the authors describe Algorithm~\ref{alg:wvguardingB} (Algorithm 2.2 in \cite{vertexGuardingWV}) which returns a set of guards that guard every vertex in a given weak-visibility polygon.

The input for this algorithm is a weak-visibility polygon $\WW$ with the common edge $uv$.
In short, the algorithm starts with a vertex $z$ is chosen (which is initially $u$), and then traverses all the vertices in the clockwise order from $z$, performing an action with respect to the property of $z$.
The process goes on until the vertex $v$ is reached.

\begin{algorithm}[htbp]
	\caption{Computing a guard set $\Gamma_B$ for all the vertices of a given weak-visibility polygon $\WW$}
	\label{alg:wvguardingB}
	
	\smallskip
	Compute  $\textit{SPT}(u)$ and  $\textit{SPT}(v)$\;
	Initialize all the vertices of $\WW$ as unmarked\; 
	Initialize $B \gets \emptyset$, $\Gamma_B \gets \emptyset$, and $z \gets su$\;
	\While{there exists an unmarked vertex in $\WW$}{
		$z \gets $ the first unmarked vertex on $\partial(u,v)$ in clockwise order from $z$\;
		\If{every unmarked vertex of $\partial(z,p_v(z))$ is visible from $p_u(z)$ or $p_v(z)$}{
			$B \gets B \cup \{z\}$ and $\Gamma_B \gets \Gamma_B \cup \{p_u(z), p_v(z)\}$\;
			Mark all vertices of $\WW$ that become visible from $p_u(z)$ or $p_v(z)$\;
			$z \gets p_v(z)$\;
		}
		\Else{
			$z' \gets$ the first unmarked vertex on $\partial(z,v)$ in clockwise order\;
			\While{every unmarked vertex of $\partial(p_u(z'),z')$ is visible from $p_u(z')$ or $p_v(z')$}{
				$z\gets z'$ and $z'\gets$ the first unmarked vertex on $\partial(z',v)$ in clockwise order\;
			}
			
			$w \gets z$\;
			\While{there exists an unmarked vertex on $\partial(u,z)$}{
				$B \gets B \cup \{z\}$ and $\Gamma_B \gets \Gamma_B \cup \{p_u(z), p_v(z)\}$\;
				Mark all vertices of $\WW$ that become visible from $p_u(w)$ or $p_v(w)$\;
				$w \gets$ the first unmarked vertex on $\partial'(z,u)$ in counterclockwise order\;
			}
		}
	}
	\Return{$\Gamma_B$}
	%	Reinitialize all the vertices of $\WW$ that are visible from some guard in $\Gamma_B$ as unmarked\;
	%	\ForEach{vertex $z \in B'$ chosen in reverse order of inclusion}{
		%		Locate and mark each unmarked vertex visible from $p_u(z)$ or $p_v(z)$\;
		%		$B'\gets B' \setminus \{z\}$ and $\Gamma'_B \gets \Gamma'_B \setminus \{p_u(w),p_v(w)\}$\;
		%	}
	%	$B' \gets B \cup B'$\;
	%	\Return{$\Gamma_B \cup \Gamma'_B$}
\end{algorithm}

\subsection{Conflict-free chromatic guarding of a weak-visibility polygon} \label{sec:cfweakvis}
In this section, we extend the scope of the studied problem of vertex-to-point 
conflict-free chromatic guarding from funnels to weak visibility polygons.
We will establish an $O(\log^2 n)$ upper bound for the number of colors
of vertex-guards on $n$-vertex weak-visibility polygons, and give the
corresponding polynomial time algorithm.

\begin{figure}[tbp]
	\centering
	\begin{tikzpicture}[scale=0.35]
		
		\tikzstyle{every node}=[draw, shape=circle, minimum size=4pt,inner sep=0pt];
		\tikzstyle{every path}=[thin, fill=none];
		\draw[fill=magenta!15!white,draw=none] (0,0)--(6,0.9)--(7.9,2)--(15,8.5)--(15,11.5)--(14.5,13)--(25,1.5)--(31,0)--(0,0);
		\node[fill=black,label=150:$u\,$](A) at (0,0) {};
		\node(B) at (6,0.9) {};
		\node(BB) at (7.9,2) {};
		\node(C) at (8,6) {};
		\node(D) at (7,8) {};
		\node[fill=orange](E) at (4.3,9) {};
		\node[fill=red](F) at (5,11.5) {};
		\node(G) at (7.8,9.5) {};
		\node(H) at (11,8) {};
		\node(I) at (15,8.5) {};
		\node(J) at (15,11.5) {};
		\node[fill=magenta](K) at (14.5,13) {};
		\node[fill=purple](L) at (16.5,12) {};
		\node[fill=blue](M) at (18,13) {};
		\node(N) at (19,10) {};
		\node(O) at (19.5,9) {};
		\node(P) at (23,8) {};
		\node(Q) at (24,10.5) {};
		\node[fill=cyan](R) at (24.2,12) {};
		\node[fill=green](S) at (26,11) {};
		\node[fill=green!60!black](T) at (31,10) {};
		\node(V) at (28.6,8) {};
		\node(W) at (27.5,6) {};
		\node(X) at (25,1.5) {};
		\node[fill=black,label=30:$\,v$](Y) at (31,0) {};
		\draw (A)--(B)--(BB)--(C)--(D)--(E)--(F)--(G)--(H)--(I)--(J)--(K)--(L)--(M)--(N)--(O)--(P)--(Q)--(R)--(S)--(T)--(V)--(W)--(X)--(Y)--(A);
		%\draw[fill=magenta,draw=none] (A)--(B)--(BB)--(I)--(J)--(K)--(X)--(Y)--(A);
		
		\tikzstyle{every path}=[very thick,dotted];
		\draw[orange] (A)--(B)--(BB)--(C)--(D)--(E)--(X)--(Y);
		\draw[red] (A)--(B)--(BB)--(C)--(D)--(F)--(G)--(X)--(Y);
		\draw[magenta] (A)--(B)--(BB)--(I)--(J)--(K)--(X)--(Y);
		\draw[purple] (A)--(B)--(BB)--(I)--(L)--(X)--(Y);
		\draw[blue] (A)--(B)--(BB)--(I)--(M)--(N)--(O)--(X)--(Y);
		\draw[cyan] (A)--(B)--(P)--(Q)--(R)--(X)--(Y);
		\draw[green] (A)--(B)--(P)--(S)--(X)--(Y);
		\draw[green!60!black] (A)--(B)--(T)--(V)--(W)--(X)--(Y);
	\end{tikzpicture}
	\caption{A weak visibility polygon $\WW$ on the base edge $uv$, and
		its collection of $8$ max funnels ordered from left to right.
		These max funnels are depicted by thick dotted lines, which
		are colored from orange on the left to green on the right
		(the apex vertex of each max funnel is colored the same as the funnel chains).
		Note that the left-most orange funnel is degenerate, i.e.
		having its two apex edges collinear.
		Moreover, the third (pink) max funnel from the left is
		emphasised as the filled region of~$\WW$.
	}
	\label{fig:maxfunnels}
\end{figure}
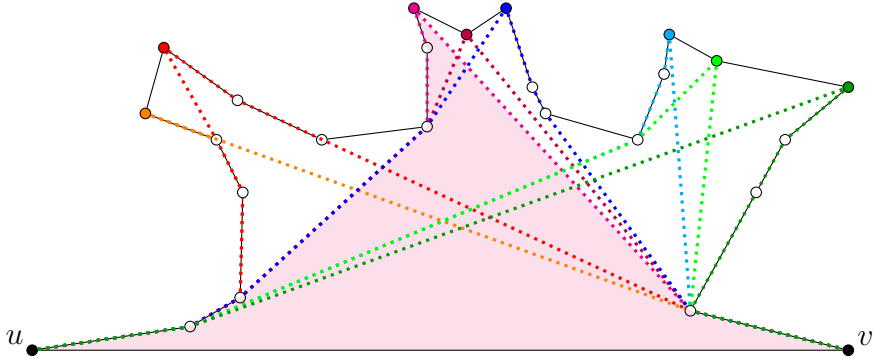

Consider a given weak visibility polygon $\WW$ on the base edge $uv$.
Each vertex $w$ of $\WW$ (other than $u,v$) is visible from some point of the segment~$\lseg uv$.
Consequently, for this $w$ one can always find a funnel $\FF \subseteq W$ with the base $uv$ 
and the apex~$w$ (such that $V(F)\subseteq V(W)$).
The chains of this funnel $\FF$ are simply the geometric shortest paths inside $\WW$
from $w$ to $u$ and~$v$, respectively.
The union of these funnels covers all vertices of $\WW$, and we consider the maximal
ones of them by inclusion -- the {\em max funnels} of $\WW$ -- ordered from left to right.
See Figure~\ref{fig:maxfunnels}.
Note that, as illustrated in the picture, some max funnels may be degenerate,
which means that their two edges incident with the apex vertex are collinear.

Our algorithm (and the corresponding upper bound) 
is based on the following two rather simple observations:

\begin{itemize}
	\item In each of the max funnels, we assign colors independently in the
	logarithmic ruler sequence.
	While doing so, we ``wastefully'' place guards at all vertices
	of both chains except the apex (which makes arguments easier).
	\item In order to prevent conflicts for points seeing guards from
	several funnels, we use generally different sets of colors for different
	max funnels.
	We distribute the color sets to the max funnels again according to the
	ruler sequence, from left to right, 
	which requires only $O(\log n)$ different color sets.
	We use the highest ruler sequence value to determine which color set
	applies to vertices which belong to several max funnels.
	
	See Algorithm~\ref{alg:weakv2pcoloringB} for details.
\end{itemize}

Let the ordered set of max funnels (possibly degenerate) of $\WW$ be $\FF = \{F_1, F_2, \ldots, F_m \}$
in the order of their apices occurring on the clockwise boundary of $\WW$ from $u$ to~$v$.
Then, as already mentioned, every vertex of $\WW$ belongs to one or more funnels of $\FF$.
Observe that $u$ and $v$ belong to all funnels of $\FF$.
Furthermore, for $F_i \in \FF$ with the apex vertex~$a_i$, the point set $W\setminus F_i$
is the union of internally-disjoint polygons such that some of them are to
the {\em left of~$F_i$} (those having their vertices between $u$ and $a_i$
in the clockwise order) and the others to the {\em right of~$F_i$}.
Observe that, for $F_j \in \FF$, we have $j<i$ if and only if $F_j\setminus F_i$
is all to the left of~$F_i$.

\begin{algorithm}[htbp]
	\caption{Computing a vertex-to-point conflict-free chromatic guarding of 
		a weak visibility polygon using $O(\log^2n)$ colors.}
	\label{alg:weakv2pcoloringB}
	\KwIn{A weak visibility polygon $\WW$, weakly visible from the base edge $uv$.}
	\KwOut{A V2P conflict-free chromatic guarding of $\WW$ using $O(\log^2 n)$
		colors, such that no guard is placed at any apex vertex of a max funnel of~$\WW$.}
	
	\smallskip
	Construct the (geom.) shortest path trees $T_u$ and $T_v$ of $\WW$ from $u$ and $v$\;
	Identify the common leaves of $T_u$ and $T_v$ as the apices of max funnels,
	and order them into a set $\FF = \{ F_1, F_2, \ldots, F_m \}$ from left to
	right, where $m\leq n$\;
	
	For $\ell=1+\lfloor\log m\rfloor$,
	consider $\ell$ color sets $C_1, C_2, \ldots, C_{\ell}$ of $2\log n$ colors
	each, i.e., let $C_i=C^l_i\cup C^r_i$ and all these sets be pairwise 
	disjoint of cardinality~$\log n$\;
	
	\ForEach{max funnel $F_i\in \FF$}{
		Assign color set $C_j$ to $F_i$, where $j$ is the largest power of $2$ dividing~$2i$%
		\tcc*{indices $j$ of the assigned sets $C_j$ form the ruler sequ.}
	}
	\ForEach{vertex $w$ of $\WW$}
	{
		Associate $w$ with the max funnel that was assigned color set $C_j$ of
		the highest index~$j$ (among the max funnels containing~$w$, this is unique)\;
		\label{lin:7}
	}
	\ForEach{max funnel $F_i\in \FF$}{
		color all vertices of the left chain of $F_i$ that are 
		associated with $F_i$, except the apex, by the ruler sequence 
		using the color set~$C^l_j$ assigned to~$F_i$\;
		color the same way the right chain of $F_i$ using the color set~$C^r_j$\;
	}
	\Return{colored $\WW$}\;
\end{algorithm}

We have the following lemma on $\WW$ and $\FF$.

\begin{lem} \label{lem:atleastone}
	Consider a weak visibility polygon $\WW$, and its max funnel $F_i \in \FF$.
	Assume that an observer $p\in W\setminus F_i$ is to the left of $F_i$ in
	$\WW$, and that $x\in W\setminus F_i$ is a point of $\WW$ to the right of $F_i$
	such that $p$ sees~$x$.
	Then $p$ sees at least one vertex of $\WW$ belonging to $F_i$ except the apex.
\end{lem}

Note that while it is immediate that the line of sight between $p$ and $x$
must cross both chains of the funnel $F_i$, this fact alone does not imply that
there is a vertex of $F_i$ visible from~$p$.
\begin{proof}
	By the definition, $p$ sees a point $b\in\overline{uv}$ on the base of~$\WW$.
	Consider the ray $\overrightarrow{pb}$ and rotate it counterclockwise until
	it first time hits a vertex $c$ on the right chain of~$F_i$.
	If $p$ sees~$c$, then we are done.
	Otherwise, the line of sight $\overline{pc}$ is blocked by a point
	$y\in\overline{pc}$ such that $y$ is outside of $\WW$, and hence to the left of $F_i$, too.
	Let $L_i$ be the point set of the left chain of $F_i$.
	
	Now, both segments $\overline{pb}$ and $\overline{px}$ cross the $L_i$, 
	and the counter-clockwise order of the rays from $p$
	is $\overrightarrow{pb}$, $\overrightarrow{pc}$, $\overrightarrow{px}$.
	Consequently, the point set 
	$\overline{pb}\cup\overline{px}\cup L_i$ (which is part of $\WW$)
	separates~$y$ (which is in the exterior of $\WW$).
	This contradicts the fact that $\WW$ is a simple polygon.
\end{proof}

Informally, Lemma~\ref{lem:atleastone} shows that the subcollection of max
funnels whose vertices are seen by the observer $p$, forms a consecutive
subsequence of $\FF$, and so precisely one of these funnels visible by $p$
gets the highest color set according to the ruler sequence.
From this color set we then get the unique color guard seen by~$p$.

The straightforward proof of correctness of our algorithm follows.

\begin{thm}\label{thm:twoApprox}
	Algorithm~\ref{alg:weakv2pcoloringB} in polynomial time computes a conflict-free
	chromatic guarding of a weak visibility polygon using $O(\log^2 n)$ colors.
\end{thm}

\begin{proof}
	Clearly, identifying the max funnels and the indexing can be done in $O(n^2)$ time. 
	We now prove the correctness of the algorithm.
	
	Any given point $p\in W$ always sees a unique color within any single funnel $F_i \in \FF$, 
	if at least some vertex of $F_i$ is seen from~$p$;
	in fact one such unique color left the left or/and one from the right chain
	of this funnel (and the color sets of the left and right chains are disjoint).
	This is because $p$ always sees a consecutive section of the left or
	right chain of $F_i$, regardless of whether $p$ is inside or outside
	of~$F_i$, and we use the ruler sequence colors on each chain.
	However, $p$ may see vertices associated with two or more max funnels assigned the same color set
	by Algorithm~\ref{alg:weakv2pcoloringB} (line~\ref{lin:7}).
	
	Suppose the latter; that the point $p$ sees vertices associated with
	two distinct max funnels $F_i,F_j \in \FF$ assigned the same color set.
	Let $k$, $i<k<j$, be such that $F_k$ is assigned the highest color
	set within the ruler subsequence from $i$ to $j$ ($k$ is unique).
	If $p\in F_k$, then $p$ sees a unique color guard from those on~$F_k$.
	Otherwise, up to symmetry, $p$ is to the left of $F_k$ and $p$ sees
	a vertex associated with $F_j$ which is to the right of $F_k$. 
	We have a situation anticipated by Lemma~\ref{lem:atleastone}, and
	so $p$ sees some vertex of $F_k$.
	Then again, $p$ gets a unique color guard from those on~$F_k$.
\end{proof}

\chapter{General Polygons} \label{chap:generalpolygons}

\section{Summary of the chapter}
In this chapter, we consider coloring and guarding simple polygons in general, and further generalize some of the results that were presented in previous chapters. This time, we start with the ordinary proper coloring (defined in Section~\ref{sec:problems})
\begin{itemize}
	\item Given the visibility graph of a simple polygon, we describe a polynomial-time algorithm to determine whether that graph admits a 4-coloring, and return the coloring if exists (Algorithm~\ref{alg:4coloring}).
	\item We show that the described algorithm can also be implemented without the representation of the polygon (Section~\ref{sec:withoutrepr}).
	\item We give NP-hardness reductions for the problems of proper 5-coloring of simple polygons (Section~\ref{sec:hardness5}), and proper 4-coloring of polygons with holes (Section~\ref{sec:polygonswithholes}).
	\item We show that the vertex-to-point conflict-free chromatic guarding of a simple polygon can be done in $O (\log^2 n)$ guards, which is a generalization of the upper bound for weak-visibility graphs (Algorithm~\ref{alg:weakv2pcoloringB}).
	\item We consider also the vertex-to-vertex chromatic guarding and give an NP-hardness reduction for that problem as well (Section~\ref{sec:v2vconflictfree}).
\end{itemize}

\section{Related work}
Visibility graphs of polygons have been studied with respect to various theoretical and practical computational problems.
The complexities of several popular optimization problems
have been determined for visibility graphs of polygons.
A geometric variation of the dominating set problem, namely polygon guarding, is one of the most studied problems
in computational geometry and is known as the Art Gallery Problem \cite{o-agta-87}. 
It has been studied extensively
for both polygons with and without holes and has been found to be NP-hard in both cases \cite{ll-ccagp-86,os-snpd-83}.
Besides, given a polygon, computing a maximum independent set is known to be hard, due to Shermer \cite{s-hpip-89},
while computing a maximum clique has been shown to be in polynomial time by Ghosh et al.~\cite{gsbg-cmcvgsp-07}.

A \emph{proper vertex colouring} of a graph is an assignment of labels or colours to the vertices of
the graph so that no two adjacent vertices have the same colours. Henceforth, 
when we say colouring a graph, we refer to proper vertex colouring.
The \emph{chromatic number} of a graph is defined as the minimum number of colours used in any proper colouring of the graph.
Visibility graph colouring has been studied for various types of visibility graphs.
K\'ara et al.\ characterized 3-colourable visibility graphs of point sets and
described a super-polynomial lower bound on the chromatic number with respect to
the clique number of visibility graphs of point sets \cite{Kara_pointVisChromatic}.
Pfender showed that, as for general graphs, the chromatic number of visibility graphs 
of point sets is also not upper-bounded by their clique numbers \cite{p-vgps-2008}.
Diwan and Roy showed that for visibility graphs of point sets, the 5-colouring
problem is NP-hard, but 4-colouring is solvable in polynomial time \cite{coloringPVG}.

The problem of {\em colouring} the visibility graphs of given polygons 
has been studied in the special context where each internal point of the
polygon is seen by a vertex, whose colour appears exactly
once among the vertices visible to that point \cite{Bartschi-2014,hoffmann:2015,FeketeFHM014}. 
However, little is known on colouring visibility graphs of polygons without such constraints.
Although 3-colouring is NP-hard for general graphs \cite{computersandintractability}, in particular it is rather trivial
to solve it for visibility graphs of polygons in polynomial time using a greedy approach. 
Already with $4$ colours the same question has been open so far.

In this chapter, we settle the complexity question of
the general problem of colouring polygonal visibility graphs,
which was declared open in 1995 by Lin and Skiena \cite{Lin_complexityaspects}.
We first provide a polynomial-time algorithm to find a $4$-colouring
of a given graph $G$ with the promise that $G$ is the visibility graph of some polygon, if $G$ is indeed $4$-colourable.
Then, we provide a reduction showing that the question of
$k$-colourability of the visibility graph of a given simple polygon 
is NP-complete for any~$k\geq5$.
We also show that $k$-colourability of visibility graphs of polygons with holes is
NP-complete for any~$k\geq4$.
Finally, we use our techniques of reduction to give a very simple proof of the NP-hardness
of independent sets of visibility graphs of simple polygons, which was originally proved by Shermer \cite{s-hpip-89}.

\section{Proper coloring of a polygon visibility graph} \label{sec:polyproper}
\subsection{Polynomial-time algorithm for proper 4-coloring} \label{sec:4coloringP}
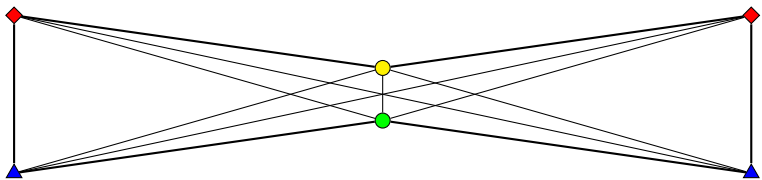
\begin{figure}
	\centering
	\begin{tikzpicture}[scale=0.7,
		diamondd/.style = {inner sep=1.6pt, diamond},	
		triangle/.style = {inner sep=1.2pt, regular polygon, regular polygon sides=3},	
		dtriangle/.style = {inner sep=1.2pt, regular polygon, regular polygon sides=3, rotate=180}	
		]	
		\tikzstyle{every node}=[draw, shape=circle, minimum size=3pt, inner sep=2pt];
		\node[fill=blue,triangle] (A) at (-5,0) {};
		\node[fill=red,diamondd] (B) at (-5,3) {};
		\node[fill=yellow] (C) at (2,2) {};
		\node[fill=red,diamondd] (D) at (9,3) {};
		\node[fill=blue,triangle] (E) at (9,0) {};
		\node[fill=green] (F) at (2,1) {};
		
		\draw[thick] (A)--(B)--(C)--(D)--(E)--(F)--(A);
		\draw (A)--(C);
		\draw (A)--(D);
		\draw (B)--(E);
		\draw (B)--(F);
		\draw (C)--(E);
		\draw (C)--(F);
		\draw (D)--(F);
		
	\end{tikzpicture}
	\caption{A visibility graph that is without a $K_5$, non planar but is 4-colorable.}
	\label{fig:nonplanar}
\end{figure}

In this section, we study the algorithmic question of $4$-colorability of the
visibility graph of a given polygon.
The full structure of $4$-colorable visibility graphs is not yet known
and it seems to be non-trivial. 
For instance, if a visibility graph is planar, it is $4$-colorable by the 4-color theorem \cite{4ct}.
Though, if such a graph contains $K_5$, then it is neither planar nor $4$-colorable, 
but a visibility graph not containing any $K_5$ may be non-planar yet $4$-colorable
(See Figure~\ref{fig:nonplanar}).

The related algorithmic problem of 3-coloring visibility graphs is rather
easy to resolve as follows.
Every simple polygon can be triangulated and, in such a triangulation, 
every non-boundary edge is contained in two triangles.
One can then proceed greedily edge by edge: 
Suppose a triangle has already been colored, and it
shares an edge with a triangle that is not fully colored. 
Then the two end vertices of the shared edge 
uniquely determine the color of the third vertex of the uncolored triangle. 
Finally, one checks whether this unique possible coloring is indeed proper.

Our algorithm essentially generalizes the 3-coloring method for 4-coloring.
We first divide the polygon into \emph{reduced polygons}.
A polygon $\PP$ is called a reduced polygon, if every chord of $\PP$ (i.e., an internal diagonal) is intersected by another chord of $\PP$.
After the division, we find and color in each reduced subpolygon a triangle (a $K_3$ subgraph) with three distinct colors.
Subsequently, whenever we find an uncolored vertex $v$ adjacent to some three
vertices colored with three distinct colors (such as, to an already
colored triangle), we can uniquely color also $v$, by the fourth color.
We will show that we can exhaust all vertices of a reduced subpolygon in this manner.
Furthermore, we check for possible coloring conflicts -- since the
coloring process is unique, this suffices to solve $4$-colorability.

Altogether, this will lead to the following theorem.
\begin{thm} \label{colthm}
	The 4-colorability problem is solvable in polynomial time for visibility graphs of simple polygons,
	and if a 4-coloring exists, then it can be computed in polynomial time from the given input graph (even without a visibility representation).
\end{thm}

We first prove that if a reduced polygon is 4-colorable, then the 4-coloring is unique up to	
a permutation of colors.
In the coming proof, consider a polygon $\PP$ and its visibility graph $G(V,E)$, embedded on $\PP$.
Hereafter we slightly abuse notation by equating $\PP$ and $G$.
Since we want to 4-color $\PP$, we assume that $G$ has no $K_5$ (or we answer `no').
We denote the clockwise polygonal chain of $\PP$ from a vertex $u$ to
a vertex $v$ as $\partial(u,v)$ (see Section~\ref{sec:graphs} for a complete definition).

One can easily see that it is enough to focus on reduced polygons in our proofs.
Indeed, assume an edge $uv$ of $G$ which is a chord of $\PP$ and not crossed
by any other chord.
We can partition $\PP$ into subpolygons $P_1$ and $P_2$, where
$P_1=(u\,\partial(u,v)\,v)$ and $P_2= (v\,\partial(v,u)\,u)$. 
Since there is no edge $xy$ of $G$ such that $x$ is in $P_1$ and $y$ is in 	
$P_2$, the polygons $P_1$ and $P_2$ can be $4$-colored separately	
and merged again (provided that $P_1$ and $P_2$ are $4$-colorable).

For a pair of distinct points $x$ and $y$, let $\overline{xy}$ denote the	
line segment with the ends $x$ and $y$, and $\overrightarrow{xy}$	
denote the half-line (ray) starting in $x$ and passing through~$y$.

Let $u$ and $v$ be two vertices of $\PP$.
The {\em shortest path} between $u$ and $v$ is a (graph) path from $u$ to
$v$ in $G$ such that the sum of the Euclidean lengths of its edges is minimized.
Such a shortest path between $u$ and $v$ is unique in $\PP$ and is denoted by $\Pi(u,v)$.
Observe that all non-terminal vertices of a shortest path are non-convex \cite{g-vap-07}.
We will assume an implicit ordering of vertices on $\Pi(u,v)$ from $u$ to $v$. 
When we say that some vertex $w$ is the first (or last) 
vertex on $\Pi(u,v)$ with a certain property, we mean that $w$ precedes 
(respectively, succeeds) all other vertices with that property on $\Pi(u,v)$.
For a point $x$, let the {\em right tangent from $x$ to $\Pi(u,v)$}	
be the ray $\overrightarrow{xy}$ such that $y\in\Pi(u,v)$ (possibly a 	
segment of $\Pi(u,v)$ belongs to $\overrightarrow{xy}$) and whole	
$\Pi(u,v)$ lies in counterclockwise direction from	
$\overrightarrow{xy}$, i.e., $\Pi(u,v)$ does not cross~$\overrightarrow{xy}$.	
A {\em left tangent} is defined analogously.

To prove Theorem~\ref{colthm}, we use the following sequence of lemmas.	
In each lemma below, we consider a $K_5$-free reduced polygon $\PP$	
and three vertices $\tau_1,\tau_2,\tau_3$ which forms a triangle~$T\subseteq G$.	
Assume that $T$ is already colored (which is unique up to a permutation of the colors).
Suppose that $i \in V$ is an uncolored vertex, such that an edge incident to $i$ intersects $T$.
Then we have the following lemmas.

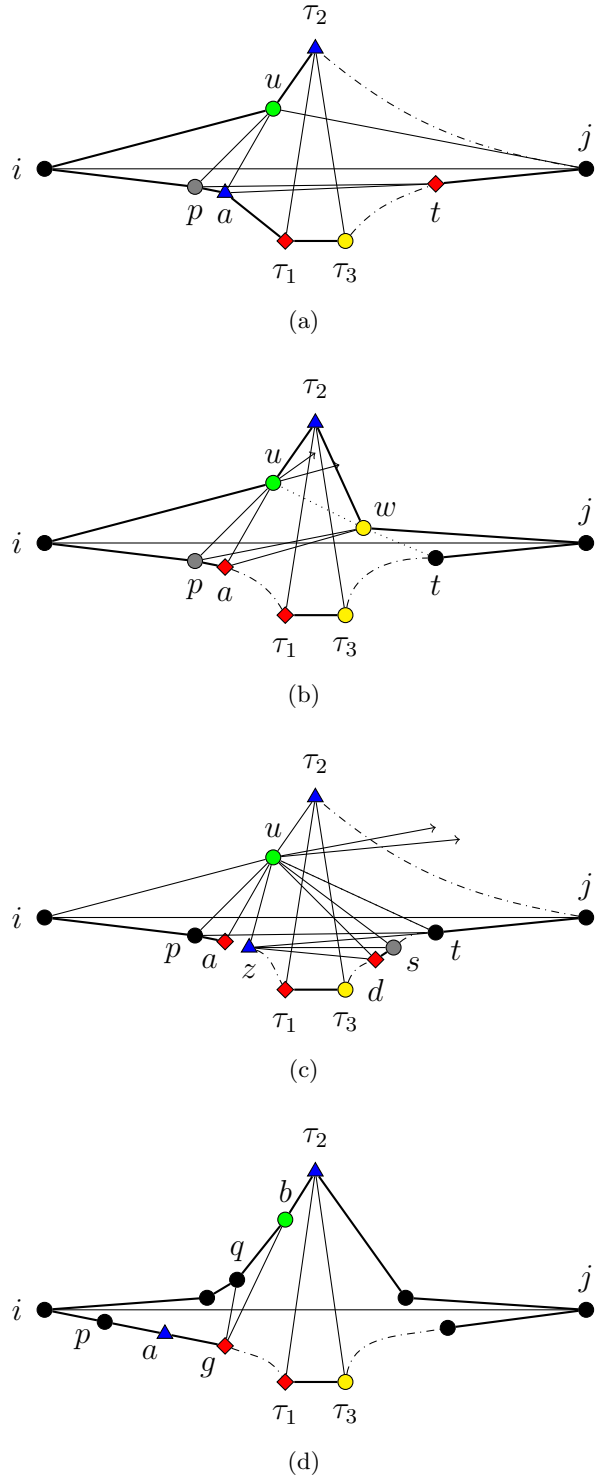
\begin{figure}[htbp]
	\vspace{-8em}
	\centering
	\subfloat[]{
		\begin{tikzpicture}[scale=0.8,
			diamondd/.style = {inner sep=1.6pt, diamond},	
			triangle/.style = {inner sep=1.2pt, regular polygon, regular polygon sides=3},	
			dtriangle/.style = {inner sep=1.2pt, regular polygon, regular polygon sides=3, rotate=180}	
			]	
			\tikzstyle{every node}=[draw, shape=circle, minimum size=3pt, inner sep=2pt];
			\node[fill=black, label=left:$i$] (vi) at (0,0) {};
			\node[fill=black, label=$j$] (vj) at (9,0) {};
			\node[fill=gray, label=below:$p$] (vp) at (2.5, -0.3) {};
			\node[fill=blue,triangle, label=below:$a$] (va) at (3,-0.4) {};
			\node[fill=red,diamondd, label=below:$\tau_1$] (t1) at (4,-1.2) {};
			\node[fill=yellow, label=below:$\tau_3$] (t3) at (5,-1.2) {};
			\node[fill=red,diamondd, label=below:$t$] (vt) at (6.5,-0.25) {};
			\node[fill=blue,triangle, label=$\tau_2$] (t2) at (4.5,2) {};
			\node[fill=green, label=$u$] (vu) at  (3.8,1) {};
			
			\draw[thick] (vi)--(vp)--(va)--(t1)--(t3);
			\draw[dash dot] (t3)  to[out=50,in=200] (vt);
			\draw[thick] (vt)--(vj);
			\draw[dash dot] (vj) to[out=170,in=320] (t2);
			\draw[thick] (t2)--(vu)--(vi);
			
			\draw (vi)--(vj);
			
			\draw (vt)--(vp);
			\draw (vt)--(va);
			\draw (vj)--(vu);
			
			\draw (vu)--(vp);
			\draw (vu)--(va);
			
			\draw (t2)--(t1);
			\draw (t2)--(t3);
		\end{tikzpicture}
	}
	
	\subfloat[]{
		\begin{tikzpicture}[scale=0.8,
			diamondd/.style = {inner sep=1.6pt, diamond},	
			triangle/.style = {inner sep=1.2pt, regular polygon, regular polygon sides=3},	
			dtriangle/.style = {inner sep=1.2pt, regular polygon, regular polygon sides=3, rotate=180}	
			]	
			\tikzstyle{every node}=[draw, shape=circle, minimum size=3pt, inner sep=2pt];
			\node[fill=black, label=left:$i$] (vi) at (0,0) {};
			\node[fill=black, label=$j$] (vj) at (9,0) {};
			\node[fill=gray, label=below:$p$] (vp) at (2.5, -0.3) {};
			\node[fill=red,diamondd, label=below:$a$] (va) at (3,-0.4) {};
			\node[fill=red,diamondd, label=below:$\tau_1$] (t1) at (4,-1.2) {};
			\node[fill=yellow, label=below:$\tau_3$] (t3) at (5,-1.2) {};
			\node[fill=black, label=below:$t$] (vt) at (6.5,-0.25) {};
			\node[fill=yellow, label=30:$w$] (vw) at (5.3,0.25) {};
			\node[fill=blue,triangle, label=$\tau_2$] (t2) at (4.5,2) {};
			\node[fill=green, label=$u$] (vu) at  (3.8,1) {};
			
			\draw[thick] (vi)--(vp)--(va);
			\draw[dash dot] (va) to[out=340,in=120] (t1);
			\draw[thick] (t1)--(t3);
			\draw[dash dot] (t3)  to[out=80,in=180] (vt);
			\draw[thick] (vt)--(vj)--(vw)--(t2)--(vu)--(vi);
			
			\draw[dotted] (vt)--(vw)--(vu);
			\draw[->] (vu)--(4.5,1.5);
			\draw[->] (vu)--(4.9,1.3);

			\draw (vi)--(vj);
			
			\draw (vw)--(vp);
			\draw (vw)--(va);

			\draw (vu)--(vp);
			\draw (vu)--(va);
			
			\draw (t2)--(t1);
			\draw (t2)--(t3);
		\end{tikzpicture}
	}
	
	\subfloat[]{
		\begin{tikzpicture}[scale=0.8,
			diamondd/.style = {inner sep=1.6pt, diamond},	
			triangle/.style = {inner sep=1.2pt, regular polygon, regular polygon sides=3},	
			dtriangle/.style = {inner sep=1.2pt, regular polygon, regular polygon sides=3, rotate=180}	
			]	
			\tikzstyle{every node}=[draw, shape=circle, minimum size=3pt, inner sep=2pt];
			\node[fill=black, label=left:$i$] (vi) at (0,0) {};
			\node[fill=black, label=$j$] (vj) at (9,0) {};
			\node[fill=black, label=200:$p$] (vp) at (2.5, -0.3) {};
			\node[fill=red,diamondd, label=250:$a$] (va) at (3,-0.4) {};
			\node[fill=blue,triangle, label=below:$z$] (vz) at (3.4,-0.5) {};
			\node[fill=red,diamondd, label=below:$\tau_1$] (t1) at (4,-1.2) {};
			\node[fill=yellow, label=below:$\tau_3$] (t3) at (5,-1.2) {};
			\node[fill=red,diamondd, label=below:$d$] (vd) at (5.5,-0.7) {};
			\node[fill=gray, label=-30:$s$] (vs) at (5.8,-0.5) {};
			\node[fill=black, label=-30:$t$] (vt) at (6.5,-0.25) {};
			\node[fill=blue,triangle, label=$\tau_2$] (t2) at (4.5,2) {};
			\node[fill=green, label=$u$] (vu) at  (3.8,1) {};
			
			\draw[dash dot] (vz) to[out=340,in=120] (t1);
			\draw[thick] (vi)--(vp)--(va);
			
			\draw[thick] (t1)--(t3);
			\draw[dash dot] (t3)  to[out=70,in=200] (vd);
			\draw[thick] (vd)--(vs);
			\draw[dash dot] (vs)  to[out=40,in=180] (vt);
			\draw[thick] (vt)--(vj);
			\draw[dash dot] (vj)  to[out=170,in=320] (t2);
			\draw (t2)--(vu)--(vi);
			
			\draw[dotted] (vt)--(vu);
			\draw[->] (vu)--(6.5,1.5);
			\draw[->] (vu)--(6.9,1.3);

			\draw (vi)--(vj);
			\draw (vp)--(vt);
			
			\draw (vu)--(vp);
			\draw (vu)--(va);
			\draw (vu)--(vz);
			\draw (vu)--(vt);
			\draw (vu)--(vs);
			\draw (vu)--(vd);
			
			\draw (vz)--(vd);
			\draw (vz)--(vs);
			\draw (vz)--(vt);
			
			\draw (t2)--(t1);
			\draw (t2)--(t3);
		\end{tikzpicture}
	}
	
	\subfloat[]{
		\begin{tikzpicture}[scale=0.8,
			diamondd/.style = {inner sep=1.6pt, diamond},	
			triangle/.style = {inner sep=1.2pt, regular polygon, regular polygon sides=3},	
			dtriangle/.style = {inner sep=1.2pt, regular polygon, regular polygon sides=3, rotate=180}	
			]	
			\tikzstyle{every node}=[draw, shape=circle, minimum size=3pt, inner sep=2pt];
			\node[fill=black, label=left:$i$] (vi) at (0,0) {};
			\node[fill=black, label=$j$] (vj) at (9,0) {};
			\node[fill=black, label=200:$p$] (vp) at (1, -0.2) {};
			\node[fill=blue,triangle, label=250:$a$] (va) at (2,-0.4) {};
			\node[fill=red,diamondd, label=250:$g$] (vg) at (3,-0.6) {};
			\node[fill=red,diamondd, label=below:$\tau_1$] (t1) at (4,-1.2) {};
			\node[fill=yellow, label=below:$\tau_3$] (t3) at (5,-1.2) {};
			\node[fill=black] (vt) at (6.7,-0.3) {};
			\node[fill=black] (vt') at (6,0.2) {};
			\node[fill=blue,triangle, label=$\tau_2$] (t2) at (4.5,2.3) {};
			\node[fill=green, label=$b$] (vb) at  (4,1.5) {};
			\node[fill=black, label=$q$] (vq) at  (3.2,0.5) {};
			\node[fill=black] (vq') at  (2.7,0.2) {};
			
			\draw[thick] (vi)--(vp)--(va)--(vg);
			\draw[thick] (t1)--(t3);
			\draw[dash dot] (vg) to[out=340,in=120] (t1);
			\draw[dash dot] (t3)  to[out=80,in=190] (vt);
			\draw[thick] (vt)--(vj)--(vt')--(t2);
			\draw[thick] (t2)--(vb)--(vq)--(vq')--(vi);
			
			\draw (vi)--(vj);
			\draw (t1)--(t2)--(t3);
			\draw (vb)--(vg)--(vq);
		\end{tikzpicture}
	}
	\caption{Illustration of the proof of Lemma~\ref{mainlem}:
		The vertices whose colors shall be uniquely determined next, are
		now drawn in gray. Polygonal boundaries containing multiple
		vertices not included in the figure are drawn with dashed lines.
		(a) $p$ forms a $K_4$ with $a$, $t$ and $u$. (b) $p$ forms a $K_4$ with $a$, $u$ and $w$.
		(c) $s$ forms a $K_4$ with $u$, $d$ and $z$. (d) $g$, $q$ and $b$ form a $K_3$.}
	\label{fig:cases}
\end{figure}

\begin{figure}
	\centering
	\begin{tikzpicture}[scale=2.8,
		diamondd/.style = {inner sep=1.6pt, diamond},	
		triangle/.style = {inner sep=1.2pt, regular polygon, regular polygon sides=3},	
		dtriangle/.style = {inner sep=1.2pt, regular polygon, regular polygon sides=3, rotate=180}	
		]	
		\tikzstyle{every node}=[draw, shape=circle, minimum size=3pt, inner sep=2pt];
		%\draw[dash dot] plot [smooth, tension=1] coordinates {(0.5, 0.5) (1.2,0.3) (2,0.6)};
		\draw[dash dot] plot [smooth, tension=1] coordinates {(1, -0.1) (1.5,0.1) (2.5,0.3)};
		\node[fill=red,diamondd, label=left:$\tau_1$] (t1) at (0,-0.1) {};
		\node[fill=blue,triangle, label=$\tau_2$] (t2) at (0.5, 0.5) {};
		\node[fill=yellow, label=right:$\tau_3$] (t3) at (1, -0.1) {};
		\node[fill=blue,triangle, label=$a$] (va) at (2,0.6) {};
		\node[fill=gray, label=$u$] (vu) at (2.3,0.8) {};
		\node[fill=green, label=$b$] (vb) at (2.9,0.6) {};
		\node[fill=red,diamondd, label=right:$d$] (vd) at (2.9,0.37) {};
		\node[fill=black] (vbl) at  (3.45,0.1) {};
		\node[fill=blue,triangle, label=$i$] (vi) at (3.7,0.6) {};
		
		\draw[thick] (t1)--(t2)--(t3)--(t1);
		\draw[thick] (t1)--(vi);
		\draw (2.5,0.3)--(vbl);
		\draw (vu)--(va);
		\draw (vu)--(vb);
		\draw (va)--(vb)--(vd);
		\draw (va)--(vd);
		\draw (vb)--(vd);
		\draw (vu)--(vd);
		\draw (va)--(vbl);
		\draw (vd)--(vbl);
		\draw[dash dot, bend right] (t2) to (va);
		\draw[dash dot] (vb)--(vi);	
	\end{tikzpicture}
	
	\caption{The vertex $a$ has an edge incident to one of the vertices of
		$a,b,c$, where $a$, $b$ and $c$ lie on the already colored 
		shortest paths from $\tau_1$ and $\tau_3$ to $i$.}
	\label{fig:algo1}
\end{figure}
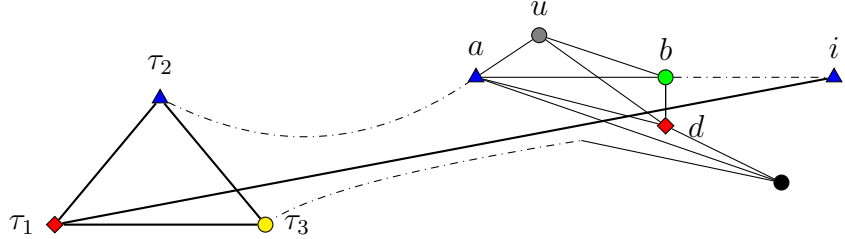
\begin{lem} \label{mainlem}
	Assume that two vertices $i \in \partial(\tau_1, \tau_2)$ and $j \in \partial(\tau_2, \tau_3)$ see each other, and the edge $ij$
	intersects $\tau_1\tau_2$ and $\tau_2 \tau_3$. Then the colors of all vertices on the four paths
	$\Pi(\tau_1, i)$, $\Pi(\tau_2, i)$, $\Pi(\tau_2, j)$ and $\Pi(\tau_3, j)$, 
	including $i,j$
	themselves, are uniquely determined by the colors of $T$.
\end{lem}

\begin{proof}
	We prove the claim by induction on the four paths.
	As the base case, the first vertices of these paths are the vertices
	of $T$, which are already assigned different colors.
	
	For the induction step, assume that $\Pi(\tau_1,i)$, $\Pi(\tau_2,i)$,
	$\Pi(\tau_2,j)$ and $\Pi(\tau_3,j)$ have been colored till vertices $a$, $b$, $c$ and
	$d$ respectively.
	Also, their immediate uncolored successors on $\Pi(\tau_1,i)$, $\Pi(\tau_2,i)$,
	$\Pi(\tau_2,j)$ and $\Pi(\tau_3,j)$ are $p$, $q$, $r$ and $s$ respectively.
	We aim to show that the colors of at least one of $p$, $q$, $r$ and $s$
	is uniquely determined by the already colored vertices.
	
	If $p$ does not see $b$ and any predecessor of $b$ on $\Pi(\tau_2,i)$, then 
	$q$ must see $a$ or some predecessor of $a$ on $\Pi(\tau_1,i)$. 
	We have the following cases (cf.~Figure~\ref{fig:cases}).
	
	\paragraph*{Case 1: $p$ sees $b$ or some predecessor of $b$ on $\Pi(\tau_2,i)$.}
	By definition, $p$ is the immediate successor of $a$ on $\Pi(\tau_1,i)$,
	and thus $p$ must see $a$.
	The right tangent of $a$ to $\Pi(\tau_2,i)$ lies to the right of the right tangent of $p$ to $\Pi(\tau_2,i)$. So, if the right tangent of $p$ touches $\Pi(\tau_2,i)$ at a vertex $u$, then $a$ sees $u$.
	Note that either $u = b$, or $u$ precedes $b$ on $\Pi(\tau_2,i)$. In any case, $u$ is already colored.
	Since $p$, $a$ and $\Pi(\tau_3,j)$ lie on the same side of $ij$, and $p$ is nearer
	to $ij$ than $a$ is, $p$ and $a$ see a vertex $t$ of $\Pi(\tau_3,j)$.
	If $u$ also sees $t$, and $t$ is already colored, then the claim is proved
	(Figure~\ref{fig:cases}(a)).
	So we consider the other two cases, namely, that $u$ does not see $t$, or that $t$ is not yet colored.
	
	\paragraph{Subcase 1.a:  $u$ does not see $t$.}
	
	Since $t$ and $u$ lie on different sides of $\overline{ij}$ and of $\overline{\tau_2\tau_3}$,	
	the edges of some vertex of $\Pi(\tau_2,j)$ must be blocking $u$ and $t$.
	Let $w$ be the first vertex of $\Pi(\tau_2,j)$ blocking the visibility between $u$ and $t$.
	
	Then $u$ sees $w$. The vertex $w$ is closer to $\overline{ij}$ than $u$ is.
	Also, $w$ lies to the right of $\overrightarrow{au}$ and $\overrightarrow{pu}$,
	and to the left of $\overrightarrow{at}$ and $\overrightarrow{pt}$. 
	Then the only possible segments blocking visibility between $w$ and either $p$ or $a$ can be from $\Pi(\tau_2,i)$.
	But all the vertices on $\Pi(\tau_2,i)$ preceding $u$ are farther from $\overline{ij}$ than $u$ is.
	So, there can be no such blocker, and $w$ must be visible from both $a$ and $p$ (Figure~\ref{fig:cases}(b)).
	If $w$ is already colored, then the claim is proved.
	
	Suppose that $w$ is not already colored.
	Then consider $r$, which now precedes $w$ on $\Pi(\tau_2,j)$.
	The vertices $r$ and $c$ are consecutive on $\Pi(\tau_2,j)$ and hence see each other.
	
	Since $\Pi(\tau_2,j)$ and $\Pi(\tau_1,i)$ are on opposite sides of the edge $ij$, the vertices $c$ and $r$ both see $a$, or some vertex preceding $a$ on $\Pi(\tau_1,i)$. Let $x$ be the last colored vertex
	of $\Pi(\tau_1,i)$ seen by both $c$ and $r$. 
	Let $x$ be the last colored vertex of $\Pi(\tau_1,i)$ seen by both $c$ and $r$.
	If $x \neq a$ then let $y$ be the last vertex of $\Pi(\tau_2,i)$ that blocks $c$ from the successor of $y$ on $\Pi(\tau_1,i)$.
	Then, $y$ must be visible from $x$, $r$ and $c$.
	
	Since $x$ precedes $a$ on $\Pi(\tau_1,i)$, and $y$ precedes $b$ on $\Pi(\tau_2,i)$, both $x$ and $y$ must already be colored, and thus the color of $r$ is uniquely determined by $T$.
	If $x=a$ then since $u$ is on the right tangent of $a$ to $\Pi(\tau_2,i)$, both $c$ and $r$ see $u$,
	and thus the color of $r$ is uniquely determined by $T$ in this case as well.
	
	Now we move to the second subcase.
	
	\paragraph{Subcase 1.b: $u$ sees $t$, but $t$ is not yet colored.}
	
	Since $u$ sees $t$, $\Pi(\tau_2,j)$ is a concave chain and the edge $\tau_1\tau_3$  exists in $\PP$,	
	$u$ must see every predecessor of $t$ on $\Pi(\tau_2,j)$.	
	This means that both $d$ and $s$ see $u$ (Figure~\ref{fig:cases}(c)).	
	Let the right tangent from $d$ touches $\Pi(\tau_1,i)$ in a vertex~$y$.	
	Then $s$ must see $y$, because the last vertices $i$ and $j$ of	
	the concave chains $\Pi(\tau_1,i)$ and $\Pi(\tau_3,j)$ see each other. 	
	Also, the left tangent of $u$ to $\Pi(\tau_1,i)$ must touch $\Pi(\tau_1,i)$ at a vertex equal to or preceding $y$.	
	Thus, all three of $s$, $d$ and $u$ see a common vertex $z$ on $\Pi(\tau_1,i)$ which precedes $a$, since $u$ and $t$ see $a$.	
	Thus, $z$ is already colored, and $u$, $d$ and $z$ form a $K_4$ with $s$, the color of $s$ is uniquely determined by $u$, $d$ and $z$.

	\paragraph*{Case 2: $p$ does not see $b$ or any predecessor of $b$ on $\Pi(\tau_2,i)$.}
	
	Since $\Pi(\tau_2,i)$ is a concave chain, this means that the tangent drawn from $p$ to $\Pi(\tau_2,i)$ in the direction of $\tau_2$, has whole $\Pi(\tau_2,b)$ to its left (refer to Figure~\ref{fig:cases}(d)).	
	Suppose that $b$ does not see any vertex of $\Pi(\tau_1,a)$.	
	Since $\Pi(\tau_1,i)$ is also a concave chain, and $\tau_1$ sees $\tau_2$, all segments	
	blocking visibility between $q$ and $\Pi(\tau_1,a)$ must come from $\Pi(p,i)$, and must include $p$.	
	But then, the aforementioned tangent drawn from $p$ to $\Pi(\tau_2,i)$	
	must have at least part of $\Pi(\tau_2,b)$ to the right, which is absurd.	
	
	So, $b$ must see $a$ or some other vertex of $\Pi(\tau_1,a)$.	
	Let $g$ denote the last vertex of $\Pi(\tau_1,i)$ seen by $b$.	
	Then, $g$ exists, and it belongs to $\Pi(\tau_1,a)$ since $p$ (the successor of $a$) does not see~$b$, and $g$ is seen by~$q$ (Figure \ref{fig:cases}(d)).	
	Since the vertex $g$ is on $\Pi(\tau_1,a)$, it is already colored.	
	
	Let us now similarly consider a vertex, say $t$ on $\Pi(\tau_2, j)$,	
	which is seen by both $b$ and $q$.	
	Suppose that $g$ or another common colored neighbor of $b$ and $q$ sees $t$.	
	Then we are immediately done if $t$ is colored, or we are in Subcase 1.b if $t$ is uncolored.	
	Otherwise, some vertex on $\Pi(\tau_3, j)$ blocks all visibilities between $t$	
	and all the common neighbors of $q$ and $b$. 	
	Then we finish as in Subcase 1.a.	
\end{proof}

\begin{cor} \label{maincor}
	If any vertex $i$ of $\PP$ sees a vertex of $T$ and their visibility edge crosses one of the edges of $T$, then the color of $i$ is uniquely determined by the colors of $T$.
\end{cor}
\begin{proof}
	Without loss of generality, suppose that $i$ sees $\tau_1$, and $i\tau_1$ crosses $\tau_2\tau_3$. Then $j=\tau_1$, $\Pi(\tau_2,j) = \tau_2\tau_1$
	and $\Pi(\tau_1,j) =\tau_1$, and Lemma~\ref{mainlem} proves the claim.
\end{proof}

\begin{thm} \label{lemalgo}
	If a reduced polygon is 4-colorable, then it has a unique 4-coloring up to permutation of colors.
\end{thm}
\begin{proof}
	Consider a triangle $T$ in a reduced polygon $\PP$. 
	If $\PP$ is not just $T$, then at least one edge of $T$ is not a boundary edge of $\PP$.
	Without loss of generality, let $\tau_1\tau_2$ be such an edge.
	
	Since $\PP$ is reduced, there must be a vertex $i$ on $\partial(\tau_1, \tau_2)$ 
	such that an edge incident to $i$ crosses $\tau_1 \tau_2$.
	By Lemma \ref{mainlem} and Corollary \ref{maincor}, if $\PP$ is 4-colorable, then 
	all vertices on the paths $\Pi(\tau_1, i)$ and $\Pi(\tau_2, i)$, including $i$ have a 4-coloring uniquely determined by $T$.
	In case $\tau_2 \tau_3$ or $\tau_3 \tau_1$ are not boundary edges of $\PP$, we can similarly find $j$ on
	$\partial(\tau_2, \tau_3)$ 	and $k$ on $\partial(\tau_3, \tau_1)$ and uniquely 4-color $\Pi(\tau_2, j)$,
	$\Pi(\tau_3, j)$, $\Pi(\tau_3, k)$ and $\Pi(\tau_1, k)$. 
	\begin{figure}	
		\centering	
		\begin{tikzpicture}[scale=3,	
			diamondd/.style = {inner sep=1.6pt, diamond},	
			triangle/.style = {inner sep=1.2pt, regular polygon, regular polygon sides=3},	
			dtriangle/.style = {inner sep=1.2pt, regular polygon, regular polygon sides=3, rotate=180}	
			]	
			\tikzstyle{every node}=[draw, shape=circle, minimum size=3pt, inner sep=2pt];	
			\node[fill=red,diamondd, label=left:$t_3$] (t1) at (0,-0.1) {};	
			\node[fill=blue,triangle, label=$t_1$] (t2) at (0.5, 0.7) {};	
			\node[fill=yellow, label=right:$t_2$] (t3) at (1, -0.1) {};	
			\node[fill=blue,triangle, label=$a$] (va) at (1.7,0.53) {};	
			\node[fill=gray, label=$u$\hbox to 0pt{~($z$)\hspace*{-2em}}] (vu) at (2.1,0.85) {};	
			\node[fill=green,dtriangle, label=below:$b$] (vb) at	
			(2.9,0.55) {};	
			\node[fill=yellow, label=right:$w$] (vd) at (2.9,0.45) {};	
			\node[fill=blue,triangle, label=$i$] (vi) at (3.7,0.6) {};	
			
			\draw[thick] (t2)--(t1)--(t3);	
			\draw (vu)--(va);	
			\draw (vu)--(vb);	
			\draw (va)--(vb)--(vd);	
			\draw (va)--(vd);	
			\draw (vb)--(vd);	
			\draw[dash dot] (t2) to[out=340,in=180] (va);	
			\draw[dash dot] (vb)--(vi);		
			\draw[dash dot] (t3) to[out=20,in=190] (vi);		
			\tikzstyle{every path}=[color=cyan];	
			\draw[thick] (0.4,0.2)--(vi);	
			\draw[thick] (t2)--(t3);	
			\draw[thick] (vu)--(vd);	
		\end{tikzpicture}	
		
		\bigskip	
		\begin{tikzpicture}[scale=3,	
			diamondd/.style = {inner sep=1.6pt, diamond},	
			triangle/.style = {inner sep=1.2pt, regular polygon, regular polygon sides=3},	
			dtriangle/.style = {inner sep=1.2pt, regular polygon, regular polygon sides=3, rotate=180}	
			]	
			\tikzstyle{every node}=[draw, shape=circle, minimum size=3pt, inner sep=2pt];	
			\node[fill=red,diamondd, label=left:$t_3$] (t1) at (0,-0.1) {};	
			\node[fill=blue,triangle, label=$t_1$] (t2) at (0.5, 0.7) {};	
			\node[fill=yellow, label=right:$t_2$] (t3) at (1, -0.1) {};	
			\node[fill=blue,triangle, label=$a$] (va) at (1.7,0.53) {};	
			\node[fill=gray, label=$u$] (vu) at (1.9,0.85) {};	
			\node[fill=green,dtriangle, label=below:$b$] (vb) at (2.9,0.55) {};	
			\node[fill=yellow, label=below:$c$] (vc) at (1.3,0.06) {};	
			\node[fill=red,diamondd, label=below:$d$] (vd) at (3.2,0.5) {};	
			\node[fill=white] (vbl) at  (2.7,0) {};	
			\node[fill=blue,triangle, label=$i$] (vi) at (3.7,0.6) {};	
			
			\draw[thick] (t2)--(t1)--(t3);	
			\draw (vu)--(va);	
			\draw (vu)--(vb);	
			\draw (va)--(vb)--(vd);	
			\draw (va)--(vc)--(vb);	
			\draw (va)--(vd);	
			\draw (vb)--(vd);	
			\draw (vd)--(vc);	
			\draw[dash dot] (t2) to[out=340,in=180] (va);	
			\draw[dash dot] (vb)--(vi);		
			\draw[dash dot] (t3) to[out=50,in=190] (vc);		
			\draw[dash dot] (vd)--(vi);		
			\tikzstyle{every path}=[color=cyan];	
			\draw[thick] (0.4,0.2)--(vi);	
			\draw[thick] (t2)--(t3);	
			\draw[thick] (vu)--(vbl);	
		\end{tikzpicture}	
		
		\caption{Illustration of the proof of Theorem~\ref{lemalgo}.	
			Top: an edge incident to $u$ sees a colored vertex on $\Pi(t_2, i)$.	
			Bottom: an edge incident to $u$ crosses an edge $cd$ of $\Pi(t_2, i)$.	
			Both cases can be resolved by an application of Lemma	
			\ref{mainlem} and Corollary \ref{maincor} to the already	
			colored triangles $abw$ and $abc$, respectively.}	
		\label{fig:algo}
	\end{figure}
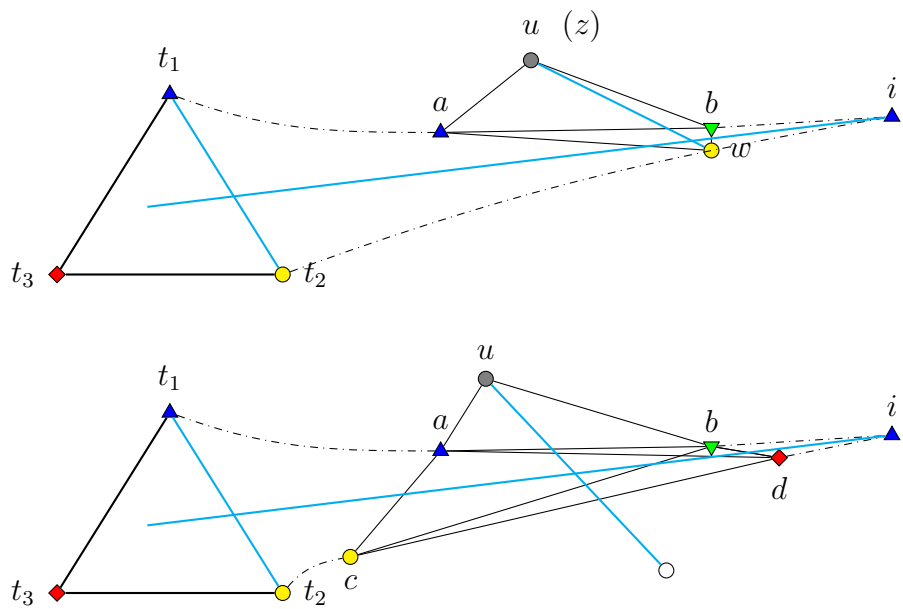
	Now, all the remaining uncolored vertices of $\PP$ are on polygonal chains of the form
	$\partial(a, b)$,
	where $a$ and $b$ are two consecutive vertices in one of the six paths mentioned above. Furthermore,
	no vertex in the polygonal chain $\partial(a, b)$, other than $a$ and $b$, is colored.
	Without loss of generality, let $a$ and $b$ be two consecutive vertices on $\Pi(\tau_1, \tau_2)$.
	If $a b$ is not a boundary edge of $\PP$, then since $\PP$ is reduced, there must be an
	uncolored vertex $u$ in $\partial(a, b)$ such that an edge incident to $u$ crosses $a b$.
	This edge is either incident to a vertex of $\Pi(\tau_2, i)$, or crosses an edge of
	$\Pi(\tau_2, i)$.
	Consider the case where such an edge to a vertex of $\Pi(\tau_2, i)$ exists.
	Then consider a vertex $w$ that is closest to $ab$ among all the vertices
	of $\Pi(\tau_2, i)$ that see a vertex (say, $z$) of $\partial(a, b)$. Since the edge $wz$ exists,
	$w$ cannot be blocked by any vertex of $\Pi(\tau_1, i)$. Due to the choice of $w$, no vertex of 
	$\Pi(\tau_2, i)$ can block $w$ from $a$ or $b$. So, $w$ sees both $a$ and $b$.
	Now consider the case where no vertex of $\partial(a, b)$ sees any vertex of $\Pi(\tau_2, i)$, but some vertex of  
	$\partial(a, b)$ sees some vertex of $\partial(c, d)$,
	where $c$ and $d$ are consecutive points on $\Pi(\tau_2, i)$. 
	Without loss of generality, assume that $c$ precedes $d$ in $\Pi(\tau_2, i)$. Then $c$ must see both $a$ and $b$ (Figure \ref{fig:algo}),
	for otherwise a vertex of $\partial(a, b)$ must have an edge with some 
	vertex of $\Pi(\tau_2, i)$ acting as a blocker for $c$, contrary to our assumption.
	Then, in the above two cases, based on the triangle $a b w$ and $a b c$, respectively,
	again Lemma \ref{mainlem} and Corollary \ref{maincor} can be used to 
	uniquely determine a 4-coloring for $\Pi (a, u)$ and $\Pi(b, u)$.
	
	Now we generalize the above procedure. Let $T_0 = \{T\}$, and let $S_0 = \{
	\Pi(\tau_1, i),\Pi(\tau_2, i)\Pi(\tau_2, j), \Pi(\tau_3, j), \Pi(\tau_3, v_k), \Pi(\tau_1, v_k)\}$.
	Note that we have assumed that none of the edges of $T$ are boundary edges. If some edges of $T$ are boundary edges then $S_0$ will have less elements.
	By the above procedure, we can uniquely 4-color all the vertices of all elements of $S_0$. Now, all the uncolored vertices lie on
	$\partial(a, b)$,
	where $a$ and $b$ are consecutive vertices of some element of $S_0$. For each such $a b$, we find a new triangle $a b c$ or $a b d$,
	and two new shortest paths of the form $\Pi (a, u)$ and $\Pi(b, u)$. Let $\tau_1$ denote the set of all such new triangles, and $S_1$ denote the set of all
	new shortest paths obtained from $T_0$ and $S_0$. Now, the remaining uncolored vertices must line on polygonal chains of the form
	$\partial(v_e, v_f)$
	where $v_e$ and $v_f$ are two consecutive vertices of some element of $S_1$. In general, following the same method
	we can always construct $T_{i+1}$ and $S_{i+1}$ from $T_i$ and $S_i$, until all vertices of $\PP$ are colored. Since in each step, the colors of vertices
	are uniquely determined, it follows that if $\PP$ has a 4-coloring, then it must be unique.
\end{proof}

\subsection{Computing a 4-coloring without a polygonal representation} \label{sec:withoutrepr}

In the previous section, we have proved that if a reduced polygon is 4-colorable, then its 4-coloring must be unique up to permutations.
Now we use the property to derive a polynomial time 4-coloring algorithms for the visibility graph of a polygon, even when the polygonal embedding
or boundary are not given.
First we need to define a few structures and operations.

\begin{dfn}[Bottleneck pair]
	A pair of adjacent vertices whose removal disconnects $G$ is called a bottleneck pair.
\end{dfn}

\begin{dfn} \label{def:bottleneck}
	Call a pair of adjacent vertices whose removal disconnects
	a given graph $G$, a \emph{bottleneck pair}.
	Consider removing all the bottleneck pairs from~$G$. We are left with connected components of $G$. 
	
	Now, consider any bottleneck pair $(x,y)$. Suppose that $x$ and $y$ were earlier adjacent to a set of vertices $S_x$ and $S_y$ of a connected component $C_i$.
	Then create a copy of $(x,y)$ and re-connect them with edges with the vertices of $S_x$ and $S_y$ respectively. Do this with every bottleneck pair of $G$.
	Call the subgraphs of $G$ so formed as \emph{reduced subgraphs} of $G$.
\end{dfn}

We remark that, in structural graph theory, the process described by
Definition~\ref{def:bottleneck} is known as a ``clique-cutset
decomposition'', specifically with cliques of size~$2$.
Since we do not assume readers closely familiar with structural graph
theory, we prefer our standalone definition.
We have the following lemma.

\begin{lem} \label{lem:bottleneck}.
	Let $G$ be the visibility graph of a polygon~$\PP$.
	Each bottleneck pair of $G$ corresponds to an internal edge of $\PP$ that is not intersected by any other internal edge of $\PP$, and vice versa.
\end{lem}

\begin{proof}
	We use the same notations for the vertices of $G$ and their corresponding vertices of $\PP$.
	Consider any internal edge $xy$ of $\PP$ such that no other internal edge of $\PP$ intersects it. 
	Then disconnecting the edge $xy$ and the vertices $x$ and $y$ disconnects $G$.
	So $(x,y)$ is a bottleneck pair.
	Conversely, suppose that $(x,y)$ is a bottleneck pair.
	Then $xy$ is an internal edge of $\PP$ since deleting a boundary edge does not disconnect~$G$.
	Let $P_1$ and $P_2$ be the two subpolygons of $\PP$ divided by $xy$.
	If there was a visibility edge from $P_1$ to $P_2$ not
	incident to $x,y$, then since the visibility graphs of $P_1$
	and of $P_2$ are connected, deleting $xy$ would again not disconnect~$G$.
	So, is an internal edge of $\PP$ not intersected by any other internal edge of $\PP$.
\end{proof}

The corollary below follows immediately from Lemma~\ref{lem:bottleneck}.

\begin{cor}
	Each reduced subgraph of $G$ is the visibility graph of some reduced subpolygon of $\PP$.
	Likewise, each reduced subpolygon of $\PP$ has a reduced subgraph of $G$ as its visibility graph.
\end{cor}

\makeatletter
\renewcommand{\@algocf@capt@plain}{above}% formerly {bottom}
\makeatother
\begin{algorithm}[htbp]
	\caption{Algorithm to decide 4-colorablity of visibility graphs of polygons~~~~~~~~~~}
	\label{alg:4coloring}
	\KwIn{A graph $G$ with the promise of being the visibility graph of
		a simple polygon}
	\KwOut{Whether $G$ is $4$-colorable or not. If so, then a proper $4$-coloring of $G$.}
	\smallskip
	Identify all edges $uv$ of $G$ such that removal of $u$ and $v$ disconnects $G$\;
	Delete all these bottleneck pairs (i.e.,~$u,v$) and partition $G$ into connected components $G_1, G_2, \ldots, G_k$\;
	To each connected component of $G$, add copies of the bottleneck pairs which were originally attached to it\;
	\For{each connected component $G_i$}
	{
		Locate a triangle $T$ in $G_i$ and assign three colors to its vertices\;
		\Repeat{each vertex in $G_i$ is colored}{
			Locate a vertex adjacent to three colored vertices of distinct colors and assign it the fourth color\;
		}
	}
	\If{two adjacent vertices receive the same color}
	{
		Output `non-4-colorable'\;
		Terminate\;
	}
	Glue the connected components back by merging the corresponding two copies of each bottleneck pair\;
	Permute the colors of the vertices after glueing so that there is no conflict\;
\end{algorithm}

% Algorithm~\ref{alg:4coloring} shows that 4-colorability of a given polygon visibility graph is decidable in polynomial time.

Now, in light of the above Algorithm~\ref{alg:4coloring} and Theorem \ref{lemalgo}, we prove Theorem~\ref{colthm}.

\begin{proof} [Proof of Theorem \ref{colthm}]
	Corollary~\ref{maincor} shows that the reduced subgraphs correspond to reduced polygons. 
	By Theorem~\ref{lemalgo} (and its proof), 4-colorable reduced polygons have unique 4-colorings
	which can be found iteratively by coloring each time a vertex with some
	three previously distinctly colored neighbours. 
	Since the algorithm always chooses a color for a vertex by
	this iterative scheme, the computed (partial) 4-coloring is the only one possible.
	%unique, if $G$ is not 4-colorable.
	So, the algorithm is correct.
	
	Let the number of vertices and edges in $G$ be $n$ and $m$ respectively. The
	bottleneck pairs that do not cross any other chord, can be found in $O(m^2)$ time.
	Thus, the decomposition of $\PP$ into reduced subpolygons takes $O(m^2)$ time.
	A vertex adjacent to every vertex of a colored triangle can be found in $O(n)$ time.
	While computing the coloring on the shortest paths, a pointer
	can be kept on each of the shortest paths, and the coloring takes $O(n)$ time. The coloring step can be iterated at most once 
	for each vertex, so the complexity for all vertices is $O(n^2)$. Checking for conflict takes $O(m)$ time. Finally, rejoining 
	the reduced subgraphs takes $O(n)$ time. Thus, the complexity of the algorithm is $O(m^2)$.
\end{proof}

\subsection{NP-hardness reduction of proper 5-coloring} \label{sec:hardness5}

In this section we prove that the problem of deciding whether the visibility
graph $G$ of a given simple polygon $\PP$ can be properly colored with 
$5$ colors, is NP-complete.

Membership of our problem in NP is trivial (since $G$ can be efficiently
computed from $\PP$ and then a coloring checked on~$G$).
We are going to present a polynomial reduction from the NP-hard problem of
$3$-colorability of general graphs.
Our reduction shares some common ideas with reductions on visibility
graphs presented in \cite{Kara_pointVisChromatic},
but the main difference is in not using the SAT problem (which makes our case even simpler).
The rough outline of the reduction is depicted in
Figure~\ref{fig:hardness-rough}.

\begin{figure}[]
	$$%\small
	\begin{tikzpicture}[scale=3]
		\tikzstyle{every node}=[minimum size=3pt, inner sep=0pt];
		\draw (2.2,0.25) node {edges of $H$} ;
		\draw (2.2,2.2) node {vertices of $H$} ;
		\draw (2.4,1.97) node {\bf\dots\dots\dots} ;
		\draw (2.4,0.52) node {\bf\dots\dots\dots} ;
		\draw (0.5,0.3) node {$v_1v_4$} ;
		\draw (0.9,0.35) node {$v_2v_6$} ;
		\draw (1.3,0.38) node {$v_4v_8$} ;
		\draw (1.7,0.40) node {$v_7v_{11}$} ;
		\tikzstyle{every path}=[draw,color=lightgray, thick];
		\draw (0,0.5) arc (101:79:11) -- (4.2,2) ;
		\draw (0,0.5) -- (0,2) arc (-101:-79:11) ;
		\tikzstyle{every path}=[draw,thick, color=black];
		\tikzstyle{every node}=[draw,shape=circle,fill=black, minimum size=2pt, inner sep=0pt];
		\draw (0,0.5) -- (0,2)
		-- (0.2,1.96) -- (0.25,2.05) node[label=above:$v_1$] {} -- (0.3,1.95)
		-- (0.4,1.93) -- (0.45,2.02) node[label=above:$v_2$] {} -- (0.5,1.915)
		-- (0.6,1.90) -- (0.65,1.99) node[label=above:$v_3$] {} -- (0.7,1.89)
		-- (0.8,1.875) -- (0.85,1.96) node[label=above:$v_4$] {} -- (0.9,1.865)
		-- (1.0,1.85) -- (1.05,1.94) node[label=above:$v_5$] {} -- (1.1,1.84)
		-- (1.2,1.83) -- (1.25,1.92) node[label=above:$v_6$] {} -- (1.3,1.825)
		-- (1.4,1.82) -- (1.45,1.9) node[label=above:$v_7$] {} -- (1.5,1.815)
		-- (1.6,1.81) -- (1.65,1.89) node[label=above:$v_8$] {} -- (1.7,1.803)
		-- (1.8,1.803) -- (1.85,1.89) node[label=above:$v_9$] {}
		-- (1.9,1.797) -- (2.0,1.797) ;
		\draw (4.2,0.5) -- (4.2,2)
		-- (4.0,1.96) -- (3.95,2.05) node[label=above:$v_n$] {} -- (3.9,1.95)
		-- (3.8,1.93) -- (3.75,2.02) node[label=above:$v_{n-1}$] {} -- (3.7,1.915)
		-- (3.6,1.90) -- (3.55,1.99) node[label=above:$\!\!\!\!\!v_{n-2}$] {}
		-- (3.5,1.89) -- (3.4,1.875) ;
		\draw (0,0.5) -- (0.5,0.58) -- (0.4,0.4) -- (0.6,0.4) -- (0.52,0.585)
		-- (0.9,0.635) -- (0.8,0.45) -- (1.0,0.45) -- (0.92,0.64)
		-- (1.3,0.67) -- (1.2,0.48) -- (1.4,0.48) -- (1.32,0.675)
		-- (1.7,0.692) -- (1.6,0.50) -- (1.8,0.50) -- (1.72,0.695) -- (1.9,0.7) ;
		\draw (4.2,0.5) -- (3.7,0.58) -- (3.8,0.4) -- (3.6,0.4) -- (3.68,0.585)
		-- (3.3,0.635) ;
		\tikzstyle{every path}=[draw, color=cyan, dashed,thick];
		\draw (0.855,1.96) -- (0.51,0.58) -- (0.25,2.05) ;
		\draw (0.445,2.02) -- (0.91,0.635) -- (1.255,1.92) ;
		\draw (0.85,1.96) -- (1.31,0.67) -- (1.655,1.89) ;
		\draw (1.45,1.9) -- (1.71,0.692) -- (2.3,1.87) ;
	\end{tikzpicture}
	$$
	\caption{A scheme of the polygon $\PP$ constructed from a
		given graph~$H$ in the proof of Theorem~\ref{thm:5hardness}.
		There are two mostly concave chains, top and bottom one.
		The top sawtooth chain features black-marked vertices
		$v_1,v_2,\ldots,v_n$ for each of the $n$ vertices of~$H$.
		The bottom chain contains, for each edge $v_iv_j$ of $H$
		(such as $v_1v_4,v_2,v_6,v_4,v_8,v_7v_{11}$ in the picture), a
		triangular pocket glued to $\PP$ by a tiny pinhole passage.
		This pocket of $v_iv_j$ is adjusted such that its lower
		corners can see precisely the top vertices $v_i$ and $v_j$,
		respectively (cf.~Figure~\ref{fig:hardness-vdetail}). 
		The important visibilities are sketched here with dashed lines.
		Altogether, we can get a proper $5$-coloring of
		the visibility graph of~$\PP$ if and only if the vertices
		$v_i$ and $v_j$ receive distinct colors for every edge
		$v_iv_j\in E(H)$.}
	\label{fig:hardness-rough}
\end{figure}
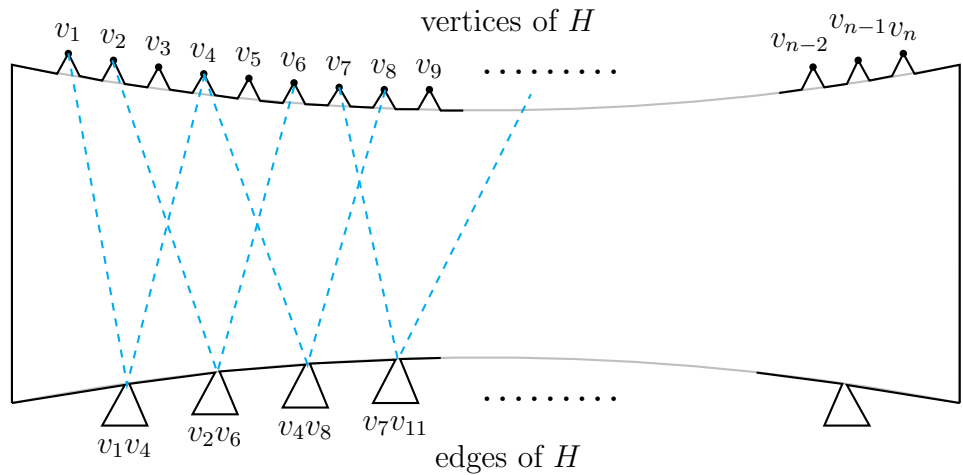

\begin{thm}
	\label{thm:5hardness}
	The problem -- given a simple polygon $\PP$ in the plane, to decide whether
	the visibility graph of $\PP$ is properly $k$-colorable -- is NP-complete for every~$k\geq5$.
\end{thm}

\begin{proof}
	As mentioned, the problem is in NP since one can construct the
	visibility graph $G$ of $\PP$ in polynomial time
	\cite{ORourke_vertexEdgeVis,Hershberger89} and then verify a coloring.
	In the opposite direction, we reduce from the NP-complete problem of
	$3$-coloring a given graph~$H$. Let $k=5$ now.
	
	Let $V(H)=\{v_1,\ldots,v_n\}$.
	The polygon $\PP$ constructed from $H$ is shaped as in Figure~\ref{fig:hardness-rough}.
	The top chain of $\PP$ consists of $3n+2$ vertices in a sawtooth
	configuration, such that the convex vertices of the teeth are marked
	by~$v_1,\ldots,v_n$. The picture is scaled such that each $v_i$ sees
	the whole bottom chain.
	The bottom chain contains, for each edge $v_iv_j\in E(H)$, $i<j$ (in an
	arbitrary order of edges), a ``pocket'' consisting of $5$ vertices
	$p_{ij}^1,p_{ij}^2,p_{ij}^3,p_{ij}^4,p_{ij}^5$ in order, as
	detailed in Figure~\ref{fig:hardness-vdetail}.
	Importantly, $p_{ij}^1$ and $p_{ij}^5$ are mutually so close that
	the vertices $p_{ij}^2,p_{ij}^3$ in the lower left corner can see
	only the vertex $v_j$ (of course, besides $p_{ij}^1$ and $p_{ij}^5$)
	and the vertex $p_{ij}^4$ in the lower right corner can see   
	only the vertex $v_i$.
	
	\begin{figure}[]
		$$%\small
		\begin{tikzpicture}[xscale=1.6, yscale=1.4]
			\tikzstyle{every path}=[draw, color=black];
			\tikzstyle{every node}=[draw,shape=circle,fill=black, minimum size=2pt, inner sep=0pt];
			\draw (2,0.95) -- (3,1) node[label=200:$\!p_{ij}^1~~~$] {} 
			-- (2.4,0.03) node[label=130:$p_{ij}^2$] (p2) {} 
			-- (2.43,0) node[label=270:$p_{ij}^3$] (p3) {} 
			-- (3.85,0) node[label=right:$p_{ij}^4$] (p4) {}
			-- (3.1,1.01) node[label=350:$~~~p_{ij}^5\!\!\!\!$] {} -- (4,1.05) ;
			\draw (0,2.5) -- (0.7,2.5) 
			-- (1.5,3) node[label=above:$v_i$] {} -- (2.3,2.5) -- (2.5,2.5) ; 
			\draw (3.3,2.53) -- (3.5,2.53) 
			-- (4.35,3.02) node[label=above:$v_j$] {} -- (5,2.55) -- (5.5,2.55) ; 
			\tikzstyle{every path}=[draw, color=cyan, dashed];
			\draw (1.6,3) -- (p4) -- (1.33,3);
			\draw (4.25,3) -- (p2) -- (4.53,3);
			\draw (4.44,3) -- (p3) -- (4.16,3);
		\end{tikzpicture}
		$$
		\caption{A detail (not to scale) of the pocket $v_iv_j$ from Figure~\ref{fig:hardness-rough}.
			Note that $v_j$ and $p_{ij}^4$ see the same four vertices
			$p_{ij}^1,p_{ij}^2,p_{ij}^3,p_{ij}^5$, and so they have to be colored the same.}
		\label{fig:hardness-vdetail}
	\end{figure}
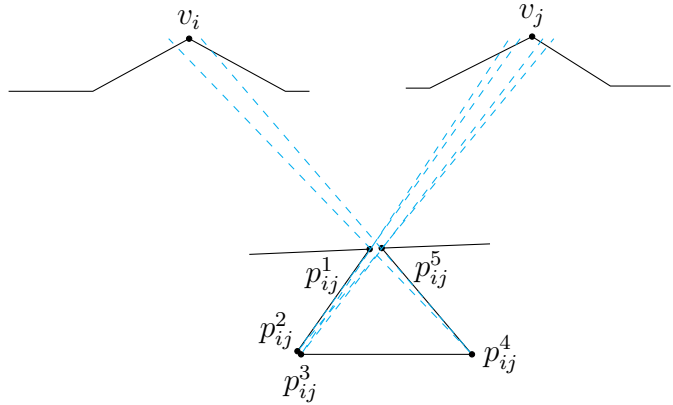
	
	Assume now that we have a proper $5$-coloring of the visibility
	graph $G$ of the constructed polygon~$\PP$.
	We easily argue the following:
	\begin{itemize}
		\item Choose any edge $v_iv_j\in E(H)$. Then the vertices $p_{ij}^1$
		and $p_{ij}^5$ of the corresponding pocket must receive distinct
		colors which we, up to symmetry, denote by $4$ and~$5$.
		Since every vertex of the top chain sees $p_{ij}^1$ and $p_{ij}^5$,
		we get that every vertex $v_b$, $b=1,2,\ldots,n$, has a color $1,2$ or~$3$.
		
		\item For each edge $v_iv_j\in E(H)$, the $5$-tuple of vertices
		$(v_j,p_{ij}^1,p_{ij}^2,p_{ij}^3,p_{ij}^5)$ of $\PP$ induces a $K_5$,
		and so does the nearly-identical $5$-tuple
		$(p_{ij}^1,p_{ij}^2,p_{ij}^3,p_{ij}^4,p_{ij}^5)$.
		Consequently, in any proper $5$-coloring of~$G$, the vertices $v_j$ 
		and $p_{ij}^4$ get the same color.
		And since $p_{ij}^4$ sees $v_i$, the colors of $v_i$ and $v_j$ must
		be distinct. 
	\end{itemize}
	Altogether, any proper $5$-coloring of the visibility
	graph $G$ of~$\PP$ implies a proper $3$-coloring of the graph~$H$.
	
	On the other hand, assume we have a proper coloring of the graph $H$ by
	colors $\{1,2,3\}$.
	We give the same colors to the vertices $v_1,\ldots,v_n$ of the top
	chain of~$\PP$, and we can always complete (e.g., greedily from left to right)
	this partial coloring to a proper $3$-coloring of the top chain of~$\PP$.
	Then we assign alternate colors $4,5,4,5,\ldots$ to the exposed
	vertices of the bottom chain.
	Finally, we color the lower corners of the bottom pockets as
	follows; for an edge $v_iv_j\in E(H)$, we give $p_{ij}^4$ the color
	of $v_j$, and to $p_{ij}^2,p_{ij}^3$ the remaining two colors among $1,2,3$.
	This gives a proper $5$-coloring of the visibility graph $G$ of~$\PP$.
	
	For each $k>5$, we naturally expand each pair of vertices
	$p_{ij}^2,p_{ij}^3$ from the above construction into a $(k-3)$-tuple
	of such vertices, and the same proof goes through. 
	
	The last bit is to show that the construction of $\PP$ can be realized
	in a grid of polynomial size in~$n=|V(H)|$ (assuming $k$ is constant).
	Both the top and bottom concave shapes can be realized as ``fat'' parabolas,
	requiring only rough resolution of $\Theta(n^2)$ in both
	horizontal and vertical directions.
	This is fully sufficient for the top chain, but realizing the pockets
	of the bottom chain is more delicate.
	Still, the precise placement of the pocket of an edge $v_iv_j$ depends
	only on the vertices $v_i$ and $v_j$ of the top chain, and not on
	other pockets. 
	Within the main scale, each pocket has dimensions $\Theta(n)$ and
	the pinhole opening is, say, $\frac1n$, and hence a sufficient precision
	for adjusting the pocket corners is $\Theta(\frac1n)$.
	Altogether, the construction of $\PP$ is
	achieved on an $\mathcal{O}(n^3)$ grid.
\end{proof}
\section{4-coloring problem for polygons with holes}
\label{sec:polygonswithholes}

Consider a polygon $\PP$ together with a collection of pairwise
disjoint polygons $Q_i$, $i=1,\ldots,k$, such that $Q_i\subseteq int(\PP)$.
Then the set $\PP \setminus int\big(Q_1\cup\ldots\cup Q_k\big)$ is
called a {\em polygon with holes}.
In this section we prove that, for polygons with holes,
$4$-colorability is already an NP-complete problem.
Given the algorithm for $4$-coloring from Section~\ref{sec:4coloringP}, it is natural
that the proof we are going to present should be very different from the
reduction in Section~\ref{sec:hardness5}.

For better clarity, we present a construction of a polygon with
holes as ``digging polygonal corridors in solid mass''.
These corridors (precisely, their topological closure) will then
form the point set of our polygon, while the ``mass trapped between''
corridors will form the holes in the polygon.
On a high level, our corridors will be composed of elementary
{\em channels}, as depicted in Figures~\ref{fig:hardness-vertstrip} and
\ref{fig:hardness-edgestrip}, placed along the lines of a large
hexagonal (honeycomb) grid in the plane. More details follow next.

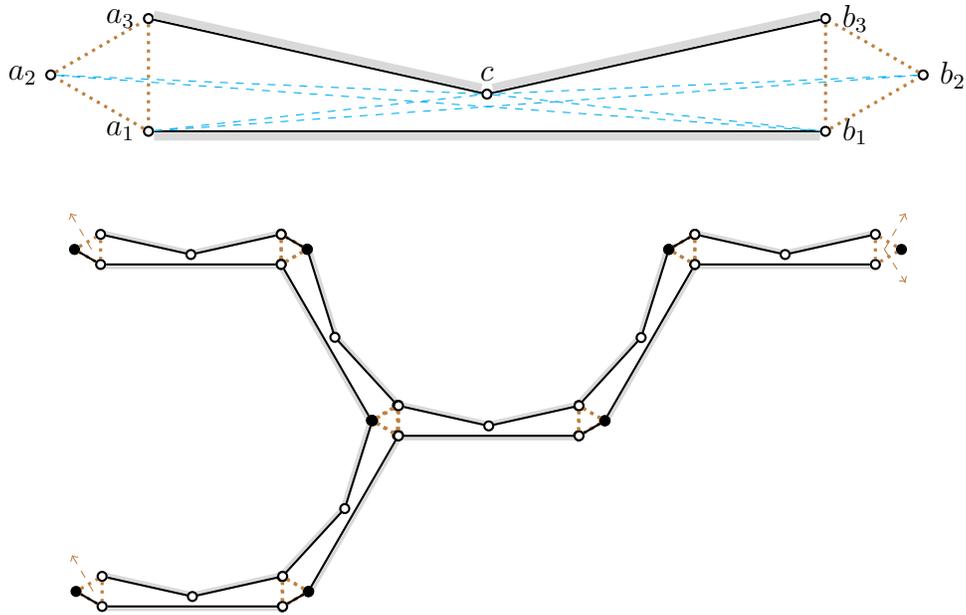
\begin{figure}[]
	\centering
	\begin{tikzpicture}[scale=1.5]
		\tikzstyle{every node}=[draw,solid,shape=circle,fill=white, inner sep=1.2pt];
		\tikzstyle{every path}=[draw, thick,solid, color=black];
		\draw (1,0) node[label=left:$a_1$] (a){}
		++(150:1) node[label=left:$a_2$] (b){}
		++(30:1) node[label=left:$a_3$] (c){} ;
		\draw (7,0) node[label=right:$b_1$] (k){}
		++(30:1) node[label=right:$b_2$] (l){}
		++(150:1) node[label=right:$b_3$] (m){} ;
		\draw (c) -- (4,0.33) node[label=above:$c$] (d){} -- (m) ;
		\draw (a) -- (k) ;
		\begin{scope}
			\tikzstyle{every path}=[draw, line width=3pt, color=gray!30];
			\draw[transform canvas={yshift=-2pt}] (a) -- (k) ;
			\draw[transform canvas={yshift=+2pt}] (c) -- (d) -- (m);
		\end{scope}
		\tikzstyle{every path}=[draw, color=brown, very thick, dotted];
		\draw (a) -- (b) -- (c) -- (a) ;
		\draw (k) -- (l) -- (m) -- (k) ;
		\tikzstyle{every path}=[draw, thin,dashed, color=cyan];
		\draw (a) -- (l) -- (d) -- (b) -- (k) ;
		\draw (a) -- (d) -- (k) ;
	\end{tikzpicture}
	\\~\\
	\begin{tikzpicture}[scale=0.4]
		\path [use as bounding box] (0,-6) rectangle (8,8);
		\def\vertchan#1#2{%
			\tikzstyle{every node}=[draw,solid,shape=circle,fill=white, inner sep=1.2pt];
			\tikzstyle{every path}=[draw, very thick,solid,	color=gray!30];
			\draw (1.08,-0.1) -- (6.92,-0.1) ;
			\draw (1.08,1.08) -- (4,0.42) -- (6.92,1.08) ;
			\tikzstyle{every path}=[draw, thick,solid, color=black];
			\draw (1,0) node (a){} ++(150:1) node[#2] (b){} ++(30:1) node (c){} ;
			\draw (7,0) node (k){} ++(30:1) node[#2] (l){} ++(150:1) node (m){} ;
			\draw (c) -- (4,0.33) node (d){} -- (m) ;
			\draw (a) -- (k) ;
			\tikzstyle{every path}=[draw, color=brown, very thick, dotted];
			\draw (a) -- (b) -- (c) -- (a) ;
			\draw (k) -- (l) -- (m) -- (k) ;
			\tikzstyle{every path}=[draw, thick,solid, color=black];
			#1 ;
		}
		\begin{scope}[transform canvas={xshift=-10mm,yshift=-24.2mm},rotate=60]
			\vertchan{}{} ;
		\end{scope}
		\begin{scope}[transform	canvas={xshift=-13.6mm,yshift=26.3mm},rotate=-60]
			\vertchan{}{} ;
		\end{scope}
		\begin{scope}[transform	canvas={xshift=29.4mm,yshift=-1.45mm},rotate=60]
			\vertchan{}{} ;
		\end{scope}
		%\begin{scope}[transform	canvas={xshift=32.55mm,yshift=4.3mm},rotate=-60]
		%	\vertchan
		%\end{scope}
		\vertchan{\draw (k)--(l)}{fill=black} ;
		\begin{scope}[transform	canvas={xshift=-39.4mm,yshift=-22.7mm}]
			\vertchan{\draw (k)--(l); \draw (b)--(a);
				\tikzstyle{every path}=[draw,dashed, color=brown,thin];
				\draw[->] (0.7,0.5)--(0,1.7);
			}{fill=black} ;
		\end{scope}
		\begin{scope}[transform	canvas={xshift=-39.6mm,yshift=22.8mm}]
			\vertchan{\draw (m)--(l); \draw (b)--(a);
				\tikzstyle{every path}=[draw,dashed, color=brown,thin];
				\draw[->] (0.7,0.5)--(0,1.7);
			}{fill=black} ;
		\end{scope}
		\begin{scope}[transform	canvas={xshift=39.4mm,yshift=22.8mm}]
			\vertchan{\draw (b)--(c);
				\tikzstyle{every path}=[draw,dashed, color=brown,thin];
				\draw[->] (7.3,0.5)--(8,1.7);
				\draw[->] (7.3,0.5)--(8,-0.6)
			}{fill=black} ;
		\end{scope}
	\end{tikzpicture}
	
	\caption{An illustration of the proof of Theorem~\ref{thm:4hardness}.
		\\Top: a picture of the {\em vertex channel}, where the
		``solid mass'' remains outside (as indicated by the shade).
		Note that the color of $c$ is unique in the picture.
		In any proper $4$-coloring,
		$b_1$ must be of the same color as $a_3$ (since both
		see the triangle $a_1a_2c$), then $b_2$ is similarly of the
		same color as $a_2$ and $b_3$ as~$a_1$.
		\\
		Bottom: How vertex channels are composed by gluing at triangle joins
		along the shape of a hexagonal grid.
		All the black vertices (the {\em flag vertices}) must
		receive the same color in any proper $4$-coloring.
	}
	\label{fig:hardness-vertstrip}
\end{figure}

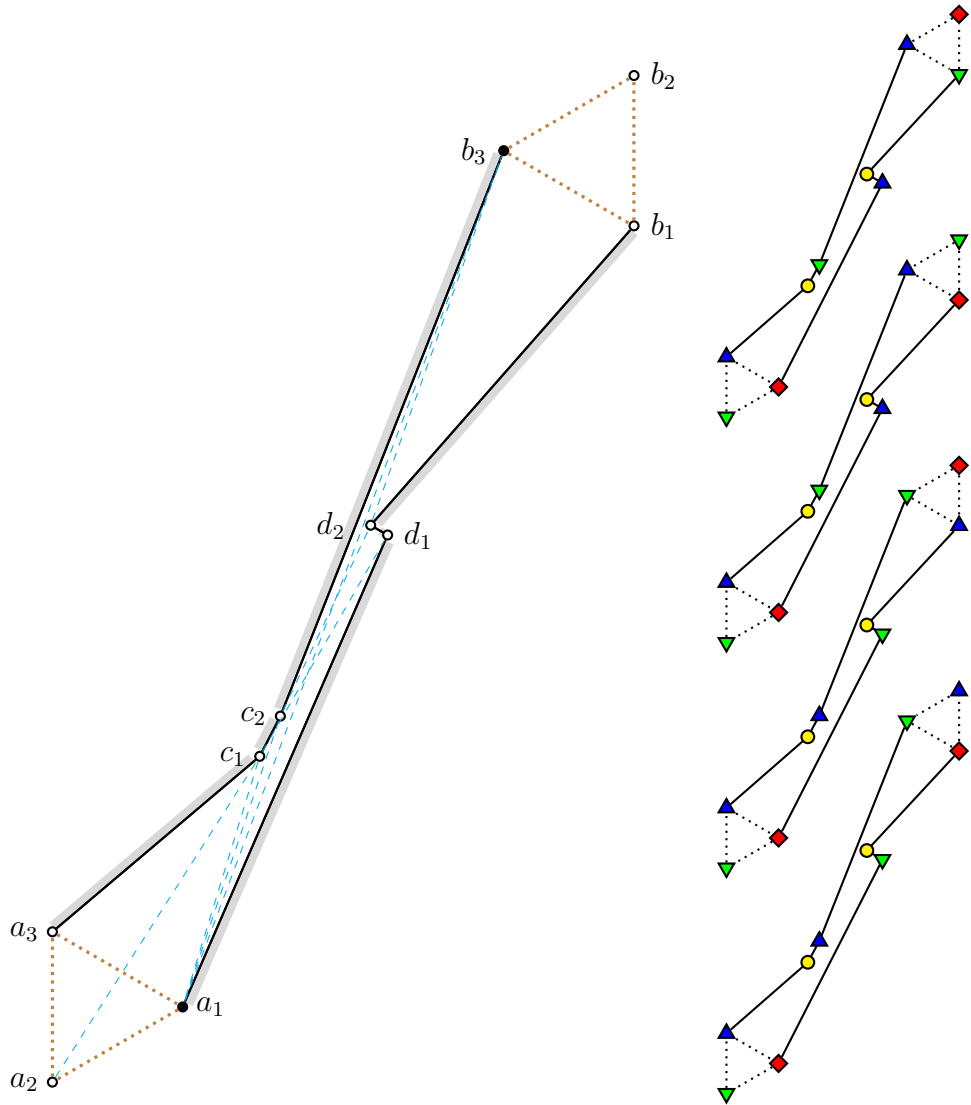
\begin{figure}[]
	\centering
	\begin{tikzpicture}[scale=2, rotate=60]
		\tikzstyle{every node}=[draw,solid,shape=circle,fill=white, inner sep=1.2pt];
		\tikzstyle{every path}=[draw, thick,solid, color=black];
		\draw (1,0) node[label=right:$a_1$, fill=black] (a){}
		++(150:1) node[label=left:$a_2$] (b){}
		++(30:1) node[label=left:$a_3$] (c){} ;
		\draw (7,0) node[label=right:$b_1$] (k){}
		++(30:1) node[label=right:$b_2$] (l){}
		++(150:1) node[label=left:$b_3$, fill=black] (m){} ;
		\draw (c) -- (2.7,0.39) node[label=left:$c_1$] (d){} 
		-- (3.0,0.405) node[label=left:$c_2$] (dd){} -- (m) ;
		\draw (a) -- (4.4,0.39) node[label=right:$d_1$] (ee){} 
		-- (4.4,0.52) node[label=left:$d_2~$] (e){} -- (k) ;
		\begin{scope}[on background layer]
			\tikzstyle{every path}=[draw, line width=3pt, color=gray!30];
			\draw[transform canvas={xshift=-1.7pt,yshift=+1pt}]
			(c) -- (d) -- (dd) -- (m) ;
			\draw[transform canvas={xshift=+1.7pt,yshift=-1pt}]
			(a) -- (ee) -- (e) -- (k) ;
		\end{scope}
		\tikzstyle{every path}=[draw, fill=none, thick, color=black];
		\draw (c) -- (d) -- (dd) -- (m);
		\draw (a) -- (ee) -- (e) -- (k);
		\tikzstyle{every path}=[draw, color=brown, very thick, dotted];
		\draw (a) -- (b) -- (c) -- (a) ;
		\draw (k) -- (l) -- (m) -- (k) ;
		\tikzstyle{every path}=[draw, thin,dashed, color=cyan];
		\draw (dd) -- (a) -- (m) ;
		\draw (b) -- (d) -- (ee) ;
		\draw (dd) -- (e) -- (m) ; \draw (d) -- (a) ;
	\end{tikzpicture}
	~
	\begin{tikzpicture}[scale=0.8, rotate=60,
		triangle/.style = {inner sep=1.2pt, regular polygon, regular polygon sides=3},
		dtriangle/.style = {inner sep=1.2pt, regular polygon, regular polygon sides=3, rotate=180}
		]
		\tikzstyle{every node}=[draw,solid,shape=circle,fill=white, inner sep=1.7pt];
		\tikzstyle{every path}=[draw, thick,solid, color=black];
		\def\onechan#1#2#3#4{%
			\draw (1,0) node[diamond,fill=red] (a){} 
			++(150:1) node[dtriangle,fill=green] (b){} 
			++(30:1) node[triangle,fill=blue] (c){} ;
			\draw (7,0) node[fill=#3] (k){} ++(30:1) node[fill=#4] (l){}
			++(150:1) node[fill=#2] (m){} ;
			\draw (c) -- (2.7,0.42) node[fill=yellow] (d){} 
			-- (3.1,0.43) node[fill=#1] (dd){} -- (m) ;
			\draw (a) -- (4.8,0.2) node[fill=#2] (ee){} 
			-- (4.8,0.5) node[fill=yellow] (e){} -- (k) ;
			\draw[dotted] (a) -- (b) -- (c) -- (a) ;
			\draw[dotted] (k) -- (l) -- (m) -- (k) ;
		}
		\onechan{blue,triangle}{green,dtriangle}{red,diamond}{blue,triangle}
		\begin{scope}[transform canvas={yshift=+30mm}]
			\onechan{blue,triangle}{green,dtriangle}{blue,triangle}{red,diamond}
		\end{scope}
		\begin{scope}[transform canvas={yshift=+60mm}]
			\onechan{green,dtriangle}{blue,triangle}{red,diamond}{green,dtriangle}
		\end{scope}
		\begin{scope}[transform canvas={yshift=+90mm}]
			\onechan{green,dtriangle}{blue,triangle}{green,dtriangle}{red,diamond}
		\end{scope}
	\end{tikzpicture}
	\caption{Left: a picture of the {\em edge channel}.
		Some important fine details (which cannot be clearly displayed in this scale) are:
		$a_1$ sees $c_2$ and $b_3$, $c_1$ sees $d_1$ but not $d_2$,
		neither of $d_1,c_2$ can see~$a_2$ and neither of $c_1,c_2$ can see~$b_2$.
		Note that in any proper $4$-coloring, $a_1$ and $b_3$ must
		receive distinct colors, while $c_1$ and $d_2$ must have
		the same color (since they both see the triangle~$a_1c_2d_1$).
		Hence, in particular, the triple of colors used on $a_1a_2a_3$ must be the
		same (up to ordering) as the triple of colors on~$b_1b_2b_3$.
		\\[3pt]
		Right: Examples of proper $4$-colorings of the edge channel.
		Note that the flexibility of these $4$-colorings is not in a
		contradiction with Theorem~\ref{lemalgo} since the chord
		$a_1c_1$ is not crossed by other chords, and likewise the chord~$b_3d_2$.
	}
	\label{fig:hardness-edgestrip}
\end{figure}

\begin{thm}\label{thm:4hardness}
	The problem -- given a polygon with holes $\PP$ in the plane, 
	to decide whether the visibility graph of $\PP$ is properly $k$-colorable --
	is NP-complete for every~$k\geq4$.
\end{thm}

\begin{proof}
	The claim follows from Theorem~\ref{thm:5hardness} for $k\geq5$, and
	so we consider only~$k=4$ here.
	Again, the problem is clearly in NP.
	In the opposite direction, we reduce from the NP-complete problem of
	$3$-coloring a given {\em planar} graph~$H$.
	
	We first recall a folklore claim that every planar graph $H$ 
	is a minor of the $3$-regular plane grid (also called the ``wall').
	By the definition of a minor (a {\em minor model}), this means that 
	there is a collection of pairwise disjoint subtrees of the grid
	$T_v$: $v\in V(H)$ (representatives of the vertices of~$H$) such that,
	for every edge $uv\in E(H)$, the grid contains an edge between
	$V(T_u)$ and $V(T_v)$ (called a representative edge of~$uv$). 
	This minor model of $H$ can be naturally embedded in the geometric hexagonal
	grid in the plane.
	To simplify our construction, we may moreover assume that we always
	choose representative edges in the grid which are not of horizontal
	direction (out of the three directions $0^\circ$, $120^\circ$ and~$240^\circ$).
	Without further optimization, the size of the resulting grid
	would be at most quadratic in the size of~$H$, which is good enough
	for the reduction.
	
	Having such a representation of the given planar graph $H$ in the
	geometric hexagonal grid, we continue as follows.
	Let a {\em vertex channel} be the polygonal fragment shown in
	Figure~\ref{fig:hardness-vertstrip}, where the triples $a_1a_2a_3$
	and $b_1b_2b_3$ are the {\em triangle joins} of the channel.
	Channels are composed, after suitable rotation, by gluing their
	triangle joins together, as illustrated in bottom part of 
	Figure~\ref{fig:hardness-vertstrip}.
	When no further channel is glued to a join, then the dotted triangle
	edge(s) is ``sealed'' by a polygon edge.
	Let an {\em edge channel} be the polygonal fragment shown in
	Figure~\ref{fig:hardness-edgestrip}, again having two triangle joins
	$a_1a_2a_3$ and $b_1b_2b_3$ at its ends.
	Edge channels are used and composed in a same way with vertex
	channels, but edge channels cannot be rotated, only mirrored by the
	vertical axis (that is why we do not use them along horizontal grid edges).
	Altogether, we construct a polygon $\PP$ from $H$ by composing copies of the
	vertex channel along all the grid edges of each $T_v$, $v\in V(H)$,
	and by further composing in copies of the edge channel (possibly mirrored)
	along the representative edges of $H$ in the grid.
	
	Assume now that we have a proper $4$-coloring of the visibility graph of $\PP$.
	For each triangle join, let the vertex with the middle $y$-coordinate
	be called the {\em flag vertex} (it is the vertex which is extreme
	to the left or right).
	One can easily check from Figure~\ref{fig:hardness-vertstrip} that,
	among all vertex channels of one $T_v$, $v\in V(H)$, all the
	triangle joins receive the same unordered triple of colors and, in
	particular, all the flag vertices have the same one color.
	The same claim can also be derived from Theorem~\ref{lemalgo} applied to
	the standalone simple polygon formed by the vertex channels of $T_v$.
	Furthermore, one can check from Figure~\ref{fig:hardness-edgestrip},
	that the edge channel also maintains the property that both its
	triangle joins must receive the same unordered triple of colors.
	
	Naturally assuming connectivity of $H$, we hence conclude that every
	triangle join constructed in $\PP$ receives the same unordered triple
	of colors, say~$\{1,2,3\}$.
	Now, to each vertex $v$ of $H$ we assign the unique color from
	$\{1,2,3\}$ which occurs on the flag vertices of~$T_v$.
	Since the two flag vertices of the edge channel see each other
	($a_1$ and $b_3$ in Figure~\ref{fig:hardness-edgestrip}), this
	ensures that for every edge $uv\in E(H)$ the colors assigned to $u$
	and $v$ are distinct, and so $H$ is $3$-colorable.
	
	In the converse direction, we assume that $H$ has a proper $3$-coloring.
	We can routinely $4$-color the polygonal fragments of each $T_v$, 
	$v\in V(H)$, such that all the flag vertices of $T_v$ get the color of~$v$.
	Then, for each $uv\in E(H)$ with distinct colors on $u$ and~$v$, 
	we can complete proper $4$-coloring of the fragment of $\PP$ made by 
	the representative edge channel of $uv$, as shown in the right part
	of Figure~\ref{fig:hardness-edgestrip}.
	Hence the visibility graph of~$\PP$ is then $4$-colorable.
	
	Finally, the construction of $\PP$ is easily done (with negligible
	distortion of the angles of hexagonal grid) within polynomial
	resolution and so in polynomial time.
\end{proof}	

\section{Conflict-free chromatic guarding of a simple polygon} \label{sec:polygonconflict}

\subsection{Vertex-to-point conflict-free chromatic guarding}
In this section, we extend the scope of the studied problem of
vertex-to-point conflict-free chromatic guarding from funnels and weak
visibility polygons to general simple polygons.  
We will establish an $O(\log^2 n)$ upper bound for the
number of colors of vertex-guards on $n$-vertex simple polygons, and give
the corresponding polynomial time algorithm.
For that we use the previous algorithm for chromatic guarding of weak visibility polygons
(Algorithm~\ref{alg:weakv2pcoloringB}).
Moreover, our algorithm is ready for further improvements in
Algorithm~\ref{alg:weakv2pcoloringB}, as it uses $O(C+\log n)$ colors
where $C$ is the number of colors used by Algorithm~\ref{alg:weakv2pcoloringB}.

We give our algorithm in two phases.
Namely, the decomposition phase, and the coloring phase.

\subsubsection{Decomposition into weak visibility polygons}

Our coloring scheme relies on the decomposition of a simple polygon into weak visibility polygons.
We utilise the decomposition algorithm described by B\"{a}rtschi \emph{et al.}
\cite{Bartschi-2014}, given in Algorithm~\ref{alg:bartdecompose}.
For an edge $f$ (or a vertex) of a polygon $\PP$,
we call the {\em visibility polygon of $f$} the set of all points of $\PP$ 
which are visible from some point of~$f$.

\begin{algorithm}[htbp]
	\KwIn{A simple polygon $\PP$, an arbitrary vertex $v$}
	\KwOut{Decomposition of $\PP$ into weak visibility polygons}
	\smallskip
	$W_1 \gets$ visibility polygon of $v$; \tcc*{i.e., the set of points visible from $v$}
	$i \gets 1$, $j \gets 2$\;
	\Repeat{$\bigcup_{W \in \mathcal{W}} = P$}{
		$\mathcal{W} \gets \mathcal{W} \cup \{W_i\}$\;
		\ForEach{edge $e$ of $W_i$ which is not a boundary edge of $\PP$}
		{	\qquad\hfill\qquad\hfill\qquad\tcc{that is, $e$ cuts $W_i$ from the rest of $\PP$}
			$W_j \gets$ visibility polygon of $e$ in the
			adjacent subpolygon of $P\setminus W_i$\;
			$j \gets j + 1$\;
		}
		% 		$P \gets P \setminus W_i$\;
		%% should not change P because of the termination condition!!!
		$i \gets i+1$\;
	}
	\Return{$\mathcal{W}$}\;
	\caption{Polygon decomposition algorithm of B\"{a}rtschi \emph{et al.} \cite{Bartschi-2014}}
	\label{alg:bartdecompose}
\end{algorithm}

The Algorithm~\ref{alg:bartdecompose} takes a simple polygon $\PP$ and an arbitrary vertex $v$ of $\PP$.
First, the algorithm computes the visibility polygon $W_1$ of $v$ and removes~$W_1$ from $\PP$.
Then, until the whole polygon is partitioned, the algorithm selects an edge $e$ of
a previously removed polygon, computes the visibility polygon $W_j$ of $e$ within
the rest of $\PP$, and then removes $W_j$, and so on.

The mentioned decomposition algorithm has been performed in
\cite{Bartschi-2014} to obtain a point-to-point guarding, 
in which the guards are not necessarily selected at vertices of the polygon.
In our case, we need to ensure that recursively chosen guards in weak
visibility subpolygons $W_j$ of $\PP$ are placed at vertices of $\PP$.
However, Algorithm~\ref{alg:bartdecompose} typically creates subpolygons
whose vertices are internal points of the edges of~$\PP$.
To overcome this problem, we slightly modify the algorithm by inserting an
intermediate phase -- creating a special subpolygon, which ``recovers'' the
property that the base edge of each subsequently constructed visibility polygon
is between two vertices of $\PP$ again.
As we will show later, this modification does not weaken the decomposition
technique much.

\begin{algorithm}[htbp]
	\KwIn{A simple polygon $\PP$, an arbitrary edge $e_0$ of $\PP$}
	\KwOut{Decomposition of $\PP$ into weak visibility polygons of two kinds}
	\smallskip
	$W_1 \gets$ visibility polygon of $e_0$ in $\PP$\;
	% 		 \tcp{weak visibility polygon with base edge $e_b$}
	% 	$\mathcal{W} \gets \mathcal{W} \cup \{W_1\}$; \tcp{set of weak visibility polygons}
	% 	$P \gets P \setminus \{W_i\}$\;
	$i \gets 1$, $j \gets 2$\;
	\Repeat{$\bigcup_{W \in \mathcal{W}} = P$}{
		$\mathcal{W} \gets \mathcal{W} \cup \{W_i\}$\;
		\ForEach{edge $e=uv$ of $W_i$ which is not a boundary edge of $\PP$}{
			$U \gets$ visibility polygon of $e$ in the
			adjacent subpolygon of $P\setminus W_i$\;
			% 			\If{$W_i$ has four or less vertices \textsc{or} both $u$ and $v$ are vertices of $\PP$}
			\If{both $u$ and $v$ are vertices of $\PP$}
			{ \label{it:uvv}
				$W_j \gets$ $U$ \tcc*{got an ordinary subpolygon $U$}
				$j \gets j + 1$\;
			}
			\Else{ \label{it:uvxy}
				% 				\tcc*{forward partitioning}
				Let $x$ and $y$ be the vertices of $U$ such that $x\not=v$
				is adjacent to $u$ and $y\not=u$ is adjacent to~$v$
				\tcc*{$x,y$ are vertices of $\PP$, too}
				$\Pi_{xy} \gets$ geometric shortest path in $U$ between $x$ and $y$\;
				$U_1$ $\gets$ the subpolygon of $U$ bounded by the
				paths $\Pi_{xy}$ and $(xu,uv,vy)$\;
				$\mathcal{W} \gets \mathcal{W} \cup \{U_1\}$
				\tcc*{adding a forward subpolygon $U_1$}
				% 				$W_i \gets$ the polygon whose edge set is $E(\Pi_{xy}) \cup \{uv,ux,vy\}$\;
				% 				\tcc{$E(\Pi_{xy})$ is the set of edges in the shortest path between $x$ and $y$}
				% 				$k \gets 1$\;
				\ForEach{edge $f$ of the path $\Pi_{xy}$}
				{ \label{it:pixy}
					$W_j \gets$ visibility polygon of $f$ in the
					adjacent subpolygon of $P\setminus U_1$\;
					% 					$W_j^k \gets$ the visibility polygon of $e_k$\;
					% 					$\mathcal{W} \gets \mathcal{W} \cup \{W_j^k\}$\;
					% 					$P \gets P \setminus \{W_j^k\}$\;
					\tcc{"forwarding" to ordinary subpolygons adjacent to $U_1$}
					$j \gets j + 1$
				}
			}
		}
		% 		$P \gets P \setminus \{W_i\}$\;
		$i \gets i+1$\;
	}
	\Return{hierarchically structured decomposition $\mathcal{W}$ of $\PP$}\;
	\caption{Adjusted polygon decomposition algorithm}
	\label{alg:ourdecompose}
\end{algorithm}

\begin{figure}[tbp]
	\centering
	\begin{tikzpicture}[xscale=-0.5, yscale=0.25]
		\coordinate (1) at (7,-1);
		\coordinate (2) at (1.5,1.5);
		\coordinate (3) at (0.5,2.5);
		\coordinate (4) at (-1,3);
		\coordinate (5) at (-5,4);
		\coordinate (6) at (-1,4);
		\coordinate (7) at (2.2,4.8);
		\coordinate (8) at (3.4,6.6);
		\coordinate (9) at (3,8.3);
		\coordinate (10) at (1.8,9.65);
		\coordinate (11) at (-3,10.3);
		\coordinate (12) at (-0.35,10.55);
		\coordinate (13) at (2.3,11.3);
		\coordinate (14) at (4,13.6);
		\coordinate (15) at (4.6,14.8);
		\coordinate (16) at (4,16);
		\coordinate (17) at (3.5,17);
		\coordinate (18) at (2.5,18);
		\coordinate (19) at (-1,17.5);
		\coordinate (20) at (1.6,18.8);
		\coordinate (21) at (2.2,19.4);
		\coordinate (22) at (3.5,21);
		\coordinate (23) at (2.5,24);
		\coordinate (24) at (12,14);
		\foreach \i in {1,...,23}
		{
			\pgfmathtruncatemacro\j{\i+1};
			\draw[thick] (\i)--(\j);
		}
		\draw[very thick,dash dot] (24)--(1);
		\coordinate (A) at (10,-2.3);
		\coordinate (B) at (12, 18);
		\draw[thick] (A)--(1);
		\draw[thick] (B)--(24);
		
		\draw[draw=none,fill=teal, opacity=0.35] (2)--(3)--(4)--(5)--(6)--(7)--(8)--(2);
		\draw[draw=none,fill=teal, opacity=0.35] (8)--(9)--(10)--(11)--(12)--(13)--(14)--(15)--(8);
		\draw[draw=none,fill=teal, opacity=0.35] (15)--(16)--(17)--(18)--(19)--(20)--(21)--(22)--(15);
		\draw[dash dot,fill=orange, opacity=0.35] (1)--(2)--(8)--(15)--(22)--(23)--(24)--(1);
		
		\node at (8,12) {$U_1$};
		\node at (1,3.5) {$W_{j+2}$};
		\node at (3,10.5) {$W_{j+1}$};
		\node at (3,18.8) {$W_j$};
		
		\tikzstyle{every node}=[draw,fill=white, shape=circle, inner sep=1pt];
		\node[label=left:$u$] at (24) {};
		\node[label=below:$v$] at (1) {};
		\node[label=below:$y$] at (2) {};
		\node[label=right:$c$] at (8) {};
		\node[label=right:$b$] at (15) {};
		\node[label=right:$a$] at (22) {};
		\node[label=right:$x$] at (23) {};
	\end{tikzpicture}
	\caption{An illustration of forward partitioning in Algorithm~\ref{alg:ourdecompose}.
		The base edge $e=uv$ of its visibility polygon $U$ (to the right) has one
		end $v$ which is not a vertex of $\PP$.
		In this situation we add an intermediate forward subpolygon $U_1$
		(colored orange) bounded by $e$ and the path $\Pi_{xy}=(x,a,b,c,y)$.
		The child ordinary subpolygons of $U_1$ are then the visibility
		polygons of the edges of $\Pi_{xy}$, denoted in the picture by $W_j,W_{j+1},W_{j+2}$.
		Notice also that, in general, the union $U_1\cup W_{j}\cup
		W_{j+2}\cup\dots$ may be larger, as witnessed in the picture by~$W_j$, than
		the visibility polygon $U$ of the edge~$e$.}
	\label{fig:forward}
\end{figure}
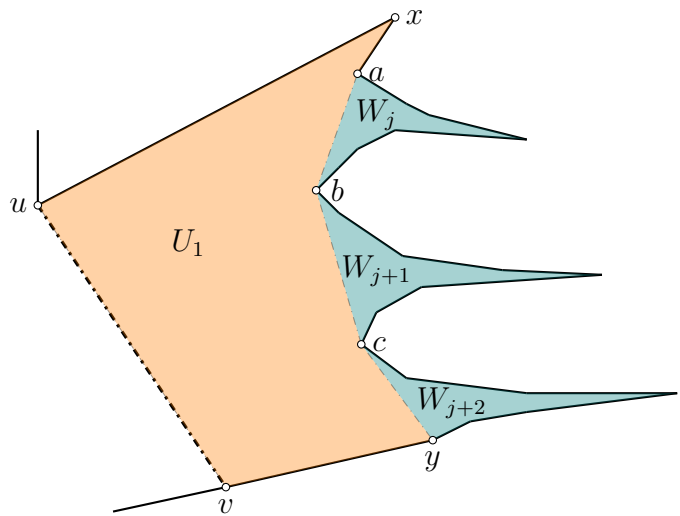

\medskip
We describe the full adjusted procedure in Algorithm~\ref{alg:ourdecompose}.
The constructed decomposition there consists of two kinds of polygons;
\begin{itemize}\item
	the {\em ordinary} weak visibility polygons whose base edge has both ends
	vertices of~$\PP$ (while some other vertices may be internal points of edges of~$\PP$),
	\item
	the {\em forward} weak visibility polygons whose base edge has at least one
	end not a vertex of $\PP$ (and, actually, exactly one end, but this fact is not
	crucial for the arguments), but all their other vertices are vertices of $\PP$
	and form a concave chain (which will be important).
\end{itemize}
This two-sorted process is called \emph{forward partitioning} (cf.~the
branch from line~\ref{it:uvxy} of the algorithm).
We illustrate its essence in Figure~\ref{fig:forward}.
Notice that it may happen that $x=y$ and $\Pi_{xy}$ is a trivial one-vertex path.

\begin{lem}\label{lem:ourdecompose}
	Algorithm~\ref{alg:ourdecompose} runs in polynomial time, and it outputs a
	decomposition $\mathcal{W}$ of $\PP$ into weak visibility polygons such that
	the following holds:
	\begin{enumerate}
		\item [a)]
		If $U\in\mathcal{W}$ is a forward polygon, then all vertices of $U$ except possibly
		the base ones are vertices of $\PP$.
		The non-base vertices form a concave chain.
		\item [b)]
		If $W\in\mathcal{W}$ is an ordinary polygon, then all vertices of $W$ 
		are vertices of $\PP$, except possibly for apex vertices of max funnels of $W$.
	\end{enumerate}
\end{lem}

\begin{proof}
	It is an easy routine to verify that the algorithm can be implemented in polynomial
	time and that the members of $\mathcal{W}$ are pairwise internally disjoint
	subpolygons which together cover~$\PP$.
	
	Claim (a) follows trivially from the choice of $\Pi_{xy}$ as a geometric
	shortest path in $U$.
	For claim (b), the base vertices of $W$ are vertices of $\PP$ by
	the condition on line \ref{it:uvv}, or the choice of base edge $f$ on line
	\ref{it:pixy} of Algorithm~\ref{alg:ourdecompose}.
	If some other vertex $w$ of $W$ is not a vertex of $\PP$, then since $W$ is a
	visibility polygon of $f$, the vertex $w$ must not be reflex, and so $w$ is
	the apex of some max funnel of~$W$.
\end{proof}

\subsubsection{Recursive coloring of the decomposition}

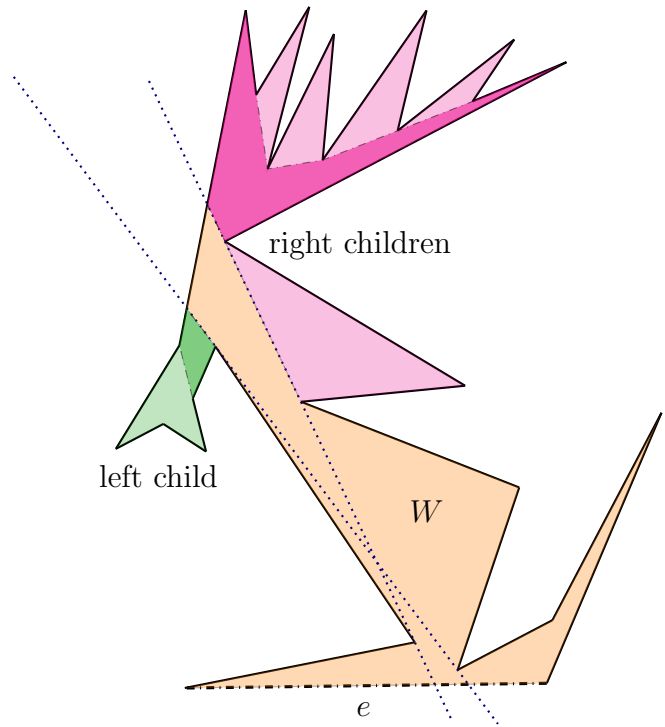
\begin{figure}[htbp]
	\centering
	\begin{tikzpicture}[xscale=1.2, yscale=0.7, rotate=-93]
		\coordinate (1) at (3.6,4);
		\coordinate (2) at (5.2,4.2);
		\coordinate (3) at (3.5,4.7);
		\coordinate (4) at (6.6,4.4);
		\coordinate (5) at (4,5);
		\coordinate (6) at (6.4,5);
		\coordinate (7) at (3.5,6);
		\coordinate (8) at (5.8,5.8);
		\coordinate (9) at (4,7);
		\coordinate (10) at (5.2,6.6);
		\coordinate (11) at (4.4,7.6);
		\coordinate (12) at (8,4);
		\coordinate (13) at (10.6,6.8);
		\coordinate (14) at (11,5);
		\coordinate (15) at (12.5,7.5);
		\coordinate (16) at (16,7);
		\coordinate (17) at (15,8);
		\coordinate (18) at (11,9);
		\coordinate (19) at (16.2,8);
		\coordinate (20) at (16.5,4);
		\coordinate (21) at (15.5,6.5);
		\coordinate (22) at (10,4);
		\coordinate (23) at (11,3.8);
		\coordinate (24) at (12,4);
		\coordinate (25) at (11.5,3.5);
		\coordinate (26) at (12,3);
		\coordinate (27) at (10,3.6);
		
		\foreach \i in {1,...,18,20,21,22,23,24,25,26}
		{
			\pgfmathtruncatemacro\j{\i+1};
			\draw[thick] (\i)--(\j);
		}
		\draw[thick] (27)--(1);
		
		\draw[dash dot, very thick] (19)--(20);
		\coordinate (A) at (16.27,7.13);
		\coordinate (B) at (16.29,6.76);
		\coordinate (C) at (9.29,3.64);
		\coordinate (D) at (7.31,3.77);
		
		\draw[dash dot,fill=orange, opacity=0.3] (12)--(14)--(15)--(16)--(17)--(18)--(19)--(20)--(21)--(22)--(C)--(D)--(12);
		
		\coordinate (P) at (5,3);
		\coordinate (Q) at (17,7);
		
		\coordinate (X) at (5,1.5);
		\coordinate (Y) at (17,7.5);
		
		\draw[thick,dotted,color=blue!50!black] (P)--(Q);
		\draw[thick,dotted,color=blue!50!black] (X)--(Y);
		
		\draw[dash dot, fill=magenta, opacity=0.25] (1)--(2)--(3)--(4)--(5)--(6)--(7)--(8)--(9)--(10)--(11)--(12)--(D)--(1);
		\draw[dash dot, fill=magenta, opacity=0.5] (1)--(2)--(4)--(6)--(8)--(10)--(11)--(12)--(D)--(1);
		
		\draw[dash dot, fill=green!60!black, opacity=0.25] (23)--(24)--(25)--(26)--(27)--(23);
		\draw[dash dot, fill=green!60!black, opacity=0.5] (22)--(23)--(27)--(C)--(22);
		
		\draw[dash dot, fill=magenta, opacity=0.25] (12)--(13)--(14)--(12);
		\node at (16.8,6) {$e$};
		\node at (13,6.5) {$W$};
		\node at (12.5,3.5) {left child};
		\node at (8,5.5) {right children};
	\end{tikzpicture}
	\caption{A weak visibility polygon $W$ of the base edge $e$ (denoted by dashed line), its left child (colored green) and two right children (colored pink).
		The two forward polygons among all three children of $W$ are filled with
		darked color.}
	\label{fig:decompositione}
\end{figure}

In the coloring phase, we traverse the (naturally rooted) decomposition tree of the
decomposition~$\mathcal{W}$ of $\PP$ computed by Algorithm~\ref{alg:ourdecompose}.
We will use the terms {\em parent} and {\em child} polygon with respect to
this rooted decomposition tree (which, essentially, is a BFS tree).
For each ordinary weak visibility polygon $W$ in $\mathcal{W}$, we apply our
Algorithm~\ref{alg:weakv2pcoloringB} for guarding, which is a correct usage since only
some apices of the max funnels of $W$ are not vertices of $\PP$, and those are
not used for guards.
On the other hand, for each forward weak visibility polygon $U$ in
$\mathcal{W}$, we apply a straightforward guarding by a ruler sequence on the
non-base vertices (which are vertices of $\PP$ and form a concave chain in~$U$).
For vertices shared between a parent and a child polygon we apply the
parental color (which automatically resolves also sibling conflicts).

For this coloring scheme, in order to avoid color conflicts between
different polygons, we need to use disjoint subsets of colors for the
ordinary and the forward polygons.
Moreover, this whole set of colors will be used in three disjoint copies
-- the first set of colors given to a parent polygon, the second one given
to all its ``left'' children and the third one to its ``right'' children.
Each child polygon will then reuse the other two color sets for its children.
The reason why this works is essentially the same as in B\"{a}rtschi
\emph{et al.} \cite{Bartschi-2014}.

We need to define what the left and right children mean.
Consider an ordinary polygon $W\in\mathcal{W}$ with the base edge $e$,
and picture $W$ such that $e$ is drawn horizontal with $W$ above it.
See Figure~\ref{fig:decompositione}.
Then every edge $f$ of $W$ which is not a boundary edge of $\PP$ is {\em not} horizontal.
Hence every child polygon $U$ of $W$ in the decomposition (regardless of whether $U$
is a forward or ordinary polygon) is either to the left of $f$ or to the
right of it, and then $U$ is called a {\em left} or {\em right} child of $W$, respectively.
If $U$ is a left child of $W$ and $U$ is a forward polygon, then the
child ordinary polygons of $U$ are also called {\em left} children.
The same applies to right children.

We can now state the precise procedure in Algorithm~\ref{alg:coloringoverall}.

\begin{algorithm}[htbp]
	\caption{Computing a vertex-to-point conflict-free chromatic guarding of 
		a simple polygon using $O(\log^2n)$ colors.}
	\label{alg:coloringoverall}
	\KwIn{A simple polygon $\PP$.}
	\KwOut{A V2P conflict-free chromatic guarding of $\PP$ using $O(C+\log n)$
		colors, where $C$ is the maximum number of colors used by calls to
		Algorithm~\ref{alg:weakv2pcoloringB}.}
	
	\smallskip
	Call Algorithm~\ref{alg:ourdecompose} to get the hierarchically
	structured decomposition $\mathcal{W}$ of $\PP$\;
	$A \gets \{1,2,\ldots,C\}$, $B \gets \{C\!+\!1,\ldots,C\!+\!\lfloor\log n\rfloor\}$
	\tcc*{color sets to be used}
	\ForEach{ $i=1,2,3$ }{
		$C_i \gets$ $C_i'\cup C_i''$ where $C_i'$ is a disjoint
		copy of $A$ and $C_i''$ a disjoint copy of $B$\;
	}
	$W \gets$ the root polygon in the decomposition $\mathcal{W}$\;
	Call procedure Recursivecolor($W$,1)
	\tcc*{as defined below}
	\Return{colored $\PP$}\;
\end{algorithm}

\begin{procedure}[tb]
	\caption{Recursivecolor($W$,\,$c$) }
	%\quad\tcc{used in Algorithm~\ref{alg:coloringoverall}}}
\If{$W$ is a forward polygon in the decomposition $\mathcal{W}$}
{
	\ForEach{child polygon $W_i$ of $W$}
	{ Call procedure Recursivecolor($W_i$,\,$c$)\; }
	color the non-base vertices of $W$ by the ruler sequence
	using colors from $C_c''$\;
	\tcc{this overrides child colors from the recursive calls}
} \Else {
	Choose $a,b$ such that $\{c,a,b\}=\{1,2,3\}$\;
	\ForEach{child polygon $W_j$ of $W$}
	{
		\If{$W_j$ is a left child of $W$}
		{ Call procedure Recursivecolor($W_j$,\,$a$)\; }
		\Else
		{ Call procedure Recursivecolor($W_j$,\,$b$)\; }
		
	}
	Call Algorithm~\ref{alg:weakv2pcoloringB} to color-guard $W$,
	using colors from $C_c'$\;
	\tcc{this again overrides child colors from the recursive calls}
}
\Return{}
\end{procedure}

\begin{thm}\label{thm:allpolygon}
Algorithm~\ref{alg:coloringoverall} in polynomial time computes a conflict-free
chromatic guarding of an $n$-vertex simple polygon using $O(\log^2 n)$ colors.
\end{thm}
\begin{proof}
Overall runtime analysis follows that of previous Algorithm~\ref{alg:ourdecompose}.
As noted above, the obtained coloring is valid since we only assign guards
to those vertices of the polygons in $\mathcal{W}$ which are at the same
time vertices of~$\PP$; this is claimed by Lemma~\ref{lem:ourdecompose} in
connection with Algorithm~\ref{alg:weakv2pcoloringB}.

It remains to prove that the resulting coloring of $\PP$ is conflict-free,
i.e., that every observer can see a unique color.
We know that the constructed coloring is conflict-free within any single
polygon of $\mathcal{W}$.
Let us say that the {\em home} polygon of an observer $x\in P$ is $W\in\mathcal{W}$
such that $x\in W$ (the parent one in case of $x$ belonging to multiple polygons).
The rest will follow if we prove that no observer in $\PP$ 
can see another polygon which has received the same color set as the home
polygon of~$x$.

Consider an ordinary polygon $W\in\mathcal{W}$ which is the visibility
polygon of its base edge $e$ (such as the one in Figure~\ref{fig:decompositione}), 
and let~$W_0$ be the nearest ancestor of $W$ which is also ordinary 
(i.e., $W_0$ is the parent of $W$, or the grandparent in case the parent is a forward polygon).
Let $x\in W_o$ be our observer, and let $W_1$ be an ordinary child or grandchild of $W$ 
which uses the same color set as~$W_0$ in Algorithm~\ref{alg:coloringoverall}.
Then the line of sight between $x$ and some non-base vertex of $W_1$ must
cross the base edge $e$ of $W$, and hence be visible from~$e$.
This is a contradiction since the observed vertex of $W_1$ does not belong to~$W$.
The same argument applies symmetrically, and also in the case of a forward
home polygon.

Consider now two sibling ordinary polygons $W_1,W_2$ which receive the same color set.
One case is that their common parent is a forward polygon $U$.
Since $W_1$ and $W_2$ are adjacent to a concave chain of $U$, an observer
$x$ may see both of them only if $x$ belongs to $U$ or some ancestor, and so
there is no conflict for~$x$.
The second case is that they have a common ordinary grandparent/parent $W$
with base edge~$e$.
Suppose that there is an observer, say $y\in W_1$, which sees a vertex $p$ of $W_2$.
Let $q$ be any vertex of $W$ ``between'' $W_1$ and $W_2$.
Then the line of sight between $e$ and $q$ must cross the line segment $yp$,
which is absurd.
Again, the same argument may also be applied to sibling forward polygons.

We have exhausted all cases.

Finally, the number of colors used by Algorithm~\ref{alg:coloringoverall}
is $3\cdot(C+\log n)$ where $C=O(\log^2 n)$ by thm~\ref{thm:twoApprox},
and so we have the bound $O(\log^2 n)$.
\end{proof}

\subsection{Vertex-to-vertex conflict-free chromatic guarding}  \label{sec:v2vconflictfree}
In the last section of this chapter, we turn to the vertex-to-vertex variant of the guarding problem.
As we have noted at the beginning, the vertex-to-vertex (V2V) weak conflict-free chromatic guarding
problem coincides with the graph conflict-free coloring problem
\cite{k-cfc-graph} on the visibility graph of the considered polygon.
While, for the latter graph problem, \cite{k-cfc-graph} provided
constructions requiring an unbounded number of colors,
it does not seem to be possible to adapt those constructions for polygon
visibility graphs (and, actually, we propose that the conflict-free
chromatic number is bounded in the case of polygon visibility graphs,
see Conjecture~\ref{cj:v2v-upper3}).

\subsubsection{Lower bound for vertex-to-vertex conflict-free chromatic guarding}

We start by showing in Figure~\ref{fig:v2v-1color} 
that one color is not always enough.
To improve this very simple lower bound further, we will then need 
a more sophisticated construction.

\begin{prp}\label{pro:v2v-3colors}
There exists a simple polygon which has no V2V conflict-free chromatic guarding with
$2$ colors.
\end{prp}

In the proof of Proposition~\ref{pro:v2v-3colors}, we elaborate on properties of the simple example 
from Figure~\ref{fig:v2v-1color} (right),
and provide the construction shown in Figure~\ref{fig:v2v-2color}; its underlying idea is that the polygon pictured on the left requires at
least one guard to be placed on $p_1$ or $p_2$.
Gluing four copies of that polygon to the picture on the right, we obtain
a polygon for which two colors are not enough.

\begin{figure}[tbp]
\centering
\begin{tikzpicture}[scale=0.9]
	\tikzstyle{every node}=[draw, shape=circle, minimum size=4pt,inner sep=0pt];
	\tikzstyle{every path}=[thick];
	\node (a) at (0,0) {};
	\node (b) at (2,0) {};
	\node (c) at (1,1.6) {};
	\node (d) at (2.5,-0.2) {};
	\node (e) at (1,2) {};
	\node (f) at (-0.4,-0.2) {};
	\draw (a) -- (d) -- (b) -- (e) -- (c) -- (f) -- (a);
\end{tikzpicture}
\qquad\qquad\qquad\qquad
\begin{tikzpicture}[scale=1]
	\tikzstyle{every node}=[draw, shape=circle, minimum size=4pt,inner sep=0pt];
	\tikzstyle{every path}=[thick];
	\node (a) at (0,0) {};
	\node (b) at (1,0.2) {};
	\node (c) at (2,0.2) {};
	\node (d) at (3,0) {};
	\node (e) at (3,2) {};
	\node (f) at (2,1.8) {};
	\node (g) at (1,1.8) {};
	\node (h) at (0,2) {};
	\draw (a) -- (b) -- (c) -- (d) -- (e) -- (f) -- (g) -- (h) -- (a);
\end{tikzpicture}
\caption{Two examples of simple polygons requiring at least $2$ colors for
	a V2V conflict-free chromatic guarding (in each, one vertex guard cannot see
	all vertices, and any two guards of the same color make a conflict).}
% \caption{Red vertex is on $\mathcal{L}$, and its upper tangent is drawn with red line segment. The blue vertex is on $\mathcal{R}$, and its upper tangent is drawn with blue line segment. The two black dots where the upper tangents touch the left and the right chains are the points of intersection.}
\label{fig:v2v-1color}
\end{figure}
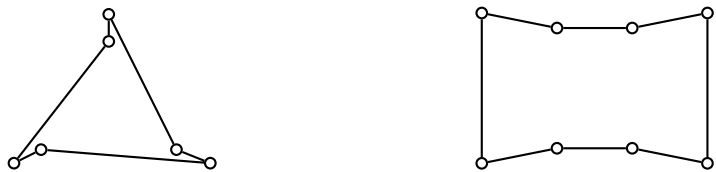

\begin{figure}[tbp]
\centering
\begin{tikzpicture}[xscale=1, yscale=1.5]
	\tikzstyle{every node}=[draw, shape=circle, minimum size=4pt,
	fill=gray, inner sep=0pt];
	\tikzstyle{every path}=[thick];
	\node[label=right:$~p_2$] (a) at (0.06,-0.3) {};
	\node[label=left:$p_1~$] (b) at (-0.06,-0.3) {};
	\node[label=left:$a_1$] (c) at (-1.9,1) {};
	\node[label=left:$a_2$] (cc) at (-1.9,1.33) {};
	\node[label=left:$a_3$] (dd) at (-1.97,1.66) {};
	\node[label=left:$a_4$] (d) at (-2.1,2) {};
	\node[label=right:$a_5$] (e) at (-1.1,2) {};
	\node[label=right:$a_6$] (ee) at (-0.89,1.66) {};
	\node[label=right:$a_7$] (ff) at (-0.64,1.33) {};
	\node[label=left:$a_8$] (f) at (-0.3,1) {};
	\node[label=below:$\!a_9~$] (g) at (-0.1,0.9) {};
	\node[label=below:$~c_9\!$] (h) at (0.1,0.9) {};
	\node[label=right:$c_8$] (i) at (0.3,1) {};
	\node[label=left:$c_7$] (ii) at (0.64,1.33) {};
	\node[label=left:$c_6$] (jj) at (0.89,1.66) {};
	\node[label=left:$c_5$] (j) at (1.1,2) {};
	\node[label=right:$c_4$] (k) at (2.1,2) {};
	\node[label=right:$c_3$] (kk) at (1.97,1.66) {};
	\node[label=right:$c_2$] (ll) at (1.9,1.33) {};
	\node[label=right:$c_1$] (l) at (1.9,1) {};
	\draw[color=blue] (a) -- (b) -- (c) -- (cc) -- (dd) -- (d)
	-- (e) -- (ee) -- (ff) -- (f) -- (g) -- (h)
	-- (i) -- (ii) -- (jj) -- (j)
	-- (k) -- (kk) -- (ll) -- (l) -- (a);
\end{tikzpicture}
\qquad
\def\bowlsh#1#2{%
	\begin{scope}[shift={#1}, rotate=#2, yscale=1, xscale=0.17]
		\draw[color=blue, thin, fill=blue!15!white] (-0.06,-0.3) -- (-1.9,1) -- (-2.1,2) -- (-1.1,2)
		-- (-0.3,1) -- (0.3,1) -- (1.1,2) -- (2.1,2) -- (1.9,1) -- (0.06,-0.3);
	\end{scope}
}
\begin{tikzpicture}[xscale=0.36,yscale=1.05]
	\tikzstyle{every node}=[draw, shape=circle, minimum size=4pt,
	fill=gray, inner sep=0pt];
	\tikzstyle{every path}=[thick];
	\node[label=below:$t$] (a) at (0,-0.3) {};
	% 	\node[label=below:$r_1$] (b) at (-1,-0.55) {};
	\node[label=below:$r_1$] (c) at (-2,-0.83) {};
	% 	\node[label=below:$r_3$] (d) at (-3,-1.15) {};
	\node[label=below:$r_2$] (e) at (-4,-1.5) {};
	\node[fill=white, label=right:$q_1$] (ee) at (-4.5,-0.4) {};
	\node[fill=white, label=right:$q_2$] (ff) at (-4.5,0.4) {};
	\node[label=above:$r_3$] (f) at (-4,1.5) {};
	% 	\node[label=above:$r_6$] (g) at (-3,1.15) {};
	\node[label=above:$r_4$] (h) at (-2,0.83) {};
	% 	\node[label=above:$r_8$] (i) at (-1,0.55) {};
	\node[label=above:$t'$] (aa) at (0,0.3) {};
	% 	\node[label=above:$s_8$] (j) at (1,0.55) {};
	\node[label=above:$s_4$] (k) at (2,0.83) {};
	% 	\node[label=above:$s_6$] (l) at (3,1.15) {};
	\node[label=above:$s_3$] (m) at (4,1.5) {};
	\node[fill=white, label=left:$q_3$] (mm) at (4.5,0.4) {};
	\node[fill=white, label=left:$q_4$] (nn) at (4.5,-0.4) {};
	\node[label=below:$s_2$] (n) at (4,-1.5) {};
	% 	\node[label=below:$s_3$] (o) at (3,-1.15) {};
	\node[label=below:$s_1$] (p) at (2,-0.83) {};
	% 	\node[label=below:$s_1$] (q) at (1,-0.55) {};
	% 	\draw (aa) -- (m);
	\draw (a) -- (c) -- (e) -- (ee) -- (ff) -- (f)
	-- (h) -- (aa) -- (k) -- (m) -- (mm)
	-- (nn) -- (n) -- (p) -- (a);
	\bowlsh{(-4.75,-0.4)}{90}
	\bowlsh{(-4.75,0.4)}{90}
	\bowlsh{(4.75,-0.4)}{270}
	\bowlsh{(4.75,0.4)}{270}
	% 	\node (a) at (0,-0.1) {};
	% 	\node (b) at (-0.9,-1) {};
	% 	\node (c) at (-0.6,0.2) {};
	% 	\node (d) at (-2,0.5) {};
	% 	\node (e) at (-0.6,0.8) {};
	% 	\node (f) at (-0.9,2) {};
	% 	\node (g) at (0,1.1) {};
	% 	\node (h) at (0.9,2) {};
	% 	\node (i) at (0.6,0.8) {};
	% 	\node (j) at (2,0.5) {};
	% 	\node (k) at (0.6,0.2) {};
	% 	\node (l) at (0.9,-1) {};
	% 	\draw (a) -- (b) -- (c) -- (d) -- (e) -- (f) -- (g)
	% 		 -- (h) -- (i) -- (j) -- (k) -- (l) -- (a);
	% 	\bowlsh{(0,-0.15)}{180}
	% 	\bowlsh{(0,1.15)}{0}
	% 	\bowlsh{(0.65,0.17)}{235}
	% 	\bowlsh{(0.65,0.84)}{305}
	% 	\bowlsh{(-0.65,0.17)}{125}
	% 	\bowlsh{(-0.65,0.84)}{55}
	% 	\bowlsh{(0.93,-1.05)}{210}
	% 	\bowlsh{(-0.93,2.05)}{30}
	% 	\bowlsh{(-0.93,-1.05)}{140}
	% 	\bowlsh{(0.93,2.05)}{320}
	% 	\bowlsh{(2.05,0.5)}{270}
	% 	\bowlsh{(-2.05,0.5)}{90}
\end{tikzpicture}

\caption{A construction of an example requiring at least $3$ colors for
	a V2V conflict-free chromatic guarding (cf.~Proposition~\ref{pro:v2v-3colors}).
	The bowl shape (see on the left), suitably squeezed and with a tiny
	opening between $p_1$ and $p_2$, is placed to four
	positions within the bowtie shape on the right.}
% \caption{Red vertex is on $\mathcal{L}$, and its upper tangent is drawn with red line segment. The blue vertex is on $\mathcal{R}$, and its upper tangent is drawn with blue line segment. The two black dots where the upper tangents touch the left and the right chains are the points of intersection.}
\label{fig:v2v-2color}
\end{figure}
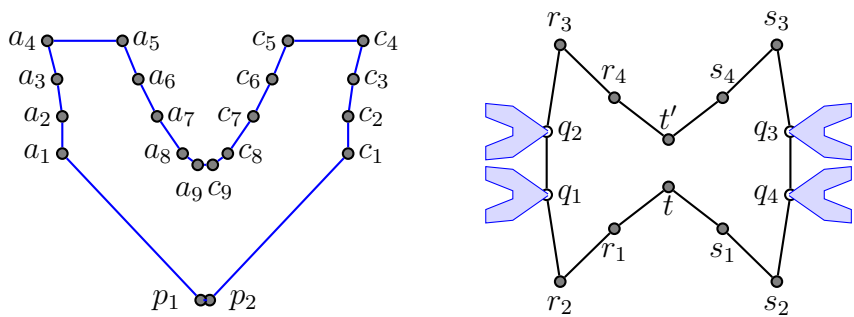

\begin{proof}
We call a ``{\em bowl\/}'' the simple polygon depicted in Figure~\ref{fig:v2v-2color} 
(note that it contains two copies of the shape from Fig.~\ref{fig:v2v-1color}).
In particular, the vertices $p_1,p_2$ see all other vertices of the bowl.
We claim the following:
\begin{itemize}
	\item[(i)] In any conflict-free $2$-coloring of the bowl there is a guard
	placed on $p_1$ or $p_2$ (or both).
\end{itemize}

Assume the contrary to (i), that is, existence of a $2$-coloring
of the bowl avoiding both $p_1$ and~$p_2$.
One can easily check that the subset $A=\{a_1,\ldots,a_8\}$ of the vertices
requires both colors, with guards possibly placed at $A\cup\{a_9\}$.
Symmetrically, there should be guards of both colors placed at the 
disjoint subset $\{c_1,\ldots,c_8,c_9\}$.
Then we have got a coloring conflict at both $p_1$ and $p_2$, thus proving (i).
(On the other hand, placing a single guard at either $p_1$ or $p_2$
gives a feasible conflict-free coloring of the bowl.)

\smallskip
The next step is to arrange four copies of the bowl within a suitable
simple polygon, as depicted on the right hand
side of Figure~\ref{fig:v2v-2color}.
More precisely, let $S$ be the (bowtie shaped) polygon on the right of the picture.
Note that the chains $C_1=(r_2,r_1,t,s_1,s_2)$ and
$C_2=(r_3,r_4,t',s_4,s_3)$ of $S$ are both concave, and each vertex
of $C_1$ sees all of~$C_2$.
We construct a polygon $S'$ from $S$ by making a tiny opening at each of the
vertices $q_1,q_2,q_3,q_4$, and gluing there a suitably rotated and squeezed copy of the
bowl, where gluing is done along a copy of the edge $\lseg{p_1}{p_2}$ of the bowl.
We call these openings at former vertices of $S$ the {\em doors} of~$S'$.
Obviously, the doors can be made so tiny that there is no accidental
visibility between a vertex inside a bowl and a vertex belonging
to the rest of~$S'$.

Assume, for a contradiction, that $S'$ admits a conflict-free coloring with
two colors, say red and blue.
Up to symmetry, let the unique color seen by vertex $t$ be blue.
Since $t$ sees all four doors, and (i) every door has a guard,
at least three guards at the doors are red.
Hence the unique color seen by symmetric $t'$ must also be blue.
Consequently, either there is only one blue guard at one of
$r_1,r_4,s_1,s_4,q_1,q_2,q_3,q_4$ (which all see both $t$ and $t'$), 
or there are two blue guards suitably placed at a pair of vertices from
$r_2,r_3,s_2,s_3$ (each of those sees one of~$t,t'$).
Moreover, a single blue guard cannot be placed at one of the doors
$q_1,q_2,q_3,q_4$, because that would leave $r_3$ or $s_3$ unguarded
(seeing two red and no blue).
Consequently, the guards placed at the doors          
$q_1,q_2,q_3,q_4$ must all be red, and so all vertices
$t,r_1,r_2,r_3,r_4,t',s_1,s_2,s_3,s_4$ must be guarded by a blue guard.
The latter is clearly impossible without a conflict
(similarly as in Fig.~\ref{fig:v2v-1color}).
\end{proof}

Our investigation of the V2V chromatic guarding problem,
although not giving further rigorous claims (yet), moreover
suggests the following conjecture:%
\begin{cnj}\label{cj:v2v-upper3}
Every simple polygon admits a weak conflict-free vertex-to-vertex 
chromatic guarding with at most~$3$ colors.
\end{cnj}

\subsubsection{Hardness of vertex-to-vertex conflict-free chromatic guarding}

In view of Conjecture~\ref{cj:v2v-upper3} and the algorithmic results for
other variants of chromatic guarding, it is natural to ask how difficult is
to decide whether using $1$ or $2$ colors in V2V guarding is enough for 
a given simple polygon.
Actually, for general graphs the question whether one can find a
conflict-free coloring with~$1$ color was investigated already long time
ago \cite{biggs-1973} (under the name of a perfect code in a graph), 
and its NP-completeness was shown by Kratochv{\'{\i}}l and K\v{r}iv{\'a}nek
in~\cite{DBLP:conf/mfcs/KratochvilK88}.
Previous lower bounds and hardness results in this area, 
including recent~\cite{k-cfc-graph}, however, do not seem to
directly help in the case of polygon visibility graphs.

Here we show that in both cases of $1$ or $2$ guard colors,
the conflict-free chromatic guarding problem on arbitrary simple polygons is NP-complete.
In each case we use a routine reduction from SAT, using a ``reflection model'' of a
formula.
This technique is, on a high approximate level, shown in
Figure~\ref{fig:hardness-rough}:
the variable values (T or F) are encoded in purely local {\em variable
gadgets}, which are privately observed by opposite {\em reflection (or copy)
gadgets} modelling the literals, and then reflected within precisely
adjustable narrow beams to again opposite {\em clause gadgets}.
Furthermore, there is possibly a special {\em guard-fix gadget} whose purpose 
is to uniquely guard the polygonal skeleton of the whole construction,
and so to prevent interference of the skeleton with the other local gadgets
(which are otherwise ``hidden'' from each other).

\begin{thm}\label{thm:v2vhardness}
For $c\in\{1,2\}$, the question whether a given simple polygon 
admits a weak conflict-free vertex-to-vertex chromatic guarding 
with at most~$c$ colors, is NP-complete.
\end{thm}

\begin{proof}
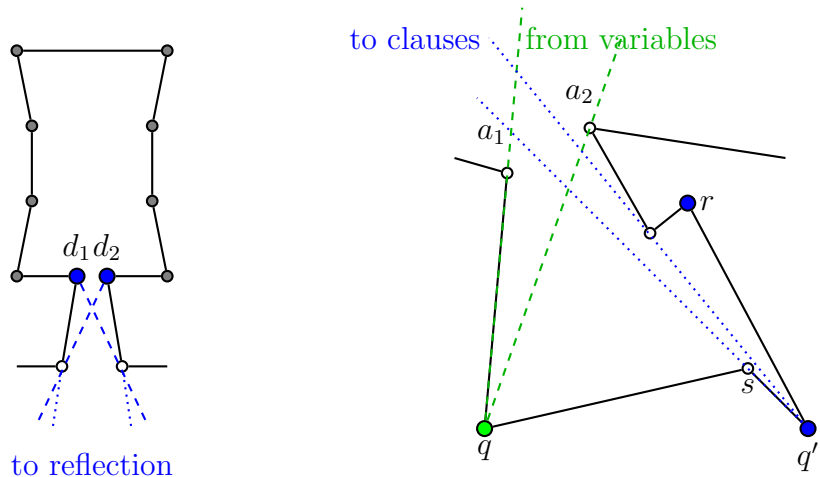
\begin{figure}[tbp]
	\centering\hfill
	\begin{tikzpicture}[scale=1]
		\tikzstyle{every path}=[thick];
		\tikzstyle{every node}=[draw, shape=circle, minimum size=4pt,
		inner sep=0pt, fill=gray];
		\node (a) at (0,0) {};
		\node (b) at (0.2,1) {};
		\node (c) at (0.2,2) {};
		\node (d) at (0,3) {};
		\node (e) at (2,3) {};
		\node (f) at (1.8,2) {};
		\node (g) at (1.8,1) {};
		\node (h) at (2,0) {};
		\tikzstyle{every node}=[draw, shape=circle, minimum size=4pt,
		inner sep=0pt, fill=none];
		\node[fill=blue,inner sep=2pt, label=above:$d_1$] (d1) at (0.8,0) {};
		\node[fill=blue,inner sep=2pt, label=above:$d_2$] (d2) at (1.2,0) {};
		\node (v1) at (0.6,-1.2) {};
		\node (v2) at (1.4,-1.2) {};
		\draw (a) -- (b) -- (c) -- (d) -- (e) -- (f) -- (g) -- (h);
		\draw (a) -- (d1) --  (v1) -- (0,-1.2) node[draw=none] {} ;
		\draw (h) -- (d2) --  (v2) -- (2,-1.2) node[draw=none] {} ;
		\draw[dotted,color=blue] (v1) -- (0.48,-2);
		\draw[dotted,color=blue] (v2) -- (1.52,-2);
		\draw[dashed,color=blue] (d2) -- (0.25,-2);
		\draw[dashed,color=blue] (d1) -- (1.75,-2);
		\node[draw=none,fill=none,color=blue, shape=rectangle] at (1,-2.5) {to	reflection};
	\end{tikzpicture}
	\qquad\hfill
	\begin{tikzpicture}[scale=2]
		\tikzstyle{every node}=[draw, shape=circle, minimum size=4pt,inner sep=0pt];
		\tikzstyle{every path}=[thick];
		\node[label=$a_1~~~$] (a) at (0.15,1.7) {};
		\node[fill=green,inner sep=2pt, label=below:$q$] (b) at (0,0) {};
		\node[label=below:$s$] (c) at (1.75,0.4) {};
		\node[fill=blue,inner sep=2pt, label=below:$q'$] (d) at (2.15,0) {};
		\node[fill=blue,inner sep=2pt, label=right:$r$] (e) at (1.35,1.5) {};
		\node (f) at (1.1,1.3) {};
		\node[label=$a_2~~$] (g) at (0.7,2) {};
		\draw (a) -- (b) -- (c) -- (d) -- (e) -- (f) -- (g) ;
		\draw (a) -- (-0.2,1.8) node[draw=none] {} ;
		\draw (g) -- (2,1.8) node[draw=none] {} ;
		\draw[dashed, color=green!70!black] (b) -- (0.25,2.8);
		\draw[dashed, color=green!70!black] (b) -- (0.91,2.6) node[draw=none] {from variables};
		\draw[dotted, color=blue] (d) -- (-0.05,2.2);
		\draw[dotted, color=blue] (d) -- (0.03,2.6) node[draw=none, label=left:to clauses] {};
	\end{tikzpicture}
	\hfill~%
	\caption{Left: the variable gadget, such that exactly one of $d_1,d_2$ must
		have a guard. ~Right: the reflection gadget, where $a_1,a_2$ are
		seen by a guard from outside, and one of $r,q'$ has a guard.}
	\label{fig:v2v-1hardness}
\end{figure}

For $c=1$ we reduce from the 1-in-3 SAT variant, which asks for an assignment
having exactly one true literal in each clause and which is NP-complete, too.
Note that in this case, there is only one color of guards, and so every
vertex must see exactly one guard.
Following the general scheme of Figure~\ref{fig:hardness-rough}, we now give
the particular variable and reflection gadgets:
\begin{itemize}
	\item The variable gadget is depicted in Figure~\ref{fig:v2v-1hardness} on the left. 
	The opening between $d_1$ and $d_2$ is sufficiently small to
	prevent accidental visibility between an inner vertex of this gadget
	(gray in the picture) and other vertices outside.
	
	In any conflict-free $1$-coloring of a polygon containing this
	gadget, precisely one of $d_1,d_2$ must have a guard (otherwise, the inner
	vertices cannot be guarded, as in Figure~\ref{fig:v2v-1color}),
	and there is no other guard on this gadget.
	
	\item The reflection gadget is depicted in Figure~\ref{fig:v2v-1hardness} on the right.
	Note that the visible angles of the vertices $q$ and $q'$ can be made
	arbitrarily tiny and fine-adjusted (independently of each other) by changing
	horizontal positions of $q,q'$ and~$s$.
	
	In any conflict-free $1$-coloring of a polygon containing this
	gadget, such that both $a_1,a_2$ see a guard (from outside),
	the following holds: a guard can be placed only at $q'$ or $r$,
	and the guard is at $q'$ if and only if $q$ sees a guard from outside
	(otherwise, there would be a conflict at $q$ or $q$ would not be guarded).
	
	\item Every clause gadget is just a single vertex positioned on a concave
	chain as shown in Figure~\ref{fig:hardness-rough}.
	There is no guard-fix gadget present in this construction.
\end{itemize}

For a given 3-SAT formula $\Phi=(x_i\vee\neg x_j\vee x_k)\wedge\ldots$,
the construction is completed as follows.
Within the frame of Figure~\ref{fig:hardness-rough}, we place a copy of the
variable gadget for each variable of $\Phi$.
We adjust these gadgets such that the combined visible angles of $d_1,d_2$
of each variable gadget do not overlap with those of other variables on the
bottom base of the frame.
(The lower left and right corners of the frame are seen by the first and
last variable gadgets, respectively.)

For each literal $\ell$ containing a variable $x_i$ we place a copy of the
reflection gadget at the bottom of the frame, such that its opening
$a_1,a_2$ is visible from both $d_1$ and $d_2$ of the gadget of $x_i$.
We adjust the narrow visible angle of $q$ (by moving $q$
horizontally) such that $q$ sees $d_1$ but not $d_2$
if the literal $\ell$ is $x_i$, and $q$ sees $d_2$ but not $d_1$
if $\ell=\neg x_i$.
Then we adjust the visible angle of $q'$ such that it sees exactly the one
vertex (on the top concave chain of the frame) which represents the clause
containing~$\ell$.

This whole construction can clearly be done in polynomial time and 
precision.
To recapitulate, each of the vertices of our constructed polygon is
\begin{itemize}
	\item coming from a copy of the variable gadget for each variable of $\Phi$, or
	\item from a copy of the reflection gadget added for each literal in~$\Phi$, or
	\item is a singleton representative of one of the clauses of $\Phi$, or
	\item is the auxiliary lower-left or right corner.
\end{itemize}

Assume that $\Phi$ has a 1-in-3 satisfying assignment.
Then we place a guard at $d_1$ of a variable gadget of $x_i$ if $x_i$ is
true, and at $d_2$ otherwise.
Moreover, at a reflection gadget of a literal $\ell$, we place a guard at
$q'$ if $\ell$ is true, and at $r$ otherwise.
No other guards are placed.
In this arrangement, every vertex of a variable or reflection gadget sees
precisely one guard, regardless of the evaluation of $\Phi$.
Moreover, since the assignment of $\Phi$ makes precisely one literal of each
clause true, every clause vertex also sees one guard.
This is a valid conflict-free $1$-coloring.

Conversely, assume we have a conflict-free $1$-coloring of our polygon.
Since precisely one of $d_1,d_2$ of each variable gadget of $x_i$ has a guard,
this correctly encodes the truth value of $x_i$ (true iff the guard is
at~$d_1$), and we know that both $a_1,a_2$ of each reflection gadget
of a literal $\ell$ see an outside guard.
Hence, by our construction, the gadget of $\ell$ has a guard at $q'$ if and
only if $\ell$ is true in our derived assignment of variables.

Furthermore, any clause vertex can see a guard only at a vertex $q'$ of some
reflection gadget. 
Otherwise (including the case of a guard at some clause vertex), we would
necessarily get a guard conflict at a vertex $a_1$ of some reflection gadget.
Consequently, every clause contains precisely one true literal
(as determined by the guard visible from this clause vertex).
The NP-completeness reduction for $c=1$ is finished.

\medskip

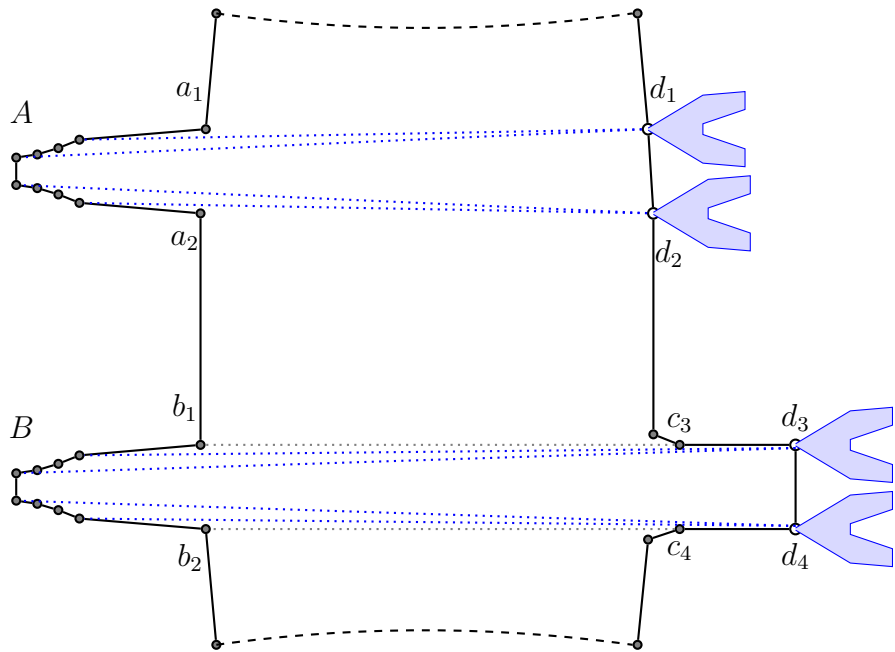
\begin{figure}[tbp]
	\def\bowlsh#1#2{%
		\begin{scope}[shift={#1}, rotate=#2, yscale=0.4, xscale=0.17]
			\begin{scope}[shift={(0,0.3)}]
				\draw[color=blue, thin, fill=blue!15!white] (-0.06,-0.3) -- (-1.9,1) -- (-2.1,2) -- (-1.1,2)
				-- (-0.3,1) -- (0.3,1) -- (1.1,2) -- (2.1,2) -- (1.9,1) -- (0.06,-0.3);
			\end{scope}
		\end{scope}
	}
	\centering
	\begin{tikzpicture}[scale=1.4]
		\tikzstyle{every node}=[draw, shape=circle, minimum size=3pt,
		inner sep=0pt, fill=gray];
		\tikzstyle{every path}=[thick];
		\node (a) at (0,0) {};
		\node (b) at (4,0) {};
		\node (c1) at (4.1,1) {};
		\node[label=below:$c_4$] (c2) at (4.4,1.1) {};
		\node[fill=white, minimum size=4pt, label=below:$d_4$] (c3) at (5.5,1.1) {};
		\node[fill=white, minimum size=4pt, label=above:$d_3$] (c4) at (5.5,1.9) {};
		\node[label=above:$c_3$] (c5) at (4.4,1.9) {};
		\node (c6) at (4.15,2) {};
		\node[fill=white, minimum size=4pt, label=below:$~~~d_2$] (d1) at (4.15,4.1) {};
		\node[fill=white, minimum size=4pt, label=above:$~~~d_1$] (d2) at (4.1,4.9) {};
		\node (e) at (4,6) {};
		\node (h) at (0,6) {};
		\node[label=$a_1~~~$] (j1) at (-0.1,4.9) {};
		\node (j2) at (-1.3,4.8) {};
		\node (j3) at (-1.5,4.72) {};
		\node (j4) at (-1.7,4.66) {};
		\node[label=$~~~A~~$] (j5) at (-1.9,4.63) {};
		\node (j6) at (-1.9,4.37) {};
		\node (j7) at (-1.7,4.34) {};
		\node (j8) at (-1.5,4.28) {};
		\node (j9) at (-1.3,4.2) {};
		\node[label=below:$\!\!a_2~~$] (j10) at (-0.15,4.1) {};
		\node[label=$b_1~~~$] (k1) at (-0.15,1.9) {};
		\node (k2) at (-1.3,1.8) {};
		\node (k3) at (-1.5,1.72) {};
		\node (k4) at (-1.7,1.66) {};
		\node[label=$~~~B~~$] (k5) at (-1.9,1.63) {};
		\node (k6) at (-1.9,1.37) {};
		\node (k7) at (-1.7,1.34) {};
		\node (k8) at (-1.5,1.28) {};
		\node (k9) at (-1.3,1.2) {};
		\node[label=below:$\!\!b_2~~$] (k10) at (-0.1,1.1) {};
		\draw (b) -- (c1) -- (c2) -- (c3) -- (c4) -- (c5)
		-- (c6) -- (d1) -- (d2) -- (e);
		\draw (h) -- (j1) -- (j2) -- (j3) -- (j4) -- (j5) 
		-- (j6) -- (j7) -- (j8) -- (j9) -- (j10) 
		-- (k1) -- (k2) -- (k3) -- (k4) -- (k5) 
		-- (k6) -- (k7) -- (k8) -- (k9) -- (k10) -- (a);
		\draw[dashed] (a) arc (98:82:14.2) -- (b) ;
		\draw[dashed] (h) arc (-98:-82:14.2) -- (e) ;
		\bowlsh{(c3)}{270}
		\bowlsh{(c4)}{270}
		\bowlsh{(d1)}{270}
		\bowlsh{(d2)}{270}
		\tikzstyle{every path}=[thick];
		\draw[dotted, color=blue] (5.5,1.13) -- (k9);
		\draw[dotted, color=blue] (5.5,1.13) -- (k6);
		\draw[dotted, color=blue] (5.5,1.87) -- (k5);
		\draw[dotted, color=blue] (5.5,1.87) -- (k2);
		\draw[dotted, color=gray] (c2) -- (k10);
		\draw[dotted, color=gray] (c5) -- (k1);
		\draw[dotted, color=blue] (d1) -- (j9);
		\draw[dotted, color=blue] (d1) -- (j6);
		\draw[dotted, color=blue] (d2) -- (j5);
		\draw[dotted, color=blue] (d2) -- (j2);
	\end{tikzpicture}
	
	\caption{The guard-fix gadget for conflict-free $2$-coloring:
		the four blue filled shapes glued to the frame on the right are copies of the bowl
		shape from Figure~\ref{fig:v2v-2color}, and the two ``pockets'' 
		on the left side of the frame enforce each pair of bowls to receive guards of both
		colors. The middle part of the frame is much wider than depicted here.
		See in the proof of Theorem~\ref{thm:v2vhardness}.}
	\label{fig:v2v-2colorF}
\end{figure}

We now move onto the $c=2$ case, which we reduce from the NP-complete
not-all-equal positive 3-SAT problem (also known as $2$-coloring of
$3$-uniform hypergraphs).
This special variant of 3-SAT requires every clause to have at least one true and one
false literal, and there are no negations allowed.
We again follow the same general scheme as for $c=1$, but this time 
the main focus will be on implementing the guard-fix gadget.
Let our guard colors be red and blue (then every vertex must see exactly
one red guard, or exactly one blue guard).

The left and right walls of the schematic frame from
Figure~\ref{fig:hardness-rough} are constructed as shown in
Figure~\ref{fig:v2v-2colorF}.
In the construction, we use four copies of the bowl
shape from Figure~\ref{fig:v2v-2color}, and we adopt the terminology
of gluing the bowls and of the {\em doors} from the proof of
Proposition~\ref{pro:v2v-3colors}.
For simplicity, while keeping in mind that the door of each bowl is a narrow
passage formed by a pair of vertices, we denote each door by a single letter
$d_i$, $i=1,2,3,4$ (as other vertices).

This guard-fix gadget is constructed such that both $d_1,d_2$ see all
the $8$ vertices of the ``pocket'' $A$ on the left, but neither of $a_1,a_2$ does so.
We recall the following property from the proof of Proposition~\ref{pro:v2v-3colors}:
\begin{itemize}
	\item[(i)] In any conflict-free $2$-coloring of the bowl there is a guard
	(or two) placed at the door.
\end{itemize}
Consequently, the guards placed at $d_1$ and $d_2$ must be of different
colors (red and blue); otherwise, say for two red guards at $d_1,d_2$, each
of the $8$ vertices in $A$ would have to be guarded by a blue guard
within $A\cup\{a_1,a_2\}$ which is not possible.
(Though, we have not yet excluded the case that, say, $d_1$ would have 
a red guard and $d_2$ a blue and a red guards.)

In the next step, we note that $d_3,d_4$ see all $8$ vertices of the 
``pocket'' $B$, but none of $b_1,b_2,c_3,c_4$ does so.
So, by analogous arguments, the guards placed at $d_3,d_4$ must be of
different colors (red and blue).
We remark that this does not necessarily cause a conflict with the guards
at $d_1,d_2$ since the visibility between $b_1,d_3$ and between $b_2,d_4$
is blocked by $c_3$ and $c_4$, respectively.
However, the vertices $b_1$ and $b_2$ are now ``exhausted'' in the sense
that one sees (at least) one red and two blue guards, and the other one blue
and two red guards.
Altogether, this implies that
\begin{itemize}\item[(ii)]
	there is exactly one red and one blue guard
	among $d_1,d_2$ and the same holds among $d_3,d_4$,
	and no other vertex visible from both $b_1$ and $b_2$ can have a guard.
\end{itemize}

\begin{figure}[tbp]
	\def\bowlsh#1#2{%
		\begin{scope}[shift={#1}, rotate=#2, yscale=0.25, xscale=0.1]
			\begin{scope}[shift={(0,0.3)}]
				\draw[color=blue, thin, fill=blue!15!white] (-0.06,-0.3) -- (-1.9,1) -- (-2.1,2) -- (-1.1,2)
				-- (-0.3,1) -- (0.3,1) -- (1.1,2) -- (2.1,2) -- (1.9,1) -- (0.06,-0.3);
			\end{scope}
		\end{scope}
	}
	\centering
	\begin{tikzpicture}[yscale=1.6,xscale=1.3]
		\tikzstyle{every node}=[draw, shape=circle, minimum size=3pt,
		inner sep=0pt, fill=gray];
		\tikzstyle{every path}=[thick];
		\node (a) at (0,0) {};
		\node (b) at (8,0) {};
		\coordinate (b1) at (8,1) {};
		\node[color=blue,fill=blue!30!white, minimum size=4pt,
		label=below:$d_4$] (b2) at (8.5,1.02) {};
		\node[color=blue,fill=blue!30!white, minimum size=4pt,
		label=above:$d_3$] (b3) at (8.5,1.18) {};
		\coordinate (b4) at (8,1.2) {};
		\node[color=blue,fill=blue!30!white, minimum size=4pt,
		label=below:$~~d_2\!\!\!\!$] (b5) at (8,2.02) {};
		\node[color=blue,fill=blue!30!white, minimum size=4pt,
		label=above:$~~d_1\!\!\!\!$] (b6) at (8,2.18) {};
		\node (e) at (8,3) {};
		\node[label=$c_1$] (h) at (0,3) {};
		\node[label=$c_2$] (h2) at (0.5,2.85) {};
		\node[label=$c_3$] (h3) at (1,2.74) {};
		\node[label=$c_4$] (h4) at (1.5,2.65) {};
		\node[label=$\dots$] (h5) at (2,2.58) {};
		\coordinate[label=below:$\!\!b_2~~$] (f1) at (0,1) {};
		\coordinate (f2) at (-0.8,1.05) {};
		\coordinate (f3) at (-0.8,1.15) {};
		\coordinate[label=$\!\!b_1~~$] (f4) at (0,1.2) {};
		\coordinate[label=below:$\!\!a_2~~$] (g1) at (0,2.0) {};
		\coordinate (g2) at (-0.8,2.05) {};
		\coordinate (g3) at (-0.8,2.15) {};
		\coordinate[label=$\!\!a_1~~$] (g4) at (0,2.2) {};
		\draw[dashed, color=gray] (a) arc (98:82:28.6) -- (b) ;
		\draw[dashed, color=gray] (h) arc (-105:-75:15.2) -- (e) ;
		\draw (a) -- (f1) -- (f2) -- (f3) -- (f4)
		-- (g1) -- (g2) -- (g3) -- (g4) -- (h);
		\draw (b) -- (b1) -- (b2) -- (b3) -- (b4) --(b5) -- (b6) -- (e);
		\draw (h) -- (h2) -- (h3) -- (h4) -- (h5) -- (2.2,2.56);
		\draw[dotted, color=blue] (b5) -- (2,2.58);
		\draw[dotted, color=blue] (b6) -- (a);
		\draw[dotted, color=gray] (f4) -- (8,2.98);
		
		\node (u1) at (4.78,2.49) {};
		\node[fill=white, label=left:$x_2$] (u2) at (4.95,3.01) {};
		\node (u3) at (5.02,2.5) {};
		\bowlsh{(u2)}{0}
		\draw[dotted, color=green] (u1) -- (4.2,0.1);
		\draw[dotted, color=green] (u3) -- (5.2,0.1);
		\node (w1) at (3.95,2.466) {};
		\node[fill=white, label=left:$x_1$] (w2) at (4.3,3) {};
		\node (w3) at (4.25,2.47) {};
		\bowlsh{(w2)}{0}
		\draw[dotted, color=red] (w1) -- (2.45,0.1);
		\draw[dotted, color=red] (w3) -- (3.85,0.1);
		\node (v1) at (5.5,2.54) {};
		\node[fill=white, label=left:$x_3$, label=right:$~~\dots$] (v2) at (5.5,3.03) {};
		\node (v3) at (5.65,2.56) {};
		\bowlsh{(v2)}{0}
		\draw[dotted, color=red] (v2) -- (5.5,0.08);
		\draw[dotted, color=red] (v2) -- (6.4,0.07);
		\draw (3.7,2.47) -- (w1) -- (w2) -- (w3) -- (u1) -- (u2) 
		--(u3) --(v1) --(v2) --(v3) -- (5.9,2.59);
		
		\node (l1) at (2.8,0.28) {};
		\node (l2) at (3.23,-0.4) {};
		\node (l3) at (2.92,0.285) {};
		\draw[dotted, color=green] (l2) -- (h4);
		\node (ll1) at (3.2,0.3) {};
		\node (ll2) at (4,-0.3) {};
		\node (ll3) at (3.35,0.3) {};
		\draw[dotted, color=green] (ll2) -- (h);
		\draw (2.6,0.27) -- (l1) -- (l2) -- (l3) -- (ll1) -- (ll2) -- (ll3)
		-- (3.6,0.3);
		\node[label=240:$l_1$] (lm1) at (4.7,0.28) {};
		\node[label=right:$l_2$] (lm2) at (5.7,-0.3) {};
		\node[label=60:$l_3$] (lm3) at (4.9,0.278) {};
		\draw[dotted, color=red] (lm2) -- (h3);
		\draw (lm1) -- (lm2) -- (lm3);
	\end{tikzpicture}
	
	\caption{Placement of the variable/reflection/clause gadgets for conflict-free $2$-coloring:
		the clause gadgets are the single vertices $c_1,c_2,\ldots$,
		the variable gadgets are formed by copies of the bowl at $x_1,x_2,\ldots$,
		and the reflection gadgets are simply triples of vertices as
		$l_1,l_2,l_3$ at the bottom line.
		In this example, the value (color red or blue) of the variable
		$x_1$ is reflected towards clauses $c_1$ and $c_4$ (which are
		thus assumed to contain literal~$x_1$), and the value of
		$x_2$ is reflected towards~$c_3$.}
	\label{fig:v2v-2colorG}
\end{figure}
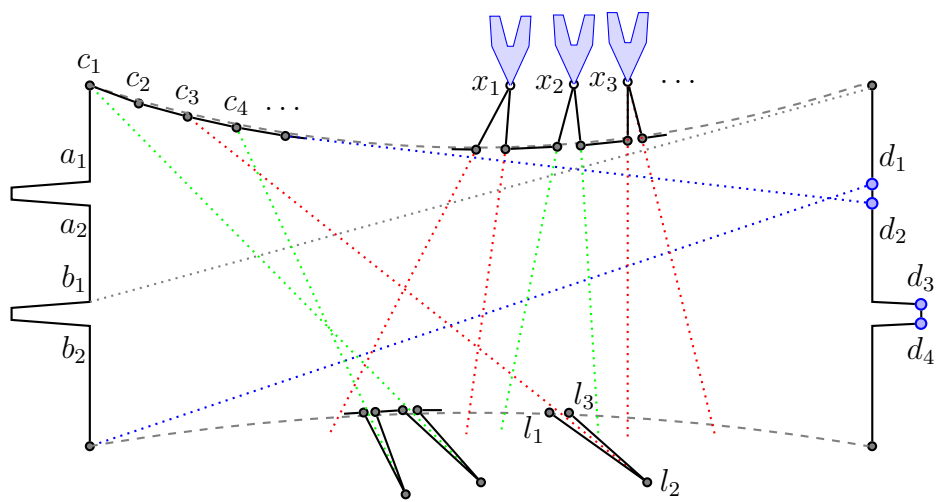

The rest of the construction is, within the frame constructed above
(Figure~\ref{fig:v2v-2colorF}), already quite easy.
Let $\Phi$ be a given 3-SAT formula without negations.
See Figure~\ref{fig:v2v-2colorG}.
\begin{itemize}
	\item We again represent each clause of $\Phi$ by a single vertex on a
	concave chain on the top of our frame.
	This chain of clause vertices is ``slightly hidden'' in a sense that it is
	not visible from $d_1$ or $d_2$, but it is all visible from $b_1$ and~$b_2$.
	\item Each variable $x_i$ of $\Phi$ is represented by a copy of the bowl,
	also placed on the top of the frame (but separate from the
	section of clause vertices).
	As before, the visible angles of the variable gadgets are adjusted 
	so that they do not overlap on the bottom line of the frame.
	\item Each literal $\ell=x_i$ is represented by a triple of vertices
	as $l_1,l_2,l_3$ in Figure~\ref{fig:v2v-2colorG}, such that $l_1,l_3$ are
	in the visible angle of the variable-$x_i$ gadget, while $l_2$ is ``deeply
	hidden'' so that $l_2$ sees only the clause vertex $c_j$ that $\ell$ belongs to.
\end{itemize}

The final argument is analogous to the $c=1$ case.
Assume we have a not-all-equal assignment of~$\Phi$.
We put red guards to $d_1$ and $d_3$ and blue guards to $d_2$ and~$d_4$.
For each variable $x_i$, we put one guard to $x_1$ of blue color if $x_i$
is true and of red color if $x_i$ is false.
Whenever $\ell=x_i$ is a literal represented by the triple $l_1,l_2,l_3$,
we put to $l_2$ a guard of the same color as of the guard at $x_i$.
Then every clause $c_j=(x_a\vee x_b\vee x_c)$ will see the colors of guards
at $x_a,x_b,x_c$, and since the values assigned to $x_a,x_b,x_c$ are not all
the same, one of the colors is unique to guard~$c_j$.
We have got a conflict-free $2$-coloring.

On the other hand, assume a conflict-free $2$-coloring.
By (ii), there are no guards at the clause vertices $c_1,c_2,\ldots$,
and so those vertices can be only guarded from the reflection gadgets.
Assume a reflection gadget of a literal $\ell=x_i$.
Then, again by (ii), the vertices $l_1,l_3$ of this gadget have no guards,
and they see a red and a blue guard from $d_1,d_2$.
On the other hand, $l_2$ cannot see any other guard except one placed at
$l_2$, and so there has to be a guard at~$l_2$.
If, moreover, $l_1,l_3$ see a red (say) guard placed at $x_i$
(and $x_i$ as the door of a glued bowl there must have a guard by (i)\,),
then the guard at $l_2$ must also be red (or $l_1$ would have a conflict).
Consequently, every clause $c_j=(x_a\vee x_b\vee x_c)$ sees the colors of
the guards placed at $x_a$, $a_b$ and $x_c$, and since the coloring is
conflict-free, there have to be both colors visible (one red plus two blue,
or one blue plus two red).
From this we can read a valid not-all-equal variable assignment of~$\Phi$.
\end{proof}

\chapter{Unit Disk Visibility Graphs} \label{chap:udvg}

\section{Summary of the chapter}
In this chapter, we introduce a new graph class called \emph{unit disk visibility graphs}.
This class aims to model the real-world scenarios more accurately, considering the distance between the geometric objects.
The results we present in this chapter are as follows:

\begin{itemize}
	\item We prove that unit disk visibility graphs are a proper superclass of the visibility graphs (Section~\ref{sec:classification}).
	\item We prove that determining whether a given unit disk visibility graph is 3-colorable is NP-complete for a set of line segments, and we remark that this result also applies for a set of points (Section~\ref{sec:3coloringSegment}).
	\item We prove that determining whether a given unit disk visibility graph is 3-colorable is NP-complete for a polygon with holes (Section~\ref{sec:withholes}).
	\item We discuss some combinatorial problems that are studied for visibility graphs, but might have different aspects in unit disk visibility graphs (Section~\ref{sec:othercomb}).
\end{itemize}

\section{Related work}
In the literature, visibility graphs were studied considering various geometric sets such as a simple polygon \cite{o-agta-87}, a polygon with holes \cite{Wein_voronoi}, a set of points \cite{Cardinal_pointcomplexity}, a set of line segments \cite{Everett_planarsegment}, along with different visibility models such as line-of-sight visibility \cite{Garey_lineOfSight}, $\alpha$-visibility \cite{Ghodsi_alpha}, and $\pi$-visibility \cite{Urrutia_artGalleryAndIllum}.

Although the visibility graphs have been extensively studied, there are several fundamental problems which are open for some type of visibility graphs \cite{Ghosh_unsolvedproblems}. Concerning different applications (e.g. VLSI design, pattern recognition and robot motion planning), the visibility relations are restricted differently. 
To model the applications more accurately, various types of visibility graphs have been introduced such as bar visibility graphs \cite{Duchet_planar}, rectangle visibility graphs \cite{Bose_rectVis} and circle visibility graphs \cite{Hutchinson_arcVis}.  
When the environments which are modeled using the mentioned models are investigated by its properties, the most general types can be listed as compact visibility graphs \cite{Kant_compact} and polygon visibility graphs \cite{Hershberger_vis}.
Moreover, the environments that are modeled using polygons were classified into severalberbaer birbirimizi sub-types to capture some specific structural properties.
Those sub-types are, but not limited to, when the polygon is orthogonal \cite{Kahn_watchmen}, staircase \cite{Abello_staircase}, spiral \cite{Everett_spiral}, terrain \cite{Evans_terrain} or funnel \cite{Choi_funnel}.
In addition to the simple polygons, the visibility graphs of polygons with holes were also studied.
A \textit{hole} inside a polygon is simply another polygon which acts as an obstacle that blocks the visibility.
The holes can be polygonal \cite{Wein_voronoi}, disk \cite{Kim_shortestPF} and even mobile \cite{Khaili_movingObs}. 

Visibility graphs are used to describe real-world scenarios majority of which concern the mobile robots and path planning \cite{Latombe_robotmotion,Berg_compgeo,Oommen_robotnavigation}.
While modeling the environment in which the robots move, a very common tool is to interpret the geometric entities and the relations among them as visibility graphs \cite{ORourke_handbook,Aichholzer_convexifying,Overmars_newmethods,Floriani_visibilityonterrain}.

However, the physical limitations of the real world are usually overlooked or ignored while using visibility graphs. Since no camera, sensor, or guard (the objects represented by vertices of the visibility graph) has infinite range, two objects might not sense each other even though there are no obstacles in-between. Based on such a limitation, we assume that if a pair of objects  are too far from each other, then they do not see each other.

Some other restrictions are related to the visibility range. For instance, in weak visibility graphs of a set of line segments, at least two mutually visible endpoints are required for a pair of line segments to see each other \cite{Ghosh_weakvis}.
Another example is when the maximum visibility angle $\alpha$ is introduced instead of assuming each object can see all directions \cite{Ghodsi_alpha}.
In that case, a line segment is visible by an endpoint if this endpoint can see that line segment with an angle of  at least $\alpha$.
Moreover, limitation on the guarding range has been considered for terrain visibility graphs \cite{Khodakarami201715DTG}.
Instead of restricting, the case where the visibility range is extended by introducing mirror vertices and edges was also studied \cite{Klee_illuminable}.
In that case, the visibility is transferred by reflecting the straight lines using the mirroring objects.  

In this chapter, we introduce a new visibility model called \emph{unit disk visibility graphs} (UDVG).
In this model, the distance between two objects (points, vertices) also affects the visibility relation as a generalization of \cite{Khodakarami201715DTG}.
More specifically, we assume that two objects are not mutually visible if the Euclidean distance between them is greater than one unit, regardless of the straight line connecting them not containing any other object.
Our aim is to offer an insight into this new model, and investigate the algorithmic aspects of some combinatorial problems, mainly the 3-colorability problem.

We introduce a new type of visibility graphs which represents the real world applications in a finer detail.
We study \emph{unit disk visibility graphs} (UDVG) where the distance between two objects is taken into account in addition to the usual visibility restrictions.
Since the sensing range of the sensors, cameras or robots can not be ignored as the study of visibility graphs originates from network applications, we expect unit disk visibility graphs to be a natural extension of visibility graphs.

\section{The classification of the visibility graphs and the unit disk visibility graphs} \label{sec:classification}
\begin{lem} \label{lem:scaledown}
	Consider a set $P = \{p_1, \dots, p_n\}$ of points, and the visibility graph $G(P)$ of $P$.
	There exists an embedding $\Sigma(P)$ of $P$, such that the Euclidean distance between every pair $p,q \in P$ is less than one unit, preserving the visibility relations.
\end{lem}
\begin{proof}
	Let $\Phi = \{x_1, y_1, x_2, y_2, \dots, x_n, y_n\}$ be the set of all coordinates used to represent the set $P$, and let $\varphi \in \Phi$ be a coordinate whose absolute value is the largest number in $\Phi$.
	
	The circle $C$ with center $(0,0)$ and radius $\varphi$ contains all the points in $P$. 
	Now imagine we shrink $C$ into a unit circle, and scale the point set accordingly.
	In order to do that, we divide every coordinate in $\Phi$ by $2\varphi$. 
	That is, the new coordinates for the points are $(x_1/2\varphi, y_1/2\varphi), \dots, (x_n/2\varphi, y_n/2\varphi)$.
	
	Considering any pair $s_i, s_j \in S$, and the straight line $\ell(i,j)$ passing through these points, the slope of $\ell$ does not change. 
	Therefore, the visibility relations are preserved for $P$.
\end{proof}

By Lemma~\ref{lem:scaledown}, a given set $P$ of points can be scaled down to obtain $P'$ so that every point in $P'$ is inside a unit circle, and the visibility graph $G(P)$ of $P$ is exactly the same as the visibility graph $G(P')$ of $P'$.
Thus, the unit disk visibility graph of $P'$ is also isomorphic to $G(P')$ since no pair of points in $P'$ has Euclidean distance greater than one unit. We easily get the following.

Now, let us show that every combinatorial problem that is NP-hard in $P$ is also NP-hard in $P'$.
\begin{lem} \label{lem:hardUDVG}
	If a problem $\mathfrak{Q}$ is NP-hard for point visibility graphs, then $\mathfrak{Q}$ is also NP-hard for unit disk point visibility graphs.
\end{lem}
\begin{proof}
	Let $P = \{p_1, \dots, p_n\}$ be a geometric set (points, vertices of a polygon, endpoints of a set of line segments) in the Euclidean plane, with coordinates $(x_1, y_1)$, $\dots$, $(x_n, y_n)$, respectively.
	Let $\mathcal{A}$ be an algorithm that solves an instance $\mathcal{Q}$ of the problem $\mathfrak{Q}$ on $P$, in polynomial time.
	
	As shown in Lemma~\ref{lem:scaledown}, $P$ can be scaled down into a unit circle, preserving the relations.
	To avoid the possible coordinates with exponentially many digits, instead of making the transformation $(x_i, y_i) \to (x_i/2\varphi, y_i/2\varphi)$ for each $(x_i, y_i)$, we let $M$ be the smallest integer larger than $|\varphi|$, which can be expressed as $2^k$.
	Since dividing a number $n$ by $2^k$ requires at most $k = O(\log n)$ many digits, applying the transformation $(x_i, y_i) \to (x_i/2M, y_i/2M)$ does not create coordinates with exponentially many digits.
	
	So, let $P' = \{p'_1, \dots, p'_n\}$ be a point set with coordinates $(x_1/2M, y_1/2M), \dots, (x_n/2M,$ $y_n/2M)$, respectively.
	If $\mathcal{A}$ solves the problem $\mathfrak{Q}$ on $P'$ in polynomial time, then $\mathcal{A}$ solves $\mathfrak{Q}$ on $P$ also in polynomial time.
	
	Therefore, any $\mathfrak{Q}$ is also NP-hard for unit disk point visibility graphs.
\end{proof}

\begin{rem}
	By Lemma~\ref{lem:hardUDVG}, the minimum vertex cover, the maximum independent set and the minimum dominating set problems which have been shown to be NP-hard for visibility graphs by \cite{Lin_complexityaspects,ll-ccagp-86} are also NP-hard for unit disk visibility graphs.
\end{rem}

\begin{figure}[htbp]
	\captionsetup[subfigure]{position=b,justification=centering}
	\centering
	\subfloat[]{
		\centering
		\begin{tikzpicture}[scale=0.8]
			\tikzstyle{every node}=[draw=black, fill=gray, shape=circle, minimum size=3pt,inner sep=0pt];
			\node[fill=red] (center) at (0,0) {};

			\foreach \i in {1,...,11} {
				\coordinate (\i) at (\i*360/11:1cm);
			}
			
			\draw (1) -- (11);
			\draw (center)--(11);
			\foreach \i in {3,...,11} {
				\pgfmathtruncatemacro\j{\i-1};
				\draw (\i) -- (\j);
				\draw (center)--(\j);
			}
			
			\foreach \i in {1,2,4,6,8,10} \draw[thick, red] (center)--(\i);
			
			\foreach \i in {3,5,7,9,11} \node at (\i) {};
			
			\tikzstyle{every node}=[draw=red, fill=red, shape=circle, minimum size=3pt,inner sep=0pt];
			\foreach \i in {1,2,4,6,8,10} \node at (\i) {};
			
		\end{tikzpicture}
		\label{fig:2a}
	}
	~
	\subfloat[]{
		\centering
		\begin{tikzpicture}[scale=3]
			%\node[draw=none, fill=none, rotate=30] at (0.75,0) {\tiny $1$ unit};
			\node[draw=none, fill=none, rotate=26] at (0.78,0) {\tiny $1$ unit};
			
			\node[draw=none, fill=none] at (0.5,-0.07) {\small $u$};
			
			\tikzstyle{every node}=[draw=black, fill=gray, shape=circle, minimum size=3pt,inner sep=0pt];
			%\draw[<->] (0.55,-0.05) -- (0.95,0.15);
			\draw[<->] (0.55,-0.05) -- (0.985,0.145);
			
			\node (1) at (0,0.2) {};
			\node (2) at (0.1,0.2) {};
			\node (3) at (0.2,0.2) {};
			\node (4) at (0.3,0.2) {};
			\node (5) at (0.4,0.2) {};
			\node (6) at (0.5,0.2) {};
			\node (7) at (0.6,0.2) {};
			\node (8) at (0.7,0.2) {};
			\node (9) at (0.8,0.2) {};
			\node (10) at (0.9,0.2) {};
			
			\node (11) at (1,0.2) {};
			
			\node (12) at (0.5,0) {};
			
			\foreach \i in {1,...,10}
			{
				\pgfmathtruncatemacro\j{\i+1};
				\draw (\i) -- (\j);
				\draw (\i) -- (12);
			}
			
			\draw (11) -- (12);
			
			%DELETE TO REMOVE RED EDGES
			\foreach \i in {1,3,5,7,9,11} \draw[thick, red] (12)--(\i);
			%DELETE TO REMOVE RED EDGES
			
		\end{tikzpicture}
		\label{fig:2b}
	}
	~
	\subfloat[]{
		\centering
		\begin{tikzpicture}[scale=3]
			\node[draw=none, fill=none, rotate=30] at (0.75,0) {\tiny $1$ unit};
			
			\node[draw=none, fill=none] at (0.4,-0.2) {\small $u$};
			\node[draw=none, fill=none] at 
			(0.4,-0.03) {\small $v$};
			
			\tikzstyle{every node}=[draw=black, fill=gray, shape=circle, minimum size=3pt,inner sep=0pt];
			\draw[<->] (0.55,-0.05) -- (0.95,0.15);
			
			\node (1) at (0,0.2) {};
			\node (2) at (0.1,0.2) {};
			\node (3) at (0.2,0.2) {};
			\node (4) at (0.3,0.2) {};
			\node (5) at (0.4,0.2) {};
			\node (6) at (0.5,0.2) {};
			\node (7) at (0.6,0.2) {};
			\node (8) at (0.7,0.2) {};
			\node (9) at (0.8,0.2) {};
			\node (10) at (0.9,0.2) {};
			\node (11) at (0.5,0) {};
			\node (12) at (0.5,-0.2) {};
			
			\foreach \i in {1,3,5,7,9,11}
			{
				\pgfmathtruncatemacro\j{\i+1};
				\draw[ultra thick] (\i) -- (\j);
				\draw (\i) -- (11);
				\draw (\j) -- (11);
			}
			
			\foreach \i in {2,4,6,8}
			{
				\pgfmathtruncatemacro\j{\i+1};
				\draw (\i) -- (\j);
			}
			
			\draw (10) -- (11);
			
			%DELETE TO REMOVE RED EDGES
			\foreach \i in {1,3,5,7,9,12} \draw[thick, red] (11)--(\i);
			%DELETE TO REMOVE RED EDGES
			
		\end{tikzpicture}
		\label{fig:2c}
	}
	~
	\subfloat[]{
		\centering
		\begin{tikzpicture}[scale=3]
			%\node[draw=none, fill=none, rotate=45] at (0.8,0.15) {\tiny $1$ unit};
			\node[draw=none, fill=none, rotate=45] at (0.8,0.15) {\tiny $1$ unit};
			
			\node[draw=none, fill=none] at (0.5,-0.07) {\small $u$};
			
			\tikzstyle{every node}=[draw=black, fill=gray, shape=circle, minimum size=3pt,inner sep=0pt];
			%\draw[<->] (0.55,-0.05) -- (0.95,0.45);
			\draw[<->] (0.55,-0.05) -- (1.04,0.475);
			
			\node (1) at (0,0.5) {};
			\node (2) at (0.1,0.46) {};
			\node (3) at (0.2,0.425) {};
			\node (4) at (0.3,0.395) {};
			\node (5) at (0.4,0.375) {};
			\node (6) at (0.5,0.375) {};
			\node (7) at (0.6,0.395) {};
			\node (8) at (0.7,0.425) {};
			\node (9) at (0.8,0.46) {};
			\node (10) at (0.9,0.5) {};
			
			\node (11) at (1,0.54) {};
			
			\node (12) at (0.5,0) {};
			
			\foreach \i in {1,...,11}
			{
				\pgfmathtruncatemacro\j{\i+1};
				\draw[ultra thick] (\i) -- (\j);			
			}
			\draw[ultra thick] (1) -- (12);			
			\foreach \i in {2,...,10}
			{
				\draw (\i)--(12);
			}
			
			%DELETE TO REMOVE RED EDGES
			\foreach \i in {1,3,5,7,9,11} \draw[thick, red] (12)--(\i);
			%DELETE TO REMOVE RED EDGES
			
		\end{tikzpicture}
		\label{fig:2d}
	}
	
	\caption{\textsc{(a)} A graph with an induced $K_{1,6}$ (indicated with red edges). Unit disk visibility graphs for \textsc{(b)} a set of points, \textsc{(c)} a set of segments, and \textsc{(d)} a simple polygon, each containing an induced $K_{1,6}$.}
	\label{fig:notUD}
\end{figure}
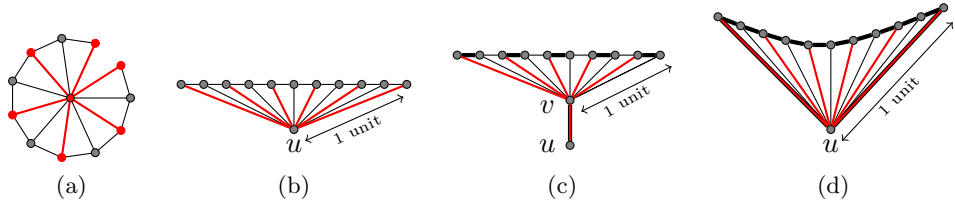			

In Figure~\ref{fig:2a}, we have a graph which contains $K_{1,6}$ as an induced subgraph depicted with red vertices and red edges. Our purpose is to show that an induced $K_{1,6}$ can be contained in the unit disk visibility graph of a set of points, a set of line segments and a simple polygon. In Figure~\ref{fig:2b}, we are given a set of twelve points where all of them except $u$ are collinear. Observe that among the collinear points, we can choose at most six pairwise non-adjacent points to be in our induced subgraph, otherwise at least two of them see each other. With the addition of $u$, we obtain an induced $K_{1,6}$ since the distance between $u$ and the farthest point is one unit. In Figure~\ref{fig:2c}, we have a similar configuration where we can choose at most five pairwise non-adjacent endpoints among the collinear segments. By adding both endpoints $u$ and $v$, we obtain an induced $K_{1,6}$ since $u$ does not see other points than $v$ due to distance and $v$ sees all endpoints. In Figure~\ref{fig:2d}, we also have a similar configuration to Figure~\ref{fig:2b}. The vertices of a polygon do not see each other if the straight line joining them is outside of the polygon. Therefore, among the vertices except $u$, we can choose at most six pairwise non-adjacent points to be in our induced subgraph. With the addition of $u$, we obtain an induced $K_{1,6}$.

\begin{rem} \label{rem:super}
	As we have shown that the unit disk visibility graphs are a superclass of visibility graphs, we emphasize that the unit disk visibility graphs are not a subclass of unit disk graphs since unit disk graphs can not contain an induced $K_{1,6}$ \cite{Marathe_heuristics} while a unit disk visibility graph of a set of points, a set of line segments, and a simple polygon can contain it as demonstrated in Figure~\ref{fig:notUD}. Therefore, the combinatorial problems that are NP-hard for unit disk graphs are not necessarily hard on unit disk visibility graphs, and the question whether unit disk graphs are a subclass of unit disk visibility graphs is left open.
\end{rem}

\begin{lem} \label{lem:K1,6}
	Unit disk visibility graphs are not a subclass of unit disk graphs. 
\end{lem}
\begin{proof}
	Unit disk graphs cannot contain an induced $K_{1,6}$ \cite{Marathe_heuristics} while a unit disk visibility graph of a set of points, a set of line segments, and a simple polygon can contain it.
	In Figure~\ref{fig:2a}, we have a graph which contains $K_{1,6}$ as an induced subgraph shown with red vertices and red edges. Our purpose is to show that an induced $K_{1,6}$ can be contained in the unit disk visibility graph of a set of points, a set of line segments and a simple polygon. In Figure~\ref{fig:2b}, we are given a set of twelve points where all of them except $u$ are collinear. Observe that among the collinear points, we can choose at most six pairwise non-adjacent points to be in our induced subgraph, otherwise at least two of them see each other. With the addition of $u$, we obtain an induced $K_{1,6}$ since the distance between $u$ and the farthest point is one unit. In Figure~\ref{fig:2c}, we have a similar configuration where we can choose at most five pairwise non-adjacent endpoints among the collinear segments. By adding both endpoints $u$ and $v$, we obtain an induced $K_{1,6}$ since $u$ does not see other points than $v$ due to distance and $v$ sees all endpoints. In Figure~\ref{fig:2d}, we also have a similar configuration to Figure~\ref{fig:2b}. The vertices of a polygon do not see each other if the straight line joining them is outside of the polygon. Therefore, among the vertices except $u$, we can choose at most six pairwise non-adjacent points to be in our induced subgraph. With the addition of $u$, we obtain an induced $K_{1,6}$.
\end{proof}

The idea used to prove Lemma~\ref{lem:K1,6} is that unit disk graphs cannot contain an induced $K_{1,6}$ \cite{Marathe_heuristics} while unit disk point, segment and polygon visibility graphs can contain it as in Figure~\ref{fig:notUD}.

\begin{lem} \label{lem:UnitUDVG}
	Unit disk graphs are a proper subclass of unit disk point visibility graphs, and \emph{not} a subclass of unit disk segment and polygon (simple or with holes) visibility graphs.
\end{lem}
\begin{proof}
	Given a representation of unit disk graphs, we can simply perturb the disk centers slightly to obtain a set of points in general position, creating a configuration in which no three points are collinear \cite{Fonesca_recognition}.
	This way, we have had obtained a setting that is exactly a unit disk graph, in which a pair of vertices are adjacent only if they are close enough. 
	It is still a question whether the new positions of the disk centers can be represented by using polynomially many decimal digits with respect to the input size. However, this shows that every unit disk graph can be represented as the unit disk visibility graph of the new set of points. Lemma~\ref{lem:UnitUDVG} shows that there are unit disk visibility graphs of a set of points which cannot be recognized as unit disk graphs. Therefore, unit disk graphs are a ``proper'' subclass of unit disk visibility graphs for a set of points.
	
	By definition, a set of $n$ (disjoint) line segments contains $2n$ endpoints. Therefore, unit disk visibility graphs for a set of segments have even number of vertices while there unit disk graphs can have odd number of vertices. Combined with Lemma~\ref{lem:UnitUDVG}, unit disk graphs are neither a subclass nor a superclass of unit disk visibility graphs of a set of line segments.
	
	Unit disk visibility graphs for a simple polygon contain a Hamiltonian cycle which acts as the border of the given polygon. However, unit disk graphs may not contain an Hamiltonian cycle. Combined with Lemma~\ref{lem:UnitUDVG}, unit disk graphs are neither a subclass nor a superclass of unit disk visibility graphs for a simple polygon.
	
	Unit disk visibility graphs for a polygon with holes contain one cycle $C$ to act as the exterior border of the given polygon and at least one other cycle $C'$ disjoint from $C$ to act as the border of a hole. By definition, a cycle of a simple graph is of length at least three and this means that unit disk visibility graphs for a polygon with holes have at least $6$ vertices. Therefore, unit disk graphs of order $< 6$ are not unit disk visibility graphs for a polygon with holes. We now show that unit disk graphs of order $\geq 6$ are not unit disk visibility graphs for a polygon with holes using the following fact: \emph{The cycle $C'$ corresponding to the border of a hole is an induced ``chordless'' cycle of the graph.} However, unit disk graphs may not contain any induced cycle of length $\geq 3$. Combined with Lemma~\ref{lem:UnitUDVG}, unit disk graphs are neither a subclass nor a superclass of unit disk visibility graphs for a polygon with holes.	
\end{proof}

\section{3-colorability of unit disk segment visibility graphs} \label{sec:3coloringSegment}

In this section, we mention our NP-hardness reductions. A polynomial-time (NP-hardness) reduction from a (NP-hard) problem $\mathfrak{Q}$ to another problem $\mathfrak{P}$ is to map any instance $\Phi$ of $\mathfrak{Q}$ to some instance $\Psi$ of $\mathfrak{P}$ such that $\Phi$ is a \textsc{YES}-instance of $\mathfrak{Q}$ if and only if $\Psi$ is a \textsc{YES}-instance of $\mathfrak{P}$, in polynomial-time and polynomial-space.
We first show that the 3-coloring problem for unit disk segment visibility graphs is NP-hard, using a reduction from the \emph{Monotone not-all-equal 3-satisfiability} (Monotone NAE3SAT) problem which is a variation of 3SAT \cite{Schaefer_complexitySAT} with no negated variables, and to satisfy the circuit, at least one variable must be true, and at least one variable must be false in each clause. 

\begin{thm} \label{thm:main}
	There is a polynomial-time reduction of any instance $\Phi$ of Monotone NAE3SAT to some instance $\Psi$ of unit disk segment visibility graphs such that $\Phi$ is a YES-instance if and only if $\Psi$ is 3-colorable.
\end{thm}

Before proving Theorem~\ref{thm:main}, we describe the gadgets to be used to construct a unit disk segment visibility graph from a given NAE3SAT formula. From now on, for the sake of simplicity, when we write  ``color of an endpoint,'' we mean the color of the vertex which corresponds to that particular endpoint.

\subsection{The long edges and NAE3SAT clauses}\label{sec:longedges}

First, let us describe how we model the long edges to transfer the colors from one vertex to the others.
In Figure~\ref{fig:longEdge}, there are six line segments on two rows, drawn with thick lines.
Their endpoints, which correspond to the vertices of the unit disk visibility graph, are shown by three different shapes: square, triangle, and circle.
The disks drawn around the endpoints are the unit disks.
The dashed red lines between the endpoints are the visibility edges.

By definition, if two unit disks do not intersect, then the corresponding endpoints are not mutually visible.
Thus, even though there is no obstacle between some endpoints, they do not share a visibility edge since they are farther than one unit.
Note that if two endpoints see each other, then they are of different shape.
Therefore, assigning a unique color for each unique shape gives us a proper 3-coloring in the unit disk visibility graph of this particular configuration.

\begin{figure}[htbp]
	\centering
	\begin{tikzpicture} [scale = 1.1,
		triangle/.style = {regular polygon, regular polygon sides=3, scale=1.2},
		square/.style = {regular polygon, regular polygon sides=4, scale=1.2}]
		
		\tikzstyle{every node}=[draw=black, fill=gray, shape=circle, minimum size=6pt,inner sep=0pt];
		
		\foreach \i in {1,2,3,4,5,6}
		{	
			\draw (\i,0) circle (0.55cm);
			\draw[gray] (\i+0.2,0.2) circle (0.55cm);
		}
		
		\node[square] (1) at (1,0) {};
		\node[triangle] (2) at (2,0) {};
		\node (3) at (3,0) {};
		\node[square] (4) at (4,0) {};
		\node[triangle] (5) at (5,0) {};
		\node (6) at (6,0) {};
		
		\node (7) at (1.2,0.2) {};
		\node[square] (8) at (2.2,0.2) {};
		\node[triangle] (9) at (3.2,0.2) {};
		\node (10) at (4.2,0.2) {};
		\node[square] (11) at (5.2,0.2) {};
		\node[triangle] (12) at (6.2,0.2) {};
		
		\foreach \i in {1,3,5}
		{
			\pgfmathtruncatemacro\j{\i+1};
			\pgfmathtruncatemacro\k{\i+6};
			\pgfmathtruncatemacro\m{\i+7};
			\draw[ultra thick] (\i) -- (\j);
			\draw[ultra thick] (\k) -- (\m);
		}
		
		\foreach \i in {2,3,4,5,6}
		{
			\pgfmathtruncatemacro\j{\i+5};
			\pgfmathtruncatemacro\k{\i+6};
			\draw[red, dashed] (\i) -- (\j);
			\draw (\i) -- (\k);
			
		}
		
		\foreach \i in {2,4,8,10}
		{
			\pgfmathtruncatemacro\j{\i+1};
			\draw[red, dashed] (\i) -- (\j);	
			
		}
		
		\draw (1) -- (7);
		
	\end{tikzpicture}
	\caption{A long edge gadget constructed by six line segments having a unique 3-coloring.}
	\label{fig:longEdge}	
\end{figure}
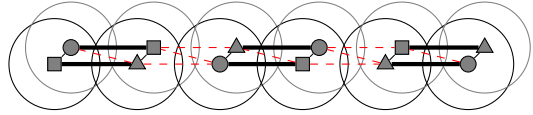

\begin{rem}
	In a proper 3-coloring of the unit disk visibility graph of the line segment configuration given in Figure~\ref{fig:longEdge}, a pair of vertices receive the same color if and only if their corresponding endpoints are of same shape.
\end{rem}

Observe that on each row, the shapes of the endpoints are repeating sequentially.
Given such a configuration with arbitrarily many line segments, if we consider the ordering $(1,2, \dots)$ of the endpoints from left to right on a single row, then the color of $i$th endpoint will have the same color with $(i+3)$th endpoint.
This helps us to transfer a color from one side of the configuration to the desired position, enabling the usage of long edges in the circuit.

\begin{figure}[htbp]
	\centering
	\begin{tikzpicture}[yscale=0.5, xscale=1.7,
		triangle/.style = {regular polygon, regular polygon sides=3, scale=1.2},
		square/.style = {regular polygon, regular polygon sides=4, scale=1.2}]
		\tikzstyle{every node}=[draw, fill=black, minimum size=4pt,inner sep=0pt];
		
		\node[color=black, draw=none, fill=none] at (0,4) {$q$};
		\node[color=black, draw=none, fill=none] at (0,3) {$r$};
		\node[color=black, draw=none, fill=none] at (0,2) {$s$};
		\node[color=black, draw=none, fill=none] at (0,1) {$t$};
		
		\foreach \i in {1,2,3,4}
		{
			\draw (0.2,\i)--(6,\i);	
			\foreach \j in {0.5, 1.1, 1.7, 2.3, 2.9, 3.5, 4.1, 4.7, 5.3}
			{
				\node[color=gray,fill=gray!50, draw=gray, square] at (\j,\i) {};
				\node[color=gray,fill=gray!50, triangle] at (\j + 0.2, \i) {};
				\node[color=gray,fill=gray!50, shape=circle] at (\j + 0.4, \i) {};
			}	
		}
		
		\clause{0.5}{4}{2}{1}
		\clause{2.3}{4}{3}{1}
		\clause{4.1}{3}{2}{1}
		
		\tikzstyle{every node}=[draw, shape=rectangle, color=black, fill=none, minimum size=5pt, minimum width=6em, inner sep=2pt];
		
		\node at (1.1,-0.4) {$A$};
		\node at (2.9,-0.4) {$B$};
		\node at (4.7,-0.4) {$C$};
	\end{tikzpicture}
	\caption{An NAE3SAT formula with variables $q,r,s,t$, and clauses $A = (q \vee s \vee t)$, $B = (q \vee r \vee t)$, and $C =(r \vee s \vee t)$.}
	\label{fig:3satCircuit}
\end{figure}
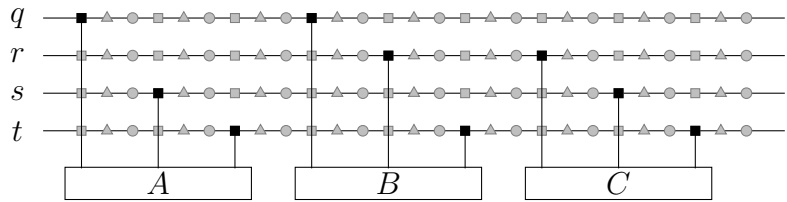

Consider the circuit given in Figure~\ref{fig:3satCircuit}.
There are four boolean variables, $q$, $r$, $s$, $t$, and three clauses, $A$, $B$, $C$.
Each variable has a long edge, representing its ``wire.''
On each row, three shapes are repeating, square, triangle, and circle, in this precise order.
These three shapes are in the same order with the configuration given in Figure~\ref{fig:longEdge}.
If a variable appears in a clause, then it is shown by an edge from the corresponding row to the clause.
Thus, when a color is assigned to a square endpoint, the same color repeats along the wire,
and eventually is transferred to the clause.

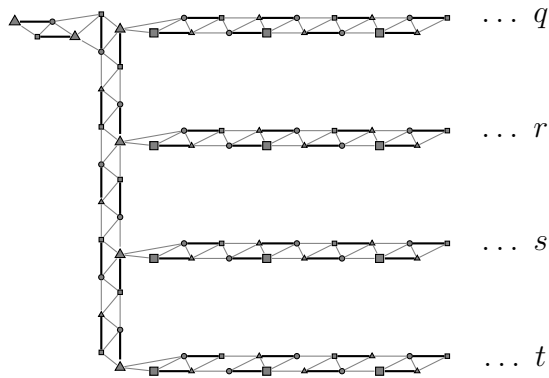
\begin{figure}[htbp]
	\centering
	\begin{tikzpicture} [scale=0.5,
		triangle/.style = {regular polygon, regular polygon sides=3, scale=1.2},
		square/.style = {regular polygon, regular polygon sides=4, scale=1.2}]
		
		\node[draw=none, fill=none] at (11,-0.9) {$\dots\ q$};
		\node[draw=none, fill=none] at (11,-3.9) {$\dots\ r$};
		\node[draw=none, fill=none] at (11,-6.9) {$\dots\ s$};
		\node[draw=none, fill=none] at (11,-9.9) {$\dots\ t$};
		
		\tikzstyle{every node}=[draw=black, fill=gray, shape=circle, minimum size=2pt,inner sep=0pt];
		\tikzstyle{every path}=[draw, gray];

		\node[triangle,scale=2] (A) at (-2.3,-1) {};
		\node (B) at (-1.3,-1) {};
		\node[square] (C) at (-1.7,-1.4) {};
		\node[triangle,scale=1.8] (D) at (-0.7, -1.4) {};
		\draw[black, thick] (A)--(B);
		\draw (A)--(C);
		\draw (B)--(C);
		\draw[black, thick] (C)--(D);

		\foreach [evaluate={\j=int(mod(\i,3));}] \i in {1,...,10}
		{
			\ifthenelse{\j = 1}
			{
				\node[square] (\i) at (0, -\i+0.2) {};
				\node[triangle,scale=1.8] (\i') at (0.5, -\i-0.2) {};
			}
			{}
			
			\ifthenelse{\j = 2}
			{
				\node (\i) at (0, -\i+0.2) {};
				\node[square] (\i') at (0.5, -\i-0.2) {};
			}
			{}
			
			\ifthenelse{\j = 0}
			{
				\node[triangle] (\i) at (0, -\i+0.2) {};
				\node (\i') at (0.5, -\i-0.2) {};
			}
			{}

		}
		
		\draw (1)--(B)--(D)--(1);
		\draw (2)--(D);

		\foreach \i in {1,3,5,7,9}
		{
			\pgfmathtruncatemacro\j{\i+1};
			\draw[black, thick] (\i)--(\j);
			\draw (\i)--(\i');
			\draw[black, thick] (\i')--(\j');
			\draw (\i')--(\j);
			\draw (\j)--(\j');	
		}
		
		\foreach \i in {2,4,6,8}
		{
			\pgfmathtruncatemacro\j{\i+1};
			\draw (\i)--(\j);
			\draw (\i')--(\j');
			\draw (\i')--(\j);
		}

		\foreach [evaluate={\j=int(mod(\i,3));}] \i in {1,...,8}
		{
			\ifthenelse{\j = 1}
			{
				\node (qu\i) at (\i+1.2,-0.9) {};
				\node (ru\i) at (\i+1.2,-3.9) {};
				\node (su\i) at (\i+1.2,-6.9) {};
				\node (tu\i) at (\i+1.2,-9.9) {};
				
				\node[square,scale=1.8] (qd\i) at (\i+0.4,-1.3) {};
				\node[square,scale=1.8] (rd\i) at (\i+0.4,-4.3) {};
				\node[square,scale=1.8] (sd\i) at (\i+0.4,-7.3) {};	
				\node[square,scale=1.8] (td\i) at (\i+0.4,-10.3) {};
			}
			{}
			\ifthenelse{\j = 2}
			{
				\node[square] (qu\i) at (\i+1.2,-0.9) {};
				\node[square] (ru\i) at (\i+1.2,-3.9) {};
				\node[square] (su\i) at (\i+1.2,-6.9) {};
				\node[square] (tu\i) at (\i+1.2,-9.9) {};
				
				\node[triangle] (qd\i) at (\i+0.4,-1.3) {};
				\node[triangle] (rd\i) at (\i+0.4,-4.3) {};				
				\node[triangle] (sd\i) at (\i+0.4,-7.3) {};		
				\node[triangle] (td\i) at (\i+0.4,-10.3) {};
			}
			{}
			
			\ifthenelse{\j = 0}
			{
				\node[triangle] (qu\i) at (\i+1.2,-0.9) {};
				\node[triangle] (ru\i) at (\i+1.2,-3.9) {};
				\node[triangle] (su\i) at (\i+1.2,-6.9) {};
				\node[triangle] (tu\i) at (\i+1.2,-9.9) {};

				\node (qd\i) at (\i+0.4,-1.3) {};
				\node (rd\i) at (\i+0.4,-4.3) {};
				\node (sd\i) at (\i+0.4,-7.3) {};		
				\node (td\i) at (\i+0.4,-10.3) {};
			}
			{}

		}
		
		\foreach \i in {1,3,5,7}
		{
			\pgfmathtruncatemacro\j{\i+1};
			
			\draw[black, thick] (qu\i)--(qu\j);
			\draw[black, thick] (qd\i)--(qd\j);
			
			\draw[black, thick] (ru\i)--(ru\j);
			\draw[black, thick] (rd\i)--(rd\j);
			
			\draw[black, thick] (su\i)--(su\j);
			\draw[black, thick] (sd\i)--(sd\j);
			
			\draw[black, thick] (tu\i)--(tu\j);
			\draw[black, thick] (td\i)--(td\j);
			
			\draw (qu\i)--(qd\i);
			\draw (ru\i)--(rd\i);
			\draw (su\i)--(sd\i);
			\draw (tu\i)--(td\i);
			
			\draw (qu\j)--(qd\j);
			\draw (ru\j)--(rd\j);
			\draw (su\j)--(sd\j);
			\draw (tu\j)--(td\j);
			
			\draw (qu\i)--(qd\j);
			\draw (ru\i)--(rd\j);
			\draw (su\i)--(sd\j);
			\draw (tu\i)--(td\j);
			
		}

		\foreach \i in {2,4,6}
		{
			\pgfmathtruncatemacro\j{\i+1};
			
			\draw (qu\i)--(qu\j);
			\draw (qd\i)--(qd\j);
			
			\draw (ru\i)--(ru\j);
			\draw (rd\i)--(rd\j);
			
			\draw (su\i)--(su\j);
			\draw (sd\i)--(sd\j);
			
			\draw (tu\i)--(tu\j);
			\draw (td\i)--(td\j);

			\draw (qu\i)--(qd\j);
			\draw (ru\i)--(rd\j);
			\draw (su\i)--(sd\j);
			
		}
		
		\draw (qu1)--(1')--(qd1);
		\draw (ru1)--(4')--(rd1);
		\draw (su1)--(7')--(sd1);
		\draw (tu1)--(10')--(td1);

	\end{tikzpicture}
	\caption{The wires transfering one of two colors.}
	\label{fig:initial}
\end{figure}

The truth assignments of the variables are determined by a pair of colors.
In this case, every square endpoint on a wire should receive one of these two colors.
We guarantee this by setting up the gadget shown in Figure~\ref{fig:initial}.
The color of the triangle vertex is the ``neutral'' color, which means that the remaining two colors represents ``true'' and ``false'' for the variables.
Thus, by picking a color for the triangle, we enforce every variable to transfer either true or false.
Important part is that a variable connects to a clause only by one of the squares.
Therefore, the truth assignment of a variable is transferred to the clause by the color of a square endpoint.
Let us assume that the triangles in Figure~\ref{fig:initial} are colored green.
In Figure~\ref{fig:transferringSegment}, we see an embedding of line segments, in which the variables transfer either blue, or red to the clauses below.

\begin{figure}[htbp]
	\centering
	\begin{tikzpicture} [scale=0.8,
		triangle/.style = {regular polygon, regular polygon sides=3, scale=0.5},
		square/.style = {regular polygon, regular polygon sides=4, scale=0.5},
		circle/.style = {shape=circle, scale=0.4}]
		
		\tikzstyle{every node}=[draw=black, fill=gray, shape=circle, minimum size=6pt,inner sep=0pt];

		\node[square] (1) at (1,0) {};
		\node[triangle] (2) at (2,0) {};
		\node[circle] (3) at (3,0) {};
		\node[square] (4) at (4,0) {};
		\node[triangle] (5) at (5,0) {};
		\node[circle] (6) at (6,0) {};
		
		\node[circle] (7) at (1.2,0.2) {};
		\node[square] (8) at (2.2,0.2) {};
		\node[triangle] (9) at (3.2,0.2) {};
		\node[circle] (10) at (4.2,0.2) {};
		\node[square] (11) at (5.2,0.2) {};
		\node[triangle] (12) at (6.2,0.2) {};
		
		\foreach \i in {1,3,5}
		{
			\pgfmathtruncatemacro\j{\i+1};
			\pgfmathtruncatemacro\k{\i+6};
			\pgfmathtruncatemacro\m{\i+7};
			\draw[thick] (\i) -- (\j);
			\draw[thick] (\k) -- (\m);
		}
		
		\foreach \i in {2,3,4,5,6}
		{
			\pgfmathtruncatemacro\j{\i+5};
			\pgfmathtruncatemacro\k{\i+6};
			\draw(\i) -- (\j);
			\draw (\i) -- (\k);
			
		}
		
		\foreach \i in {2,4,8,10}
		{
			\pgfmathtruncatemacro\j{\i+1};
			\draw (\i) -- (\j);	
			
		}
		
		\draw (1) -- (7);
		
		\node[triangle] (13) at (3.3, -0.2) {};
		\node[circle] (14) at (3.6, -0.525) {};
		
		\node[square] (15) at (3.6, -1.5) {};
		\node[triangle] (16) at (3.6, -2.5) {};
		
		\node[circle] (17) at (3.6, -3.5) {};
		\node[square, scale=2, fill=black] (18) at (3.6, -4.5) {};
		
		\node[triangle] (19) at (3.8,-0.9) {};
		\node[circle] (20) at (3.8,-1.9) {};
		
		\node[square] (21) at (3.8,-2.9) {};
		\node[triangle] (22) at (3.8,-3.9) {};
		
		\foreach \i in {1,3,5}
		{
			\pgfmathtruncatemacro\j{\i+1};
			\pgfmathtruncatemacro\k{\i+6};
			\pgfmathtruncatemacro\m{\i+7};
			\draw[thick] (\i) -- (\j);
			\draw[thick] (\k) -- (\m);
		}
		
		\foreach \i in {2,4,8,10}
		{
			\pgfmathtruncatemacro\j{\i+1};
			\draw (\i) -- (\j);
		}
		
		\foreach \i in {13,15,17,19,21}
		{
			\pgfmathtruncatemacro\j{\i+1};
			\draw[thick] (\i) -- (\j);	
		}
		
		\draw (1)--(7)--(2)--(8)--(3)--(9)--(4)--(10)--(5)--(11)--(6);
		\draw (3)--(13);
		\draw (4)--(19);
		\draw (13)--(4)--(14)--(15);
		\draw (14)--(19)--(15)--(20)--(16)--(21)--(17)--(22)--(18);
		\draw (16)--(17);
		\draw (20)--(21);
	\end{tikzpicture}
	\caption{Setting up the long edge gadget to transfer the color of square endpoints given in Figure~\ref{fig:3satCircuit} and Figure~\ref{fig:initial}.}
	\label{fig:transferringSegment}
\end{figure}
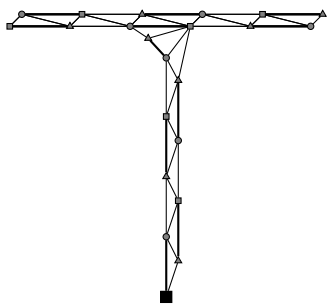

Three components given in Figure~\ref{fig:3satCircuit}, Figure~\ref{fig:initial}, and Figure~\ref{fig:transferringSegment}
guarantee that two colors representing true and false are transferred properly to the NAE3SAT clauses.
In a NAE3SAT instance, a clause should be not-all-equal i.e., in a clause at least one variable must be false, and at least one variable must be true.
Now, let us show that if a clause receives three variables with the same truth assignment, then the gadget is not 3-colorable.

Regardless of the number of variables, and number of clauses, the long vertical edges in Figure~\ref{fig:3satCircuit} never intersect.
This is because a variable transfers its color from different endpoints on the corresponding row, for every clause it appears in.
In Figure~\ref{fig:3satCircuit}, we see that the variables are connected to clauses via vertical edges these are close to each other.
In other words, every clause has its own area, and no two vertical edges intersect.
This allows us to design a gadget that is able to connect three variables, even if they are far apart.

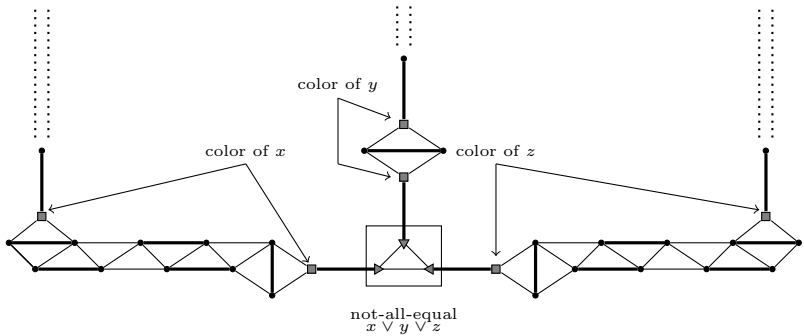
\begin{figure}[htbp]
	\centering
	\begin{tikzpicture}[scale=1.75, 
		square/.style = {regular polygon, regular polygon sides=4, fill=gray, scale=2.2},
		triangle/.style = {regular polygon, regular polygon sides=3, fill=gray, scale=2.2}
		]
		
		%\draw (0,0) to[grid with coordinates] (6,2);
		\tikzstyle{every node}=[draw=black, fill=black, shape=circle, minimum size=2pt,inner sep=0pt];
		
		\foreach \i in {0,1,2,3}
		{	
			\node (\i1) at (\i*0.5,0.4) {}; 
			\node (\i0) at (\i*0.5 + 0.2,0.2) {};
		}

		\node (A) at (2,0.4) {};
		\node (B)  at (2,0) {};
		\node[square] (C) at (2.3,0.2) {};
		\node (D) at (2.8,0.2) {};
		\node (E) at (3,0.4) {};
		\node[square] (F) at (3,0.9) {};
		\node (G) at (3.2,0.2) {};
		\node[square] (H) at (3.7,0.2) {};
		\node (I) at (4,0.4) {};
		\node (J)  at (4,0) {};
		
		\foreach \i in {8,9,10,11}
		{	
			\node (\i1) at (\i*0.5+ 0.5,0.4) {}; 
			\node (\i0) at (\i*0.5 + 0.3,0.2) {};
		}
		
		\foreach \i in {0,2,8,10}
		{
			\pgfmathtruncatemacro\j{\i+1};
			\draw[very thick] (\i0)--(\j0);
			\draw[very thick] (\i1)--(\j1);
		}
		
		\draw[very thick] (A)--(B);		
		\draw[very thick] (C)--(D);
		\draw[very thick] (E)--(F);
		\draw[very thick] (G)--(H);
		\draw[very thick] (I)--(J);

		\node[square] (P) at (0.25, 0.6) {};
		\node (Q) at (0.25, 1.1) {};
		\node (U) at (2.7, 1.1) {};
		\node (V) at (3.3, 1.1) {};
		\node[square] (R) at (3,1.3) {};
		\node (S) at (3,1.8) {};
		\node[square] (X) at (5.75, 0.6) {};
		\node (Y) at (5.75, 1.1) {};
		
		\draw[very thick] (P)--(Q);
		\draw[very thick] (U)--(V);
		\draw[very thick] (R)--(S);
		\draw[very thick] (X)--(Y);	
		
		\node[fill=none, shape=rectangle, minimum width=1cm, minimum height=0.8cm] at (3,0.3) {};

		\foreach \i in {1,9}
		{	
			\pgfmathtruncatemacro\j{\i+1};
			\pgfmathtruncatemacro\k{\i-1};
			\draw (\i0)--(\i1);
			\draw (\j0)--(\j1);
			\draw (\k0)--(\k1);
			\draw (\i0)--(\j0);
			\draw (\i1)--(\j1);
		}
		
		\draw (00)--(01);
		\draw (00)--(11);
		\draw (00)--(01);
		\draw (10)--(21);
		\draw (20)--(31);
		\draw (30)--(31);
		
		\draw (81)--(90);
		\draw (91)--(100);
		\draw (101)--(110);
		\draw (111)--(110);
		
		\draw (01)--(P)--(11);
		\draw (B)--(30)--(A)--(31);
		\draw (A)--(C)--(B);
		\draw (U)--(R)--(V);
		\draw (U)--(F)--(V);
		\draw (J)--(H)--(I);
		
		\draw (J)--(80)--(I)--(81);
		\draw (101)--(X)--(111);

		\draw (D)--(E)--(G)--(D);
		\node[triangle, rotate=-90] at (D.center) {};
		\node[triangle, rotate=180] at (E.center) {};
		\node[triangle,rotate=90] at (G.center) {};

		\draw[thick, dotted] (0.2,1.2)--(0.2,2.2);
		\draw[thick, dotted] (0.3,1.2)--(0.3,2.2);
		\draw[thick, dotted] (2.95,1.9)--(2.95,2.2);
		\draw[thick, dotted] (3.05,1.9)--(3.05,2.2);
		\draw[thick, dotted] (5.7,1.2)--(5.7,2.2);
		\draw[thick, dotted] (5.8,1.2)--(5.8,2.2);
		
		\node[draw=none, fill=none] at (1.8,1.1) {\tiny color of $x$};
		\node[draw=none, fill=none] at (2.5,1.6) {\tiny color of $y$};
		\node[draw=none, fill=none] at (3.7,1.1) {\tiny color of $z$};
		\node[draw=none, fill=none] at (3,-0.15) {\tiny not-all-equal};
		\node[draw=none, fill=none] at (3,-0.25) {\tiny $x \vee y \vee z$};
		\draw [->] (1.8,1)--(2.275,0.25);
		\draw [->] (1.8,1)--(0.3,0.65);
		
		\draw [->] (3.7,1)--(3.7,0.3);
		\draw [->] (3.7,1)--(5.7,0.65);
		
		\draw [->] (2.5,1.5)--(2.9,1.35);
		\draw (2.5,1.5)--(2.5,1);
		\draw[->] (2.5,1)--(2.9,0.9);
	\end{tikzpicture}
	\caption{A NAE3SAT clause $x \vee y \vee z$. Assuming that the square-shaped nodes can receive only one of two colors (either the color that corresponds to \textsc{true}, or the color that corresponds to \textsc{false}); the three triangle-shaped nodes inside by the rectangular frame can be properly colored using three colors if and only if at least one of three square-shaped nodes have a different color than the other two.}
	\label{fig:naeclause}
\end{figure}

Up to this point, we showed how we model the long edges to transfer the colors.
However, given an instance of the Monotone NAE3SAT problem, there is no guarantee that the circuit can be drawn without edge crossings.
Note that the circuit given in Figure~\ref{fig:3satCircuit}, the vertical edges passes through several long edges.
Since our embedding is in 2D, we need a gadget which has no edge crossings, and also can be used to transfer the colors properly when two edges intersect in the circuit.

\subsection{Edge crossings}\label{sec:edgecrossings}

We now describe a certificate that transfers the colors safely in case of edge crossings.
The certificate that we use is similar to the certificate given by Gr{\"a}f et al. for the reduction of unit disk graph coloring problem (see Figure 4 in \cite{Graf_udgColoring}).

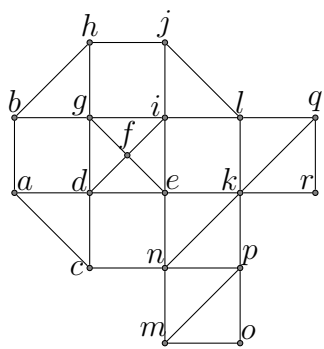
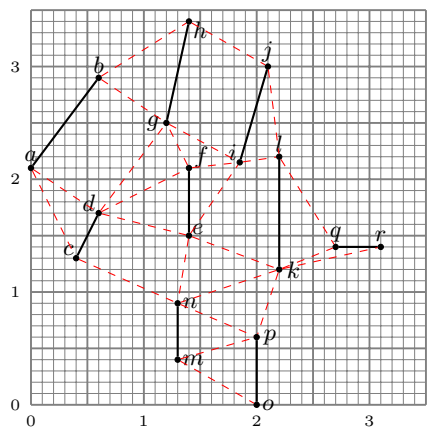
\begin{figure}[htbp]
	\centering
	\subfloat[The edge crossing gadget transferring the color from $a$ to $r$, and from $h$ to $o$.]{
		\centering
		\begin{tikzpicture}[scale=0.5]
			\tikzstyle{every node}=[draw, fill=gray, shape=circle, minimum size=2pt,inner sep=0pt]
			\node[label=35:$a$] (A) at (0,2) {}; %a
			\node[label=$b$] (B) at (0,4) {}; %b
			\node[label=left:$c$] (C) at (2,0) {}; %c
			\node[label=100:$d$] (D) at (2,2) {}; %d
			\node[label=80:$e$] (E) at (4,2) {}; %e
			\node[label=$f$] (F) at (3,3) {}; %f
			\node[label=100:$g$] (G) at (2,4) {}; %g
			\node[label=$h$] (H) at (2,6) {}; %h
			\node[label=100:$i$] (I) at (4,4) {}; %i
			\node[label=$j$] (J) at (4,6) {}; %j
			\node[label=100:$k$] (K) at (6,2) {}; %k
			\node[label=$l$] (L) at (6,4) {}; %l
			\node[label=100:$m$] (M) at (4,-2) {}; %l
			\node[label=100:$n$] (N) at (4,0) {}; %l
			\node[label=80:$o$] (O) at (6,-2) {}; %l
			\node[label=80:$p$] (P) at (6,0) {}; %l
			\node[label=$q$] (Q) at (8,4) {}; %l
			\node[label=100:$r$] (R) at (8,2) {}; %l
			
			\draw (A)--(B);
			\draw (A)--(C);
			\draw (A)--(D);
			
			\draw (B)--(G);
			\draw (B)--(H);
			
			\draw (C)--(D);
			\draw (C)--(N);
			
			\draw (D)--(E);
			\draw (D)--(F);
			\draw (D)--(G);
			
			\draw (E)--(F);
			\draw (E)--(K);
			\draw (E)--(N);
			\draw (E)--(I);
			
			\draw (F)--(G);
			\draw (F)--(I);
			
			\draw (G)--(H);
			\draw (G)--(I);
			
			\draw (H)--(J);
			
			\draw (I)--(J);
			\draw (I)--(L);
			
			\draw (J)--(L);
			
			\draw (K)--(L);
			\draw (K)--(N);
			\draw (K)--(P);
			\draw (K)--(Q);
			\draw (K)--(R);
			
			\draw (L)--(Q);
			
			\draw (M)--(P);
			\draw (M)--(O);
			\draw (M)--(N);
			
			\draw (N)--(P);
			
			\draw (O)--(P);
			
			\draw (Q)--(R);
			
		\end{tikzpicture}
		\label{fig:crossingSegments}
	}
	~\hfill~
	\subfloat[The segment embedding of (\textsc{e}) on a grid. Only endpoint that is not on the grid is $i$, which is at $(1.85,2.15)$.]{
		\centering
		\hspace{-0.5em}
		\begin{tikzpicture}[scale=1.5]
			\draw (0,0) to[grid with coordinates] (3.5,3.5);
			
			\tikzstyle{every node}=[draw, fill=black, shape=circle, minimum size=2pt,inner sep=0pt];
			
			\node[label={\footnotesize $a$}] (A) at (0,2.1) {}; % A
			\node[label={\footnotesize $b$}] (B) at (0.6,2.9) {}; % B
			
			\node[label=135:{\footnotesize $c$}] (C) at (0.4,1.3) {}; % C
			\node[label=135:{\footnotesize $d$}] (D) at (0.6,1.7) {}; % D
			
			\node[label=20:{\footnotesize $e$}] (E) at (1.4,1.5) {}; % E
			\node[label=20:{\footnotesize $f$}] (F) at (1.4,2.1) {}; % F
			
			\node[label=left:{\footnotesize $g$}] (G) at (1.2,2.5) {}; % G
			\node[label=-10:{\footnotesize $h$}] (H) at (1.4,3.4) {}; % H
			
			\node[label=100:{\footnotesize $i$}] (I) at (1.85,2.15) {}; % I
			\node[label=above:{\footnotesize $j$}] (J) at (2.1,3) {}; % J
			
			\node[label=0:{\footnotesize $k$}] (K) at (2.2,1.2) {}; % K
			\node[label={\footnotesize $l$}] (L) at (2.2,2.2) {}; % L
			
			\node[label=0:{\footnotesize $m$}] (M) at (1.3,0.4) {}; % M
			\node[label=0:{\footnotesize $n$}] (N) at (1.3,0.9) {}; % N
			
			\node[label=0:{\footnotesize $o$}] (O) at (2,0) {}; % M
			\node[label=0:{\footnotesize $p$}] (P) at (2,0.6) {}; % N
			
			\node[label={\footnotesize $q$}] (Q) at (2.7,1.4) {}; % M
			\node[label={\footnotesize $r$}] (R) at (3.1,1.4) {}; % N

			\draw[thick] (A)--(B);
			\draw[thick] (C)--(D);
			\draw[thick] (E)--(F);
			\draw[thick] (G)--(H);
			\draw[thick] (I)--(J);
			\draw[thick] (K)--(L);
			\draw[thick] (M)--(N);
			\draw[thick] (O)--(P);
			\draw[thick] (Q)--(R);

			\tikzstyle{every path}=[dashed, draw, color=red];
			
			\draw (A)--(C);
			\draw (A)--(D);
			
			\draw (B)--(G);
			\draw (B)--(H);
			
			\draw (C)--(N);
			
			\draw (D)--(E);
			\draw (D)--(F);
			\draw (D)--(G);
			
			\draw (E)--(K);
			\draw (E)--(N);
			\draw (E)--(I);
			
			\draw (F)--(G);
			\draw (F)--(I);
			
			\draw (G)--(I);
			
			\draw (H)--(J);
			
			\draw (I)--(L);
			
			\draw (J)--(L);
			
			\draw (K)--(N);
			\draw (K)--(P);
			\draw (K)--(Q);
			\draw (K)--(R);
			
			\draw (L)--(Q);
			
			\draw (M)--(P);
			\draw (M)--(O);
			
			\draw (N)--(P);
			
		\end{tikzpicture}
		\label{fig:crossingGadgetSeg}
	}
	
	\caption{Gadget that replaces edge crossings of the 3SAT circuit that corresponds to a unit disk segment visibility graph.}
	\label{fig:edgeCrossingsSeg}
\end{figure}

In Figure~\ref{fig:crossingSegments} and Figure~\ref{fig:crossingGadgetSeg}, we give the gadget to replace the edge crossings in the 3SAT circuit, and its embedding as line segments, respectively. As we have mentioned previously, the colors repeat.

\begin{clm} \label{2uniqueCol}
	There are exactly two distinct (proper) 3-colorings of the edge crossing gadget gadget given in Figure~\ref{fig:crossingSegments}, all of which require the vertex pairs $a, r$ and $h, o$ to receive the same color.
\end{clm}

\begin{proof}
	Observe that there is a unique coloring of the induced subgraph $G'$ by the vertex subset $\{d, e, f, g, i\}$ up to a permutation of colors. Since $f$ is adjacent to all vertices in $G'$, it has to obtain a unique color $c_g$. Then, coloring any of $\{d, e, g, i\}$ uniquely determines the colors of the rest of these vertices where $e$ and $g$ get the same color $c_b \neq c_g$, and $d$ and $i$ get the same color $c_r \neq c_g, c_b$. We continue by choosing any of the triangles $\{e, k, n\}$, $\{i, j, l\}$, $\{a, c, d\}$  and $\{b, g, h\}$ whose one vertex is already colored by the coloring of $G'$. Assume that we choose $\{e, k, n\}$ where $e$ has the color $c_b$. There are two cases which may occur. 
	
	The first case is when $k$ is colored with $c_g$ and $n$ is colored with $c_r$.
	In this case, the color of every uncolored vertex is determined and as in Figure~\ref{fig:case1}. Note that this case is the same as the following cases: \emph{i)} choosing the triangle $\{i, j, l\}$ and coloring $j$ with $c_g$ and $l$ with $c_b$, \emph{ii)} choosing the triangle $\{a, c, d\}$ and coloring $a$ with $c_b$ and $c$ with $c_g$, and \emph{iii)} choosing the triangle $\{b, g, h\}$ and coloring $b$ with $c_g$ and $h$ with $c_r$.
	
	The second case is when $k$ is colored with $c_r$ and $n$ is colored with $c_g$.
	In this case, the color of every uncolored vertex is determined and as in Figure~\ref{fig:case2} (b). Note that this case is the same as the following cases: \emph{i)} choosing the triangle $\{i, j, l\}$ and coloring $j$ with $c_b$ and $l$ with $c_g$, \emph{ii)} choosing the triangle $\{a, c, d\}$ and coloring $a$ with $c_g$ and $c$ with $c_b$, and \emph{iii)} choosing the triangle $\{b, g, h\}$ and coloring $b$ with $c_r$ and $h$ with $c_g$.

	In both cases, the color of $a$ is the same as the color of $r$, and the color of $h$ is the same as the color of $o$. Therefore, our gadget transfers the color from $a$ to $r$ and the color from $h$ to $o$ correctly. Moreover, in case 1, $a$ and $h$ are assigned different colors while in case 2 they are assigned the same color. Therefore, when $a$ and $h$ are forced to be true or false and not neutral, our gadget corresponds to all possible truth assignments for these vertices. 
	In Figure~\ref{fig:twoCases}, these two unique colorings are demonstrated.
	
\end{proof}

\begin{figure}[h]
	\centering
	\captionsetup[subfigure]{position=b}
	\subfloat[Case 1 after the coloring of $G'$ is fixed.]{
		\centering
		\begin{tikzpicture}[scale=0.6,
			square/.style = {regular polygon, regular polygon sides=4, scale=4},
			circle/.style = {shape=circle, scale=1.8}]
			\tikzstyle{every node}=[draw, color=black, shape=circle, minimum size=3pt,inner sep=0pt]
			
			\node[label=35:$a$, square, fill=blue] (A) at (0,2) {}; %a
			\node[label=$b$, circle, fill=green] (B) at (0,4) {}; %b
			\node[label=left:$c$, circle, fill=green] (C) at (2,0) {}; %c
			\node[label=100:$d$, circle, fill=red] (D) at (2,2) {}; %d
			\node[label=80:$e$, circle, fill=blue] (E) at (4,2) {}; %e
			\node[label=$f$, circle, fill=green] (F) at (3,3) {}; %f
			\node[label=100:$g$, circle, fill=blue] (G) at (2,4) {}; %g
			\node[label=$h$, square, fill=red] (H) at (2,6) {}; %h
			\node[label=100:$i$, circle, fill=red] (I) at (4,4) {}; %i
			\node[label=$j$, circle, fill=green] (J) at (4,6) {}; %j
			\node[label=100:$k$, circle, fill=green] (K) at (6,2) {}; %k
			\node[label=$l$, circle, fill=blue] (L) at (6,4) {}; %l
			\node[label=100:$m$, circle, fill=green] (M) at (4,-2) {}; %l
			\node[label=100:$n$, circle, fill=red] (N) at (4,0) {}; %l
			\node[label=80:$o$, square, fill=red] (O) at (6,-2) {}; %l
			\node[label=80:$p$, circle, fill=blue] (P) at (6,0) {}; %l
			\node[label=$q$, circle, fill=red] (Q) at (8,4) {}; %l
			\node[label=100:$r$, square, fill=blue] (R) at (8,2) {}; %l
			
			\draw (A)--(B);
			\draw (A)--(C);
			\draw (A)--(D);
			
			\draw (B)--(G);
			\draw (B)--(H);
			
			\draw (C)--(D);
			\draw (C)--(N);
			
			\draw (D)--(E);
			\draw (D)--(F);
			\draw (D)--(G);
			
			\draw (E)--(F);
			\draw (E)--(K);
			\draw (E)--(N);
			\draw (E)--(I);
			
			\draw (F)--(G);
			\draw (F)--(I);
			
			\draw (G)--(H);
			\draw (G)--(I);
			
			\draw (H)--(J);
			
			\draw (I)--(J);
			\draw (I)--(L);
			
			\draw (J)--(L);
			
			\draw (K)--(L);
			\draw (K)--(N);
			\draw (K)--(P);
			\draw (K)--(Q);
			\draw (K)--(R);
			
			\draw (L)--(Q);
			
			\draw (M)--(P);
			\draw (M)--(O);
			\draw (M)--(N);
			
			\draw (N)--(P);
			
			\draw (O)--(P);
			
			\draw (Q)--(R);
			
			%\node (1) at (4,-3) {\textbf{(a)}};
			
		\end{tikzpicture}
		\label{fig:case1}
	}
	~\hfill~
	\subfloat[Case 2 after the coloring of $G'$ is fixed.]{
		\centering
		\begin{tikzpicture}[scale=0.6,
			square/.style = {regular polygon, regular polygon sides=4, scale=4},
			circle/.style = {shape=circle, scale=1.8}]
			\tikzstyle{every node}=[draw, color=black, shape=circle, minimum size=3pt,inner sep=0pt]
			
			\node[label=35:$a$, square, fill=green] (A) at (6,2) {}; %a
			\node[label=$b$, circle, fill=red] (B) at (6,4) {}; %b
			\node[label=left:$c$, circle, fill=blue] (C) at (8,0) {}; %c
			\node[label=100:$d$, circle, fill=red] (D) at (8,2) {}; %d
			\node[label=80:$e$, circle, fill=blue] (E) at (10,2) {}; %e
			\node[label=$f$, circle, fill=green] (F) at (9,3) {}; %f
			\node[label=100:$g$, circle, fill=blue] (G) at (8,4) {}; %g
			\node[label=$h$, square, fill=green] (H) at (8,6) {}; %h
			\node[label=100:$i$, circle, fill=red] (I) at (10,4) {}; %i
			\node[label=$j$, circle, fill=blue] (J) at (10,6) {}; %j
			\node[label=100:$k$, circle, fill=red] (K) at (12,2) {}; %k
			\node[label=$l$, circle, fill=green] (L) at (12,4) {}; %l
			\node[label=100:$m$, circle, fill=red] (M) at (10,-2) {}; %l
			\node[label=100:$n$, circle, fill=green] (N) at (10,0) {}; %l
			\node[label=80:$o$, square, fill=green] (O) at (12,-2) {}; %l
			\node[label=80:$p$, circle, fill=blue] (P) at (12,0) {}; %l
			\node[label=$q$, circle, fill=blue] (Q) at (14,4) {}; %l
			\node[label=100:$r$, square, fill=green] (R) at (14,2) {}; %l
			
			\draw (A)--(B);
			\draw (A)--(C);
			\draw (A)--(D);
			
			\draw (B)--(G);
			\draw (B)--(H);
			
			\draw (C)--(D);
			\draw (C)--(N);
			
			\draw (D)--(E);
			\draw (D)--(F);
			\draw (D)--(G);
			
			\draw (E)--(F);
			\draw (E)--(K);
			\draw (E)--(N);
			\draw (E)--(I);
			
			\draw (F)--(G);
			\draw (F)--(I);
			
			\draw (G)--(H);
			\draw (G)--(I);
			
			\draw (H)--(J);
			
			\draw (I)--(J);
			\draw (I)--(L);
			
			\draw (J)--(L);
			
			\draw (K)--(L);
			\draw (K)--(N);
			\draw (K)--(P);
			\draw (K)--(Q);
			\draw (K)--(R);
			
			\draw (L)--(Q);
			
			\draw (M)--(P);
			\draw (M)--(O);
			\draw (M)--(N);
			
			\draw (N)--(P);
			
			\draw (O)--(P);
			
			\draw (Q)--(R);
			
		\end{tikzpicture}
		\label{fig:case2}
	}
	\caption{Exactly two possible colorings of the edge crossing gadget (up to a permutation of colors).}
	\label{fig:twoCases}
\end{figure}
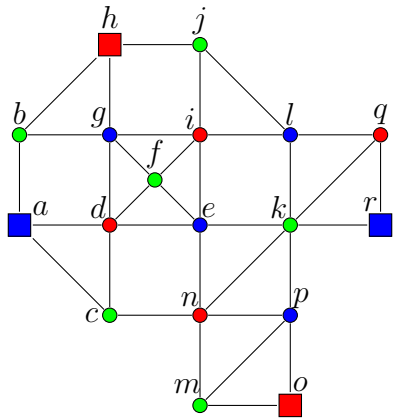
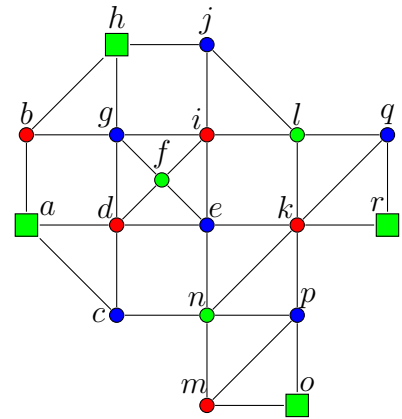

\subsection{The reduction of NP-hardness}
Given an instance of the Monotone NAE3SAT problem, we construct a circuit where every boolean variable is a wire placed horizontally in the plane, on a different row.
The clauses are placed on the bottom-most row. 
The truth assignment of the variables are transferred to the clauses via vertical wires, one end connected to the horizontal wire, other end connected to the clause (See Figure~\ref{fig:3satCircuit}).

We replace the wires in the described construction by a series of line segments such that the set of segments on a wire has a unique 3-coloring. Three main components of our reduction described in detail can be summarized as follows.

\begin{enumerate}
	\item \emph{A long edge} shown in Figure~\ref{fig:longEdge} is used to transfer a color from one end to the other (similar to transferring the truth assignment of a variable). The horizontal line segments are placed on two rows in the Euclidan plane. This configuration, no matter how long, always has a unique 3-coloring. Thus, by assigning a color to any endpoint, one automatically decides which color to appear in the clause gadget.
	
	\item \emph{A clause gadget} shown in Figure~\ref{fig:naeclause} is modeled by three line segments $xx'$, $yy'$, and $zz'$, not allowing three variables to have the same truth assignment. In our case, the transferred colors. 
	One endpoint from each segment, $x', y', z'$ yield a $K_3$.
	The remaining endpoints, $x, y, z$ do not see each other.
	The color of $x$ is transferred using long edges, and since $x'$ is the other endpoint, the color of $x'$ must be different.
	The same rule applies to $y$ and $z$ as well.
	Thus, if all $x, y, z$ have the same color, then this clause gadget cannot be 3-colored.
	
	\item \emph{An edge crossing gadget} shown in Figure~\ref{fig:crossingSegments} describes a certificate for an edge crossing in the circuit so that it can be realized as a set of non-intersecting line segments.		
\end{enumerate}

The truth assignments of the variables are determined by a pair of colors.
In this case, every square endpoint on a wire should receive one of these two colors.
We guarantee this by setting up the gadget shown in Figure~\ref{fig:initial}.
The color of the triangle vertex is the ``neutral'' color, which means that the remaining two colors represents ``true'' and ``false'' for the variables.
Thus, by picking a color for the triangle, we enforce every variable to transfer either true or false.
Important part is that a variable connects to a clause only by one of the squares.
Therefore, the truth assignment of a variable is transferred to the clause by the color of a square endpoint.

Given a Monotone NAE3SAT formula with $m$ clauses $C_1, \dots, C_m$ and $n$ variables $q_1, \dots, q_n$, we construct the corresponding unit disk segment visibility graph $G$ as follows:	
\begin{itemize}
	\item For each variable $q_i$, add a vertex $v_i$ to $G$ together with a long horizontal edge $H_i$ transferring its color.
	\item For each clause $C_i$ and each variable $q_j$ in $C_i$, add a triangle $T_i$ to $G$ together with a long vertical edge $V_j$ transferring the color of $v_j$. 
	\item For each $V_j$ crossing a $H_i$, add an edge crossing gadget (certificate) to $G$ replacing the vertices in the intersection $V_j \cup H_i$.	
\end{itemize}

Since we have shown that our gadgets interpret a given Monotone NAE3SAT formula and transfer colors correctly, the constructed unit disk segment visibility graph is 3-colorable if and only if the given Monotone NAE3SAT formula has a satisfying assignment.

\subsubsection{The time and space complexity.}
For $n$ variables and $m$ clauses, the number of segments on a vertical long ``wire'' is at most $O(m)$, since the colors repeat every three endpoints. 
Since there are $n$ variables, total number of segments for vertical edges is at most $O(nm)$.
For every edge crossing, and for every clause, a constant number of line segments is needed.
In the worst case, there will be $O(nm)$ edge crossings, hence $O(nm)$ line segments for a constant $c$.
There are $m$ clauses, and thus $O(m)$ edges are needed for the clauses.
In total, $O(m + nm + nm + m) = O(m + nm)$ segments are enough to model a NAE3SAT instance with $n$ variables and $m$ clauses. 

It is trivial to see that the configurations given in Figures~\ref{fig:longEdge} takes up polynomial space.
The edge crossing gadget given in Figure~\ref{fig:crossingSegments} has an embedding with polynomially many digits which can be verified by the coordinate system given in Figure~\ref{fig:crossingGadgetSeg}.
Notice that when a horizontal edge and a vertical edge cross, because of the embedding, two ends of the edge crossing gadget have slightly different $y$-coordinates for the endpoints of the horizontal segments ($a$ and $r$ in Figure~\ref{fig:crossingGadgetSeg}) and slightly different $x$-coordinates for the endpoints of the vertical segments ($h$ and $o$ in Figure~\ref{fig:crossingGadgetSeg}).
At first, it seems like the total space that the gadget uses will grow with respect to the number of edge crossings.
However, this is not the case since we can simply use a pair of diagonal segments to shift the position of the upcoming horizontal (resp. vertical) segments back to the initial $y$-coordinate (resp. $x$-coordinate). 
Because of the repeating color patterns,  the distance between a pair of adjacent variables can be adjusted such that the crossings have enough space to be embedded in.

As we proved the correctness of our reduction and showed that it is a polynomial-time reduction, the theorem holds. Since the Monotone NAE3SAT problem is NP-complete \cite{Schaefer_complexitySAT}, the 3-coloring problem for unit disk segment visibility graphs is also NP-complete.

Since the Monotone NAE3SAT problem is NP-complete \cite{Schaefer_complexitySAT}, the 3-coloring problem for unit disk segment visibility graphs is also NP-complete by Theorem~\ref{thm:main}.	

\begin{rem}
	Since the 3-coloring problem for unit disk graphs \cite{Graf_udgColoring} is NP-complete, it is also NP-complete for unit disk point visibility graphs by Lemma~\ref{lem:UnitUDVG}. For an alternative reduction to \cite{Graf_udgColoring}, the mentioned gadgets can be utilized with small modifications.
\end{rem}

\subsubsection{An example embedding}

We show an example embedding of a single clause $(q \vee s \vee t)$.
Note that the figures given throughout Section~\ref{sec:3coloringSegment} are to describe the idea behind the proof.
In an actual embedding, we might need some supplementary segments to transfer the colors properly.
In Figure~\ref{fig:exampleSegmentEmb}, there are some extra segments around the edge crossing gadgets.
These segments have no function other than transferring the last seen color.
This can also be done by altering the length of each segment.

In Figure~\ref{fig:crossingZoom}, we show a zoom-in view of a crossing which appear in the example embedding given in Figure~\ref{fig:exampleSegmentEmb}.
The shaded area is the edge crossing gadget described in Figure~\ref{fig:crossingSegments}.
The bold lines denote the segments, and the thin, red lines denote how a color from a long edge is transferred to and from the edge crossing gadget.

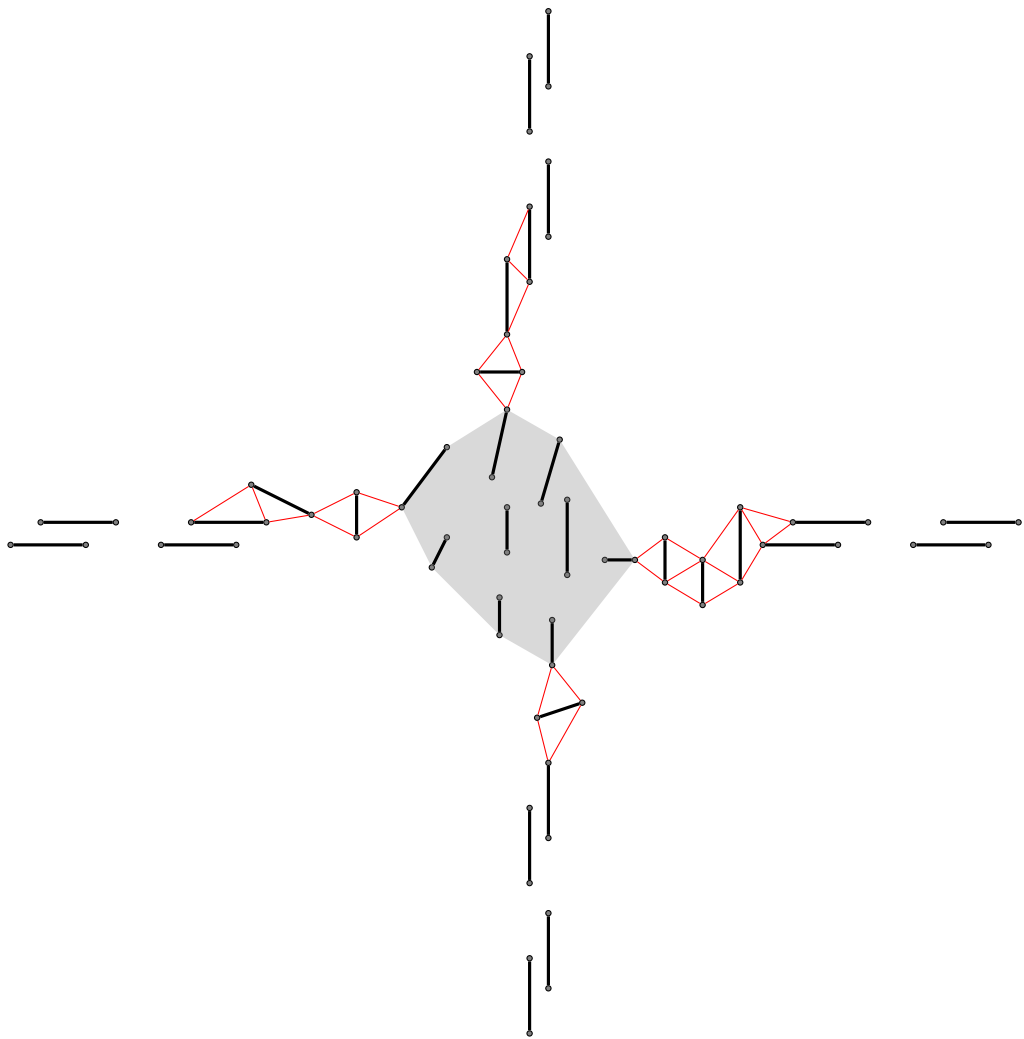
\begin{figure}[htbp]
	\centering
	\begin{tikzpicture}[RED/.style = {fill=red},
		BLUE/.style = {fill=blue},
		GREEN/.style = {fill=green}]
		\tikzstyle{every node}=[draw=black, fill=gray, shape=circle, minimum size=2pt,inner sep=0pt];
		\node (A) at (0,0.7) {}; % A
		\node (B) at (0.6,1.5) {}; % B
		
		\node (C) at (0.4,-0.1) {}; % C
		\node (D) at (0.6,0.3) {}; % D
		
		\node (E) at (1.4,0.1) {}; % E
		\node (F) at (1.4,0.7) {}; % F
		
		\node (G) at (1.2,1.1) {}; % G
		\node (H) at (1.4,2) {}; % H
		\node (I) at (1.85,0.75) {}; % I
		\node (J) at (2.1,1.6) {}; % J
		
		\node (K) at (2.2,-0.2) {}; % K
		\node (L) at (2.2,0.8) {}; % L
		
		\node (M) at (1.3,-1) {}; % M
		\node (N) at (1.3,-0.5) {}; % N
		
		\node (O) at (2,-1.4) {}; % M
		\node (P) at (2,-0.8) {}; % N
		
		\node (Q) at (2.7,0) {}; % M
		\node (R) at (3.1,0) {}; % N 
		
		\node (L1) at (-2,1) {};
		\node (L2) at (-1.2,0.6) {};
		
		\node (L3) at (-0.6,0.9) {};
		\node (L4) at (-0.6,0.3) {};
		
		\node (R1) at (3.5,0.3) {};
		\node (R2) at (3.5,-0.3) {};
		
		\node (R3) at (4,0) {};
		\node (R4) at (4,-0.6) {};
		
		\node (R5) at (4.5,0.7) {};
		\node (R6) at (4.5,-0.3) {};
		
		\node (U1) at (1,2.5) {};
		\node (U2) at (1.6,2.5) {};
		
		\node (U3) at (1.4,3) {};
		\node (U4) at (1.4,4) {};
		
		\node (D1) at (1.8,-2.1) {};
		\node (D2) at (2.4,-1.9) {};

		\fill[gray,opacity=0.3] (A.center)--(B.center)--(H.center)--(J.center)--(R.center)--(O.center)--(M.center)--(C.center)--(A.center);
		
		\draw[very thick] (A)--(B);
		\draw[very thick] (C)--(D);
		\draw[very thick] (E)--(F);
		\draw[very thick] (G)--(H);
		\draw[very thick] (I)--(J);
		\draw[very thick] (K)--(L);
		\draw[very thick] (M)--(N);
		\draw[very thick] (O)--(P);
		\draw[very thick] (Q)--(R);
		\draw[very thick] (L1)--(L2);
		\draw[very thick] (L3)--(L4);
		\draw[very thick] (R1)--(R2);
		\draw[very thick] (R3)--(R4);
		\draw[very thick] (R5)--(R6);
		\draw[very thick] (D1)--(D2);
		\draw[very thick] (U1)--(U2);
		\draw[very thick] (U3)--(U4);
		
		\begin{scope}[shift={(-4,0.2)}]
			\node (H1) at (0.8,0) {};
			\node (H1') at (1.8,0) {};
			\node (H2) at (1.2,0.3) {};
			\node (H2') at (2.2,0.3) {};
		\end{scope}
		
		\begin{scope}[shift={(-6,0.2)}]
			\node (x) at (0.8,0) {};
			\node (y) at (1.8,0) {};
			\node (z) at (1.2,0.3) {};
			\node (t) at (2.2,0.3) {};
			\draw[very thick] (x)--(y);
			\draw[very thick] (z)--(t);
		\end{scope}
		
		\begin{scope}[shift={(4,0.2)}]
			\node (H3) at (0.8,0) {};
			\node (H3') at (1.8,0) {};
			\node (H4) at (1.2,0.3) {};
			\node (H4') at (2.2,0.3) {};
		\end{scope}
		
		\begin{scope}[shift={(6,0.2)}]
			\node (x) at (0.8,0) {};
			\node (y) at (1.8,0) {};
			\node (z) at (1.2,0.3) {};
			\node (t) at (2.2,0.3) {};
			\draw[very thick] (x)--(y);
			\draw[very thick] (z)--(t);
		\end{scope}
		
		\begin{scope}[shift={(1.7,6)}]
			\node (V1) at (0,-1.3) {};
			\node (V1') at (0,-2.3) {};
			\node (V2) at (0.25,-0.7) {};
			\node (V2') at (0.25,-1.7) {};
		\end{scope}
		
		\begin{scope}[shift={(1.7,8)}]
			\node (x) at (0,-1.3) {};
			\node (y) at (0,-2.3) {};
			\node (z) at (0.25,-0.7) {};
			\node (t) at (0.25,-1.7) {};
			\draw[very thick] (x)--(y);
			\draw[very thick] (z)--(t);
		\end{scope}
		
		\begin{scope}[shift={(1.7,-2)}]
			\node (V3) at (0,-1.3) {};
			\node (V3') at (0,-2.3) {};
			\node (V4) at  (0.25,-0.7){};
			\node (V4') at (0.25,-1.7) {};
		\end{scope}
		
		\begin{scope}[shift={(1.7,-4)}]
			\node (x) at (0,-1.3) {};
			\node (y) at (0,-2.3) {};
			\node (z) at (0.25,-0.7) {};
			\node (t) at (0.25,-1.7) {};
			\draw[very thick] (x)--(y);
			\draw[very thick] (z)--(t);
		\end{scope}

		\draw[very thick] (H1)--(H1');
		\draw[very thick] (H2)--(H2');
		\draw[very thick] (H3)--(H3');
		\draw[very thick] (H4)--(H4');
		
		\draw[very thick] (V1)--(V1');
		\draw[very thick] (V2)--(V2');
		\draw[very thick] (V3)--(V3');
		\draw[very thick] (V4)--(V4');
		
		\tikzstyle{every path}=[draw=red];
		
		\draw (L2)--(H2')--(L1)--(H2);
		\draw (L2)--(L3)--(A)--(L4)--(L2);
		\draw (V1)--(U4)--(V1')--(U3);
		\draw (H)--(U1)--(U3)--(U2)--(H);
		
		\draw (H3)--(H4)--(R5)--(H3)--(R6);
		\draw (R1)--(R3)--(R2)--(R4)--(R6)--(R3)--(R5);
		\draw (R1)--(R)--(R2);
		
		\draw (O)--(D1)--(V4)--(D2)--(O);
		
	\end{tikzpicture}
	\caption{An embedding of an edge crossing gadget zoomed in.}
	\label{fig:crossingZoom}
\end{figure}

\begin{figure}
	\vspace{-6em}
	\hspace{-6em}
	\begin{tikzpicture} [scale=0.42,
		RED/.style = {fill=red},
		BLUE/.style = {fill=blue},
		GREEN/.style = {fill=green},
		]
		
		\tikzstyle{every node}=[draw=black, fill=gray, shape=circle, minimum size=2pt,inner sep=0pt];
		
		\node[GREEN] (A) at (-2,3.1) {};
		\node[BLUE] (B) at (-1,3.1) {};
		\node[RED] (C) at (-1.4,2.7) {};
		\node[GREEN] (D) at (-0.4, 2.7) {};
		\draw[thick] (A)--(B);
		\draw[thick] (C)--(D);
		
		\vedgeL{(0.2,4)}{38}	
		\hedgeT{(0,2.5)}{40}
		
		\vedgeCT{(9.7,2.2)}{8}
		\vedgeCF{(21.7,-21.8)}{8}
		\vedgeCF{(33.7,-33.8)}{10}
		
		\vedgeT{(9.7,-11.7)}{6}
		\vedgeT{(9.7,-23.7)}{6}
		\vedgeT{(9.7,-35.7)}{8}
		\vedgeL{(22,-35.7)}{6}
		
		\hedgeT{(0,-9.5)}{6}
		\hedgeT{(12,-9.5)}{28}
		\hedgeF{(0,-21.5)}{6}
		\hedgeF{(12,-21.5)}{28}
		\hedgeF{(0,-33.5)}{6}
		\hedgeF{(12,-33.5)}{6}
		\hedgeF{(24,-33.5)}{16}
		\crossing{(8,-9.7)}{1}
		\crossing{(8,-21.7)}{1}
		\crossing{(8,-33.7)}{1}
		\crossing{(20,-33.7)}{1}
		
		\hedgeL{(8.6,-44.8)}{10}
		\hedgeF{(24.3,-44.8)}{10}
		
		\nae{(22,-45)}
		
	\end{tikzpicture}
	\caption{An example embedding and coloring of a single clause $(q \vee s \vee t)$ using line segments.}
	\label{fig:exampleSegmentEmb}
\end{figure}
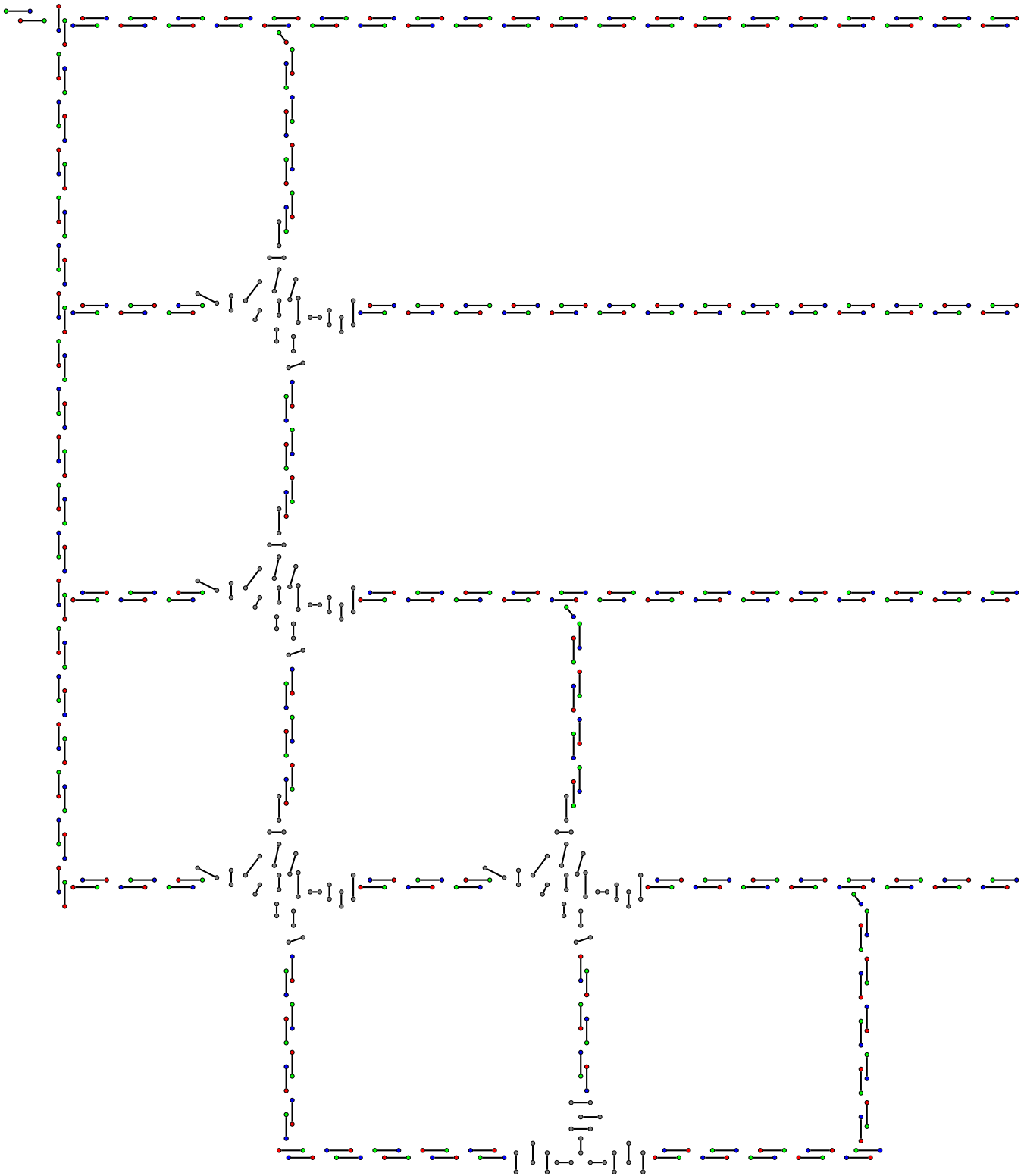

\begin{rem}
	K\'{a}ra et al. showed that there are exactly five cases when the visibility graph of a set of points is 3-colorable (see Figure 3 in \cite{Kara_pointVisChromatic}).
	However, note that our model considers also the Euclidean distances and thus this particular result does not apply to our case when the set of points are not bounded by a circle of diameter 1.
	
	The NP-completeness reduction for unit disk point visibility graphs are straightforward from Gr\"{a}f et al.'s proof of 3-colorability of unit disk graphs \cite{Graf_udgColoring}.	
\end{rem}

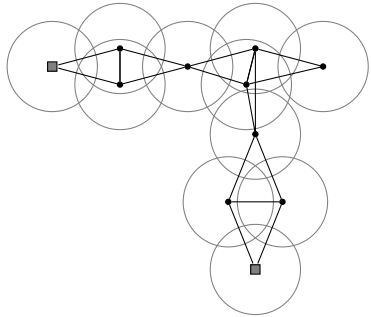
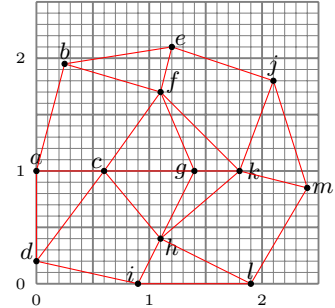
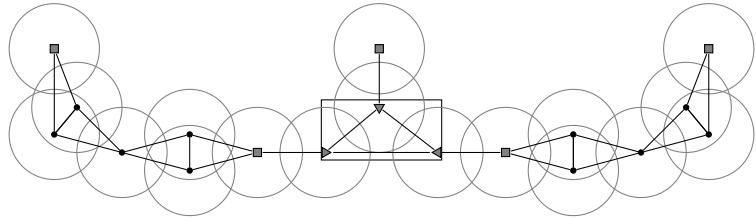
\begin{figure}[tp]
	\centering
	\captionsetup[subfigure]{position=b}
	\subfloat[Long edges for unit disk point visibility graphs.]{
		\centering
		\begin{tikzpicture}[scale=0.6,
			square/.style = {regular polygon, regular polygon sides=4, fill=gray, scale=2.5}	
			]
			\tikzstyle{every node}=[draw=black, fill=black, shape=circle, minimum size=2pt,inner sep=0pt];

			\coordinate (1) at (0,1.5);
			\coordinate (2) at (1.5,1.9);
			\coordinate (3) at (1.5,1.1);
			\coordinate (4) at (3,1.5);
			\coordinate (5) at (4.5,1.9);
			\coordinate (6) at (4.3,1.1);
			\coordinate (7) at (4.5,0);
			\coordinate (8) at (6,1.5);
			\coordinate (9) at (3.9,-1.5);
			\coordinate (10) at (5.1,-1.5);
			\coordinate (11) at (4.5,-3);

			\foreach \i in {2,...,10}
			{
				\node (\i) at (\i) {};
				\draw[gray] (\i) circle (1cm);
			}
			\draw[gray] (1) circle (1cm);
			\draw[gray] (11) circle (1cm);
			
			\node[square] (1) at (1) {};
			\node[square] (11) at (11) {};
			
			\draw (1)--(2)--(3)--(1);
			\draw (2)--(3)--(4)--(2);
			\draw (4)--(5)--(6)--(4);
			\draw (6)--(7)--(5);
			\draw (5)--(6)--(8)--(5);
			\draw (7)--(9)--(10)--(7);
			\draw (9)--(11)--(10);

		\end{tikzpicture}
		\label{fig:longEdgePt}
	}
	\hfill~\hfill
	\subfloat[The point configuration to be replaced with edge crossings which transfers the color from $a$ to $m$, and from $e$ to $i$.]{
		\centering
		\hspace{-0.5em}
		\begin{tikzpicture}[scale=1.5]
			\draw (0,0) to[grid with coordinates] (2.5,2.5);
			
			\tikzstyle{every node}=[draw, fill=black, shape=circle, minimum size=2pt,inner sep=0pt];
			
			\node[label={\footnotesize $a$}] (A) at (0,1) {}; % A
			\node[label={\footnotesize $b$}] (B) at (0.25,1.95) {}; % B
			\node[label=135:{\footnotesize $c$}] (C) at (0.6,1) {}; % C
			\node[label=135:{\footnotesize $d$}] (D) at (0,0.2) {}; % D			
			\node[label=20:{\footnotesize $e$}] (E) at (1.2,2.1) {}; % E
			\node[label=20:{\footnotesize $f$}] (F) at (1.1,1.7) {}; % F
			\node[label=left:{\footnotesize $g$}] (G) at (1.4,1) {}; % G
			\node[label=-10:{\footnotesize $h$}] (H) at (1.1,0.4) {}; % H			
			\node[label=100:{\footnotesize $i$}] (I) at (0.9,0) {}; % I
			\node[label=above:{\footnotesize $j$}] (J) at (2.1,1.8) {}; % J			
			\node[label=0:{\footnotesize $k$}] (K) at (1.8,1) {}; % K
			\node[label={\footnotesize $l$}] (L) at (1.9,0) {}; % L			
			\node[label=0:{\footnotesize $m$}] (M) at (2.4,0.85) {}; % M

			\tikzstyle{every path}=[draw, color=red];
			
			\draw (A)--(B);
			\draw (A)--(C);
			\draw (A)--(D);
			
			\draw (B)--(E);
			\draw (B)--(F);
			
			\draw (C)--(D);
			\draw (C)--(H);
			\draw (C)--(G);
			\draw (C)--(F);
			
			\draw (D)--(I);
			
			\draw (E)--(F);
			\draw (E)--(J);
			
			\draw (F)--(G);
			\draw (F)--(K);
			
			\draw (G)--(H);
			\draw (G)--(K);
			
			\draw (H)--(I);
			\draw (H)--(L);	
			\draw (H)--(K);		
			
			\draw (I)--(L);
			
			\draw (J)--(K);
			\draw (J)--(M);
			
			\draw (K)--(M);
			
			\draw (L)--(M);
		\end{tikzpicture}
		\label{fig:crossingGadgetPt}
	}
	
	\subfloat[Monotone NAE3SAT clause for a unit disk point visibility graph.]{
		\centering
		\begin{tikzpicture}[scale=0.6, 
			square/.style = {regular polygon, regular polygon sides=4, fill=gray, scale=2.2},
			triangle/.style = {regular polygon, regular polygon sides=3, fill=gray, scale=2.2}
			]
			\tikzstyle{every node}=[draw, fill=black, shape=circle, minimum size=2pt,inner sep=0pt]
			\node[fill=none, shape=rectangle, minimum width=1.6cm, minimum height=0.8cm] at (8.75,2) {};	
			
			\node[square]  (1) at (1.5,3.8) {};
			\node (2) at (1.5,1.9) {};
			\node (3) at (2,2.5) {};
			\node (4) at (3,1.5) {};
			\node (5) at (4.5,1.9) {};
			\node (6) at (4.5,1.1) {};
			\node[square] (7) at (6,1.5) {};
			\node[triangle, rotate=-90] (8) at (7.5,1.5) {};
			\node[triangle, rotate=180] (9) at (8.7,2.5) {};
			\node[triangle, rotate=90] (10) at (10,1.5) {};
			\node[square] (11) at (11.5,1.5) {};
			\node (12) at (13,1.9) {};
			\node (13) at (13,1.1) {};
			\node (14) at (14.5, 1.5) {};
			\node (15) at (15.5, 2.5) {}; 
			\node (16) at (16, 1.9) {};
			\node[square] (17) at (16, 3.8) {};
			\node[square] (18) at (8.7,3.8) {};

			\foreach \i in {1,...,18}
			{
				\draw[gray] (\i) circle (1cm);
			}

			\draw (1)--(2)--(3)--(1);
			\draw (2)--(3)--(4)--(2);
			\draw (4)--(5)--(6)--(4);
			\draw (6)--(7)--(5);
			\draw (7)--(8)--(9)--(10)--(8);
			\draw (10)--(11);
			\draw (11)--(12)--(13)--(11);
			\draw (9)--(18);
			\draw (12)--(13)--(14)--(12);
			\draw (14)--(15)--(16)--(14);
			\draw (15)--(16)--(17)--(15);
			
		\end{tikzpicture}
	}
	\caption{Three main components to build the NP-hardness gadget for unit disk visibility graph of a set of points.}
	\label{fig:clauseGadgetPt}
\end{figure}

In Figure \ref{fig:longEdgePt}, we show the gadget to transfer the color on a long edge for a unit disk point visibility graph. In Figure~\ref{fig:crossingGadgetPt}, we give the gadget to replace edge crossings. In Figure~\ref{fig:clauseGadgetPt}, we show the embedding of the points in the Euclidean plane. Then, the same reduction given in Section~\ref{sec:3coloringSegment} for unit disk segment visibility graphs can be utilized to prove the NP-completeness of 3-coloring problem for unit disk point visibility graphs.

\subsection{The 3-coloring problem for unit disk visibility graphs for polygons with holes} \label{sec:withholes}
In this section, we prove that determining whether the unit disk visibility graph of a given polygon with holes is 3-colorable or not is an NP-complete problem.
Let us state our theorem.
\begin{thm}  \label{thm:withholes}
	There is a polynomial-time reduction from the 3-coloring problem for 4-regular planar graphs to the 3-coloring problem for unit disk visibility graphs of polygons with holes.
\end{thm}

Before proving Theorem~\ref{thm:withholes}, we describe the gadgets used to construct a polygon with holes (which has a corresponding unit disk visibility graph for the constructed polygon) from a given 4-regular planar graph in more detail, and show that they correctly transform an instance of the 3-coloring problem for 4-regular planar graphs to an instance of the 3-coloring problem for unit disk visibility graphs of polygons with holes.

\subsubsection{The corridors}\label{sec:corridor}

We first describe how we model the edges of a given 4-regular planar graph. In Figure~\ref{fig:polygonEdge}, there are two nodes, $u$ and $v$, and a ``corridor'' which connects them.
The interior of the polygon is shaded, and the boundaries are indicated with bold lines.
The visibility edges are indicated using thin lines, and colored red. When the number of vertices on each side of the corridor (excluding $u$ and $v$) is a positive multiple of 3, it is trivial to verify that $u$ and $v$ receive different colors in a proper 3-coloring.

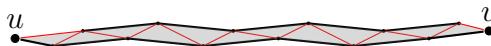
\begin{figure}[h]
	\centering
	\begin{tikzpicture}
		
		\tikzstyle{every node}=[draw=black, fill=black, shape=circle, minimum size=1pt,inner sep=0pt];

		\foreach [evaluate={\j=int(mod(\i,2));}] \i in {0,...,5}
		{
			\ifthenelse{\j = 1}
			{
				\coordinate (d\i) at (\i,0);
				\coordinate (u\i) at (\i + 0.4 ,0.2);
			}
			{
				\coordinate (d\i) at (\i,-0.1);
				\coordinate (u\i) at (\i + 0.4 ,0.1);
			}		
		}

		\foreach [evaluate={\j=\i+1;}] \i in {0,...,5}
		{
			\ifthenelse{\i = 5}{\draw[red] (u\i)--(d\i);}
			{
				\draw[thick,black] (u\i)--(u\j);
				\draw[thick,black] (d\i)--(d\j);
				\draw[red] (u\i)--(d\i);
				\draw[red] (u\i)--(d\j);				
			}
			\node at (u\i) {};
			\node at (d\i) {};
		}
		
		\fill[gray, opacity=0.3] (-0.5,0)--(u0)--(u1)--(u2)--(u3)--(u4)--(u5)--(5.8,0.1)--(d5)--(d4)--(d3)--(d2)--(d1)--(d0)--(-0.5,0);
		\node[label=$u$,fill=black,scale=3] (u) at (-0.5,0) {};
		\node[label=$v$,fill=black,scale=3] (v) at (5.8,0.1) {};
		\draw[red] (u)--(u0);
		\draw[thick] (u)--(d0);
		\draw[red] (v)--(u5);
		\draw[thick] (v)--(d5);
	\end{tikzpicture}
	\caption{A corridor modeling the edges in a planar graph.}
	\label{fig:polygonEdge}
\end{figure}

Therefore, \emph{a corridor} shown in Figure~\ref{fig:polygonEdge} replaces the edges in a given planar graph.
We use the same idea which we used to model edges in unit disk segment visibility graphs.
However, unlike the wires, instead of transferring a color along a long edge, our gadget makes sure that two ends of an edge receives different colors since these ends correspond to adjacent vertices of the given 4-regular planar graph.
A corridor consists of two polygonal chains $A$ and $B$ with edges $a_1a_2, a_2a_3, \dots, a_{k-1}a_k$ and $b_1b_2, b_2b_3, \dots, b_{k-1}b_k$, respectively. It is trivial to see that we can obtain a unit disk visibility graph for a polygon with holes, where for each $i$, the visibility edges $a_ib_i$, and $b_ia_{i+1}$ exist as visibility edges\footnote{Even if we assume that no three vertices can be collinear, we can slightly perturb the vertices by $\varepsilon$ units where $\varepsilon$ is some positive number which can be represented using polynomially many decimal digits.}.
This basically describes an induced subgraph with $2k$ vertices, $2k-2$ boundary edges, and $2k-1$ visibility edges.
Moreover, the largest induced cycle is 3 (which means this is a chordal graph), and each triplet $(a_i,a_{i+1},b_i)$ and $(b_i,b_{i+1},a_i)$ yields a $C_3$.
Now, suppose that the polygonal chain $A$ has two neighboring vertices $u$ and $v$ where $ua_1$ and $va_k$ are two polygonal edges.
Assuming that $k = 3c$ for some constant $c \in \mathbb{N}^+$, $u$ and $v$ receive different colors in a 3-coloring of the described subgraph.

\subsubsection{The chambers}\label{sec:chamber}

Now, let us describe the gadget which replaces the vertices in a given planar graph, which we refer to as a ``chamber''. \emph{A chamber} is an induced subgraph with 12 vertices $c_1, \dots, c_{12}$ with boundary edges $c_1c_2$, $c_3c_4$, $c_4c_5$, $c_6c_7$, $c_7c_8$, $c_9c_{10}$, $c_{10}c_{11}$, and $c_{12}c_1$. 
Figure~\ref{fig:polygonVertex} shows an embedding of a chamber of a polygon $P$ where the interior of the polygon is shaded. The vertex $c_1$ is at the center of some circle $\mathcal{C}$ with radius 1.

\begin{figure}[h]
	\centering
	\begin{tikzpicture} [scale = 2]
		\draw (-1,-1) to[grid with coordinates] (1,1);
		\draw[blue] (0,0) circle (1cm);
		
		\begin{scope}[shift = {(-1.4,-0.8)}]
			\tikzstyle{every node}=[draw=black, fill=gray, shape=circle, minimum size=3pt,inner sep=0pt];
			\coordinate (0) at (1.4,0.8);
			\coordinate (1) at (2,0);
			\coordinate (2) at (2.2,0.2);
			\coordinate (3) at (2.8,0.8);
			\coordinate (4) at (2.2,1.4);
			\coordinate (5) at (2,1.6);
			\coordinate (6) at (1.4,2.2);
			\coordinate (7) at (0.8,1.6);
			\coordinate (8) at (0.6,1.4);
			\coordinate (9) at (0,0.8);
			\coordinate (10) at (0.6,0.2);
			\coordinate (11) at (0.8,0);
			
			\coordinate (e1) at (2.4,-0.4);
			\coordinate (e2) at (2.6,-0.2);
			\coordinate (e4) at (2.6,1.8);
			\coordinate (e5) at (2.4,2);
			\coordinate (e7) at (0.4,2);
			\coordinate (e8) at (0.2,1.8);
			\coordinate (e10) at (0.2,-0.2);
			\coordinate (e11) at (0.4,-0.4);	
			
			\draw[ultra thick] (0)--(1)--(e1);
			\draw[ultra thick] (e2)--(2)--(3)--(4)--(e4);
			\draw[ultra thick] (e5)--(5)--(6)--(7)--(e7);
			\draw[ultra thick] (e8)--(8)--(9)--(10)--(e10);
			\draw[ultra thick] (e11)--(11)--(0);
			
			\foreach \i in {2,4,5,7,8,10}
			{
				\draw[red] (0)--(\i);
			}
			
			\foreach \i in {1,...,11}
			{
				\node at (\i) {};
				\ifthenelse { \i = 3 \OR \i = 6 \OR \i = 9}
				{}
				{
					\node at (\i) {};
				}
			}
			
			\node[scale=3] at (0) {};
			
			\fill[gray, opacity=0.3] (0)--(1)--(e1)--(e2)--(2)--(3)--(4)--(e4)--(e5)--(5)--(6)--(7)--(e7)--(e8)--(8)--(9)--(10)--(e10)--(e11)--(11)--(0);

		\end{scope}
	\end{tikzpicture}
	\caption{The chamber gadget that replaces a vertex in a given planar graph.}
	\label{fig:polygonVertex}
\end{figure}
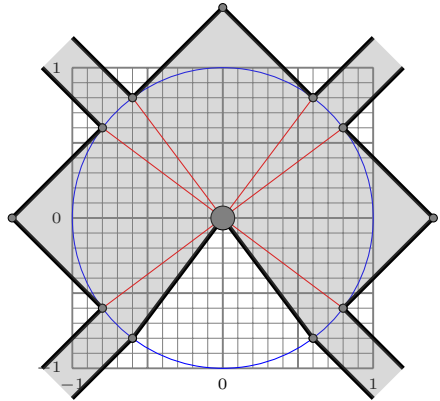

Let us refer to such a vertex as \emph{the central vertex} of the chamber. The vertices $c_i$ for $i=2,3,5,6,8,9,11,12$ are on the boundary of $\mathcal{C}$, and the remaining vertices $c_4, c_7$, and $c_{10}$ are outside $\mathcal{C}$, which means they do not see $c_1$.
The vertices that are on the boundary of $\mathcal{C}$ are four pair of ``openings'' to the corridors which connect chambers together since the given planar graph is 4-regular. A unit disk is drawn around the central vertex to demonstrate the visibility relations between it and the opening of the corridors. The eight vertices on that unit disk are called the corridor vertices of a chamber, and the remaining three vertices are called the connecting vertices of a chamber. The connecting vertices are essential because the color of the corridor vertices must be dependent only on the central vertex. In this case, if the input graph has two adjacent vertices $u$ and $v$, then there exists a pair of chambers $U$ and $V$, and a corridor with $3c$ vertices which connects $U$ and $V$, and the central vertices $c_u$ and $c_v$ of $U$ and $V$ must receive different colors.

\subsubsection{The proof of Theorem~\ref{thm:withholes}}

We now show that the 3-coloring problem for unit disk visibility graphs of polygons with holes is NP-hard by giving a reduction from the 3-coloring problem for 4-regular planar graphs \cite{Dailey_4regularplanar}.

Given a 4-regular planar graph, we construct a polygon with holes. Two main components of our reduction are as follows.

\begin{enumerate}
	\item  \emph{A corridor} shown in Figure~\ref{fig:polygonEdge} replaces the edges in a given planar graph.
	We use the same idea which we used to model edges in unit disk segment visibility graphs.
	However, instead of transferring a color along a long edge, our gadget makes sure that two ends of an edge receives different colors since the colors of these ends are determined by the colors of the corresponding adjacent vertices of the given 4-regular planar graph. Assuming that $k = 3c$ for some constant $c \in \mathbb{N}^+$, $u$ and $v$ receive different colors in a 3-coloring of the described subgraph.
	
	\item \emph{A chamber} shown in Figure~\ref{fig:polygonVertex} replaces the vertices in a given planar graph. Since we give a reduction from 4-regular planar graphs, each chamber has exactly four corridors connected to it.
	The big vertex in the center, which is called the central vertex of the chamber, corresponds to a vertex of the given planar graph. In this case, if the input graph has two adjacent vertices $u$ and $v$, then there exists a pair of chambers $U$ and $V$, and a corridor with $3c$ vertices which connects $U$ and $V$, and the central vertices $c_u$ and $c_v$ of $U$ and $V$ must receive different colors. 
\end{enumerate}

Given a 4-regular planar graph $H$ on $n$ vertices $v_1, \dots, v_n$, we construct the corresponding polygon $P$ with holes as follows:

\begin{itemize}
	\item For each vertex $v_i$, add a chamber to $P$ whose central vertex is vertex $u_i$.
	\item For each pair of adjacent vertices $(v_i,v_j)$, add a corridor to $P$ between the chambers with central vertices $u_i$ and $u_j$. 
\end{itemize}

Considering any 3-coloring of $H$, the color given to the vertex $v_i \in H$ can be given to the central vertex $u_i$ of the chamber of $P$ replacing $v_i$, and the colors of central vertices determines the colors of the vertices of corridor, thus a 3-coloring of $P$. Considering any 3-coloring of $P$, the color given to the central vertex $u_i$ of the chamber of $P$ can be given to the vertex $v_i \in H$ replaced by $u_i$. Therefore, $P$ has a 3-coloring if and only if the corresponding color given to the vertices of $H$ yields a 3-coloring. \hfill $\qed$

See Figure~\ref{fig:PolHol} for an example 3-regular planar graph and its corresponding polygon with holes in our reduction.

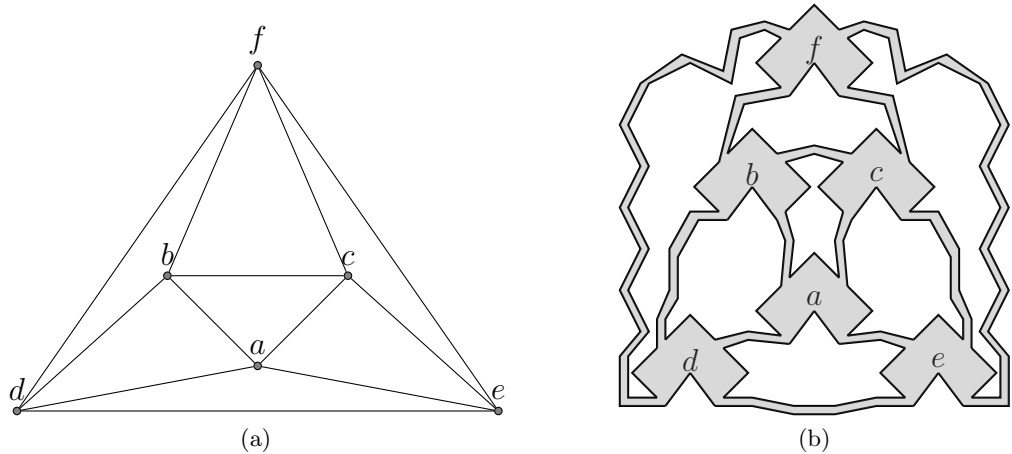
\begin{figure} [htbp]
	\captionsetup[subfigure]{position=b}
	\centering
	\subfloat[]{
		\centering
		\begin{tikzpicture}[scale=2]
			\tikzstyle{every node}=[draw=black, fill=gray, shape=circle, minimum size=1.5pt,inner sep=1pt];
			\node[label=$a$] (a) at (0,0) {};
			\node[label=$b$] (b) at (-0.6,0.6) {};
			\node[label=$c$] (c) at (0.6,0.6) {};
			\node[label=$d$] (d) at (-1.6,-0.3) {};
			\node[label=$e$] (e) at (1.6,-0.3) {};
			\node[label=$f$] (f) at (0,2) {};
			
			\draw (a)--(b)--(c)--(a);
			\draw (d)--(e)--(f)--(d);
			\draw (a)--(d);
			\draw (a)--(e);
			\draw (b)--(d);
			\draw (b)--(f);
			\draw (c)--(e);
			\draw (c)--(f);
			
		\end{tikzpicture}
		\label{fig:exampleGraph}
	}
	~\hfill~
	\subfloat[]{
		\centering
		\begin{tikzpicture}[scale=0.55]
			\node at (0,0.3) {$a$};
			\node at (-1.5,3.3) {$b$};
			\node at (1.5,3.3) {$c$};
			\node at (-3,-1.2) {$d$};
			\node at (3,-1.2) {$e$};
			\node at (0,6.3) {$f$};
			\tikzstyle{every node}=[draw=black, fill=gray, shape=circle, minimum size=1.5pt,inner sep=0pt];
			\polyV{(0,0)}{0};
			\polyV{(-1.5,3)}{0};
			\polyV{(1.5,3)}{0};
			\polyV{(-3,-1.5)}{0};
			\polyV{(3,-1.5)}{0};
			\polyV{(0,6)}{0};
			\tikzstyle{every path}=[thick];
			
			%%a--b%%
			\draw (-0.6,0.8)--(-0.5,1.7)--(-0.7,2.4);
			\draw (-0.8,0.6)--(-0.7,1.7)--(-0.9,2.2);
			\fill[gray,opacity=0.3] (-0.6,0.8)--(-0.5,1.7)--(-0.7,2.4)--(-0.9,2.2)--(-0.7,1.7)--(-0.8,0.6);
			
			%%a--c%%
			\draw (0.6,0.8)--(0.5,1.7)--(0.7,2.4);
			\draw (0.8,0.6)--(0.7,1.7)--(0.9,2.2);
			\fill[gray,opacity=0.3] (0.6,0.8)--(0.5,1.7)--(0.7,2.4)--(0.9,2.2)--(0.7,1.7)--(0.8,0.6);
			
			%%a--d%%
			\draw (-0.8,-0.6)--(-1.6,-0.5)--(-2.4,-0.7);
			\draw (-0.6,-0.8)--(-1.6,-0.7)--(-2.2,-0.9);
			\fill[gray,opacity=0.3] (-0.8,-0.6)--(-1.6,-0.5)--(-2.4,-0.7)--(-2.2,-0.9)--(-1.6,-0.7)--(-0.6,-0.8);
			
			%%a--e%%
			\draw (0.8,-0.6)--(1.6,-0.5)--(2.4,-0.7);
			\draw (0.6,-0.8)--(1.6,-0.7)--(2.2,-0.9);
			\fill[gray,opacity=0.3] (0.8,-0.6)--(1.6,-0.5)--(2.4,-0.7)--(2.2,-0.9)--(1.6,-0.7)--(0.6,-0.8);
			
			%%b--c%%			
			\draw (-0.9,3.8)--(0,4)--(0.9,3.8);
			\draw (-0.7,3.6)--(0,3.8)--(0.7,3.6);
			\fill[gray,opacity=0.3] (-0.9,3.8)--(0,4)--(0.9,3.8)--(0.7,3.6)--(0,3.8)--(-0.7,3.6);
			
			%%b--d%%			
			\draw (-2.3,2.4)--(-3,2.4)--(-3.5,1.5)--(-3.5,0.5)--(-3.8,-0.2)--(-3.8,-0.9);
			\draw (-2.1,2.2)--(-2.8,2.2)--(-3.3,1.3)--(-3.3,0.3)--(-3.6,-0.2)--(-3.6,-0.7);
			\fill[gray,opacity=0.3] (-2.3,2.4)--(-3,2.4)--(-3.5,1.5)--(-3.5,0.5)--(-3.8,-0.2)--(-3.8,-0.9)--(-3.6,-0.7)--(-3.6,-0.2)--(-3.3,0.3)--(-3.3,1.3)--(-2.8,2.2)--(-2.1,2.2);
			
			%%b--f%%
			\draw (-2.1,3.8)--(-1.7,5)--(-0.6,5.2);
			\draw (-2.3,3.6)--(-1.9,5.2)--(-0.8,5.4);
			\fill[gray,opacity=0.3] (-2.1,3.8)--(-1.7,5)--(-0.6,5.2)--(-0.8,5.4)--(-1.9,5.2)--(-2.3,3.6);
			
			%%c--f%%
			\draw (2.1,3.8)--(1.7,5)--(0.6,5.2);
			\draw (2.3,3.6)--(1.9,5.2)--(0.8,5.4);
			\fill[gray,opacity=0.3] (2.1,3.8)--(1.7,5)--(0.6,5.2)--(0.8,5.4)--(1.9,5.2)--(2.3,3.6);
			
			%%c--e%%
			
			\draw (2.3,2.4)--(3,2.4)--(3.5,1.5)--(3.5,0.5)--(3.8,-0.2)--(3.8,-0.9);
			\draw (2.1,2.2)--(2.8,2.2)--(3.3,1.3)--(3.3,0.3)--(3.6,-0.2)--(3.6,-0.7);
			\fill[gray,opacity=0.3] (2.3,2.4)--(3,2.4)--(3.5,1.5)--(3.5,0.5)--(3.8,-0.2)--(3.8,-0.9)--(3.6,-0.7)--(3.6,-0.2)--(3.3,0.3)--(3.3,1.3)--(2.8,2.2)--(2.1,2.2);
			
			%%d--e$$
			\draw (-2.2,-2.1)--(-1.5,-2.1)--(-0.5,-2.3)--(0.5,-2.3)--(1.5,-2.1)--(2.2,-2.1);
			\draw (-2.4,-2.3)--(-1.5,-2.3)--(-0.5,-2.5)--(0.5,-2.5)--(1.5,-2.3)--(2.4,-2.3);
			\fill[gray,opacity=0.3] (-2.2,-2.1)--(-1.5,-2.1)--(-0.5,-2.3)--(0.5,-2.3)--(1.5,-2.1)--(2.2,-2.1)--(2.4,-2.3)--(1.5,-2.3)--(0.5,-2.5)--(-0.5,-2.5)--(-1.5,-2.3)--(-2.4,-2.3);
			
			%%d--f%%
			\draw (-3.8,-2.1)--(-4.5,-2.1)--(-4.5,-1.1)--(-4,-0.5)--(-4.5,0.5)--(-4,1.5)--(-4.5,2.5)--(-4,3.5)--(-4.5,4.5)--(-4,5.5)--(-3,6)--(-2,5.5)--(-1.8,6.6)--(-1.2,6.8)--(-0.8,6.6);
			\draw (-3.6,-2.3)--(-4.7,-2.3)--(-4.7,-1.1)--(-4.2,-0.5)--(-4.7,0.5)--(-4.2,1.5)--(-4.7,2.5)--(-4.2,3.5)--(-4.7,4.5)--(-4.2,5.5)--(-3.2,6.2)--(-2.2,5.8)--(-2,6.8)--(-1.2,7)--(-0.6,6.8);
			\fill[gray,opacity=0.3] (-3.8,-2.1)--(-4.5,-2.1)--(-4.5,-1.1)--(-4,-0.5)--(-4.5,0.5)--(-4,1.5)--(-4.5,2.5)--(-4,3.5)--(-4.5,4.5)--(-4,5.5)--(-3,6)--(-2,5.5)--(-1.8,6.6)--(-1.2,6.8)--(-0.8,6.6)--(-0.6,6.8)--(-1.2,7)--(-2,6.8)--(-2.2,5.8)--(-3.2,6.2)--(-4.2,5.5)--(-4.7,4.5)--(-4.2,3.5)--(-4.7,2.5)--(-4.2,1.5)--(-4.7,0.5)--(-4.2,-0.5)--(-4.7,-1.1)--(-4.7,-2.3)--(-3.6,-2.3);
			
			%%e--f%%
			\draw (3.8,-2.1)--(4.5,-2.1)--(4.5,-1.1)--(4,-0.5)--(4.5,0.5)--(4,1.5)--(4.5,2.5)--(4,3.5)--(4.5,4.5)--(4,5.5)--(3,6)--(2,5.5)--(1.8,6.6)--(1.2,6.8)--(0.8,6.6);
			\draw (3.6,-2.3)--(4.7,-2.3)--(4.7,-1.1)--(4.2,-0.5)--(4.7,0.5)--(4.2,1.5)--(4.7,2.5)--(4.2,3.5)--(4.7,4.5)--(4.2,5.5)--(3.2,6.2)--(2.2,5.8)--(2,6.8)--(1.2,7)--(0.6,6.8);
			\fill[gray,opacity=0.3] (3.8,-2.1)--(4.5,-2.1)--(4.5,-1.1)--(4,-0.5)--(4.5,0.5)--(4,1.5)--(4.5,2.5)--(4,3.5)--(4.5,4.5)--(4,5.5)--(3,6)--(2,5.5)--(1.8,6.6)--(1.2,6.8)--(0.8,6.6)--(0.6,6.8)--(1.2,7)--(2,6.8)--(2.2,5.8)--(3.2,6.2)--(4.2,5.5)--(4.7,4.5)--(4.2,3.5)--(4.7,2.5)--(4.2,1.5)--(4.7,0.5)--(4.2,-0.5)--(4.7,-1.1)--(4.7,-2.3)--(3.6,-2.3);
			
		\end{tikzpicture}
		\label{fig:examplePolygon}
	}
	
	\caption{(a) An example 4-regular planar graph, and (b) the polygon that corresponds to the graph given in (a).} 
	\label{fig:PolHol}
\end{figure}

\subsubsection{The time and space complexity.} Given a 4-regular planar graph $H$ on $n$ vertices, we add $n$ chambers to $P$, each having 12 vertices. The positions of the centers of chambers can be determined with respect to any planar embedding of $H$. For each pair of adjacent vertices in $H$, we add a corridor to $P$ between the chambers corresponding to these vertices. The number of vertices on each polygonal chain of a corridor is at most $O(n)$, therefore at most $O(n + n) = O(n)$ vertices in a corridor, in total. Thus, both chambers and corridors take up polynomial space. Since there are $2n$ edges in $H$, there are at most $O(n + 2n^2) = O(n^2)$ vertices in $P$, thus polynomially many edges in $P$. As a result, the given reduction can be done in polynomial time and space.

As we proved the correctness of our reduction and showed that it is a polynomial-time reduction, the theorem holds. Since the 3-coloring problem for 4-regular planar graphs is NP-complete \cite{Dailey_4regularplanar}, the  3-coloring problem for unit disk visibility graphs of polygons with holes is also NP-complete.

Since the 3-coloring problem for 4-regular planar graphs is NP-complete \cite{Dailey_4regularplanar}, the 3-coloring problem for unit disk visibility graphs of polygons with holes is also NP-complete by Theorem~\ref{thm:withholes}.

\section{Some other combinatorial problems on unit disk visibility graphs} \label{sec:othercomb}

In previous sections, we considered the 3-coloring problem (in general, chromatic number problem) for the unit disk visibility graphs of a set of line segments, and a set of points.
There are of course many other combinatorial problems to consider. 
In this section, our aim is to give some insight about some other famous combinatorial problems, and argue why it would be interesting to study them.

In the gadget used to prove NP-completeness of 3-coloring of unit disk segment visibility graphs, all line segments can be exactly one unit long except the edge crossings. Moreover, the rest of the gadget contains line segments either horizontal or vertical (parallel to $x$ or $y$-axis). Considering these facts, we pose these two interesting questions for reader's consideration: 

\begin{open}
	Is the 3-colorability of unit disk visibility graphs of line segments NP-hard when all the segments are exactly 1 unit long?
\end{open}

\begin{open}
	Is the 3-colorability of unit disk visibility graphs of line segments NP-hard when all the segments are either vertical or horizontal?
\end{open}

As the above results show that unit disk visibility graphs are not included in the (hierarchic) intersection of unit disk graphs and visibility graphs, we would like to study the following problems which may have interesting results on unit disk visibility graphs.

\begin{open}
	Maximum clique problem on unit disk visibility graphs.
\end{open}

It is a long known result that the maximum clique problem of a unit disk graph can be found in polynomial time when the representation is given \cite{Clark_UDmaxclique} and even without a representation \cite{Raghavan_robust}.
However, the algorithm described by Clark et al. \cite{Clark_UDmaxclique} relies on the fact that if two disk centers are close enough, then the corresponding disks intersect, and thus the corresponding vertices are adjacent. 
Whereas for unit disk visibility graph, this is not the case.
Even though two disk centers are close enough, there might be an obstacle in-between, which means the corresponding vertices might not be adjacent.
It is therefore an interesting problem to study for unit disk visibility graphs.

\begin{open}
	Chromatic number problem on unit disk visibility graphs of polygons.
\end{open}

In Section~\ref{sec:3coloringSegment}, we showed that the determining whether a visibility graph of a set of line segments is 3-colorable is NP-complete.
These reductions are possible because of the edge-crossing certificates (see Figure~\ref{fig:crossingGadgetSeg} and Figure~\ref{fig:crossingGadgetPt}).
In the certificates, we use a special structure in the middle of the gadget, to transfer the colors properly.
Specifically, this structure is $W_4$: consists of five vertices: $a,b,c,d,e$ and 8 edges $ab, bc, cd, da, ea, eb, ec, ed$.
In Section~\ref{sec:polyproper}, it was proven that deciding whether the visibility graph of a simple polygon is 4-colorable can be done in polynomial time, and 5-colorability problem is NP-complete.
The complexity of deciding the 3-colorability (even 4-colorability) for the unit disk visibility graph of a simple polygon is yet to be solved.

\begin{open}
	Determining whether a unit disk visibility graph of segments yields a Hamiltonian cycle.
\end{open}

Hoffman and T\'{o}th showed that every segment visibility graph yields a Hamiltonian cycle \cite{Hoffman_segment}.
It is clearly not the case for unit disk visibility graphs of segments (consider two segments with entpoints on $(0,0)$, $(1,0)$, $(0,1)$, $(0,2)$).
Thus, we ask the question: can we determine whether a unit disk visibility graph of line segments yield a Hamiltonian cycle in polynomial time?

\begin{open}
	Conflict-free coloring problem on a unit disk visibility graph of segments, a unit disk visibility graph of points, and a unit disk visibility graph of polygons.
\end{open}

A correct solution to the conflict-free coloring problem on a graph is assigning colours to (not necessarily all) vertices of that graph such that the neighborhood of each vertex contains at least one unique color \cite{cf-app}.
This problem is shown to be NP-hard for the visibility graphs of simple polygons in Section~\ref{sec:polygonconflict}.
To the best of our knowledge, the problem is open for the visibility graphs of a set of points, and the visibility graphs of line segments.
As for the unit disk graphs, it is NP-complete to decide whether a conflict-free coloring with only one color exists, and six colors are always sufficient \cite{Fekete_conflictfree}.
However, as we mentioned in Section~\ref{sec:classification}, the result for the unit disk graphs does not apply to unit disk visibility graphs.
All considered, studying the conflict-free coloring problem for UDVG of segments, UDVG of points, and UDVG of polygons is an interesting research direction.

\chapter{Axes-Parallel Unit Disk Graphs} \label{chap:apud}

\section{Summary of the chapter}
In this chapter, we first describe a polynomial-time reduction which shows that deciding whether a graph can be realized as unit disks onto given straight lines is NP-hard, 
when the given lines are parallel to either $x$-axis or $y$-axis.
Using the reduction we described, we also show that this problem is NP-complete when the given lines are only parallel to $x$-axis (and one another).
We obtain those results using the idea of the logic engine introduced by Bhatt and Cosmadakis in 1987.

\begin{itemize}
	\item If the disks are forced to be centered on straight lines of which any pair is either parallel or perpendicular, then the recognition problem is NP-hard.
	\item If any pair of the pre-given lines are parallel to each other, then the problem is NP-complete.
	\item If there are only two lines that are perpendicular, then the recognition problem is interesting to study, yet does not yield trivial results.
\end{itemize}

\section{Related work}	
In this chapter, we are particularly interested in the unit disk graph recognition problem i.e. given a simple graph, deciding whether there exists an embedding of disks onto the plane whose intersections corresponds to the given graph.

Breu and Kirkpatrick showed that the unit disk graph recognition problem is NP-hard in general\cite{Breu_UDrecog}. 
Later on, this result was extended, and it was proved that the problem is also $\exists\mathbb{R}$-complete \cite{sphereAndDotProduct,integerRealization}.
Aspnes et al. showed that, even when we know the precise pairwise distances between the adjacent vertices in a unit disk graph, it is NP-hard to find a unique embedding of unit disks \cite{complexityOfWSN}\footnote{Note that this problem is different than embeddability of an edge weighted graph \cite{saxe}, since two disks must intersect if their centers are close enough.}.
Kuhn et al. showed that finding a ``good'' embedding is not approximable when the problem is parameterized by the maximum distance between any pair of disks' centers \cite{udgApprox}.
Besides the recognition problem, there are other famous combinatorial problems that remain NP-hard when the input graph is a unit disk graph.
Some of these problems are, but not limited to, maximum independent set \cite{Clark_UDmaxclique}, $k$-coloring \cite{Graf_udgColoring}, minimum dominating set \cite{mCenterProblems},  and Hamiltonian cycle \cite{hamiltonGrid}.

Ito and Kadoshita tackled the above mentioned problems by assuming that the centers of disks are in a predefined square-shaped region \cite{udgParameterized}.
They found out that some of these problems, namely, Hamiltonian cycle and $k$-coloring are fixed-parameter tractable, whereas some other problems, namely, the maximum independent set and the minimum dominating set are both W[1]-complete.
The result of \cite{udgParameterized} is a nice example to show that restricting the domain does not cause radical changes on the complexities of all major combinatorial problems.
But how does the recognition problem behave when the domain is restricted?

Intuitively, the most restricted domain for unit disk graphs is when the disks are centered on a single straight line in the Euclidean plane.
In this case, the unit disks become \emph{unit intervals} on the line, and they yield a \emph{unit interval graph} \cite{unitIntervalGraphs}.
To recognize whether a given graph is a unit interval graph is a linear-time task \cite{intervalRecognition}.
Breu studied the problem of unit disk graph recognition when each disk is centered on a ``thick line'' of width $c$ \cite{udgConstrained}.
In other words, the disks are restricted to be centered in the area between two straight lines with Euclidean distance $c$.
Such a configuration is called a \emph{$c$-strip graph}.
He showed that the recognition problem can be solved in polynomial time with this constraint, when $c \leq \sqrt{3}/2$.
That is, the mapping function for a given unit disk graph $G = (V,E)$ is $\Sigma: V \to (-\infty,+\infty) \times (0,c)$ where $0 \leq c \leq \sqrt{3}/2$.
Later on, Hayashi et al. introduced the class of \emph{thin strip graphs} \cite{thinStrip}.
A graph is a thin strip graph, if it is a $c$-strip graph for every $c>0$.
The main result of \cite{thinStrip} is that no constant $t$ exists where the $t$-strip graphs are exactly thin strip graphs.

In our work, we introduce \emph{axes-parallel unit disk graphs}.
Our domain is restricted to the set of straight lines given by their equations.
Given a simple graph, and a set of straight lines, we ask the question ``can this graph be realized as unit disks on the given straight lines?''
The answer to this question is ``yes'' if the given graph can be realized as unit disks, onto the straight lines that are given as input.
We show that even though these lines are restricted to be parallel to either $x$-axis or $y$-axis, it is NP-hard to recognize whether $G$ is an axes-parallel unit disk graph.

\section{APUD(k,m) recognition} \label{sec:nphard}

We prove that axes-parallel unit disk recognition ($\APUD(k,m)$ recognition with $k$ and $m$ given as input) is NP-hard by giving reduction from the \emph{Monotone not-all-equal 3-satisfiability} (NAE3SAT) problem\footnote{This problem is equivalent to the 2-coloring of 3-uniform hypergraphs. We choose to give the reduction from Monotone NAE3SAT as it is more intuitive to construct for our problem}.
NAE3SAT is a variation of 3SAT a satisfying assignment must not give the thee literals in each clause the same truth value, and due to Schaefer's dichotomy theory, The problem remains NP-complete when all clauses are monotone (i.e. none of the literals are negated) \cite{Schaefer_complexitySAT}.
Our main theorem is as follows.

\begin{thm} \label{thm:apud}
	There is a polynomial-time reduction of any instance $\Phi$ of Monotone NAE3SAT to some instance $\Psi$ of $\APUD(k,m)$ such that $\Phi$ is a YES-instance if, and only if $\Psi$ is a YES-instance.
\end{thm}

We construct our hardness proof using the scheme called a \emph{logic engine}, which is used to prove the hardness of several geometric problems \cite{logicEngine}.
For a given instance $\Phi$ of Monotone NAE3SAT, there are two main components in our reduction.
First, we construct a backbone for our gadget.
The backbone models only the number of clauses and the number of literals.
Next, we model the relationship between the clauses and literals, i.e. which literal appears in which clause.

Let us begin by describing the input graph.
For the sake of simplicity, we assume that the given formula has $3$ clauses, $A, B, C$, and $4$ literals, $q,r,s,t$ for the moment.
In general, we denote the clauses by $C_1, \dots, C_k$, and the literals by $x_1, \dots, x_m$.
Later on, with the help of Figure~\ref{fig:configuration}, we explain how to generalize the input graph according to any given instance of Monotone NAE3SAT formula.
For the following part, we describe the input graph given in Figure~\ref{fig:skeleton}.
Throughout the manuscript, we index the vertices from left to right, and from bottom to top, in ascending order.

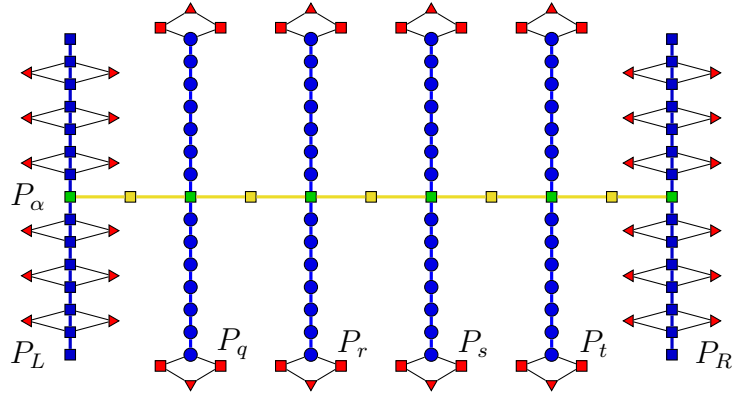
\begin{figure}[htbp]
	\centering
	\begin{tikzpicture}[xscale=0.8, yscale=0.3,
		triangle/.style = {regular polygon, regular polygon sides=3},
		ltriangle/.style = {regular polygon, regular polygon sides=3, rotate=90},
		rtriangle/.style = {regular polygon, regular polygon sides=3, rotate=270},
		dtriangle/.style = {regular polygon, regular polygon sides=3, rotate=180}]
		
		\node[draw=none] at (-0.7,-7) {$P_L$};
		\node[draw=none] at (10.7,-7) {$P_R$};
		\node[draw=none] at (-0.7,0) {$P_\alpha$};
		\node[draw=none] at (2.7,-6.5) {$P_q$};
		\node[draw=none] at (4.7,-6.5) {$P_r$};
		\node[draw=none] at (6.7,-6.5) {$P_s$};
		\node[draw=none] at (8.7,-6.5) {$P_t$};

		\tikzstyle{every node}=[draw, shape=circle, minimum size=5pt,inner sep=0pt];

		\foreach \i in {0,2,4,6,8,10} \node[shape=rectangle, fill=green!80!black, minimum size=4pt] (\i0) at (\i,0) {};
		
		\foreach \i in {1,3,5,7,9}	\node[shape=rectangle, fill=yellow!90!black, minimum size=4pt] (\i0) at (\i,0) {};
		
		\foreach \i in {0,...,9}
		{
			\pgfmathtruncatemacro\j{\i+1};
			\draw[very thick, color=yellow!90!black] (\i0)--(\j0);
		}
		
		\foreach \i in {2,4,6,8} %vertices of literal lines
		{
			\foreach \j in {-7,-6,-5,-4,-3,-2,-1,1,2,3,4,5,6,7}
			{
				\node[fill=blue!90!black] (\i\j) at (\i,\j) {};
				
			}
		}
		
		\foreach \i in {2,4,6,8} %edges of literal lines
		{
			\foreach \j in {-7,-6,-5,-4,-3,-2,-1,0,1,2,3,4,5,6}
			{
				\pgfmathtruncatemacro\k{\j+1};
				\draw[very thick, color=blue] (\i\j)--(\i\k);
				
			}
		}
		
		\foreach \i in {0,10} %vertices on L and R
		{
			\foreach \j in {-7,-6,-5,-4,-3,-2,-1,1,2,3,4,5,6,7}
			{
				\node[shape=rectangle,fill=blue!80!black, minimum size=4pt] (\i\j) at (\i,\j) {};
				
			}
		}
		
		\foreach \i in {0,10} %edges on L and R
		{
			\foreach \j in {-7,-6,-5,-4,-3,-2,-1,0,1,2,3,4,5,6}
			{
				\pgfmathtruncatemacro\k{\j+1};
				\draw[very thick, color=blue!80!black] (\i\j)--(\i\k);
				
			}
		}

		\foreach \i in {-5.5,-3.5,-1.5,1.5,3.5,5.5}
		{
			%These two lines are required since tikz cannot parse floating point numbers as variable names
			\pgfmathtruncatemacro\j{\i-0.5}; 
			\pgfmathtruncatemacro\k{\i+0.5};
			%These two lines are required since tikz cannot parse floating point numbers as variable names

			\node[ltriangle, fill=red] (-1\j) at (-0.7,\i) {};		
			\node[rtriangle, fill=red] (1\j) at (0.7,\i) {};
			
			\node[ltriangle, fill=red] (9\j) at (9.3,\i) {};
			\node[rtriangle, fill=red] (11\j) at (10.7,\i) {};
			\draw (1\j)--(0\j);
			\draw (1\j)--(0\k);
			\draw (-1\j)--(0\j);
			\draw (-1\j)--(0\k);
			\draw (9\j)--(10\j);
			\draw (9\j)--(10\k);
			\draw (11\j)--(10\j);
			\draw (11\j)--(10\k);
		}

		\foreach \i in {2,4,6,8} %4-cycles top
		{
			\node[triangle, fill=red] (\i9) at (\i,8.3) {};
			\node[dtriangle, fill=red] (\i-9) at (\i,-8.3) {};
		}
		
		\foreach \i in {1.5,2.5,3.5,4.5,5.5,6.5,7.5,8.5} %4-cycles middle
		{
			\pgfmathtruncatemacro\j{\i+0.5}
			\node[shape=rectangle, fill=red, minimum size=4pt] (\j8) at (\i,7.5) {};
			\node[shape=rectangle, fill=red, minimum size=4pt] (\j-8) at (\i,-7.5) {};
		}
		
		\foreach \i in {2,4,6,8} %4-cycles edges
		{
			\pgfmathtruncatemacro\j{\i+1}
			\draw (\i7)--(\i8)--(\i9)--(\j8)--(\i7);
			\draw (\i-7)--(\i-8)--(\i-9)--(\j-8)--(\i-7);
		}
		
	\end{tikzpicture}
	\caption{ Skeleton of the input graph for $\Phi$. The consecutive induced paths, labeled as $P_q, P_r, P_s, P_t$, are to be embedded on the literal lines $\ell_q, \ell_r, \ell_s, \ell_t$ in Figure~\ref{fig:frame} respectively. The vertices in the long induced paths $P_L$ and $P_R$ in \ref{fig:skeleton} (indicated by rectangles) must be embedded on the lines $\ell_L$ and $\ell_R$ given in \ref{fig:frame}. Similarly, the vertices in $P_\alpha$ (indicated by blue and green rectangles) must be embedded on the line $\ell_\alpha$ given in \ref{fig:frame}.}
	\label{fig:skeleton}
\end{figure}

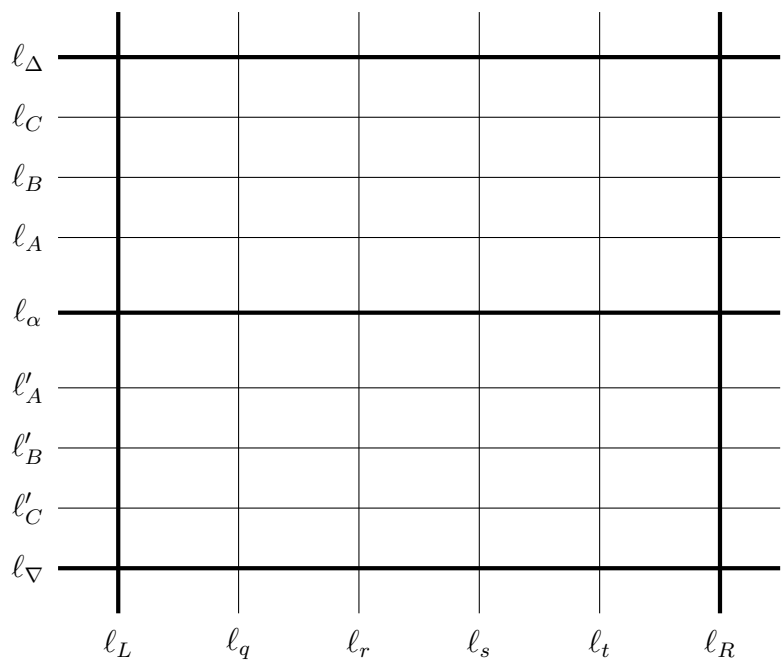
\begin{figure}[htbp]
	\centering
	\begin{tikzpicture}[scale=0.4]
		
		\draw[ultra thick] (-2,8.5) -- (22, 8.5); %Top
		\node[draw=none] at (-3,8.5) {$\ell_\Delta$};
		\draw[ultra thick] (-2,0) -- (22,0); %Middle
		\node[draw=none] at (-3,0) {$\ell_\alpha$};
		\draw[ultra thick] (-2,-8.5) -- (22, -8.5); %Bottom
		\node[draw=none] at (-3,-8.5) {$\ell_\nabla$};

		\draw[ultra thick] (0,-10) -- (0, 10); %fixing line left
		\node[draw=none] at (0,-11) {$\ell_L$};
		\draw[ultra thick] (20,-10) -- (20, 10); %fixing line right
		\node[draw=none] at (20,-11) {$\ell_R$};
		
		\node[draw=none] at (-3,2.5) {$\ell_A$};
		\node[draw=none] at (-3,4.5) {$\ell_B$};
		\node[draw=none] at (-3,6.5) {$\ell_C$};
		\node[draw=none] at (-3,-2.5) {$\ell_A'$};
		\node[draw=none] at (-3,-4.5) {$\ell_B'$};
		\node[draw=none] at (-3,-6.5) {$\ell_C'$};
		\node[draw=none] at (4,-11) {$\ell_q$};
		\node[draw=none] at (8,-11) {$\ell_r$};
		\node[draw=none] at (12,-11) {$\ell_s$};
		\node[draw=none] at (16,-11) {$\ell_t$};
		
		\foreach \i in {1,2,3} %horizontal lines
		{
			\draw (-2,2*\i+0.5) -- (22,2*\i+0.5);
			\draw (-2,2*-\i-0.5) -- (22,2*-\i-0.5);
			
		}

		\foreach \i in {1,2,3,4} %vertical lines
		{
			\draw (4*\i,-10) -- (4*\i, 10);
			
		}
	\end{tikzpicture}
	\caption{The line set of the configuration for a Monotone NAE3SAT formula $\Phi$ with $4$ literals ($q,r,s,t$) and $3$ clauses ($A,B,C$).}
	\label{fig:frame}
\end{figure}

Three essential components of the input graph is the following induced paths
$P_\alpha = (\alpha_1, \alpha_2, \dots, \alpha_{11})$, 
$P_L = (L_1, L_2, \dots, L_{15})$, 
and $P_R = (R_1, R_2, \dots, R_{15})$.
The length of $P_\alpha$ is $2m+3$ for $m$ literals. In our case, $(2\times 4) +3 = 11$.
The lengths of $P_L$ and $P_R$ are the same, equal to $3 + 4k$ for $k$ clauses. In our case, $3 + (4\times 3) = 15$.

The middle vertices of $P_L$ and $P_R$ are the end vertices of $P_\alpha$.
That is, $\alpha_1 = L_8$, and $\alpha_{11} = R_8$.
The paths $P_L$ and $P_R$ define the left and the right boundary for our gadget, respectively.

For $i  = q,r,s,t$, there is an induced path $P_i = (i_1, \dots, i_{15})$ for each literal, with 15 vertices.
In general, we denote those paths by $P^1, P^2, \dots, P^m$ for $m$ literals.
The vertices of these paths are denoted by blue circles in the Figure~\ref{fig:skeleton}, they are mutually disjoint, but each of them shares one vertex with $P_\alpha$.
The shared vertices are precisely the middle vertices, those are indicated by green rectangle vertices in the figure.
That is, $\alpha_3 = q_8$, $\alpha_5 = r_8$, $\alpha_7 = s_8$, and $\alpha_9 = t_8$.
Moreover, $i_1$ is a vertex of an induced 4-cycle, and $i_{15}$ is a vertex of another induced 4-cycle for $i = q,r,s,t$.
The three vertices in a 4-cycle, except the one in one of the induced paths, are indicated by the red color in the figure.
Precisely two of them, that are adjacent to a blue vertex (either $i_1$ or $i_{15}$) are indicated by squares, and the remaining is indicated by a triangle.

\begin{figure}[htbp]
	\centering
	\begin{tikzpicture}[scale=0.5]
		
		\node[draw=none] at (-3.5,0) {$\ell_\alpha$};
		%\node[draw=none] at (-3.5,1.9) {$\beta$};
		\node[draw=none] at (0,-12) {$\ell_L$};
		\node[draw=none] at (19,-12) {$\ell_R$};

		\node[draw=none] at (3.8,-12) {$\ell_p$};
		\node[draw=none] at (7.6,-12) {$\ell_q$};
		\node[draw=none] at (11.4,-12) {$\ell_r$};
		\node[draw=none] at (15.2,-12) {$\ell_s$};
		
		\node[draw=none] at (-3.5,-9) {$\ell_\nabla$};
		\node[draw=none] at (-3.5,-7.05) {$\ell'_C$};
		\node[draw=none] at (-3.5,-4.85) {$\ell'_B$};
		\node[draw=none] at (-3.5,-2.65) {$\ell'_A$};
		\node[draw=none] at (-3.5,2.65) {$\ell_A$};
		\node[draw=none] at (-3.5,4.85) {$\ell_B$};
		\node[draw=none] at (-3.5,7.05) {$\ell_C$};
		\node[draw=none] at (-3.5,9) {$\ell_\Delta$};
		
		\tikzstyle{every node}=[draw, fill=black, shape=circle, minimum size=2pt,inner sep=0pt];

		\draw[thick] (-3,0)--(22,0);
		%\draw (-3,1.9)--(22,1.9);
		
		\foreach \i in {0,1.9,3.8,5.7,7.6,9.5,11.4,13.3,15.2,17.1,19} %disks on \alpha
		{
			\node at (\i,0) {};
			\draw (\i,0) circle[radius=1];			
		}

		\foreach \i in {3.8,7.6,11.4,15.2} %disks on literal lines
		{
			\foreach \j in {-7.6,-6.5,-5.4,-4.3,-3.2,-2.1,-1.05,1.05,2.1,3.3, 4.4,5.5, 6.6,7.7}
			{
				\node at (\i,\j) {};
				\draw (\i,\j) circle[radius=1];
				
			}
		}
		
		\foreach \i in {0,19} %disks on R and L
		{
			\foreach \j in {-7.6,-6.5,-5.4,-4.3,-3.2,-1.9,1.9,3.3, 4.4,5.5, 6.6,7.7}
			{
				\node at (\i,\j) {};
				\draw (\i,\j) circle[radius=1];
				
			}
		}
		
		\foreach \i in {0,3.8,7.6,11.4,15.2,19} %literal lines
		\draw (\i,-11.5)--(\i,11.5);	
		
		\foreach \i in {-7.05,-4.85,-2.65,2.65,4.85,7.05} %fixing flags and clause lines
		{	
			\node at (1.5,\i) {};
			\node at (-1.5,\i) {};
			\node at (17.5,\i) {};
			\node at (20.5,\i) {};		
			\draw (1.5,\i) circle[radius=1];
			\draw (-1.5,\i) circle[radius=1];
			\draw (17.5,\i) circle[radius=1];
			\draw (20.5,\i) circle[radius=1];
			\draw (-3,\i)--(22,\i);
		}
		
		\draw[thick] (-3,9)--(22,9); %top squeezing line
		\draw[thick] (-3,-9)--(22,-9); %bottom squeezing line
		
		\foreach \i in {0,2.32, 4.38, 6.43, 8.48, 10.53, 12.58, 14.63, 16.68,19}
		{
			\node at (\i,9) {};
			\draw (\i,9) circle[radius=1];
			\node at (\i,-9) {};
			\draw (\i,-9) circle[radius=1];
		}
		
		\foreach \i in {3.8,7.6,11.4,15.2} %top disks of 4-cycles
		{
			\foreach \j in {-9.8,9.8}
			{
				\node at (\i,\j) {};
				\draw (\i,\j) circle[radius=1];
			}
		}

	\end{tikzpicture}
	\caption{Realization of the graph given in \ref{fig:skeleton} onto the lines given in \ref{fig:frame}.}
	\label{fig:skeletonRealization}
\end{figure}
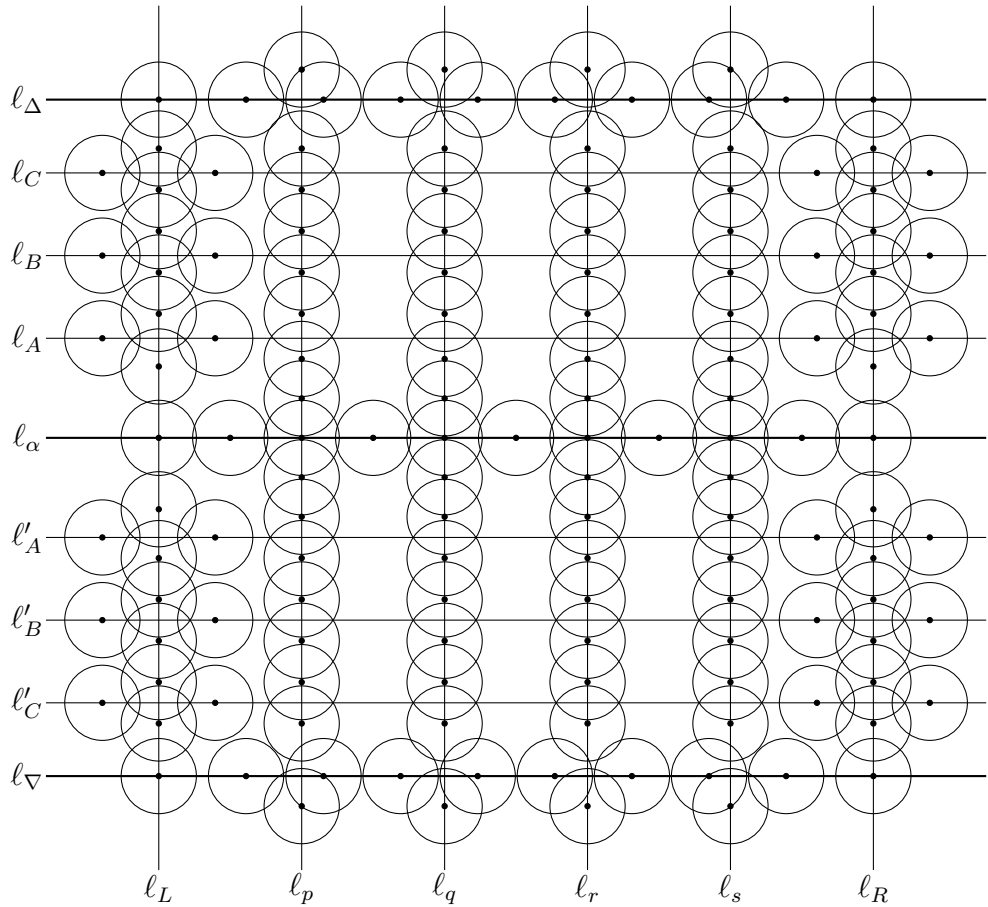

Starting from the second edge of $P_L$ (respectively $P_R$), every second edge is a chord of a 4-cycle ($C_4$).
Throughout the paper, we refer to such 4-cycles with a chord as a \emph{diamond}.
Two vertices of these diamonds are of $P_L$ (respectively $P_R$), and remaining two are denoted by red triangles in Figure~\ref{fig:skeleton}.

Remember that the problem takes two inputs: a graph, and a set of lines determined by their equations (or rather by two sets of rational numbers, since every line is parallel to either $x$- or $y$- axis).
For a Monotone NAE3SAT formula with $3$ clauses and $4$ literals, we have described the input graph above.
Now, let us discuss the input lines of our gadget.
The input graph is given in Figure~\ref{fig:skeleton}, and the corresponding lines are given in Figure~\ref{fig:frame}.
We claim that the given graph can be embedded onto the given lines with $\varepsilon$ flexibility, and the resulting realization looks like the set of unit disks given in Figure~\ref{fig:skeletonRealization}.

In order to force such embedding, we adjust the Euclidean distance between each pair of parallel lines carefully.
We start by defining the horizontal line $\ell_\alpha$.
This line is the axis of horizontal symmetry for our line configuration.
Thus, it is safe to assume that $\ell_\alpha$ is the $x$-axis.
On the positive side of the $y$-axis, for each clause $A$, $B$, and $C$, there is a straight line parallel to $\ell_\alpha$, and another horizontal line acting as the top boundary of the configuration.
These lines are denoted by $\ell_A, \ell_B$, $\ell_C$, and $\ell_\Delta$, and their equations are $y = a$, $y = b$, $y = c$, and $y = \Delta$, respectively, where $a < b < c < \Delta$.
For every pair of consecutive horizontal lines, the Euclidean distance between them is precisely $2.01$ units.
That is, $a = 2.01$, $b = 4.02$ and $c = 6.03$, and $\Delta = 8.04$.
Note that $2.01$ is sufficiently small for the case where there are three clauses.
In general, this value is $(2 + \varepsilon)$ where $\varepsilon \leq 1/4k$.
For every horizontal line described above, there is another horizontal line symmetric to it about the $x$-axis.
These lines are $\ell'_A, \ell'_B, \ell'_C$, and $\ell_\nabla$ (see Figure~\ref{fig:frame}).

The leftmost vertical line is $\ell_L$, which is the left boundary of our configuration.
We can safely assume that $\ell_L$ is the $y$-axis for the sake of simplicity.
For each literal $q$, $r$, $s$ and $t$, there exists a vertical line parallel to $\ell_L$, and another vertical line that defines the right boundary of our configuration.
These lines are denoted by $\ell_q, \ell_r$, $\ell_s$, $\ell_t$, and $\ell_R$, and their equations are $x = q$, $x = r$, $x = s$, $x = t$ and $x = R$, respectively, where $q < r < s < t < R$.
The Euclidean distance between each pair of consecutive vertical lines is precisely $3.8$ units.
That is $q = 3.8$, $r = 7.6$, $s = 11.4$, $t = 15.2$, and $R = 19$.

Up to this point, we have described the input graph, and the input lines for a given Monotone NAE3SAT formula with 3 clauses and 4 literals.
In general, for a given Monotone NAE3SAT formula $\Phi$ with $k$ clauses $C_1, C_2, \dots, C_k$, and $m$ literals $x_1, x_2, \dots, x_m$, the $\APUD(2k+1, m+2)$ instance has the components that are listed in Figure~\ref{fig:configuration}.

\begin{figure}
	\begin{mdframed}
		\begin{enumerate}
			\item An induced path $P_\alpha = (\alpha_1, \alpha_2, \dots, \alpha_{2m+3})$ with $2m+3$ vertices.
			
			\item $m$ induced paths $P^1 = (P^1_1, P^1_2, \dots, P^1_{4k+3})$, $\dots$, $P^m = (P^m_1, P^m_2, \dots, P^m_{4k+3})$, each with $4k+3$ vertices, where $\alpha_3 = P^1_{2k+2}$, $\alpha_5 = P^2_{2k+2}$, $\dots$, $\alpha_{2k+1} = P^m_{2k+2}$, and induced 4-cycles containing the first and the last vertices of each of these paths.
			
			\item Two induced paths $P_L = (L_1, \dots, L_{4k+3})$ and $P_R = (R_1, \dots, R_{4k+3})$, each with $4k+3$ vertices, where the edges 
			$L_2L_3$, $L_4L_5$, $\dots$, $L_{2k}L_{2k+1}$, $L_{2k+3}L_{2k+4}$, $\dots$, $L_{4k+1}L_{4k+2}$, and
			
			$R_2R_3$, $R_4R_5$, $\dots$, $R_{2k}R_{2k+1}$, $R_{2k+3}R_{2k+4}$, $\dots$, $R_{4k+1}R_{4k+2}$ are chords of disjoint 4-cycles.
			\item $2k+3$ horizontal lines $\ell_\nabla, \ell'^C_k, \ell'^C_{k-1}, \dots, \ell_\alpha, \ell^C_1, \ell^C_2, \dots, \ell^C_k, \ell_\Delta$, with equations
			
			$\ell_\nabla: y= -(2+\varepsilon)(k+1)$
			
			$\ell_\Delta: -\ell_\nabla$
			
			$\ell_\alpha: y = 0$
			
			$\ell'^C_{i} = -(2+\varepsilon)i$, 
			
			and $\ell^C_{i} = (2+\varepsilon)i$ for $i = 1,2,\dots,k$  and for $\varepsilon \leq 1/4k$.
			\item $m+2$ vertical lines $\ell_L, \ell^x_1, \ell^x_2, \dots, \ell^x_m, \ell_R$, with equations $\ell_L: x = 0$, 
			
			$\ell_R: x = 3.8(m+1)$, 
			
			and $\ell^x_i: 3.8i$ for $i = 1,2,\dots,m$.
		\end{enumerate}
	\end{mdframed}
	\caption{The straight line configuration for the gadget that corresponds to a NAE3SAT formula with $k$ clauses and $m$ lines.}
	\label{fig:configuration}
\end{figure}

In total, for the given formula $\Phi$ with $k$ clauses and $m$ literals, our gadget is an instance of $\APUD(2k+3,m+2)$.

Now, let us show that the given graph has a unique embedding onto the given lines, up to $\varepsilon$ flexibility.

\begin{prp} \label{prp:boldLines}
	The vertices indicated by rectangles in Figure~\ref{fig:skeleton} can only be embedded on the bold lines in Figure~\ref{fig:frame}.
\end{prp}
We prove this proposition by giving a series of lemmas.

Let us start by discussing the embedding of $P_L$ onto $\ell_L$ (and respectively $P_R$ onto $\ell_R$).
We give the following two trivial lemmas as preliminaries for the proof of our claim.

\begin{lem} \label{lem:triangle}
	Consider two disks $A$ and $B$, centered on $(a,0)$ and $(b,0)$ with $0<|a|<|b|$.
	Another disk, $C$ that is centered on $(0,c)$ cannot intersect $B$ without intersecting $A$.
\end{lem}

\begin{proof}
	
	Consider the triangle whose corners are $(a,0)$, $(b,0)$, and $(0,c)$.
	If $|a|<|b|$, then $\sqrt{a^2+c^2} < \sqrt{b^2+c^2}$ holds.
	For $C$ to intersect $B$, $\sqrt{b^2+c^2} \leq 2$ must hold.
	However, since $|a| < |b|$, if $\sqrt{b^2+c^2} \leq 2$ holds, then $\sqrt{a^2+c^2} \leq 2$ also holds.
	Thus, $C$ intersects $B$ if, and only if $C$ intersects $A$. 
\end{proof}

\begin{lem} \label{lem:s4}
	An induced 4-star ($K_{1,4}$) can be realized as a unit disk graph on two perpendicular lines, but not on two parallel lines.
\end{lem}

\begin{proof}
	
	First part of the proof is trivial.
	Consider the $x=0$ and $y=0$ lines as two perpendicular lines.
	Four unit disks centered on $(0,-(2-\varepsilon)$, $((2-\varepsilon),0)$, $(0, (2-\varepsilon)$, and $(-(2-\varepsilon),0)$ where $0 < \varepsilon \ll 1$ form an induced 4-star ($K_{1,4}$).
	
	Now, let us show that an induced 4-star cannot be realized as unit disks on two parallel lines.
	The disks that correspond to the vertices of an induced claw ($K_{1,3}$) can be embedded on two parallel lines if the centers of three disks are collinear.
	Suppose that four disks, $A$, $B$, $C$, and $U$ form an induced claw, where $U$ is the central vertex, and $A,B,C$ are the rays.
	Without loss of generality, suppose that $A$, $U$, and $C$ are centered on $(a,0)$, $(0,0)$ and $(c,0)$, respectively.
	Thus, $B$ must be on the second parallel line, say $y=k$.
	To complete a 4-star, we need one more disk, $D$, centered on either $y=0$ or $y=k$, such that $D$ intersects $U$, but none of $A$, $B$, and $C$ 
	
	Clearly, $D$ cannot be on $y=0$ line because $U$ is enclosed by $A$ and $C$ from both sides.
	So, suppose that $D$ is centered on $(d,k)$.
	In this case, we show that no such $k$ exists by contradiction.
	
	Place two more disks, $B'$ and $D'$, centered on $(-b,k)$ and $(-d,k)$, respectively.
	If $B$ and $D$ do not intersect, then $B'$ and $D'$ also do not intersect.
	Moreover, since $A$ does not intersect with $B$, $A$ also does not intersect with $B'$.
	Symmetrically, $C$, $D$, $C'$, and $D'$ have no pairwise intersections.
	
	The described configuration is a $K_{1,6}$ with disks $U;A,B,C,D,B',D'$ which cannot be realized as a unit disk graph.
	Therefore, we have a contradiction. 
\end{proof}

\begin{lem} \label{lem:c4}
	If an induced 4-cycle ($C_4$) can be realized as a unit disk graph on two perpendicular lines that intersect at $\omega$, then the following hold:
	\begin{enumerate}[i)]
		\item A pair of non-intersecting disks must be on the same line, but at the different sides of $\omega$.
		\item The the Euclidean distance between $\omega$ and any disk center is strictly greater than $2 - \sqrt{3}$ and strictly less than $\sqrt{3}$.
	\end{enumerate}
\end{lem}
\begin{proof}
	Without loss of generality, let two perpendicular lines be the $x$- and $y$-axes, and thus let $\omega$ be $(0,0)$.
	Denote the vertices of 4-cycle by $a,b,c,d$ in clockwise order.
	That is, the edges $ab$, $bc$, $cd$, and $ad$ exist.
	Denote the disks which correspond to these vertices by $A$, $B$, $C$, and $D$.
	
	Suppose that two non-intersecting disks, $A$ and $C$ are centered at $(a,0)$ and $(c,0)$, respectively where $0 \leq a \leq c$.
	Since $c-a > 2$ and $a > 0$, then $c > 2$, which means that any disk which intersects $C$ cannot be on $y$-axis.
	However, both $B$ and $D$ cannot be on $x$-axis because then we have a unit interval graph, in which an induced 4-cycle cannot appear.
	Thus, if $ac < 0$ must hold and symmetrically, $bd <0$ must also hold.
	Therefore (i) holds.
	
	Up to symmetry, assume that $A$, $B$, $C$ and $D$ are centered at $(a,0)$, $(0,b)$, $(c,0)$ and $(0,d)$, respectively, where $a,c < 0$ and $b,d > 0$.
	Then, the following equations hold:
	\begin{align}
		|a - c| > 2\\
		|b - d| > 2\\
		\sqrt{a^2+b^2} \leq 2\\
		\sqrt{b^2+c^2} \leq 2\\
		\sqrt{c^2+d^2} \leq 2\\
		\sqrt{d^2+a^2} \leq 2
	\end{align}
	Let us minimize the value of $|d|$, which is simply the Euclidean distance between the intersection point of two perpendicular lines and the center of $D$.
	Due to (2), $|d|$ is minimized when $|b|$ is maximized.
	Due to (3) and (4), the maximum value of $|b|$ is achieved when $|a|-|c|$ is minimum.
	Due to (1), $|a| - |c|$ is minimum when $|a| = |c| = 1 + 2\sqrt{\varepsilon}$ where $\varepsilon$ is sufficiently small.
	In order to maximize the value of $|b|$, we must assign $a = -(1 + \sqrt{\varepsilon})$, and $c = 1 + \sqrt{\varepsilon}$. 
	Then, the value of $b$ is $\sqrt{3 - \varepsilon}$ which is strictly less than $\sqrt{3}$, and due to (2), the value of $d$ is strictly greater than $2 - \sqrt{3}$.
	Since we maximized the value of $b$ and minimized the value of $d$, we can now conclude that (ii) holds.

\end{proof}

Now, with the help of Lemmas~\ref{lem:triangle} and \ref{lem:s4}, we state the following lemmas, and finish the proof of Proposition~\ref{prp:boldLines}.

\begin{lem} \label{lem:LRa}
	The induced paths $P_L$, $P_R$ and $P_\alpha$ in the input graph (Figure~\ref{fig:skeleton}) can only be embedded onto $\ell_L$, $\ell_R$, and $\ell_\alpha$, respectively (Figure~\ref{fig:frame}).
\end{lem}

\begin{proof}
	Assume that $P_L = (L_1, \dots, L_{4k+3})$ is realized on a single line, say $d$.
	Denote the disks $\mathcal{D}_L^1, \dots \mathcal{D}_L^{4k+3}$ that correspond to these vertices, respectively.
	Denote the disks $\mathcal{D}_\alpha^1, \dots, \mathcal{D}_\alpha^{2m+3}$ that correspond to the vertices ($\alpha_1, \dots, \alpha_{2m+3}$) of $P_\alpha$.
	
	Consider the diamond $u,v,L_{4k-1}, L_{4k}$, where $u$ and $v$ are indicated by red triangles in Figure~\ref{fig:skeleton}.
	Neither $U$ nor $V$ intersect any disk that corresponds to a vertex of $P_L$.
	Therefore, neither of them can be embedded onto $d$.
	Similarly, if they are embedded onto some other line perpendicular to $d$, then they intersect either $\mathcal{D}_L^{4k-2}$ or $\mathcal{D}_L^{4k+1}$ to intersect $\mathcal{D}_L^{4k-1}$ and $\mathcal{D}_L^{4k}$.
	Thus, every such pair of disks must be embedded in a way such that their centers lie onto the same line that intersects $d$, and at the different sides of $d$.
	This property also holds for the diamond that includes the vertices $L_{4k+1}$ and $L_{4k+2}$, as there is one extra vertex $L_{4k+3}$ at the end of $P_L$.
	
	There are $2k$ disjoint diamonds in which $4k$ vertices of $P_L$ are included.
	Each of these diamonds must be embedded around the intersection of two perpendicular lines.
	In addition, there is one vertex that $P_L$ and $P_\alpha$ have in common ($\alpha_1 = L_{2k+2}$).
	Since $P_\alpha$ is an induced path, each disk $\mathcal{D}_\alpha^i$ such that $1 \leq i \leq 2m+3$ cannot be centered on $d$.
	Thus, the center of the disk $\mathcal{D}_L^{2k+2} = \mathcal{D}_\alpha^1$, must be the closest center to some line $f$ on which $\mathcal{D}_\alpha^2$ is centered.
	Otherwise, due to Lemma~\ref{lem:triangle}, $\mathcal{D}_\alpha^2$ intersects with a closer to $f$.
	
	Therefore, to embed $P_L$, there must be a straight line, and precisely $2k+1$ parallel lines that are perpendicular to this line.
	Similarly, $P_R$ also requires $2k+1$ intersection points to be realized.
	Thus, both $P_L$ and $P_R$ must be realized on two distinct vertical lines.
	
	Now recall Lemma~\ref{lem:s4}.
	Note that each one of $\alpha_3, \alpha_5, \dots \alpha_{2m+1}$ is the central vertex of some 4-star, and thus must be embedded near the intersection of two perpendicular lines.
	Considering that $P_L$ and $P_R$ are realized on two vertical lines, we need $m+2$ intersection points to realize $P_\alpha$ as unit disks.
	Since $\alpha_1 = L_{2k+2}$ and $\alpha_{2m+3} = R_{2k+2}$, the remaining four intersections for 4-stars are the intersections between $\ell_alpha$ and each of $\ell_q, \ell_r, \ell_s, \ell_t$.
	It is now trivial to see that our initial assumption of $P_L$ being realized on a single line always holds, since if $P_L$ is realized on more than one straight lines, then $P_\alpha$ cannot be realized due to the number of intersections.
	Of course, one also needs to take into account the other induced paths, $P^1, \dots, P^m$ to verify this statement.
	
	Therefore, $P_L$, $P_R$ and $P_\alpha$ can only be embedded onto $\ell_L$, $\ell_R$, and $\ell_\alpha$, respectively. 
\end{proof}

\begin{lem} \label{lem:4cycles}
	A pair of non-intersecting disks that are included in an induced 4-cycle, but not included in any of $P^1, P^2, \dots, P^m$, must lie on either $\ell_\Delta$ or $\ell_\nabla$ (red rectangles in Figure~\ref{fig:skeleton}). 
\end{lem}
\begin{proof}
	In any 4-cycle, two of the disks must be centered on the same line, and the remaining two must be centered a line which is perpendicular to the line on which the first two are centered.
	Moreover, the centers must lie at four different directions from the intersection point.
	
	Therefore, a non-intersecting pair of disks from each induced 4-cycle must be centered on $\ell_\Delta$ and $\ell_\nabla$.
	The vertices to which those pairs of disks correspond are indicated by red rectangles in Figure~\ref{fig:skeleton}. 
\end{proof}

The following three propositions altogether show that for a given instance of NAE3SAT with $k$ clauses and $m$ literals, the corresponding graph has a unique embedding up to $\varepsilon$ flexibility onto the corresponding line configuration.
Proposition~\ref{lem:inducedPaths} shows that the long induced paths $P^1, \dots, P^m$ in the input graph must be embedded onto the vertical lines.
Proposition~\ref{lem:squeeze} shows that on a vertical line, the disks which correspond to these paths must be embedded between the top line $\ell_\Delta$ and the bottom line $\ell_\nabla$ in the line configuration.

\begin{lem} \label{lem:inducedPaths}
	The induced paths
	$P^1 = (P^1_1, \dots, P^1_{4k+3})$,
	$P^2 = (P^2_1, \dots, P^2_{4k+3})$, $\dots$,
	$P^m = (P^m_1, \dots, P^m_{4k+3})$ in the input graph can only be embedded onto $\ell^x_1$, $\ell^x_2$, $\dots$ $\ell^x_m$, respectively.
\end{lem}
\begin{proof}
	Due to Lemma~\ref{lem:LRa}, we know that $P_\alpha$ is realized on $\ell_\alpha$.
	Also recall that for $1 \leq i \leq k$; $\alpha_{2i+1} = P^i_{2k+2}$ holds.
	Since every $P^i_{2k+2}$ is a central vertex of some 4-star, as well as the middle vertex of the path $P^i$, the induced path $P^1, P^2, \dots, P^m$ must be realized on $\ell^x_1, \ell^x_2, dots, \ell^x_m$, respectively.
\end{proof}

Before stating the next proposition, let us remind that a 4-cycle is an induced subgraph in a unit interval graph and hence can be realized as a unit disk intersection graph onto two distinct straight lines.
In our case, two lines are either parallel or perpendicular.
Since every pair of parallel lines has the Euclidean distance strictly greater than $2$ units in between, it is clear that we have to embed an induced 4-cycle onto two perpendicular lines.
Moreover, considering the clockwise order $a,b,c,d$ on an induced 4-cycle, remember that by Lemma~\ref{lem:c4}, the vertex pairs $a,c$ must be on one of the perpendicular lines, and $b,d$ must be on the other line.

\begin{lem} \label{lem:squeeze}
	The center of each disk that correspond to a vertex of the induced paths $P^1, \dots, P^m$ must be between $\ell_\Delta$ and $\ell_\nabla$. 
\end{lem}
\begin{proof}
	By Lemma~\ref{lem:inducedPaths}, we already know that the disk which corresponds to the middle vertex $P^i_{2k+2}$ of some path $P^i$ must be on the line $\ell_\alpha$.
	
	Consider the $2k+2$ disks which correspond to the vertices $P^i_{2k+2}, \dots, P^i_{4k+3}$ for some $i$.
	Let us denote these disks by $D_1, \dots, D_{2k+2}$.
	Note that $P^i_{2k+2}$ is the central vertex of a 4-star, and $P^i_{4k+3}$ is a vertex of a 4-cycle.
	That is, the disks which correspond to$P^i_{2k+3}$
	These disks yield an induced path, and thus the Euclidean distance between the centers of the first disk and the last disk on the path must be strictly greater than $2k$ units.
	Thus, the disk which corresponds to the vertex $P^i_{4k+3}$ must be embedded above the horizontal line  $\ell^C_{k-1} : y = (2+\varepsilon)(k-1)$ which is strictly less than $2k$ for a sufficiently small value of $\varepsilon$.
	
	Let us denote the disk which corresponds to the vertex $P^i_{4k+3}$ by $\mathcal{D}$, and $y$-coordinate of the center of that disk by $\psi$.
	By Lemma~\ref{lem:c4}, in order to form an induced 4-cycle, $\mathcal{D}$ must be below the horizontal line $\ell^C_k : y = (2+\varepsilon)k$, and the Euclidean distance between the center of $\mathcal{D}$ and $\ell^C_k$ must be greater than $2- \sqrt{3}$.
	\begin{align*}
		\psi & > 2k\\
		(2+\varepsilon)k - \psi & > 2 - \sqrt{3}
	\end{align*}
	Which can be rewritten as
	\begin{align*}
		2k &< \psi < (2+\varepsilon)k - 2 + \sqrt{3}\\
		2k &< (2+\varepsilon)k - 2 + \sqrt{3}\\
		2 - \sqrt{3} &< \varepsilon k
	\end{align*}
	which is a contradiction for $\varepsilon \leq 1/4k$.
\end{proof}

\begin{prp} \label{prp:KclausesMliterals}
	For the given input graph that corresponds to a NAE3SAT formula with $k$ clauses and $m$ literals, the following hold:
	\begin{enumerate}[i)]
		\item A pair of non-intersecting disks that are included in an induced 4-cycle, but not included in any of $P^1, P^2, \dots, P^n$, must lie on either $\ell_\Delta$ or $\ell_\nabla$. 
		
		\item The induced paths
		$P^1 = (P^1_1, \dots, P^1_{4k+3})$,
		$P^2 = (P^2_1, \dots, P^2_{4k+3})$, $\dots$,
		$P^n = (P^m_1, \dots, P^m_{4k+3})$ in the input graph can only be embedded onto $\ell^x_1$, $\ell^x_2$, $\dots$ $\ell^x_m$, respectively.
		
		\item The center of each disk that correspond to a vertex of those induced paths must be between $\ell_\Delta$ and $\ell_\nabla$. 
		
	\end{enumerate}
\end{prp}

\begin{proof}
	Directly follows from Lemma~\ref{lem:4cycles}, Lemma~\ref{lem:inducedPaths} and Lemma~\ref{lem:squeeze}.
\end{proof}

With the Lemma~\ref{lem:LRa} and Proposition~\ref{prp:KclausesMliterals}, we have shown that the vertices denoted by rectangles in Figure~\ref{fig:skeleton} must be embedded onto the bold lines in Figure~\ref{fig:frame}, which was stated in Proposition~\ref{prp:boldLines}.
For a quick visual verification of these lemmas, we refer the reader to the Figures \ref{fig:skeleton} and \ref{fig:frame} with values $k=3$ and $m=4$.

\begin{figure}[htbp]
	\centering
	
	\subfloat[]{
		\centering
		
		\begin{tikzpicture}[xscale=0.45, yscale=0.3,
			triangle/.style = {regular polygon, regular polygon sides=3},
			ltriangle/.style = {regular polygon, regular polygon sides=3, rotate=90},
			rtriangle/.style = {regular polygon, regular polygon sides=3, rotate=270},
			dtriangle/.style = {regular polygon, regular polygon sides=3, rotate=180}]
			
			\node[draw=none] at (-0.7,-7) {$P_L$};
			\node[draw=none] at (10.7,-7) {$P_R$};
			\node[draw=none] at (-0.7,0) {$P_\alpha$};
			\node[draw=none] at (2.7,-6.5) {$P_q$};
			\node[draw=none] at (4.7,-6.5) {$P_r$};
			\node[draw=none] at (6.7,-6.5) {$P_s$};
			\node[draw=none] at (8.7,-6.5) {$P_t$};

			\tikzstyle{every node}=[draw, shape=circle, minimum size=5pt,inner sep=0pt];

			\foreach \i in {0,2,4,6,8,10} \node[shape=rectangle, fill=green!80!black, minimum size=4pt] (\i0) at (\i,0) {};
			
			\foreach \i in {1,3,5,7,9}	\node[shape=rectangle, fill=yellow!90!black, minimum size=4pt] (\i0) at (\i,0) {};
			
			\foreach \i in {0,...,9}
			{
				\pgfmathtruncatemacro\j{\i+1};
				\draw[very thick, color=yellow!90!black] (\i0)--(\j0);
			}
			
			\foreach \i in {2,4,6,8} %vertices of literal lines
			{
				\foreach \j in {-7,-6,-5,-4,-3,-2,-1,1,2,3,4,5,6,7}
				{
					\node[fill=blue!90!black] (\i\j) at (\i,\j) {};
					
				}
			}
			
			\foreach \i in {2,4,6,8} %edges of literal lines
			{
				\foreach \j in {-7,-6,-5,-4,-3,-2,-1,0,1,2,3,4,5,6}
				{
					\pgfmathtruncatemacro\k{\j+1};
					\draw[very thick, color=blue] (\i\j)--(\i\k);
					
				}
			}
			
			\foreach \i in {0,10} %vertices on L and R
			{
				\foreach \j in {-7,-6,-5,-4,-3,-2,-1,1,2,3,4,5,6,7}
				{
					\node[shape=rectangle,fill=blue!80!black, minimum size=4pt] (\i\j) at (\i,\j) {};
					
				}
			}
			
			\foreach \i in {0,10} %edges on L and R
			{
				\foreach \j in {-7,-6,-5,-4,-3,-2,-1,0,1,2,3,4,5,6}
				{
					\pgfmathtruncatemacro\k{\j+1};
					\draw[very thick, color=blue!80!black] (\i\j)--(\i\k);
					
				}
			}

			\foreach \i in {-5.5,-3.5,-1.5,1.5,3.5,5.5}
			{
				%These two lines are required since tikz cannot parse floating point numbers as variable names
				\pgfmathtruncatemacro\j{\i-0.5}; 
				\pgfmathtruncatemacro\k{\i+0.5};
				%These two lines are required since tikz cannot parse floating point numbers as variable names

				\node[ltriangle, fill=red] (-1\j) at (-0.7,\i) {};		
				\node[rtriangle, fill=red] (1\j) at (0.7,\i) {};
				
				\node[ltriangle, fill=red] (9\j) at (9.3,\i) {};
				\node[rtriangle, fill=red] (11\j) at (10.7,\i) {};
				\draw (1\j)--(0\j);
				\draw (1\j)--(0\k);
				\draw (-1\j)--(0\j);
				\draw (-1\j)--(0\k);
				\draw (9\j)--(10\j);
				\draw (9\j)--(10\k);
				\draw (11\j)--(10\j);
				\draw (11\j)--(10\k);
			}

			\foreach \i in {2,4,6,8} %4-cycles top
			{
				\node[triangle, fill=red] (\i9) at (\i,8.3) {};
				\node[dtriangle, fill=red] (\i-9) at (\i,-8.3) {};
			}
			
			\foreach \i in {1.5,2.5,3.5,4.5,5.5,6.5,7.5,8.5} %4-cycles middle
			{
				\pgfmathtruncatemacro\j{\i+0.5}
				\node[shape=rectangle, fill=red, minimum size=4pt] (\j8) at (\i,7.5) {};
				\node[shape=rectangle, fill=red, minimum size=4pt] (\j-8) at (\i,-7.5) {};
			}
			
			\foreach \i in {2,4,6,8} %4-cycles edges
			{
				\pgfmathtruncatemacro\j{\i+1}
				\draw (\i7)--(\i8)--(\i9)--(\j8)--(\i7);
				\draw (\i-7)--(\i-8)--(\i-9)--(\j-8)--(\i-7);
			}
			
			\foreach \i in {-5.5,-3.5,-1.5}
			{
				%These two lines are required since tikz cannot parse floating point numbers as variable names
				\pgfmathtruncatemacro\j{\i-0.5}; 
				\pgfmathtruncatemacro\k{\i+0.5};
				\pgfmathtruncatemacro\l{\k-1}; 
				%These two lines are required since tikz cannot parse floating point numbers as variable names
				\foreach \x in {3,5,7}
				{
					\pgfmathtruncatemacro\y{\x-1}; 
					\pgfmathtruncatemacro\z{\j+1};
					\node[shape=diamond,fill=orange!90!black] (\x\j) at (\x,\i) {};
					\draw[very thick, color=orange!90!black] (\x\j)--(\y\j);
					\draw[very thick, color=orange!90!black] (\x\j)--(\y\z);	
				}
				\node[shape=diamond,fill=orange!90!black] (9\k) at (9,\k) {};
				\draw[very thick, color=orange!90!black] (9\k)--(8\k);
				\draw[very thick, color=orange!90!black] (9\k)--(8\k);
				\draw[very thick, color=orange!90!black] (9\k)--(8\l);
				
			}

			\node[shape=diamond, fill=orange!90!black] (52) at (5, 1.5) {};
			\draw[very thick, color=orange!90!black] (52)--(41);
			\draw[very thick, color=orange!90!black] (52)--(42);
			
			\node[shape=diamond,fill=orange!90!black] (74) at (7, 3.5) {};
			\draw[very thick, color=orange!90!black] (74)--(63);
			\draw[very thick, color=orange!90!black] (74)--(64);
			
			\node[shape=diamond,fill=orange!90!black] (75) at (7,5.5) {};
			\draw[very thick, color=orange!90!black] (75)--(85);
			\draw[very thick, color=orange!90!black] (75)--(86);
			
		\end{tikzpicture}
		\label{fig:withFlags}
	}
	\subfloat[]{
		\centering
		\begin{tikzpicture}[xscale=0.45, yscale=0.3,
			triangle/.style = {regular polygon, regular polygon sides=3},
			ltriangle/.style = {regular polygon, regular polygon sides=3, rotate=90},
			rtriangle/.style = {regular polygon, regular polygon sides=3, rotate=270},
			dtriangle/.style = {regular polygon, regular polygon sides=3, rotate=180}]
			
			\node[draw=none] at (-0.7,-7) {$P_L$};
			\node[draw=none] at (10.7,-7) {$P_R$};
			\node[draw=none] at (-0.7,0) {$P_\alpha$};
			\node[draw=none] at (2.7,-6.5) {$P_q$};
			\node[draw=none] at (4.7,-6.5) {$P_r$};
			\node[draw=none] at (6.7,-6.5) {$P_s$};
			\node[draw=none] at (8.7,-6.5) {$P_t$};

			\tikzstyle{every node}=[draw, shape=circle, minimum size=5pt,inner sep=0pt];

			\foreach \i in {0,2,4,6,8,10} \node[shape=rectangle, fill=green!80!black, minimum size=4pt] (\i0) at (\i,0) {};
			
			\foreach \i in {1,3,5,7,9}	\node[shape=rectangle, fill=yellow!90!black, minimum size=4pt] (\i0) at (\i,0) {};
			
			\foreach \i in {0,...,9}
			{
				\pgfmathtruncatemacro\j{\i+1};
				\draw[very thick, color=yellow!90!black] (\i0)--(\j0);
			}
			
			\foreach \i in {2,4,6,8} %vertices of literal lines
			{
				\foreach \j in {-7,-6,-5,-4,-3,-2,-1,1,2,3,4,5,6,7}
				{
					\node[fill=blue!90!black] (\i\j) at (\i,\j) {};
					
				}
			}
			
			\foreach \i in {2,4,6,8} %edges of literal lines
			{
				\foreach \j in {-7,-6,-5,-4,-3,-2,-1,0,1,2,3,4,5,6}
				{
					\pgfmathtruncatemacro\k{\j+1};
					\draw[very thick, color=blue] (\i\j)--(\i\k);
					
				}
			}
			
			\foreach \i in {0,10} %vertices on L and R
			{
				\foreach \j in {-7,-6,-5,-4,-3,-2,-1,1,2,3,4,5,6,7}
				{
					\node[shape=rectangle,fill=blue!80!black, minimum size=4pt] (\i\j) at (\i,\j) {};
					
				}
			}
			
			\foreach \i in {0,10} %edges on L and R
			{
				\foreach \j in {-7,-6,-5,-4,-3,-2,-1,0,1,2,3,4,5,6}
				{
					\pgfmathtruncatemacro\k{\j+1};
					\draw[very thick, color=blue!80!black] (\i\j)--(\i\k);
					
				}
			}

			\foreach \i in {-5.5,-3.5,-1.5,1.5,3.5,5.5}
			{
				%These two lines are required since tikz cannot parse floating point numbers as variable names
				\pgfmathtruncatemacro\j{\i-0.5}; 
				\pgfmathtruncatemacro\k{\i+0.5};
				%These two lines are required since tikz cannot parse floating point numbers as variable names

				\node[ltriangle, fill=red] (-1\j) at (-0.7,\i) {};		
				\node[rtriangle, fill=red] (1\j) at (0.7,\i) {};
				
				\node[ltriangle, fill=red] (9\j) at (9.3,\i) {};
				\node[rtriangle, fill=red] (11\j) at (10.7,\i) {};
				\draw (1\j)--(0\j);
				\draw (1\j)--(0\k);
				\draw (-1\j)--(0\j);
				\draw (-1\j)--(0\k);
				\draw (9\j)--(10\j);
				\draw (9\j)--(10\k);
				\draw (11\j)--(10\j);
				\draw (11\j)--(10\k);
			}

			\foreach \i in {2,4,6,8} %4-cycles top
			{
				\node[triangle, fill=red] (\i9) at (\i,8.3) {};
				\node[dtriangle, fill=red] (\i-9) at (\i,-8.3) {};
			}
			
			\foreach \i in {1.5,2.5,3.5,4.5,5.5,6.5,7.5,8.5} %4-cycles middle
			{
				\pgfmathtruncatemacro\j{\i+0.5}
				\node[shape=rectangle, fill=red, minimum size=4pt] (\j8) at (\i,7.5) {};
				\node[shape=rectangle, fill=red, minimum size=4pt] (\j-8) at (\i,-7.5) {};
			}
			
			\foreach \i in {2,4,6,8} %4-cycles edges
			{
				\pgfmathtruncatemacro\j{\i+1}
				\draw (\i7)--(\i8)--(\i9)--(\j8)--(\i7);
				\draw (\i-7)--(\i-8)--(\i-9)--(\j-8)--(\i-7);
			}
			
			\foreach \i in {-5.5,-3.5,-1.5}
			{
				%These two lines are required since tikz cannot parse floating point numbers as variable names
				\pgfmathtruncatemacro\j{\i-0.5}; 
				\pgfmathtruncatemacro\k{\i+0.5};
				\pgfmathtruncatemacro\l{\k-1}; 
				%These two lines are required since tikz cannot parse floating point numbers as variable names
				\foreach \x in {3}
				{
					\pgfmathtruncatemacro\y{\x-1}; 
					\pgfmathtruncatemacro\z{\j+1};
					\node[shape=diamond,fill=orange!90!black] (\x\j) at (\x,\i) {};
					\draw[very thick, color=orange!90!black] (\x\j)--(\y\j);
					\draw[very thick, color=orange!90!black] (\x\j)--(\y\z);	
				}
				\node[shape=diamond,fill=orange!90!black] (9\k) at (7,\i) {};
				\draw[very thick, color=orange!90!black] (9\k)--(8\k);
				\draw[very thick, color=orange!90!black] (9\k)--(8\k);
				\draw[very thick, color=orange!90!black] (9\k)--(8\l);
				
			}
			
			\foreach \i in {5.5,3.5,1.5}
			{
				%These two lines are required since tikz cannot parse floating point numbers as variable names
				\pgfmathtruncatemacro\j{\i-0.5}; 
				\pgfmathtruncatemacro\k{\i+0.5};
				\pgfmathtruncatemacro\l{\k-1}; 
				%These two lines are required since tikz cannot parse floating point numbers as variable names
				\foreach \x in {3,5}
				{
					\pgfmathtruncatemacro\y{\x+1}; 
					\pgfmathtruncatemacro\z{\j+1};
					\node[shape=diamond,fill=orange!90!black] (\x\j) at (\x,\i) {};
					\draw[very thick, color=orange!90!black] (\x\j)--(\y\j);
					\draw[very thick, color=orange!90!black] (\x\j)--(\y\z);	
				}
				%		\node[shape=diamond,fill=orange!90!black] (9\k) at (7,\i) {};
				%		\draw[very thick, color=orange!90!black] (9\k)--(8\k);
				%		\draw[very thick, color=orange!90!black] (9\k)--(8\k);
				%		\draw[very thick, color=orange!90!black] (9\k)--(8\l);
				
			}

			\node[shape=diamond, fill=orange!90!black] (5-2) at (5, -1.5) {};
			\draw[very thick, color=orange!90!black] (5-2)--(4-1);
			\draw[very thick, color=orange!90!black] (5-2)--(4-2);
			
			\node[shape=diamond,fill=orange!90!black] (7-4) at (5, -3.5) {};
			\draw[very thick, color=orange!90!black] (7-4)--(6-3);
			\draw[very thick, color=orange!90!black] (7-4)--(6-4);
			
			\node[shape=diamond,fill=orange!90!black] (75) at (7,5.5) {};
			\draw[very thick, color=orange!90!black] (75)--(85);
			\draw[very thick, color=orange!90!black] (75)--(86);
			
		\end{tikzpicture}
		\label{fig:truthAssignment}
	}
	
	\subfloat[]{
		\centering
		
		\begin{tikzpicture}[scale=0.3]
			
			\node[draw=none] at (-3.5,0) {$\ell_\alpha$};
			%\node[draw=none] at (-3.5,1.9) {$\beta$};
			\node[draw=none] at (0,-12) {$\ell_L$};
			\node[draw=none] at (19,-12) {$\ell_R$};

			\node[draw=none] at (3.8,-12) {$\ell_q$};
			\node[draw=none] at (7.6,-12) {$\ell_r$};
			\node[draw=none] at (11.4,-12) {$\ell_s$};
			\node[draw=none] at (15.2,-12) {$\ell_t$};
			
			\node[draw=none] at (-3.5,-9) {$\ell_\nabla$};
			\node[draw=none] at (-3.5,-7.05) {$\ell'_C$};
			\node[draw=none] at (-3.5,-4.85) {$\ell'_B$};
			\node[draw=none] at (-3.5,-2.65) {$\ell'_A$};
			\node[draw=none] at (-3.5,2.65) {$\ell_A$};
			\node[draw=none] at (-3.5,4.85) {$\ell_B$};
			\node[draw=none] at (-3.5,7.05) {$\ell_C$};
			\node[draw=none] at (-3.5,9) {$\ell_\Delta$};
			
			\tikzstyle{every node}=[draw, fill=black, shape=circle, minimum size=2pt,inner sep=0pt];

			\draw[thick] (-3,0)--(22,0);
			%\draw (-3,1.9)--(22,1.9);
			
			\foreach \i in {0,1.9,3.8,5.7,7.6,9.5,11.4,13.3,15.2,17.1,19} %disks on \alpha
			{
				\node at (\i,0) {};
				\draw (\i,0) circle[radius=1];			
			}

			\foreach \i in {3.8,7.6,11.4,15.2} %disks on literal lines
			{
				\foreach \j in {-7.6,-6.5,-5.4,-4.3,-3.2,-2.1,-1.05,1.05,2.1,3.3, 4.4,5.5, 6.6,7.7}
				{
					\node at (\i,\j) {};
					\draw (\i,\j) circle[radius=1];
					
				}
			}
			
			\foreach \i in {0,19} %disks on R and L
			{
				\foreach \j in {-7.6,-6.5,-5.4,-4.3,-3.2,-1.9,1.9,3.3, 4.4,5.5, 6.6,7.7}
				{
					\node at (\i,\j) {};
					\draw (\i,\j) circle[radius=1];
					
				}
			}
			
			\foreach \i in {0,3.8,7.6,11.4,15.2,19} %literal lines
			\draw (\i,-11.5)--(\i,11.5);	
			
			\foreach \i in {-7.05,-4.85,-2.65,2.65,4.85,7.05} %fixing flags and clause lines
			{	
				\node at (1.5,\i) {};
				\node at (-1.5,\i) {};
				\node at (17.5,\i) {};
				\node at (20.5,\i) {};		
				\draw (1.5,\i) circle[radius=1];
				\draw (-1.5,\i) circle[radius=1];
				\draw (17.5,\i) circle[radius=1];
				\draw (20.5,\i) circle[radius=1];
				\draw (-3,\i)--(22,\i);
			}
			
			\draw[thick] (-3,9.1)--(22,9.1); %top squeezing line
			\draw[thick] (-3,-9.1)--(22,-9.1); %bottom squeezing line
			
			\foreach \i in {0,2.32, 4.38, 6.43, 8.48, 10.53, 12.58, 14.63, 16.68,19}
			{
				\node at (\i,9.1) {};
				\draw (\i,9.1) circle[radius=1];
				\node at (\i,-9.1) {};
				\draw (\i,-9.1) circle[radius=1];
			}
			
			\foreach \i in {3.8,7.6,11.4,15.2} %top disks of 4-cycles
			{
				\foreach \j in {-9.8,9.8}
				{
					\node at (\i,\j) {};
					\draw (\i,\j) circle[radius=1];
				}
			}
			
			\foreach \i in {2.65,4.85,7.05} %flag vertices on top
			{
				\node at (6,\i) {};
				\draw (6,\i) circle[radius=1];
				\node at (9.8,\i) {};
				\draw (9.8,\i) circle[radius=1];
			}
			
			\foreach \i in {-2.65,-4.85,-7.05} %flag vertices on bottom
			{
				\node at (5.3,\i) {};
				\draw (5.3,\i) circle[radius=1];
				\node at (13.7,\i) {};
				\draw (13.7,\i) circle[radius=1];
			}
			
			\node at (13.7,7.05) {};
			\draw (13.7,7.05) circle[radius=1];
			
			\node at (9.8,-4.85) {};
			\draw (9.8,-4.85) circle[radius=1];
			
			\node at (9,-2.65) {};
			\draw (9,-2.65) circle[radius=1];
		\end{tikzpicture}
		\label{fig:truthAssignmentRealization}
	}
	\caption{(a) The input graph for the Monotone NAE3SAT formula with $\Phi$ with literals $q,r,s,t$, and clauses $A,B,C$.
		The formula is written as $\Phi = A \wedge B \wedge C$  where $A = (q \vee s \vee t)$, $B = (q \vee r \vee t)$, and $C =(q \vee r \vee s)$. The flag vertices, indicated by orange diamonds, are adjacent to the vertices that correspond to a clause on an induced path, if the literal does not appear in that clause.\\
		(b) A truth assignment that satisfies the formula given in \ref{fig:withFlags}: $q = \textsc{true}$, $r = \textsc{false}$, $s = \textsc{false}$, and $t = \textsc{true}$.\\
		(c) Realization of the graph given in \ref{fig:truthAssignment}.}
	\label{fig:fullGadget}
\end{figure}
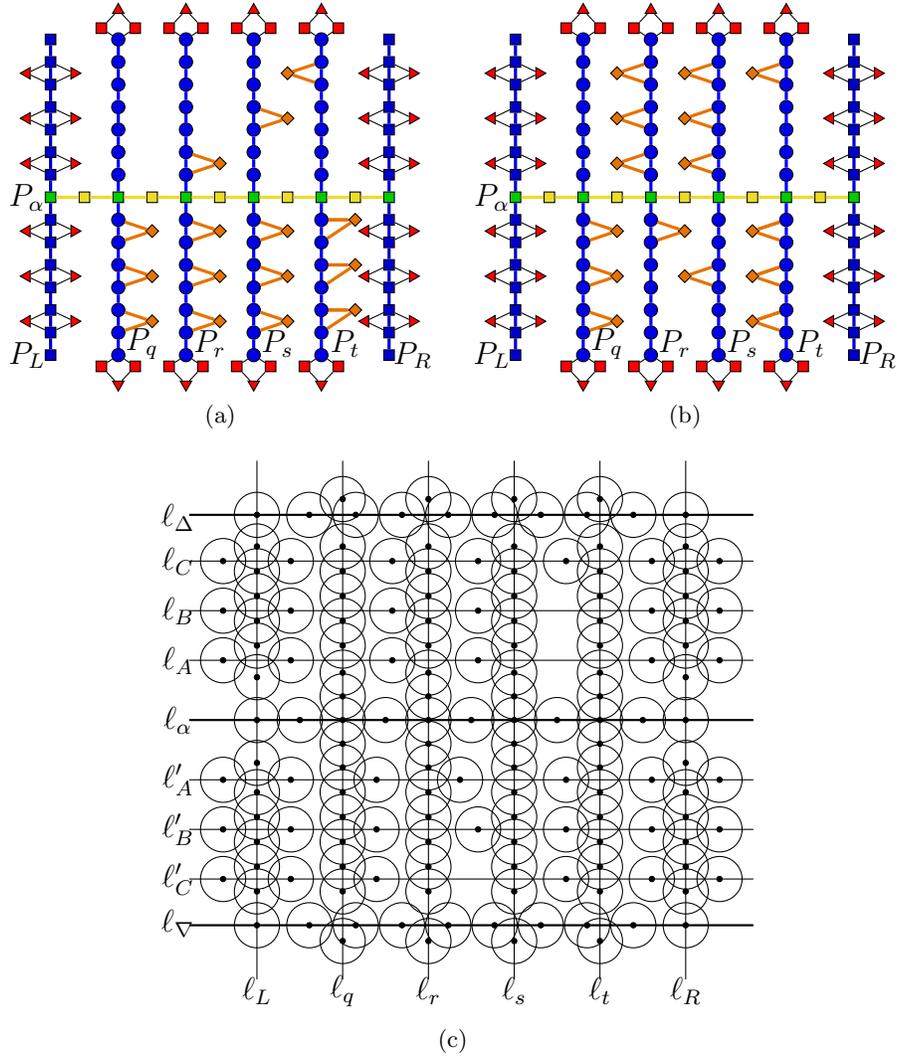

Using the backbone we have described, we now show how to model the relationship between the clauses and the literals.
To make it easier to follow, we also refer to Figure~\ref{fig:skeleton} in parentheses in the following description.
\begin{itemize}
	\item Consider a sub-path $(P^i_{2k+3}, P^i_{2k+4}$, $\dots$, $P^i_{4k+3})$ of the induced path $P^i$. This part corresponds to the literal $x_i$ of the given Monotone NAE3SAT formula (corresponds to $(q_9, \dots, q_{15})$ of $P_q$ in our example).
	
	\item The edges $P^i_{2k+3}P^i_{2k+4}$, $P^i_{2k+5}P^i_{2k+6}$, $\dots$, $P^i_{4k+1}P^i_{4k+2}$ (correspond to $q_9q_{10}$, $q_{11}q_{12}$, and $q_{13}q_{14}$ in $P_q$ in our example) are used to model membership of $x_i$ in the clauses $C_1$, $C_2$, $\dots$, $C_k$ (correspond to the clauses $A$, $B$, and $C$ in our example), respectively.
	
	\item If $x_i$ appears in a clause $C_j$, then we do nothing for the edges correspond do those clauses.
	
	\item Otherwise, if $x_i$ does not appear in $C_j$, then we introduce a \emph{flag vertex} in the graph, which is adjacent to $P^i_{2(k+j)+1}$ and $P^i_{2(k+j)+2}$.
	
	\item Due to the rigidity of the backbone (up to $\varepsilon$ flexibility), this flag vertex lies on $\ell^C_j$.
	Similarly, in our example, if $q$ appears in $B$, then $q_{11}q_{12}$ stays as is, but otherwise, a flag vertex is introduced, adjacent to both $q_{11}$ and $q_{12}$.
	
	\item Every clause has 3 literals. Thus, on each horizontal line, 3 out of $m$ possible flag vertices will be missing.
	That sums up to a total of $k(m-3)$ flag vertices for this part of the graph.
	
	\item For the remaining sub-path $(P^i_1,\dots,P^i_{2k+2})$ of $P^i$ (corresponds to $(q_1,\dots,q_8)$ of $P_q$ in our example), we introduce the flag vertices for the pairs $(P^i_2,P^i_3)$, $(P^i_4,P^i_5)$, $\dots$, $P^i_{2k}$, $P^i_{2k+1}$ (correspond to $(q_6,q_7)$ $(q_2,q_3)$, $(q_4,q_5)$, $(q_6,q_7)$ in our example).
	
	\item That is a total number of $km$ flag vertices for this part of the graph. In the whole graph, there are precisely $2km-3k$ flag vertices.
	
	\item Finally, let us note that the value of $\varepsilon$ does not require more than polynomially many digits with respect to the input size, where $(2+\epsilon)$ is the distance between two consecutive parallel lines.
	We assume that $\varepsilon < 1/4k$ where $k$ is the number of clauses in the given NAE3SAT formula. 
	Since $k$ is a part of the input, multiplying it by $0.25$ will not result a number whose decimal representation requires exponentially many digits.
\end{itemize}

Realize that the embeddings on some vertical lines must be flipped upside-down to create space for the flag vertices.
This operation corresponds to the truth assignment of the literal that corresponds to that vertical line.
The configuration forces at least one literal to have a different truth assignment, because for a pair of symmetrical horizontal lines, say $\ell_A$ and $\ell'_A$, there must be at least one, and at most two missing flags for the disks to fit between $\ell_L$ and $\ell_R$.
Here, we conclude the proof of Theorem~\ref{thm:apud}. 

The input graph, a YES-instance, and the realization of the YES-instance of the Monotone NAE3SAT formula $\Phi = (q \vee s \vee t) \wedge (q \vee r \vee t) \wedge (q \vee r \vee s)$ is given in Figure~\ref{fig:fullGadget}.

Before moving onto the next section, we would like to state a trivial corollary which is a natural result derived from Theorem~\ref{thm:apud}.
To the best of our knowledge, there are no results on unit disk graph recognition problem with the input restricted by the size of the largest cycle.
A slightly related result was studied by Fomin, Lokshtanov and Saurabh in 2012 \cite{bidimensionality}, which considers the cycle-packing problem on disk graphs with no large cliques.

\begin{cor} \label{cor:largestCycle}
	Given a graph $G = (V,E)$, deciding whether $G$ is a unit disk graph is an NP-hard problem when the size of the largest induced cycle in $G$ is of length $4$.
\end{cor}

\begin{proof}
	In the NP-hardness proof given in Section~\ref{sec:nphard}, the largest cycles in the input graph are diamonds, and induced 4-cycles.
	The rest of the graph consists of long paths.
	Since axes-parallel unit disk graph recognition is a more restricted version of unit disk graph recognition problem, the claim holds. 
\end{proof}

\section{APUD(k,0) recognition is NP-complete} \label{sec:npcomplete}
In this section, we show that the recognition of axes-parallel unit disk graphs is NP-complete, when all the given lines are parallel to each other.
This version of the problem is referred to as $\APUD(k,0)$, as there are $k$ horizontal lines given as input, but no vertical lines.
We use the reduction given in Section~\ref{sec:nphard}.

\begin{thm} \label{thm:apud0m}
	$\APUD(k,0)$ recognition is NP-hard.
\end{thm}

\begin{proof}
	Consider the realization given in Figure~\ref{fig:skeletonRealization}.
	Notice that the length of the paths $P^1, P^2, \dots, P^m$ ($P_q, P_r, P_s, P_t$ in our example), and thus the number of disks on vertical lines is equal.
	Lemma~\ref{lem:squeeze} shows that the induced paths can only be realized by unit disks that are centered between $\ell_\nabla$ and $\ell_\Delta$.
	Therefore, for the disks that correspond to the vertices on these paths, we do not need any vertical line.
	Instead, we use triplets of parallel horizontal lines to simulate a vertical line.
	For a horizontal line  $\ell^C_i$ in the line configuration we use in Secion~\ref{sec:nphard}, we add two horizontal lines, namely $\ell^{Ca}_i$ and $\ell^{Cb}_i$, where $\ell^{Ca}_i$ is above $\ell$, and $\ell^{Cb}_i$ is below $\ell$ both of which are symmetrically with respect to $\ell$.
	The Euclidean distance between each $\ell_a$ and $\ell_b$ is $1 + \varepsilon$ where $\varepsilon$ is a sufficiently small value.
	The disks which correspond to the long induced paths $P^1, \dots, P^m$ are placed on the extra two parallel lines, and thus the whole configuration is realized as described in Section~\ref{sec:nphard}.
	
	As a result, for a given instance $\Phi$ of Monotone NAE3SAT formula with $k$ clauses and $m$ literals; we have an instance $\Psi$ of $\APUD(2k+3,m+2)$ to prove NP-hardness with vertical lines, and an instance $\Psi'$ of $\APUD(4k+3,0)$ to prove NP-completeness without vertical lines.
	In $\Psi'$, only disks that can ``jump'' from one horizontal line to another are the ones that are on the top line of $\ell^{Ca}_1$ and  bottom line of $\ell'^{Cb}_1$.
	And those jumps do not change the overall configuration.
\end{proof}

\begin{figure}[htbp]
	\centering
	\begin{tikzpicture}[scale=0.5]
		
		\node[draw=none] at (-3.5,0) {$\alpha$};
		%\node[draw=none] at (-3.5,1.9) {$\beta$};
		%\node[draw=none] at (0,-12) {$L$};
		%\node[draw=none] at (19,-12) {$R$};

		\node[draw=none] at (3.8,-12) {$\ell_q$};
		\node[draw=none] at (7.6,-12) {$\ell_r$};
		\node[draw=none] at (11.4,-12) {$\ell_s$};
		\node[draw=none] at (15.2,-12) {$\ell_t$};
		
		\node[draw=none] at (-3.5,-9) {$\ell_\Delta$};
		\node[draw=none] at (-3.5,-7.05) {$\ell'_C$};
		\node[draw=none] at (-3.5,-4.85) {$\ell'_B$};
		\node[draw=none] at (-3.5,-2.65) {$\ell'_A$};
		\node[draw=none] at (-3.5,2.65) {$\ell_A$};
		\node[draw=none] at (-3.5,4.85) {$\ell_B$};
		\node[draw=none] at (-3.5,7.05) {$\ell_C$};
		\node[draw=none] at (-3.5,9) {$\ell_\nabla$};
		
		\tikzstyle{every node}=[draw, fill=black, shape=circle, minimum size=2pt,inner sep=0pt];

		\draw[thick] (-3,0)--(22,0);
		%\draw (-3,1.9)--(22,1.9);
		
		\foreach \i in {0,1.9,3.8,5.7,7.6,9.5,11.4,13.3,15.2,17.1,19} %disks on \alpha
		{
			\node at (\i,0) {};
			\draw (\i,0) circle[radius=1];			
		}

		\foreach \i in {3.8,7.6,11.4,15.2} %disks on literal lines
		{
			\foreach \j in {-7.6,-6.5,-5.4,-4.3,-3.2,-1.9,1.9,3.3, 4.4,5.5, 6.6,7.7}
			{
				\node at (\i,\j) {};
				\draw (\i,\j) circle[radius=1];		
			}
		}
		
		\foreach \i in {-7.6,-6.5,-5.4,-4.3,-3.2,-1.9,1.9,3.3, 4.4,5.5, 6.6,7.7} %
		\draw[red] (-3,\i)--(22,\i);
		
		\foreach \i in {0,19} %disks on R and L
		{
			\foreach \j in {-7.6,-6.5,-5.4,-4.3,-3.2,-1.9,1.9,3.3, 4.4,5.5, 6.6,7.7}
			{
				\node at (\i,\j) {};
				\draw (\i,\j) circle[radius=1];
				
			}
		}
		
		%\foreach \i in {0,3.8,7.6,11.4,15.2,19} %literal lines
		%\draw (\i,-11.5)--(\i,11.5);	
		
		\foreach \i in {-7.05,-4.85,-2.65,2.65,4.85,7.05} %fixing flags and clause lines
		{	
			\node at (1.5,\i) {};
			\node at (-1.5,\i) {};
			\node at (17.5,\i) {};
			\node at (20.5,\i) {};		
			\draw (1.5,\i) circle[radius=1];
			\draw (-1.5,\i) circle[radius=1];
			\draw (17.5,\i) circle[radius=1];
			\draw (20.5,\i) circle[radius=1];
			\draw (-3,\i)--(22,\i);
		}
		
		\draw[thick] (-3,9)--(22,9); %top squeezing line
		\draw[thick] (-3,-9)--(22,-9); %bottom squeezing line
		
		\foreach \i in {0, 2.32, 4.38, 6.43, 8.48, 10.53, 12.58, 14.63, 16.68, 19}
		{
			\node at (\i,9) {};
			\draw (\i,9) circle[radius=1];
			\node at (\i,-9) {};
			\draw (\i,-9) circle[radius=1];
		}

		\foreach \i in {3.8,7.6,11.4,15.2} %top disks of 4-cycles
		{
			\foreach \j in {-9.8,9.8}
			{
				\node at (\i,\j) {};
				\draw (\i,\j) circle[radius=1];
			}
		}
		
	\end{tikzpicture}
	\caption{Realization of the graph given in Figure~\ref{fig:skeleton} on a set of parallel lines.}
	\label{fig:allHorizontal}
\end{figure}

To show that $\APUD(k,0)$ recognition is in NP, we need to prove that a given solution can be verified in polynomial time as well as any feasible input will have a solution that takes up polynomial space, with respect to the input size.
Thus, we show that for any graph $G \in \APUD(0,k)$, there exists an embedding where the disk centers are represented using polynomially many decimals with respect to the input size.

For the following lemmas, define the set $\mathcal{H} = \{0, h_1, h_2, \dots, h_k\}$ where each element of the set corresponds to the Euclidean distance between a horizontal line, and the $x$-axis.
Without loss of generality, we assume that the bottom-most line is $x$-axis itself, and $h_i < h_j$ iff $i<j$.

\begin{dfn}[Extension disk]
	Let $A$ and $B$ a pair of intersecting unit disks centered at $(a,h_i)$ and $(b,h_j)$, respectively.
	The \emph{extension disk} $A_\EX$ of $A$ is a disk centered at $(a,h_i)$, with radius $2$.
	Then, $(b,h_j)$ is contained inside $A_\EX$, and symmetrically, $(a,h_i)$ is contained inside the extension disk $B_\EX$ of $B$.
\end{dfn}

We give the trivial definition of extension disk for the sake of simplicity in the following lemmas.
Essentially, we claim that the center of a disk has some freedom of movement, and this freedom is determined by precisely two disks.
To show this dependence, we utilize the intervals on a line, defined by the intersection points between the extension disks with that line.

\begin{lem}
	Consider a disk $A$ centered at $(a,h_i)$. 
	Each pair of extension disks, $B_\EX$ and $C_\EX$ that intersect $y = h_i$ line determines an interval $I_A$, in which the center of $A$ can move without changing the relationship among $A$, $B$, and $C$.
\end{lem}

\begin{proof}
	Let $\beta$ and $\gamma$ be two positive numbers, $B$ and $C$ be a pair of disks, and $B_\EX$ and $C_\EX$ be the extension disks, respectively.
	Without loss of generality, let $(a-\beta, h_i)$ be the intersection point between $B_\EX$ and $y=h_i$ which is the closest intersection point to the center of $A$.
	Similarly, let $(a+\gamma, h_i)$ be the intersection point between $C_\EX$ and $y=h_i$ which is the closest intersection point to the center of $A$.
	
	If $A$ intersects both $B$ and $C$, then $I_A = [a-\beta, a+\gamma]$.
	If $A$ intersects $B$ but not $C$, then $I_A = [a-\beta, a+\gamma)$.
	If $A$ intersects $C$ but not $B$, then $I_A = (a-\beta, a+\gamma]$.
	If $A$ intersects neither $B$ nor $C$, then $I_A = (a-\beta, a+\gamma)$. 
\end{proof}

\begin{lem} \label{lem:interval}
	Consider a disk $A$ centered at $(a,h_i)$, and assume that the Euclidean distance between any pair of parallel lines is different than $2$.
	Let $I_A$ be an interval on $y = h_i$ line, in which the center of $A$ can move freely without changing the relationship between $A$ and any other disk.
	Then, $I_A$ is determined by at most two extension disks that intersect $y = h_i$ line.
\end{lem}

\begin{proof}
	Let $p_1, p_2, \dots$ be the $x$-coordinates of the intersection points between $y=h_i$ line and extension disks that intersect $y = h_i$ line, such that $p_j < p_{j+1}$.
	Suppose that $a \neq p_j$ for any $j$.
	Then, there exists two intersection points whose $x$-coordinates are $p_j, p_{j+1}$ such that $p_j < a < p_{j+1}$ holds.
	The interval $A_I$ is determined precisely by the extension disks whose intersection points are $(p_j,h_i)$ and $(p_{j+1}, h_i)$.
	
	Observe that if $a = p_j$ holds for some $j$, this is the same case where $A$ intersects with the corresponding disk, say $B$.
	Depending on the the center of $B$ being to the left or to the right of the center of $A$, the interval can still be defined by either the pair $p_{j-1}, p_j$ or $p_j,p_{j+1}$. 
\end{proof}

Now, let us show that the intervals are large enough if the disk centers have coordinates that are represented using polynomially many decimals.
We denote this interval by $I_A$ for the center of a disk $A$.
Let us refer to a number in the output as \emph{good number} if its decimal representation has polynomially many decimals with respect to the input size, and refer to a number as \emph{bad number}, otherwise.

\begin{lem} \label{lem:intersection}
	Let $A$ be a disk centered at $(a,h_i)$, and let $B$ and $C$ be two disks that defines $I_A$.
	If the centers of $B$ and $C$ are good numbers, then there exists a point $p$ in $I_A$, such that $p$ is also a good number.
\end{lem}

\begin{proof}
	Let $(b,h_j)$ and $(c,h_m)$ be the centers of $B$ and $C$, respectively.
	Define $\delta_{ij} = |h_i - h_j|$, and define $\delta_{im} = |h_i - h_m|$.
	The intersection points on the line $y = h_i$ are $(\sqrt{4-\delta_{ij}^2} - b, h_i)$ and $(\sqrt{4-\delta_{im}^2} - c, h_i)$.
	And thus by Pythagorean theorem we have
	
	\begin{align*}
		&\left|(\sqrt{4-\delta_{ij}^2} - b) - (\sqrt{4-\delta_{im}^2} - c)\right|\\
		=& \left|(b+c) + (\sqrt{4-\delta_{ij}^2} - \sqrt{4-\delta_{im}^2})\right|
	\end{align*}
	as the length of the interval $I_A$.
	
	Since $b$ and $c$ are both good numbers, $b+c$ is also a good number.
	So, even though the expression $(\sqrt{4-\delta_{ij}^2}- \sqrt{4-\delta_{im}^2})$ is a very small, it does not change the fact that there exists a number $p$ inside the interval $A_I$ that is represented with polynomially many decimals.
	Therefore, $p$ is a good number. 
\end{proof}

\begin{lem} \label{lem:representation}
	The coordinates for the disk centers can be chosen in the form of $b + \Sigma_{i} \pm \sqrt{a_i}$, where $a_i = \sqrt{4-\delta_i^2}$ and $\delta_i$ is the distance between two of the given parallel lines, and $b$ is a good number.
\end{lem}

\begin{proof}
	Suppose that we shift all the disk centers as far left to right as possible.
	This defines a partial order $D_1 \prec D_2 \prec \dots D_n$ of dependence among the disks.
	Now, let us fix one disk, and then respect this partial order among other disks.
	
	For each disk $D_{i}$ one of the following holds:
	\begin{enumerate}[(i)]
		\item The Euclidean distance between $D_i$ and $D_j$ is exactly $2$ for some $j < i$.
		\item There are two other disks, $D_j$ and $D_m$ for some $j,m <i$ that defines the interval for the center of $D_i$ (by Lemma~\ref{lem:interval}).
	\end{enumerate}
	
	In case (i) holds, then the center of $D_i$ has the exactly same $x$-coordinate with the center of $D_j$.
	In case (ii) holds, then $a_i = \sqrt{4-\delta_i^2}$, and by Lemma~\ref{lem:intersection}, the center of $D_i$ can have an $x$-coordinate which is a good number. 
\end{proof}

Now, we finally can show that $\APUD(k,0)$ is in NP by utilizing the previous lemmas.

\begin{lem} \label{lem:decimals}
	For every graph $G$ and a set of parallel lines $\mathcal{L}$, if $G$ can be realized on lines from $\mathcal{L}$ as unit disks, then there exists such a realization using polynomial number of decimals with respect to the input size. 
\end{lem}

\begin{proof}
	Since the input consists of only horizontal lines, $y$-coordinates of every disk center is fixed.
	It remains to prove that $x$-coordinates of the disks in a solution cannot be forced to be in very small intervals.
	That is, if a solution consists of disks whose centers are bad numbers, then we can perturb the solution to a new solution contains only good numbers.
	The algorithm described below guarantees existence of an embedding of $G$ on parallel lines as unit disks, whose centers are good numbers.
	
	\begin{enumerate}
		\item Let $G_0, G_1, \dots, G_k$ denote the disjoint induced subgraphs of $G$, such that the vertices of $G_i$ correspond to the disks centered on line $y = h_i$.
		\item Embed $G_0$ on $x$-axis with small perturbations which results as all disks on $y=0$ having polynomially many decimals.
		\item For each $1 \leq i \leq k$; find an embedding of $G_i$ onto $y=h_i$ line by only considering neighbors from $\bigcup_{j<i}G_j$.
	\end{enumerate}
	
	Let us now show that the algorithm is correct.
	
	Observe that each $G_i$ is a unit interval graph, and by \cite{unitIntervalGraphs}, we know that every unit interval graph has an embedding with polynomially many decimals.
	Thus, we can embed $G_0$ on $x$-axis with small perturbations which results as all disks on $y=0$ having good number as their centers.
	
	Now, consider $G_0 \cup G_1$.
	If $h_1 > 2$, then $G_1$ is a standalone unit interval graph, without any relationship with $G_0$.
	If, on the other hand, $h_1 \leq 2$, then by Lemma~\ref{lem:intersection} and Lemma~\ref{lem:interval}, we know that every vertex of $G_1$ can be embedded on $h_1$ as unit disks, of whose centers are good numbers.
	
	For each $2 \leq i \leq k$, the algorithm processes every $G_i$ with respect to the previous embeddings, and thus, by Lemma~\ref{lem:interval}, $G_i$ has an embedding on $y = h_i$ as unit disks, whose centers are good numbers.
	By Lemma~\ref{lem:representation}, we know that the perturbations can be done in a way such that every disk keeps its relationships with the other disks.
	Note that the number of lines is bounded by the number of disks, as empty lines do not have any effect on the embedding.
	Thus, even if we have the coordinates on the line $y = h_k$ as nested square roots, the resulting coordinates will have $a(1 + 2 + \dots + 2^{k})$ bits where $a$ is a good number.
	Thus, the coordinates can be represented by polynomially many decimals, and hence are good numbers.
	
	Therefore, if $G$ has an embedding on parallel lines, then we can obtain an embedding that can be represented by polynomially many decimals with respect to the input by starting from $G_0$, and gradually moving up to $G_k$, processing each subgraph iteratively. 
\end{proof}

\begin{cor} \label{col:polytime}
	Every yes-instance of $\APUD(k,0)$ has a polynomial witness that can be verified in polynomial time
\end{cor}
\begin{proof}
	By Lemma~\ref{lem:decimals}, we know that for every graph that has a feasible embedding on $k$ parallel lines, there exists an embedding that takes up polynomial space in decimal representation with respect to the input size.
	Then, it remains to verify whether the solution gives the intersection graph, and the disks are centered on given lines, which takes only polynomial number of usual arithmetic operations. 
\end{proof}

\begin{thm} \label{thm:completeness}
	$\APUD(k,0)$ is NP-complete.
\end{thm}

\begin{proof}
	Directly follows by Theorem~\ref{thm:apud0m}, Lemma~\ref{lem:decimals}, and Corollary~\ref{col:polytime}. 
\end{proof}

At this point, we would like to pose the question whether the technique used in this section can also be used for the problem of recognizing $\APUD(k,m)$ to determine NP membership.

\begin{conj}
	$\APUD(k,m)$ recognition is NP-complete.
\end{conj}

\section{APUD(1,1) recognition is open} \label{sec:apud11}
In this section, we discuss a natural basis for $\APUD(k,m)$ recognition problem.
That is, $k = 1$, and $m = 1$.
For the sake of simplicity, we can assume that given two lines are the $x$-axis and the $y$-axis.
First, we give some forbidden induced subgraphs for $\APUD(1,1)$.
Namely, those subgraphs are 5-cycle ($C_5$), a 4-sun ($S_4$), and a 5-star ($K_{1,5}$).

\begin{figure}
	\subfloat[$C_5$]{
		\begin{tikzpicture}[scale=0.4]
			\tikzstyle{every node}=[draw, fill=black, shape=circle, minimum size=4pt,inner sep=0pt];
			\node (A) at (0,4) {};
			\node (B) at (2,2) {};
			\node (C) at (1,0) {};
			\node (D) at (-1,0) {};
			\node (E) at (-2,2) {};
			\draw (A)--(B)--(C)--(D)--(E)--(A);
		\end{tikzpicture}
	}\hfill
	\subfloat[$S_4$]{
		\begin{tikzpicture}[scale=0.4]
			\tikzstyle{every node}=[draw, fill=black, shape=circle, minimum size=4pt,inner sep=0pt];
			\node (A) at (0,4) {};
			\node (B) at (1,3) {};
			\node (C) at (2,2) {};
			\node (D) at (1,1) {};
			\node (E) at (0,0) {};
			\node (F) at (-1,1) {};
			\node (G) at (-2,2) {};
			\node (H) at (-1,3) {};
			\draw (A)--(B)--(C)--(D)--(E)--(F)--(G)--(H)--(A);
			\draw (B)--(D)--(F)--(H)--(B);
			\draw (B)--(F);
			\draw (D)--(H);
		\end{tikzpicture}
	}\hfill
	\subfloat[$K_{1,5}$]{
		\begin{tikzpicture}[scale=0.4]
			\tikzstyle{every node}=[draw, fill=black, shape=circle, minimum size=4pt,inner sep=0pt];
			\node (A) at (0,2) {};
			\node (B) at (0,4) {};
			\node (C) at (2,2) {};
			\node (D) at (1,0) {};
			\node (E) at (-1,0) {};
			\node (F) at (-2,2) {};
			\draw (A)--(B);
			\draw (A)--(C);
			\draw (A)--(D);
			\draw (A)--(E);
			\draw (A)--(F);
			
		\end{tikzpicture}
	}
	\caption{Three induced subgraphs that cannot be embedded on $\APUD(1,1)$.}
\end{figure}
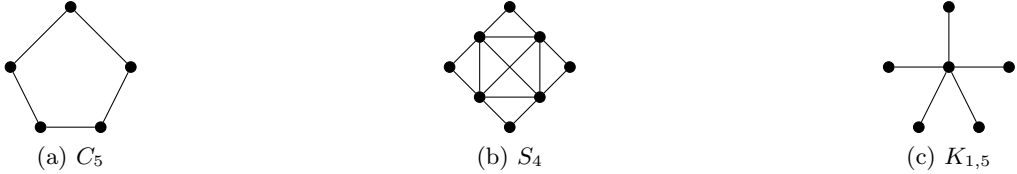
\begin{lem} \label{lem:symmetry}
	Consider two disks $A$ and $B$, centered on $(a,0)$ and $(b,0)$ with $a<0<b$.
	If $|a| = |b|$, then another disk, $C$ that is centered on $(0,c)$ intersects either both, or none.
\end{lem}

\begin{proof}
	Consider the triangle whose corners are $(a,0)$, $(b,0)$, and $(0,c)$.
	If $|a|=|b|$, then $\sqrt{a^2+c^2} = \sqrt{b^2+c^2}$ holds.
	For $C$ to intersect $A$, $\sqrt{a^2+c^2} \leq 2$ must hold.
	However, since $|a| = |b|$, if $\sqrt{a^2+c^2} \leq 2$ holds, then $\sqrt{b^2+c^2} \leq 2$ also holds.
	Thus, $C$ intersects $A$ if, and only if $C$ intersects $B$.
	For the same reason, if $\sqrt{a^2+c^2} > 2$ holds, then $\sqrt{b^2+c^2} > 2$ also holds.
	Thus, $C$ does not intersect $A$ if, and only if $C$ does not intersect $B$. 
\end{proof}

\begin{lem} \label{lem:c5s4k15}
	$C_5, S_4, K_{1,5} \not\in\APUD(1,1)$.
\end{lem}

\begin{proof}
	Let $A,B,C,D,E$ be the disks that form an induced 5-cycle on two perpendicular lines,
	with the intersections between the pairs $(A,B)$, $(B,C)$, $(C,D)$, $(D,E)$, $(E,A)$.
	By the pigeon hole principle, three of these disks must be on the same line.
	Without loss of generality, let $A$, $C$, and $D$ be on $y=0$ line with centers $(a,0)$ and $(c,0)$, and $(d,0)$, respectively.
	Up to symmetry, assume that $a < 0 < c < d$.
	The remaining two disks, $B$ and $E$ must be centered on $x=0$.
	Without loss of generality, assume that $e < 0 < b$.
	However, $E$ cannot intersect $D$ without intersecting $C$ by Lemma~\ref{lem:triangle}.
	Therefore, a 5-cycle cannot be realized as unit disks on two perpendicular lines.
	
	Let $A,B,C,D$ be the disks that form the central clique of an induced 4-sun on two perpendicular lines with.
	By Lemma \ref{lem:triangle}, two of these disks must be on one line, and the remaining two must be on the other line.
	Without loss of generality, assume that $A$ and $C$ are on $y=0$ line, and $B$ and $D$ are on $x=0$ line.
	Denote the centers of $A,B,C,D$ with $(a,0),(0,b),(c,0),(0,d)$, respectively, and assume that $a<0<c$ and $d<0<b$.
	Let $X$ and $Y$ be two disks centered on two given perpendicular lines.
	Assume that $X$ intersects $A,D$, and, $Y$ intersects $B,C$.
	Clearly, $X$ and $Y$ should be on the same line, and on the different sides of the clique to avoid intersections with other rays.
	By Lemma~\ref{lem:symmetry}, if $|b|=|d|$, then $X$ cannot intersect $B$ or $D$ independently.
	By Lemma~\ref{lem:triangle}, if $|b|<|d|$, then $X$ cannot intersect $D$ without intersecting $B$. 
	Similarly, if $|b| > |d|$, then $Y$ cannot intersect $B$ without intersecting $D$.
	Therefore, a 4-sun cannot be realized as unit disks on two perpendicular lines.
	
	Four rays $a,b,c,d$ of a 4-star $u;a,b,c,d$ must be on four different sides of the central vertex $u$.
	To complete a 5-star, there must be one more ray, say $e$ centered on $(e,0)$.
	If we can embed $e$ on one of the lines without intersecting any rays, then we can place another disk on $(-e,0)$ to form a $K_{1,6}$.
	However, a $K_{1,6}$ cannot be realized as a unit disk graph.
	Therefore, a 5-star cannot be realized as unit disks on two perpendicular lines. 
\end{proof}

\begin{lem} \label{lem:bending}
	A given graph $G$ can be embedded on $x$-axis and $y$-axis as a unit disk intersection graph, without using negative coordinates for the disk centers if, and only if $G$ is a unit interval graph.
\end{lem}

\begin{proof}
	In this proof, let us denote the class of graphs that can be embedded on $x$- and $y$- axes as unit disks, using positive coordinates only by ${(xy)}^+$.
	We show that the disks on the $y$-axis can be rotated by $\pi/2$ degrees counterclockwise, and the intersection relationships can be preserved as given in $G$. 
	
	$\mathbf{G \in {(xy)}^+ \Rightarrow G \in \mathbf{UIG}:}$
	Consider two disks, $A$ and $B$, whose centers are $(a,0)$ and $(0,b)$, respectively, where $a$ and $b$ are both positive numbers.
	If $A$ and $B$ do not intersect, then $\sqrt{a^2 + b^2} > 2$.
	After the rotation, the center of $B$ will be on $(-b,0)$.
	The new distance between the centers is $a + b$.
	Since $(a+b)^2 > a^2 + b^2 > 4$, the inequality $a + b > 2$ holds.
	
	If $A$ and $B$ intersect, then $\sqrt{a^2 + b^2} \leq 2$.
	After the rotation, it might still be the case that $a + b > 2$.
	However, now we can safely move the center of $B$ and the other centers that have negative coordinates closer to the center of $A$, recovering the intersection.
	Note that if a disk $C$ is centered in between the centers of $A$ and $B$ after the rotation, then both $A$ and $B$ must intersect $C$ by Lemma~\ref{lem:triangle}.
	
	$\mathbf{G \in \mathbf{UIG} \Rightarrow G \in {(xy)}^+:}$
	Since $G$ is a unit interval graph, we can assume that every interval is a unit disk, and the graph is embedded on $x$-axis.
	Consider two disks, $A$ and $B$, whose centers are $(-a,0)$ and $(b,0)$, respectively, where $a$ and $b$ are both positive numbers.
	If $A$ and $B$ are intersecting, then $a + b \leq 2$
	Then, after the rotation, since $a + b \leq 2$ holds, then $\sqrt{a^2 + b^2} \leq 2$ also holds.
	
	If $A$ and $B$ are not intersecting, then $a + b > 2$.
	After the rotation $\sqrt{a^2 + b^2} \leq 2$ might hold, creating an intersection between $A$ and $B$.
	However, we can simply shift the center of $A$ (along with the other centers that are on $y$-axis) far away from the center of $B$, separating $A$ and $B$. 
\end{proof}

Lemma~\ref{lem:bending} shows that if we use only non-negative coordinates, then\\
$\APUD^{+}(1,1) = \APUD^{+}(1,0) = \UIG$.
This also applies if we use only non-positive coordinates.
Thus, a given $\APUD(1,1)$ can always be partitioned into two unit interval graphs.
Considering the embedding, one of these two partitions contains the disks that are centered on the positive sides of $x$- and $y$- axes, and the other partition contains the disks that are centered on the negative sides of $x$- and $y$- axes.

\begin{lem} \label{lem:partition}
	A graph $G \in \APUD(1,1)$ can be vertex-partitioned into four parts, such that any two form a unit interval graph.
\end{lem}

\begin{proof}
	Let $\Sigma(G)$ be an embedding of $G$ onto $x$- and $y$- axes as unit disks.
	Denote the set of unit disks in $\Sigma(G)$ that are centered on the positive side of the $x$-axis, and the positive side of the $y$-axis by $\Sigma^+(G)$.
	Similarly, denote the set of unit disks that are centered on the negative side of the $x$-axis, and the negative side of the $y$-axis $\Sigma^-(G)$.
	By Lemma~\ref{lem:bending}, both $\Sigma^+(G)$ and $\Sigma^-(G)$ yield separate interval graphs.
	The vertices that correspond to the disks in $\Sigma^+(G)$ yield a unit interval graph, as well as the vertices that correspond to the disks in $\Sigma^-(G)$.
	Note that in case there exists a disk which is centered at $(0,0)$, then it can be included in one of the partitions arbitrarily.
	Therefore, we can vertex-partition $G$ into two unit interval graphs. 
\end{proof}

Up to this point, we showed that if a unit disk graph can be embedded onto two orthogonal lines, then it can be partitioned into two interval graphs.
However, this implication obviously does not hold the other way around.
Thus, we now identify some structural properties of $\APUD(1,1)$.

\begin{rem} \label{rem:c4}
	Consider four unit disks $A,B,C,D$, that are embedded onto $x$-axis and $y$-axis.
	If they induce a $4$-cycle, then the centers of those disks will be at $(a,0)$, $(0,b)$, $(-c,0)$, $(0,-d)$, respectively, where $a,b,c,d$ are non-negative numbers.
\end{rem}

For the upcoming lemma, we will utilize the observation given in Remark~\ref{rem:c4}.
The lemma is an important step to describe a characterization of $\APUD(1,1)$.

\begin{lem} \label{lem:twoc4s}
	Consider eight unit disks embedded onto $x$-axis and $y$-axis, around the origin, whose intersection graph contains two induced $4$-cycles.
	Then, this intersection contains at least four 4-cycles, each with a chord, not necessarily as induced subgraphs.
	Moreover, those 4-cycles are formed by pairs of disks on the same direction ($+x$, $+y$, $-x$, $-y$) with respect to the origin.
\end{lem}

\begin{proof}
	Let those $4$-cycles be $(A,B,C,D)$, and $(U,V,W,X)$, in counterclockwise order, precisely as given in Remark~\ref{rem:c4}.
	Consider the disks $A,B,C$ centered at $(a,0)$, $(0,b)$, $(-c,0)$, where $a,b,c$ are non-negative numbers.
	Due to the configuration, $B$ intersects both $A$ and $C$, but $A$ and $C$ do not intersect.
	If $a=0$, then by Lemma~\ref{lem:triangle}, $C$ intersect $A$ if it intersects $B$.
	Thus, $a>0$, and by symmetry, $b,c>0$ holds.
	Since $A$ intersects $B$, and $b>0$, then $a<2$ also holds.
	Again, up to symmetry, $b,c < 2$ holds.
	
	Now consider the disks $A,B,U,V$ centered at $(a,0)$, $(0,b)$, $(u,0)$, $(0,v)$, respectively.
	With the same reasoning, $0<a,b,u,v<2$ holds.
	Thus, the disks that are centered on the same axis intersect.
	That is, $A$ and $U$, and, $B$ and $V$ intersect.
	If $A$ and $V$ do not intersect, then $\sqrt{a^2 + v^2} > 2$.
	Since $U$ and $V$ intersect, $\sqrt{u^2 + v^2} \leq 2$ holds, and $u < a$ should also hold.
	Thus, by Lemma~\ref{lem:triangle}, $B$ and $U$ intersects.
	Symmetrically, if $B$ and $U$ do not intersect, then $A$ and $V$ intersect.
	With the same reasoning, since $B$ and $C$ intersect, $C$ and $V$ also intersect.
	Therefore, these eight disks have four subsets, $\{A,B,U,V\}$, $\{B,C,V,W\}$, $\{C,D,W,X\}$, $\{D,A,X,U\}$, such that in each set, the induced graph is a $4$-cycle with at least one chord, which gives us four 4-cycles, each with a chord. 
\end{proof}

The given lemmas imply that for a connected graph $G \in \APUD(1,1)$, the following hold.
\begin{enumerate}[(i)]
	\item $G$ does not contain either of 4-sun ($S_4$) or 5-star($K_{1,5}$) as an induced subgraph, and the largest induced cycle in $G$ is of length $4$  (by Lemma \ref{lem:c5s4k15}).
	\item $G$ can be vertex partitioned into into four parts, such that any two form a unit interval graph (by Lemma~\ref{lem:partition}).
	\item Given two $4$-cycles, $(a,b,c,d)$ and $(u,v,w,x)$ in $G$, each one of the quadruplets $\{a,b,u,v\}$, $\{b,c,v,w\}$, $\{c,d,w,x\}$, and $\{d,a,x,u\}$ is either a diamond or a $K_4$ (by Lemma~\ref{lem:twoc4s}).
\end{enumerate}

Although this characterization gives a rough idea about the structure of a graph $G \in \APUD(1,1)$, it is not clear whether this is a full characterization and the recognition can be done in polynomial time.
Note that the characterization is a necessary step through recognition, but it is not yet known whether it is sufficient.
Therefore, we state the following conjecture.
\begin{conj}
	Given a graph $G$, it can be determined whether $G \in \APUD(1,1)$ in polynomial time.
\end{conj}

Since the very essence oft this paper is restricting the solution space, we leave an open problem for the readers before finishing this section.
The following question can also be interpreted as instead of having two intersecting lines as the solution space, we have a vertical line, and a horizontal ray which is perpendicular to the line, and the origin of the ray is on the vertical line.

\begin{open}
	Can $\APUD(1,1)$ recognition be solved in polynomial time if the disk centers are not allowed to have negative $x$-coordinates?
\end{open}

\chapter{Conclusion} \label{chap:conclusion}

\section{Summary of the Contributions}
In this thesis, we tackled some combinatorial problems in the field of computational geometry.
We summarized our contribution to the literature in Table~\ref{table:contributions}.

\begin{table}[h!]
	\begin{small}
		\begin{tabular}{|c c p{9cm}|} 
			\hline
			& &\\[-1em]
			\textbf{Paper} & \textbf{Chapter(s)} & \textbf{Summary of the paper}\\ [0.5ex] 
			\hline\hline
			& &\\[-1em]
			\cite{Cagirici_conflictfree}& \ref{chap:restrictedpolygons} and \ref{chap:generalpolygons}& We designed an algorithm for producing a vertex-to-point guarding of funnels that is 
			optimal in the number of guards.
			Also, we designed an algorithm for a vertex-to-point conflict-free chromatic guarding 
			for funnels, which gives only a small additive error
			with respect to the minimum number of colors required.\\ 
			& &\\[-1em]
			\hline
			& &\\[-1em]
			\cite{Cagirici_chromaticpolygon} & \ref{chap:generalpolygons}& We showed that the problem of deciding 5-colourability for visibility graphs of simple polygons, is NP-complete.
			We also showed that the 4-colouring problem can be solved for visibility graphs of simple polygons, in polynomial time,
			whereas for visibility graphs of polygons with holes, it becomes NP-complete.\\
			& &\\[-1em]
			\hline
			& &\\[-1em]
			\cite{udVis} & \ref{chap:udvg}& We introduced the unit disk visibility graphs, showed that the visibility graphs are a proper subclass of the unit disk visibility graphs, and proved that the 3-coloring problem is  NP-complete both for unit disk segment visibility graphs and unit disk visibility graphs of polygons with holes.\\
			& &\\[-1em]
			\hline
			& &\\[-1em]
			\cite{UDembed} & \ref{chap:apud} & Studied a variation of unit disk recognition problem which requires the disk centers to be centered on pre-given straight lines instead of anywhere in the Euclidean plane, and showed that the problem is NP-hard when any pair of lines is  either parallel or perpendicular, and NP-complete when the lines are parallel to each other.\\
			\hline
		\end{tabular}
	\end{small}
	\caption{The contributions of this thesis to the literature.}
	\label{table:contributions}
\end{table}

In Chapter~\ref{chap:restrictedpolygons}, we considered two combinatorial problems, namely, polygon guarding and conflict-free chromatic guarding, and studied their aspects on funnels, and weak visibility polygons.
The main purpose behind studying restricted types of polygons is to describe an inductive scheme to obtain a conflict-free chromatic guarding of a simple polygon. 

We have described an algorithm for producing a vertex-to-point guarding of funnels that is 
optimal in the number of guards.
We have also designed an algorithm for a vertex-to-point conflict-free chromatic guarding 
for funnels, which gives only a small additive error with respect to the minimum number of colors required. 

Moreover, we have given a simple efficient upper bound of $O(\log^2 n)$, for the special case of weak-visibility polygons.
This still leaves room for improvement down to $O(\log n)$, which is
the worst-case scenario already for funnels and which would match
the previous P2P upper bound for simple polygons of \cite{Bartschi-2014}.
We believe that such an improvement is achievable by a more careful
distribution of the colours sets (and re-use of colours)
to the max funnels in Algorithm~\ref{alg:weakv2pcoloringB}.

In Chapter~\ref{chap:generalpolygons}, we dealed with the conflict-free chromatic guarding on the general polygons.
We have utilized the results on the chromatic guarding from Chapter~\ref{chap:restrictedpolygons} in order to obtain a conflict-free chromatic guarding in a given simple polygon with $O(n \log^2 n)$ colors.

In addition, we have studied the proper coloring problem on the general polygons, and on the polygons with holes. 
we have showed that the problem of deciding 5-colorability for visibility graphs of simple polygons, is NP-complete.
We have also showed that the 4-coloring problem can be solved for visibility graphs of simple polygons, in polynomial time,
whereas for visibility graphs of polygons with holes, it becomes NP-complete.
However, it still remains to be explored whether approximation or parameterized algorithms exist for coloring problems of visibility graphs of polygons.

In Chapter~\ref{chap:udvg}, we considered the visibility graphs of various geometric settings with real-world restrictions, i.e., two points can see each other only if they are close enough in the Euclidean plane, and there is no obstacle in-between.
Thus, we have introduced the concept of ``unit disk visibility graphs,'' which models the real-world scenarios more accurately compared to the conventional visibility graphs.

The phrase ``unit disk'' indicates that every point in the Euclidean plane has a range, modeled by unit disks.
Since the main motivation behind in Chapter~\ref{chap:restrictedpolygons} and Chapter~\ref{chap:generalpolygons} are wireless sensors, and wireless sensor networks are usually modeled by unit disk graphs, our results in this chapter can be interpreted as an intersection between the areas of visibility graphs and unit disk graphs. 
We have studied three particular types of unit disk graphs, namely, unit disk visibility graphs of a set of line segments, simple polygons and polygons with holes.
We showed that the unit disk graphs are a proper subclass of unit disk point visibility graphs while they are neither a subclass nor a superclass of these classes.
Moreover, we proved that the visibility graphs are a proper subclass of the unit disk visibility graphs.

We have also studied the problem of proper coloring for the newly introduced graph class, and showed that the 3-coloring problem for unit disk segment visibility graphs, and for the unit disk visibility graphs of polygons with holes is NP-complete.

In Chapter~\ref{chap:restrictedpolygons} and Chapter~\ref{chap:generalpolygons}, we have studied the visibility graphs, and in Chapter~\ref{chap:udvg}, we have studied a proper superclass of visibility graphs -- unit disk visibility graphs in order to model real-world restrictions of wireless sensors more accurately.
In Chapter~\ref{chap:apud}, without steering away from our motivation based on wireless sensors, we study the ``axes-parallel unit disk graphs.'' 
In this newly introduced class, we tackled the problem by restricting the solution space.
We restricted the disks to be centered on pre-given straight lines instead of anywhere in the Euclidean plane.
The result in this chapter revealed three main aspects of the problem. 
First, we showed that if the disks are forced to be centered on straight lines of which any pair is either parallel or perpendicular, then the recognition problem is NP-hard.
Then, we showed that if any pair of the pre-given lines are parallel to each other, then the problem is NP-complete.
Finally, we showed that if there are only two lines that are perpendicular, then the recognition problem is interesting to study, yet does not yield trivial results.

\begin{table}[tbph!]
	\begin{small}
		\begin{tabular}{|c p{5cm} p{6cm}|} 
			\hline
			& &\\[-1em]
			\textbf{Paper} & \textbf{Name of the Paper} & \textbf{Author's contributions}\\ [0.5ex] 
			\hline\hline
			& &\\[-1em]
			\cite{Cagirici_conflictfree}& On conflict-free chromatic guarding of simple polygons& Collaborated with all the co-authors in the design of the algorithm that produces an approximate chromatic guarding for a given funnel (Algorithm~\ref{alg:apxcfreefunnel}), and collaborated with Petr Hlin\v{e}n\'{y} on the final form of the algorithm which generalizes this scheme to the general case (Algorithm~\ref{alg:coloringoverall}).\\ 
			& &\\[-1em]
			\hline
			& &\\[-1em]
			\cite{Cagirici_chromaticpolygon} & On Colourability of Polygon Visibility Graphs& Collaborated with Bodhayan Roy in the design of the algorithm that decides whether a visibility graph of a simple polygon admits 4-coloring (Algorithm~\ref{alg:4coloring})\\
			& &\\[-1em]
			\hline
			& &\\[-1em]
			\cite{udVis} & Unit Disk Visibility Graphs& Worked alone on the construction of the NP-completeness reductions for the 3-coloring problem on unit disk visibility graphs of line segments (Theorem~\ref{thm:main}) and polygons with holes (Theorem~\ref{thm:withholes}).\\
			& &\\[-1em]
			\hline
			& &\\[-1em]
			\cite{UDembed} & On embeddability of unit disk graphs onto straight lines & Single author paper.\\
			\hline
		\end{tabular}
	\end{small}
	\caption{The contributions of the author of this thesis to the associated papers.}
	\label{table:authorsContributions}
\end{table}

In Table~\ref{table:authorsContributions}, we summarize the contribution of the author of this thesis to every paper that his name appears on.
For the sake of simplicity, instead of listing every lemma, theorem, etc. we have listed the core parts in which the author contributed to, and this list should be interpreted as ``contribution to this part and the affiliated propositions of the corresponding part.''

\section{Other Research}

In this section, we mention two other publications which were not detailed in this thesis, but definitely worth mentioning.

\subsection{Clique-Width of Point Configurations}
First of these papers are entitled as \textbf{``Clique-Width of Point Configurations''} \cite{cliqueWidth}, which was published in \emph{Workshop on Graph-Theoretic Concepts in Computer Science} on June 2020, and the extended version was accepted in \emph{Journal of Combinatorial Theory} in June 2021.

In this paper, we have studied the famous structural width parameter clique-width on geometric point configurations.
We showed that the basic properties of this clique-width notion can be applied to a set of points in the plane, and we relate it to the monadic second-order logic of point configurations.
As an application, we have provided several linear FPT time algorithms for geometric point problems which are NP-hard in general.

The paper describes how clique-width can be used to measure ``complexity'' of the combinatorial structure of point configurations in the plane.
Promising applications of the clique-width of annotated point configurations include mainly parameterized algorithms for solving some generally hard problems in discrete geometry.

The author of this thesis has contributed to this paper by reviewing the literature and providing the combinatorial problems mentioned in the paper which demonstrate the significance and relevance of studying clique-width parameter on point sets.

\subsection{On upward straight-line embeddings of oriented paths}
The second paper we would like to mention is entitled as \textbf{``On upward straight-line embeddings of oriented paths''} \cite{upwardStraight}, which was published in the proceedings of VII Spanish Meeting on Computational Geometry.

In this paper, we investigate upward straight-line embeddings of oriented paths. 
Along the lines of similar results in the literature, we find a condition -- related to the number of vertices in between sources and sinks of an oriented path -- which guarantees that an oriented path satisfying the condition on nvertices admits an upward straight-line embedding  into  any $n$-point set in general position. We also show that the following holds for every $\varepsilon >0$.
We show that if $S$ is a set of $n$ points chosen uniformly at random in the unit square, and $P$ is an oriented path on at most $(1/3 - \varepsilon)n$ vertices, then with high probability, $P$ has an upward straight-line embedding into $S$.

The author of this thesis has contributed to this paper by describing an algorithm to obtain an upward straight-line embedding for a convex point set.
Based on this algorithm, the main theorem of the paper, which proves that a random set of points in the unit square admits an upward straight-line embedding with high probability.

\section{Concluding Remarks}

The visibility graphs have become more and more relevant as the recent technological advancements emphasize on the importance of wireless networks.
Studying the computational aspects of some combinatorial problems is needed for engineers to produce more efficient products.

The relationship between a pair of objects can be interpreted and modeled in many different ways.
We picked two major ones for our doctoral studies:
The visibility, which considers the obstacles between two object, and unit disk graphs, which considers the distance between two objects.
Moreover, we bridged the gap between these two models, and introduced a new concept which considers both the obstacles, and the distance.

Considering that the application areas and the scales of wireless networks vary dearly -- from nanosensors which are injected into human body to giant satellites which communicate with us from space, we believe that our contributions to the literature can be considered worthwhile.

While our research is not backed up with real-world experiments,  the underlying theory is, or will be useful for the relevant applications.
Keeping that in mind, we look forward to seeing the technology which utilizes the theoretical results in this paper, creating new challenges to overcome.

%
%\chapter*{List of Publications}
%\subfile{publications/publications}

\bibliographystyle{alpha}
\bibliography{main}
\end{document}